\pdfoutput=1
\documentclass[11pt,fleqn]{article}
\usepackage{xcolor}
\usepackage[english]{babel}
\usepackage[utf8]{inputenc}
\usepackage[T1]{fontenc}
\usepackage{lmodern}
\usepackage{amssymb}
\usepackage{mathtools}
\usepackage[lmargin=1in,rmargin=1in,bottom=1in,top=1in,twoside=False]{geometry}
\usepackage{enumitem}
\usepackage{amsfonts}
\usepackage{amsmath}
\usepackage{amsthm}
\usepackage{complexity}
\usepackage{mathrsfs}
\usepackage{graphics}
\usepackage[pagewise]{lineno}
\usepackage{microtype}
\usepackage[defaultlines=4,all]{nowidow}
\usepackage[many]{tcolorbox}
\usepackage[colorinlistoftodos,prependcaption,textsize=scriptsize]{todonotes}
\usepackage{xspace}
\usepackage{makecell}
\usepackage{pifont}
\usepackage{bm}
\usepackage{scalerel}
\usepackage{tabularx}
\usepackage{afterpage}
\usepackage[hidelinks]{hyperref}
\usepackage[nameinlink]{cleveref}

\definecolor{myPurple}{RGB}{230,230,255}
\definecolor{interesting}{RGB}{140,0,60}
\definecolor{brown}{RGB}{210,180,140}

\usetikzlibrary{calc,arrows, automata, fit, shapes, arrows.meta, positioning, decorations.pathmorphing, patterns, shapes.geometric, shapes.symbols, shapes.arrows, shapes.misc, decorations.shapes,decorations.pathreplacing}
\usepackage{float}

\tcolorboxenvironment{theorem}{
    colframe=cyan, sharp corners=west, 
    colback=cyan!10,
	rightrule=0pt,
	toprule=0pt,
	bottomrule=0pt,
	top=0pt,
	right=0pt,
	bottom=0pt,
 }
\tcolorboxenvironment{corollary}{
    colframe=blue, sharp corners=west, 
    colback=blue!10,
	rightrule=0pt,
	toprule=0pt,
	bottomrule=0pt,
	top=0pt,
	right=0pt,
	bottom=0pt,
 }
\tcolorboxenvironment{lemma}{
    colframe=red, sharp corners=west,
	colback=red!10,
	rightrule=0pt,
	toprule=0pt,
	bottomrule=0pt,
	top=0pt,
	right=0pt,
	bottom=0pt,
 }

 \tcolorboxenvironment{result}{
    colframe=yellow, sharp corners=west, 
    colback=yellow!10,
	rightrule=0pt,
	toprule=0pt,
	bottomrule=0pt,
	top=0pt,
	right=0pt,
	bottom=0pt,
 }
\tcolorboxenvironment{proposition}{
    colframe=green, sharp corners=west, 
    colback=green!10,
	rightrule=0pt,
	toprule=0pt,
	bottomrule=0pt,
	top=0pt,
	right=0pt,
	bottom=0pt,
 }

\renewcommand{\epsilon}{\varepsilon}
\renewcommand{\phi}{\varphi}

\newcommand{\ie}{i.\,e.\xspace}

\newcommand{\Imc}{\ensuremath{\mathcal{I}}\xspace}

\newcommand{\Omc}{\ensuremath{\mathcal{O}}\xspace}

\newcommand{\Qmc}{\ensuremath{\Pi}\xspace}
\newcommand{\Rmc}{\ensuremath{\mathcal{R}}\xspace}

\newcommand{\Cc}{\mathscr{C}}
\newcommand{\Oof}{\Omc}
\newcommand{\Nbb}{\ensuremath{\mathbb{N}}\xspace}
\newcommand{\Zbb}{\ensuremath{\mathbb{Z}}\xspace}

\newcommand{\wone}{\ensuremath{\mathsf{W}[1]}}
\newcommand{\iljk}{\ensuremath{1 \leq i < j \leq \kappa}}
\newcommand{\kctwo}{\ensuremath{\binom{\kappa}{2}}\xspace}

\newcommand{\budget}{b\xspace}
\newcommand{\N}{\mathbb{N}}
\newcommand{\Z}{\mathbb{Z}}

\usepackage{xspace}

\newcommand{\mcc}{\textsc{Multi-Colored Clique}}

\newcommand{\varS}{S}
\newcommand{\conn}{\mathcal{C}}
\newcommand{\cp}{\ensuremath{\mathsf{coNP/poly}}}
\newcommand{\XNLP}{\ensuremath{\mathsf{XNLP}}}

\theoremstyle{remark}
\newtheorem{theorem}{Theorem}[section]
\newtheorem{corollary}{Corollary}[section]
\newtheorem{definition}{Definition}[section]
\newtheorem{lemma}{Lemma}[section]

\newtheorem{observation}{Observation}[section]

\newtheorem{claim}{Claim}
\crefname{corollary}{Corollary}{Corollaries}
\crefname{lemma}{Lemma}{Lemmas}
\crefname{section}{Section}{Sections}

\newenvironment{claimproof}[1][\proofname]{%
	\begin{proof}[#1]%
	}{%
	\end{proof}%
}

\newtheorem*{result*}{}
\newtheorem*{remark*}{Remark}

\DeclarePairedDelimiter\ceil{\lceil}{\rceil}

\begin{document}
\title{Kernelization complexity of solution discovery problems}
\author{Mario Grobler\\University of Bremen, Germany \and Stephanie Maaz\thanks{Funded by a grant from the Natural Sciences and Engineering Research Council of Canada.}\\University of Waterloo, Canada \and Amer E.~Mouawad\\American University of Beirut, Lebanon \and Naomi Nishimura\footnotemark[1]\\University of Waterloo, Canada\and Vijayaragunathan Ramamoorthi\thanks{Funded by the ``Mind, Media, Machines'' high-profile area at the University of Bremen.}\\University of Bremen, Germany\and Sebastian Siebertz\\University of Bremen, Germany}
\date{}
\maketitle

\begin{abstract}
\noindent In the solution discovery variant of a vertex (edge) subset problem $\Qmc$ on graphs, we are given an initial configuration of tokens on the vertices (edges) of an input graph $G$ together with a budget~$\budget$. 
The question is whether we can transform this configuration into a feasible solution of~$\Qmc$ on~$G$ with at most $\budget$ modification steps. 
We consider the token sliding variant of the solution discovery framework, where each modification step consists of sliding a token to an adjacent vertex (edge). 
The framework of solution discovery was recently introduced by Fellows et al. [Fellows et al., ECAI 2023] and for many solution discovery problems the classical as well as the parameterized complexity has been established. 
In this work, we study the kernelization complexity of the solution discovery variants of \textsc{Vertex Cover}, \textsc{Independent Set}, \textsc{Dominating Set}, \textsc{Shortest Path}, \textsc{Matching}, and \textsc{Vertex Cut} with respect to the parameters number of tokens $k$, discovery budget $\budget$, as well as structural parameters such as~pathwidth.
\end{abstract}
\section{Introduction}
In the realm of optimization, traditional approaches revolve around computing optimal solutions to problem instances from scratch. 
However, many practical scenarios can be formulated as the construction of a
feasible solution from an infeasible starting state.
Examples of such scenarios include reactive systems involving human interactions. 
The inherent dynamics of such a system is likely to lead to an infeasible state.
However, computing a solution from scratch may lead to a solution that may
differ arbitrarily from the starting state.
The modifications required to reach such a solution from the starting state may be costly, difficult to implement, or sometimes unacceptable.

Let us examine a specific example to illustrate. 
A set of workers is assigned tasks so that every task is handled by a qualified worker. 
This scenario corresponds to the classical matching problem in bipartite graphs. 
Suppose one of the workers is now no longer available (e.g.\ due to illness); hence, the schedule has to be changed. 
An optimal new matching could be efficiently recomputed from scratch, but it is desirable to find one that is as close to the original one as possible, so that most of the workers keep working on the task that they were initially assigned. 

Such applications can be conveniently modeled using the \emph{solution discovery} framework, which is the central focus of this work. 
In this framework, rather than simply finding a feasible solution to an instance $\Imc$ of a source problem $\Qmc$, we investigate whether it is possible to transform a given infeasible configuration into a feasible one by applying a limited number of transformation steps.
In this work we consider vertex (edge) subset problems $\Pi$ on graphs, where the \emph{configurations} of the problem are sets of vertices (edges). 
These configurations are represented by the placement of tokens on the vertices (edges) of the configuration. 
An atomic \emph{modification step} consists of moving one of the tokens and the question is whether a feasible configuration is reachable after at most $\budget$ modification steps. 
Inspired by the well-established framework of combinatorial reconfiguration~\cite{bousquet2022survey,nishimura2018introduction,van2013complexity}, commonly allowed modification steps are the addition/removal of a single token, the jumping of a token to an arbitrary vertex/edge, or the slide of a token to an adjacent vertex (edge). 

Problems defined in the solution discovery framework are useful and have been appearing in recent literature. 
Fellows et al.~\cite{fellows2023solution} introduced the term \emph{solution discovery}, and along with Grobler et al.~\cite{grobler2023solution} initiated the study of the (parameterized) complexity of solution discovery problems for various \NP-complete source problems including \textsc{Vertex Cover~(VC)}, \textsc{Independent Set~(IS)}, \textsc{Dominating Set (DS)}, and \textsc{Coloring~(Col)} as well as various source problems in $\P$ such as \textsc{Spanning Tree (ST)}, \textsc{Shortest Path (SP)}, \textsc{Matching (Mat)}, and \textsc{Vertex Cut (VCut) / Edge Cut (ECut)}. 

Fellows et al.~\cite{fellows2023solution} and Grobler et al.~\cite{grobler2023solution} provided a full classification of polynomial-time solvability vs.\ \NP-completeness of the above problems in all token movement models (token addition/removal, token jumping, and token sliding). 
For the \NP-complete solution discovery problems, they provided a classification of fixed-parameter tractability vs.\ $\W[1]$-hardness. 
Recall that a \textit{fixed-parameter tractable algorithm} for a problem $\Qmc$ with respect to a parameter $p$ is one that solves~$\Qmc$ in time $f(p) \cdot n^{\Omc(1)}$, where $n$ is the size of the instance and $f$ is a computable function dependent solely on $p$, while $\W[1]$-hardness provides strong evidence that the problem is likely not fixed-parameter tractable (\ie, does not admit a fixed-parameter tractable algorithm)~\cite{downey2013fundamentals}.

A classical result in parameterized complexity theory is that every problem $\Qmc$ that admits a fixed-parameter tractable algorithm necessarily admits a kernelization algorithm as well~\cite{cai1997advice}. 
A \textit{kernelization algorithm} for a problem $\Qmc$ is a polynomial-time preprocessing algorithm that, given an instance~$x$ of the problem $\Qmc$ with parameter $p$, produces a \emph{kernel} -- an equivalent instance~$x'$ of the problem $\Qmc$ with a parameter $p'$, where both the size of~$x'$ and the parameter $p'$ are bounded by a computable function  depending only on $p$~\cite{downey2013fundamentals}.
Typically, kernelization algorithms generated using the techniques of Cai et al.~\cite{cai1997advice} yield kernels of exponential (or even worse) size. 
In contrast, designing problem-specific kernelization algorithms frequently yields more efficiently-sized kernels, often quadratic or even linear with respect to the parameter.
Note that once a decidable problem~$\Qmc$ with parameter $p$ admits a kernelization algorithm, it also admits a fixed-parameter tractable algorithm, as a kernelization algorithm always produces a kernel of size that is simply a function of~$p$.
The fixed-parameter tractable solution discovery algorithms of Fellows et al.~\cite{fellows2023solution} and Grobler et al.~\cite{grobler2023solution} are not based on kernelization algorithms.
 
Unfortunately, it is unlikely that all fixed-parameter tractable problems admit polynomial kernels. 
Bodlaender et al.~\cite{bodlaender2009problems,bodlaender2014kernelization} developed the first framework for proving kernel lower bounds and Fortnow and Santhanam~\cite{fortnow2008infeasibility} showed a connection to the hypothesis $\NP \not\subseteq \cp$. 
Specifically, for several \NP-hard problems, a kernel of polynomial size with respect to a parameter would imply that $\NP \subseteq \cp$, and thus an unlikely  collapse of the polynomial hierarchy to its third level~\cite{yap1983some}. 
Driven by the practical benefits of kernelization algorithms, we explore the size bounds on kernels for most of the above-mentioned solution discovery problems in the token sliding model, particularly those identified as fixed-parameter tractable in the works of Fellows et al.~\cite{fellows2023solution} and Grobler et al.~\cite{grobler2023solution}.

\medskip \noindent
\textbf{Overview of our results.} We focus on the kernelization complexity of  solution discovery in the token sliding model for the following source problems: \textsc{Vertex Cover}, \textsc{Independent Set}, \textsc{Dominating Set}, \textsc{Shortest Path}, \textsc{Matching}, and \textsc{Vertex Cut}. 
For a base problem $\Pi$ we write $\Pi$-\textsc{D} for the discovery version in the token sliding model. 

\Cref{fig:summary-figure} summarizes our results. 
All graph classes and width parameters appearing in this introduction are defined in the preliminaries. 
Fellows et al.~\cite{fellows2023solution} and Grobler et al.~\cite{grobler2023solution} gave fixed-parameter tractable algorithms with respect to the parameter $k$ for \textsc{IS-D} on nowhere dense graphs, for \textsc{VC-D}, \textsc{SP-D}, \textsc{Mat-D}, and \textsc{VCut-D} on general graphs and for \textsc{DS-D} on biclique-free graphs.

We show that \textsc{IS-D}, \textsc{VC-D}, \textsc{DS-D}, and \textsc{Mat-D} parameterized by $k$ admit polynomial size kernels (on the aforementioned classes), while \textsc{VCut-D} does not admit kernels of size polynomial in~$k$. For \textsc{SP-D}, we show that the problem does not admit a kernel of polynomial size  parameterized by $k + \budget$ unless $\NP \subseteq \cp$.

As \NP-hardness provides strong evidence that a problem admits no polynomial-time algorithm, $\W[t]$-hardness (for a positive integer $t$) with respect to a parameter $p$ provides strong evidence that a problem admits no fixed-parameter tractable algorithm with respect to $p$.
Fellows et al.~\cite{fellows2023solution} proved that \textsc{VC-D}, \textsc{IS-D}, and \textsc{DS-D} are $\W[1]$-hard with respect to parameter $\budget$ on $d$-degenerate graphs but provided fixed-parameter tractable algorithms on nowhere dense graphs. 
They also showed that these problems are slicewise polynomial (\XP) with respect to the parameter treewidth and left open the parameterized complexity of these problems with respect to the parameter treewidth alone. 
We show that these problems remain \XNLP-hard (which implies $\W[t]$-hardness for every positive integer~$t$) for parameter pathwidth (even if given a path decomposition realising the pathwidth), 
which is greater than or equal to treewidth, and that they admit no polynomial kernels (even if given a path decomposition realising the pathwidth) 
with respect to the parameter $\budget + \textit{pw}$, where $\textit{pw}$ is the pathwidth of the input graph, unless $\NP \subseteq \cp$. 

Finally, we also consider the parameter feedback vertex set number (\emph{fvs}), which is an upper bound on the treewidth of a graph, but is incomparable to pathwidth. 
We complement the parameterized complexity classification for the results of Fellows et al.~\cite{fellows2023solution} by showing that \textsc{IS-D, VC-D}, and \textsc{DS-D} are $\W[1]$-hard for the parameter \emph{fvs}.

\begin{figure}
\centering
\begin{tikzpicture}
\draw[line width=0.5pt, draw=cyan] (0,-2) -- (8, -2);
\draw[line width=0.5pt]  (8, -2) -- (\linewidth,-2);
\draw[line width=0.5pt] (\linewidth,-2) -- (\linewidth,2);
\draw[line width=0.5pt, draw=red] (\linewidth,2) -- (\linewidth,5);; 
\draw[line width=0.5pt, draw=red] (\linewidth,5) -- (0,5);
\draw[line width=0.5pt, draw=red] (0,5) -- (0,3.75);
\draw[line width=0.5pt] (0,3.75) -- (0,2);
\draw[line width=0.5pt, draw=cyan] (0,2) -- (0,-2);
\fill[pattern=vertical lines, pattern color=red!30] (0,5) -- (0,3.75) .. controls (2,3.75) and (5,2.85) .. (\linewidth,2) -- (\linewidth,5) -- cycle;
\fill[pattern=grid, pattern color=cyan] (0,2) .. controls (3,2) and (8,2.5) .. (8,-2) -- (0,-2) -- cycle;
\fill[fill=white] (0,3.75) .. controls (2,3.55) and (5,2.85) .. (\linewidth,2) -- (\linewidth,-2) -- (8,-2) .. controls (8,2.5) and (3,2.5) .. (0,2);
\draw[line width=0.5pt, draw=red] (0,3.75) .. controls (2,3.75) and (5,2.85) .. (\linewidth,2);
\draw[line width=0.5pt, draw=cyan] (0,2) .. controls (3,2) and (8,2.5) .. (8,-2);
\node[align=center, fill=white, text width=4cm, rounded corners, inner sep=0.5mm, draw=none] at (5.7,-1.75) (label1) {\scriptsize \textit{polynomial size kernels}};
\node[circle, inner sep=2pt, draw=black, fill=white, label=right:{\scriptsize \textsc{IS-D}: ($k$)-nowhere dense \textit{(Thm.~\ref{thm:nowheredense-IS-k})}}] at (1.2,1) {};
\node[draw=red, inner sep=1.5pt, regular polygon, fill=white, regular polygon sides=3, label=right:{\scriptsize \textsc{VC-D}: ($k$)-general \textit{(Thm.~\ref{thm:VC-parameter-k})}}] at (1.25,-0.6) {};
\node[rectangle, inner sep=3pt, draw=teal, fill=white, label=right:{\scriptsize \textsc{DS-D}: ($k$)-biclique-free \textit{(Thm.~\ref{thm:DS-K-biclique-semi-ladder})}
}] at (2.3,0.3) {};
\node[signal, signal to=east, signal from=west, inner sep=2pt, draw={rgb,255:red,231; green,7; blue,128}, fill=white, label=right:{\scriptsize \textsc{Mat-D}: ($k$)-general \textit{(Thm.~\ref{thm:Mat-k-general})}}]  at (2.1,-1.2) {};
\node[align=center, pattern = grid, pattern color=cyan, text width=6.75cm, rounded corners, inner sep=0.5mm, draw=none] at (13,-1.75) (label2) {\scriptsize \textit{no poly kernels (assuming $\NP \not\subseteq \cp$)}};
\node[diamond, inner sep=2pt, draw={rgb:red,188; green,236; blue,50}, fill=white, label=right:{\scriptsize \textsc{VCut-D}: ($k$)-general \textit{(Thm.~\ref{thm:cross_composition-Vcut-k},}~\cite{grobler2023solution}\textit{)}}] at (4.7,2.1){};
\node[star, star points=5, star point ratio=2, inner sep=1.5pt, draw={rgb:red,100; green,64; blue,64}, fill=white, label={[align=left, text width=6cm] right:{\scriptsize \vspace{-1em} \textsc{SP-D}: ($k+\budget$)-general \textit{(Thm.~\ref{thm:SP-kb-general},}~\cite{grobler2023solution}\textit{)}}}] at (10.8,-0.35) {};
\node[draw=red, inner sep=1.5pt, regular polygon, regular polygon sides=3, label={[align=left, text width=6cm] right:{\scriptsize \vspace{-1em} \textsc{VC-D}: ($\budget + pw$)-general \textit{(Thm.~\ref{thm:cross_composition-VC-bpw},}~\cite{fellows2023solution}\textit{)}}}] at (9.25,1.2) {};
\node[rectangle, inner sep=3pt, draw=teal, fill=white, label={[align=left, text width=6cm] right:{\scriptsize \vspace{-1em} \textsc{DS-D}: ($\budget + pw$)-general 
\textit{(Thm.~\ref{thm:cross_composition-DS-bpw},}~\cite{fellows2023solution}\textit{)}}}] at (7.65,0.5) {};
\node[circle, inner sep=2pt, draw=black, fill=white, label={[align=left, text width=6cm] right:{\scriptsize \vspace{-1em} \textsc{IS-D}: ($\budget + pw$)-general \textit{(Thm.~\ref{thm:cross_composition-IS-bpw},}~\cite{fellows2023solution}\textit{)}}}] at (9,-1.15){};
\node[align=center, fill=white, text width=4.75cm, rounded corners, inner sep=0.5mm, draw=none] at (14,4.75) (label3) {\scriptsize \textit{$\W[1]$-hard}};
\node[circle, inner sep=2pt, draw=black, fill=white, label={[align=left, text width=4.5cm] right:{\scriptsize \vspace{-0.5em} \textsc{IS-D}: (fvs)-general \textit{(Thm.~\ref{thm:IS-D-fvs})}, \vspace{-1em}~~~~~~~~($pw$)-general \textit{(Thm.~\ref{thm:IS-D-pathwidth})}}}] at (6.5,4) {};
\node[draw=red, inner sep=1.5pt, regular polygon, fill=white, regular polygon sides=3, label={[align=left, text width=4.5cm] right:{\scriptsize \vspace{-0.5em} \textsc{VC-D}: (fvs)-general \textit{(Thm.~\ref{thm:VC-D-fvs})}, \vspace{-1em}~~~~~~~~~($pw$)-general \textit{(Thm.~\ref{thm:VC-D-pathwidth})}}}] at (11,3.2) {};
\node[rectangle, inner sep=3pt, draw=teal, fill=white, label={[align=left, text width=4.5cm] right:{\scriptsize \vspace{-0.5em} \textsc{DS-D}: (fvs)-general \textit{(Thm.~\ref{thm:DS-D-fvs})}, \vspace{-1em}~~~~~~~~~($pw$)-general \textit{(Thm.~\ref{thm:DS-D-pathwidth})}}}] at (1.25,4.3) {};
\end{tikzpicture}
\vspace{-1em}
\caption{\footnotesize A classification of problems into three categories: (blue, alternatively grid) problems for which we obtain polynomial kernels, (white) those for which polynomial kernels are unlikely, and (red, alternatively lines) those for which fixed-parameter tractable algorithms are unlikely. Each entry in a category mentions a solution discovery problem, one or more parameters (in parentheses and followed by a dash), and the graph class with respect to which the problem falls into the category. A reference in the parentheses indicates that the fixed-parameter tractability of that problem was established in the cited work. $pw$ denotes the pathwidth and \emph{fvs} denotes the feedback vertex set number of the input graph.}
\label{fig:summary-figure}
\end{figure}
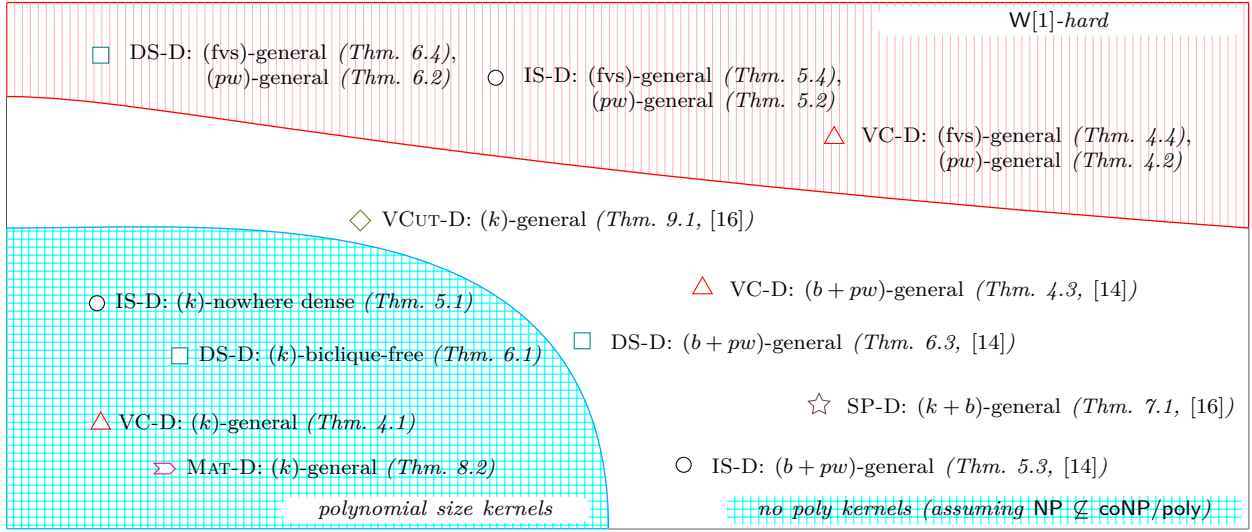

Several interesting questions remain open. 
For instance, while their parameterized complexity was determined, the kernelization complexity of \textsc{Col-D} and \textsc{ECut-D} remains unsettled.
Similarly, the kernelization complexity of \textsc{IS-D} and \textsc{DS-D} with respect to parameter $k$ is unknown on $d$-degenerate and semi-ladder-free graphs, respectively, where the problems are known to be fixed-parameter tractable.
In addition, it remains open whether \textsc{VCut-D} parameterized by $k + \budget$ admits a polynomial kernel or whether \textsc{Mat-D} parameterized by~$b$ admits polynomial kernels on restricted classes of graphs. 

\medskip
\noindent
\textbf{Organization of the paper.} We introduce all relevant notation in \Cref{sec:prelims}. In \Cref{sec:foundational}, we provide fundamental graph gadgets that appear in many constructions presented in the paper and provide several lemmas describing useful properties of those gadgets. Afterwards, we present our results for \textsc{VC-D} in \Cref{sec:vc}, \textsc{IS-D} in \Cref{sec:is}, \textsc{DS-D} in \Cref{sec:ds}, \textsc{SP-D} in \Cref{sec:sp}, \textsc{Mat-D} in \Cref{sec:mat}, and \textsc{VCut-D} in \Cref{sec:cut}.

\section{Preliminaries}
\label{sec:prelims}
We use the symbol $\N$ for the set of non-negative integers (including $0$), $\Z$ for the set of all integers, and $\Z_+$ for the set of positive non-zero integers. 
For $k \in \mathbb{N}$, we define $[k] = \{1, \ldots, k\}$ with the convention that $[0] = \varnothing$.

\medskip
\noindent \textbf{Graphs.}~
We consider finite and simple graphs only. 
We denote the vertex set and the edge set of a graph $G$ by $V(G)$ and $E(G)$, respectively, and denote an undirected edge between vertices~$u$ and~$v$ by $uv$ (or equivalently $vu$) and a directed edge from~$u$ to~$v$ by $(u,v)$. 
We use $N(v)$ to denote the set of all neighbors of $v$ and $E(v)$ to denote the set of all edges incident with $v$. Furthermore, we define the closed neighborhood of $v$ as $N[v] = N(v) \cup \{v\}$.
For a set $X$ of vertices we write $G[X]$ for the subgraph induced by $X$. 

A sequence of edges $e_1 \dots e_\ell$ for some $\ell \geq 1$ is a (simple) path of length $\ell$ if every two consecutive edges in the sequence share exactly one endpoint and each other pair of edges share no endpoints. 
For vertices $u$ and $v$, we denote the length of a shortest path $e_{1} \dots e_\ell$ that connects $u$ to $v$ by~$d(u,v)$, where $d(v,v) = 0$ for all $v \in V(G)$.
For edges $e, e' \in E(G)$, we denote by $d(e, e')$ the length of a shortest path $e_1 \dots e_\ell$ with $e_1$ being incident to $e$ and $e' = e_\ell$.
For a vertex $v \in V(G)$ and a non-negative integer $i$, we denote by $V(v, i) = \{u \in V(G) \mid d(u,v) = i\}$.
For an edge $e \in E(G)$, we denote by $E(e, i) = \{e' \in E(G) \mid d(e,e') = i\}$.

The complete graph (clique) on $n$ vertices is denoted by $K_{n}$ and a complete bipartite graph (biclique) with parts of size $m$ and $n$, respectively, by $K_{m,n}$.
For an in-depth review of general graph theoretic definitions we refer the reader to the textbook by Diestel~\cite{editionreinhard}.

\medskip
\noindent\textbf{Pathwidth and treewidth.}~
A \emph{tree decomposition} of a graph $G$ is a pair $\mathcal{T}=(T, (X_i)_{i \in V(T)})$ where $T$ is a tree and $X_i \subseteq V(G)$ for each $i \in V(T)$, such that 
\begin{enumerate}[itemsep=0pt, label=\roman*.]
    \item $\bigcup_{i \in V(T)} X_i = V(G)$,
    \item for every edge $uv = e \in E(G)$, there is an $i \in V(T)$ such that $u, v \in X_i$, and
    \item for every $v \in V(G)$, the subgraph $T_v$ of $T$ induced by $\{i \in V(T) \text{ } |\text{ } v \in X_i\}$ is connected, \ie,~$T_v$ is a tree.
\end{enumerate}

We refer to the vertices of $T$ as the \emph{nodes} of $T$. 
For a node $i$, we say that the corresponding set~$X_i$ is the \emph{bag} of $i$. 
The \emph{width} of the tree decomposition $(T, (X_i)_{i \in V(T)})$ is max$_{i \in V(T)} |X_i| - 1$. 
The \emph{treewidth} of $G$, denoted \textit{tw}($G$), is the smallest width of any tree decomposition of $G$. 

A \emph{path decomposition} of a graph $G$ is a tree decomposition $\mathcal{P}=(T, (X_i)_{i \in V(T)})$ in which $T$ is a path.
We represent a path decomposition $\mathcal{P}$ by the sequence of its bags only.
The \emph{pathwidth} of $G$, denoted \textit{pw}($G$), is the smallest width of any path decomposition of $G$. 
A \emph{nice path decomposition} of $G$ is one that begins and ends with nodes corresponding to empty bags and such that each other node in the decomposition corresponds to a bag that either \emph{introduces} a vertex $v \in V(G)$ ($X_i = X_{i-1} \cup \{v\}$ for $v \not\in X_i$) or \emph{forgets} one ($X_i = X_{i-1} \setminus \{v\}$ for $v \in X_i$). 
Every path decomposition can be efficiently turned into a nice path decomposition of the same width~\cite{cygan2015parameterized}. 
Subdividing or deleting edges of a graph $G$ preserves its path- or treewidth~\cite{robertson1986graph}. 
Additionally, the following holds.
\begin{observation}\label{obs:treewidth}
Let $G$ be a graph and $X \subseteq V(G)$. Then \textit{pw}($G$) $\le$ \textit{pw}($G - X$) $+$ $|X|$ and \textit{tw}($G$) $\le$ \textit{tw}($G - X$) $+$ $|X|$.
\end{observation}

\begin{definition}
A class $\Cc$ of graphs has bounded treewidth (bounded pathwidth) if there exists a constant $t$ such that all $G\in \Cc$ have treewidth (pathwidth) at most $t$. 
\end{definition}

\noindent{\bf Feedback vertex set number (fvs).} For a graph $G$, by \emph{fvs}$(G)$ we mean the minimum size of a vertex set whose deletion leaves the graph acyclic. 

\medskip
\noindent\textbf{Nowhere dense graphs.}~
A graph $H$ is a \emph{minor} of a graph $G$, denoted $H \preceq G$, if there exists a mapping that associates each vertex $v$ of $H$ with a non-empty connected subgraph~$G_v$ of~$G$ such that~$G_u$ and~$G_v$ are disjoint for $u \neq v$ and whenever there is an edge between $u$ and $v$ in $H$, there is an edge between a vertex of $G_u$ and a vertex of $G_v$. 
The subgraph $G_v$ is referred to as the \emph{branch set} of $v$. We call
$H$ a \emph{depth-$r$ minor} of $G$, denoted $H \preceq_r G$, if each branch set of the mapping induces a graph of radius at most $r$.

\begin{definition}
A class $\Cc$ is \emph{nowhere dense} if there exists a function $t: \mathbb{N} \rightarrow \mathbb{N}$ such that $K_{t(r)} \not\preceq_r G$ for all $r \in \mathbb{N}$ and all $G \in \Cc$. 
\end{definition}

An \emph{$r$-independent set} in a graph $G$ is a set of vertices $I$ such that the distance between any two vertices of $I$ is at least $r + 1$.
We make use of the fact that nowhere dense classes are uniform quasi-wide, as clarified by the following theorem. 
\begin{theorem}[\cite{kreutzer2018polynomial,pilipczuk2018number}]
\label{thm:quasi-wide-bounds}
Let $\Cc$ be a nowhere dense class of graphs. 
For all $r \in \mathbb{N}$, there is a polynomial $N_r: \mathbb{N} \rightarrow \mathbb{N}$ and a constant $x_r \in \mathbb{N}$ such that following holds. Let $G \in \Cc$ and let $A \subseteq V(G)$ be a vertex subset of size at least $N_r(m)$, for a given $m \in \mathbb{N}$. 
Then there exists a set $X \subseteq V(G)$ of size $|X| \le x_r$ and a set $B \subseteq A \setminus X$ of size at least $m$ that is $r$-independent in $G - X$. 
Moreover, given $G$ and $A$, such sets $X$ and $B$ can be computed in time $\Omc(|A| \cdot |E(G)|)$. 
\end{theorem} 

\medskip
\noindent 
\textbf{\small Biclique-free graphs.}~ A graph is said to be \emph{$d$-biclique-free} it excludes the biclique $K_{d,d,}$ as a subgraph. 

\begin{definition}
    A class $\Cc$ of graphs is \emph{biclique-free} if there exists a number $d$ such that all $G\in \Cc$ are $d$-biclique-free. 
\end{definition}

An inclusion diagram of all presented graph classes is depicted in \Cref{fig:classes}.

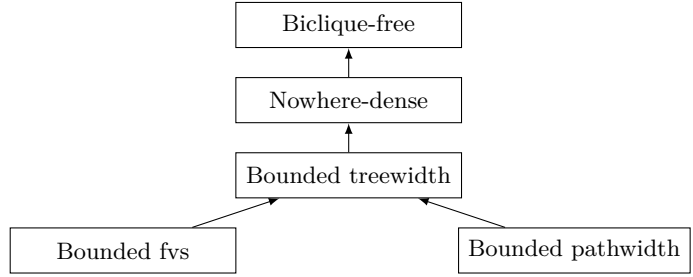
\begin{figure}
    \centering
    \begin{tikzpicture}
    \node[draw, rectangle, minimum width=3cm, minimum height=0.6cm] (nd) at (0,0) {\footnotesize Nowhere-dense};
    \node[draw, rectangle, minimum width=3cm, minimum height=0.6cm] (bf) at (0,1) {\footnotesize Biclique-free};
    \node[draw, rectangle, minimum width=3cm, minimum height=0.6cm] (bt) at (0,-1) {\footnotesize Bounded treewidth};
    \node[draw, rectangle, minimum width=3cm, minimum height=0.6cm] (bp) at (3,-2) {\footnotesize Bounded pathwidth};
    \node[draw, rectangle, minimum width=3cm, minimum height=0.6cm] (bd) at (-3,-2) {\footnotesize Bounded fvs};
    \draw[-latex] (nd) -- (bf);
    \draw[-latex] (bd) to (bt);
    \draw[-latex] (bt) -- (nd);
    \draw[-latex] (bp) -- (bt);
    \end{tikzpicture}
    \caption{Graph classes considered in this paper. Arrows indicate inclusion.}
\end{figure}\label[figure]{fig:classes}

\medskip
\noindent\textbf{Solution discovery.}~
A vertex (edge) subset problem $\Qmc$ is a problem defined on graphs such that a solution consists of a subset of vertices (edges) satisfying certain requirements.
For a vertex (edge) subset problem $\Qmc$ on an instance with an input graph $G$, a \emph{configuration} $C$ on $G$ is a subset of its vertices (edges).
Alternatively, a configuration can be seen as the placement of tokens on a subset of vertices (edges) in $G$. 
In the \emph{token sliding} model, a configuration $C'$ can be obtained (in one step) from a configuration $C$, written $C\vdash C'$, if $C' = (C \setminus \{y\}) \cup \{x\}$ for elements $y \in C$ and $x \notin C$ such that $x$ and $y$ are neighbors in $G$, that is, if $x,y \in V(G)$, then $xy \in E(G)$; and if $x,y \in E(G)$, then they share an endpoint. 
Alternatively, when a token \emph{slides} from a vertex to an adjacent one or from an edge to an incident one, we get $C\vdash C'$.
A \emph{discovery sequence} of length $\ell$ in $G$ is a sequence of configurations $C_0 C_1 \dots C_\ell$ of $G$ such that $C_i \vdash C_{i+1}$ for all $0 \leq i < \ell$. 

The \textsc{$\Pi$-Discovery} problem is defined as follows. We are
given a graph $G$, a configuration $S \subseteq V(G)$ (resp.\ $S\subseteq E(G)$) of size $k$ (which at this point is not necessarily a solution
for $\Pi$), and a budget~$\budget$ (a non-negative integer). 
We denote instances of \textsc{$\Pi$-Discovery} by $(G,S,\budget)$. 
The goal is to decide whether there exists a discovery sequence $C_0 C_1 \dots C_\ell$ in $G$ for some $\ell \leq \budget$ such that $S = C_0$ and~$C_\ell$ is a solution for $\Pi$.
When a path decomposition is given as part of the input, the instances are denoted by $(G,\mathcal{P}_G,S,\budget)$ to highlight that the path decomposition $\mathcal{P}_G$ of $G$ is provided. 

\medskip
\noindent \textbf{Parameterized complexity and kernelization.}~ Downey and Fellows~\cite{downey2012parameterized} developed a framework for parameterized problems which include a parameter $p$ in their input.
A parameterized problem~$\Qmc$ has inputs of the form $(x, p)$ where $|x| = n$ and $p \in \Nbb$. 
Fixed-parameter tractable problems belong to the complexity class $\FPT$.
The class $\XNLP$ consists of the parameterized problems that can be solved with a non-deterministic algorithm that uses $f(p)\cdot \log n$ space and $f(p)\cdot n^{\Oof(1)}$ time.
The \emph{\W-hierarchy} is a collection of parameterized complexity classes $\FPT \subseteq \W[1] \subseteq \W[2] \subseteq \ldots \subseteq \XNLP$ where inclusions are conjectured to be strict.

For parameterized problems $\Qmc$ and $\Qmc'$, an \emph{\FPT-reduction} from~$\Qmc$ to~$\Qmc'$ is a reduction that given an instance $(x,p)$ of $\Qmc$ produces $(x',p')$ of $\Qmc'$ in time $f(p) \cdot |x|^{\Omc(1)}$ and such that $p' \leq g(p)$ where $f,g$ are computable functions. 
A \emph{pl-reduction} from~$\Qmc$ to~$\Qmc'$ is one that additionally computes $(x',p')$ using $\Oof(h(p) + \log |x|)$ working space where $h$ is a computable function.
We write $\Qmc \leq_{\textsc{fpt}} \Qmc'$ (resp.\ $\Qmc \leq_{\text{pl}} \Qmc'$) if there is an \FPT-reduction (resp.\ pl-reduction) from $\Qmc$ to~$\Qmc'$.
If~$\Qmc$ is $\W[t]$-hard for a positive integer $t$ and $\Qmc \leq_{\textsc{fpt}} \Qmc'$, then $\Qmc'$ is also $\W[t]$-hard.
If~$\Qmc$ is $\XNLP$-hard and $\Qmc \leq_{\text{pl}} \Qmc'$, then~$\Qmc'$ is \XNLP-hard and, in particular, $\W[t]$-hard for all $t\geq 1$. 

Every problem that is in $\FPT$ admits a kernel, although it may be of exponential size or larger. 
Under the complexity-theoretic assumption that $\NP \not\subseteq \cp$, we can rule out the existence of a polynomial kernel for certain fixed-parameter tractable problems $\Qmc$.
The machinery for such kernel lower bounds heavily relies on composing instances that are equivalent according to a polynomial equivalence relation~\cite{cygan2015parameterized}.

\begin{definition}\label{def:equivalence-relation}
An equivalence relation $\Rmc$ on the set of instances of a problem $\Qmc$ is called a \emph{polynomial equivalence relation} if the following two conditions hold.
\begin{enumerate}[itemsep=0pt]
    \item There is an algorithm that given two instances $x$ and $y$ of $\Qmc$ decides whether $x$ and $y$ belong to the same equivalence class in time polynomial in $|x| + |y|$.
    \item For any finite set $S$ of instances of $\Qmc$, the equivalence relation $\Rmc$ partitions the elements of $S$ into at most $(\max_{x \in S}$ $|x|)^{\Oof(1)}$ classes.
\end{enumerate}
\end{definition}

We can compose equivalent instances in more than one way. 
We focus here on or-cross-compositions.
\begin{definition}[\cite{bodlaender2014kernelization}]
\label{def:or-cross-composition}
Let $\Qmc'$ be a problem and let $\Qmc$ be a parameterized problem. We say that~$\Qmc$ \emph{or-cross-composes} into $\Qmc'$ if there is a polynomial equivalence relation $\Rmc$ on the set of instances of $\Qmc$ and an algorithm that, given $t$ instances (where $t \in \mathbb{Z}_+$) $x_1, x_2, \ldots , x_t$ belonging to the same equivalence class of $\mathcal{R}$, computes an instance $(x^*, k^*)$ in time polynomial in $\Sigma^t_{i=1} |x_i|$ such that the following properties hold.
\begin{enumerate}[itemsep=0pt]
    \item  $(x^*, k^*)$ $\in \Qmc$ if and only if there exists at least one index $i$ such that $x_i$ is a yes-instance of $\Qmc'$.
    \item $k^*$ is bounded above by a polynomial in $\max^t_{i=1} |x_i| + \log t$.
\end{enumerate}
\end{definition}

The inclusion $\NP \subseteq \cp$ holds if an \NP-hard problem or-cross-composes into a parameterized problem $\Qmc$ having a polynomial kernel. As this inclusion is believed to be false, we will constantly make use of the following theorem to show that the existence of a polynomial kernel is unlikely.

\begin{theorem}[\cite{bodlaender2014kernelization}]
\label{thm:no-poly-kernel-theorem}
If a problem $\Qmc'$ is \NP-hard and $\Qmc'$ or-cross-composes into the parameterized problem $\Qmc$, then there is no polynomial kernel for $\Qmc$ unless $\NP \subseteq \cp$.
\end{theorem}

\noindent We refer the reader to textbooks~\cite{cygan2015parameterized,downey2013fundamentals} for more on parameterized complexity and kernelization.
\section{An Auxiliary Problem and Foundational Gadgets}\label{sec:foundational}
In this section, we describe foundational gadgets used in our reductions and compositions and explain how combining such gadgets preserves a bound on the pathwidth of the constructed graphs (assuming we start with  graphs of bounded pathwidth). We show first that starting from a graph of bounded pathwidth $H$, we can construct new graphs~$G_H$,~$\Tilde{G}_H$,~$G_t$, and~$\hat{G}_t$, using our gadgets such that~$G_H$,~$\Tilde{G}_H$,~$G_t$, and~$\hat{G}_t$ still have bounded pathwidth (in addition to other useful properties). 

The following problem will be used in the reductions that establish the \XNLP-hardness of \mbox{\textsc{IS-D}}, \textsc{VC-D} and \textsc{DS-D} with respect to parameter $pw$ and subsequently in the or-cross-compositions that render it unlikely for any of these problems to have a polynomial kernel with respect to parameter~$\budget + pw$.
We denote by \emph{orientation} of a graph $G$ a mapping $\lambda: E(G) \rightarrow V(G) \times V(G)$ such that $\lambda(uv) \in \{(u,v), (v,u)\}$.
\\[1ex]
\textsc{Minimum Maximum Outdegree (\textsc{MMO}):}
\newline
\noindent\textbf{Input}: Undirected weighted graph $H$, a path decomposition $\mathcal{P}_H$ of $H$ of width $pw$, an edge weighting $\sigma: E(H) \rightarrow \Zbb_{+}$ and a positive integer~$r$ (all integers are given in unary).
\newline
\textbf{Question}: Is there an orientation of $H$ such that for each $v \in V(H)$, the total weight of the edges directed away from $v$ is at most $r$?
\\[1ex]
Bodlaender et al.~\cite{bodlaender2022problems} showed that \textsc{MMO} is \XNLP-complete with respect to pathwidth given a path decomposition realising the pathwidth. 
If all edge weights are identical, then \textsc{MMO} (on general graphs) can be solved in polynomial time using network flows~\cite{asahiro2011graph}. 

For an instance $(H, \mathcal{P}_H, \sigma, r)$ of \textsc{MMO}, we define $\bm{\sigma} = \sum_{e \in E(H)} \sigma(e)$, $n = |V(H)|$ and $m = |E(H)|$. 
We construct for an instance $(H, \mathcal{P}_H, \sigma, r)$ of \textsc{MMO}, a graph $G_H$ consisting of disjoint subgraphs~$G_e$ for each $e \in E(H)$ and $G_v$ for each $v \in V(H)$. 
We refer to the edge-based and vertex-based subgraphs as \textit{MMO-edge-gadgets} and \textit{MMO-vertex-gadgets}, respectively. For an edge $e \in E(H)$ we refer to $G_e$ as \textit{MMO-edge-$e$}. Similarly, for a vertex $v \in V(H)$ we refer to $G_v$ as \textit{MMO-vertex-$v$}. 

\medskip
\noindent\textbf{MMO-edge-e.} For an edge $e = uv \in E(H)$, an \textsc{MMO}-edge-$e$ $G_e$ contains $\sigma(e) + 1$ edges with endpoints $a_e^i$ and $b_e^i$ for $i \in [\sigma(e) + 1]$, and an edge $e^ue^v$ such that $b_e^{\sigma(e) + 1}$ is adjacent to each of $e^u$ and $e^v$. 
We define $A_e = \cup_{i \in [\sigma(e)]} \text{ } \text{ } a_e^i$ and $B_e = \cup_{i \in [\sigma(e)]} \text{ } \text{ } b_e^i$ (see \Cref{fig:mmo-edge}).
We refer to the connected component inside $G_e$ (or any subdivision of $G_e$) containing $e^u$ and $e^v$ by $G^{sel}_e$.

\begin{figure}[h]
    \centering
    \begin{tikzpicture}[scale=0.7]
    \fill[blue!10] (3.5,-1.5) -- (5,-1.5) -- (5,1.5) -- (3.5,1.5) -- cycle; 
    \fill[blue!10] (4.5,-0.5) -- (8,-2) -- (8,0) -- (4.5, 0) -- cycle;
    \fill[blue!10] (4.5,0.5) -- (8,2) -- (8,0) -- (4.5, 0) -- cycle; 
    \draw[black] (4,-1) -- (4.5,-1) node[circle, draw=black, fill=white, inner sep=1pt, label=left:{\scalebox{0.5}{$a^1_e$}}, pos=0]{} node[circle, draw=black, fill=white, inner sep=1pt, label=right:{\scalebox{0.5}{$b^1_e$}}, pos=1]{};
    \draw[black] (4,-0.5) -- (4.5,-0.5) node[circle, draw=black, fill=white, inner sep=1pt, label=left:{\scalebox{0.5}{$a^2_e$}}, pos=0]{} node[circle, draw=black, fill=white, inner sep=1pt, label=right:{\scalebox{0.5}{$b^2_e$}}, pos=1]{};
    \draw[black] (4,0.5) -- (4.5,0.5) node[circle, draw=black, fill=white, inner sep=1pt, label=left:{\scalebox{0.5}{$a^3_e$}}, pos=0]{} node[circle, draw=black, fill=white, inner sep=1pt, label=right:{\scalebox{0.5}{$b^3_e$}}, pos=1]{};
    \draw[black] (4,1) -- (4.5,1) node[circle, draw=black, fill=white, inner sep=1pt, label=left:{\scalebox{0.5}{$a^4_e$}}, pos=0]{} node[circle, draw=black, fill=white, inner sep=1pt, label=right:{\scalebox{0.5}{$b^4_e$}}, pos=1]{};
    \draw[black] (7.5,1.5) -- (7.5, -1.5);
    \draw[black] (4.5,0) -- (7.5,1.5) node[circle, draw=black, fill=white, inner sep=1pt, label=right:{\scalebox{0.5}{$e^u$}}, pos=1]{};
    \draw[black] (4.5,0) -- (7.5,-1.5) node[circle, draw=black, fill=white, inner sep=1pt, label=right:{\scalebox{0.5}{$e^v$}}, pos=1]{};
    \draw[black] (4,0) -- (4.5,0) node[circle, draw=black, fill=white, inner sep=1pt, label=left:{\scalebox{0.5}{$a^5_e$}}, pos=0]{} node[circle, draw=black, fill=white, inner sep=1pt, label=right:{\scalebox{0.5}{$b^5_e$}}, pos=1]{};
    \end{tikzpicture}
    \caption{\footnotesize An \textsc{MMO}-edge-$e$  $G_e$ for an edge $uv = e \in E(H)$ for a graph $H$ and edge weight function $\sigma$ of an \textsc{MMO} instance, with $\sigma(e) = 4$.}
    \label{fig:mmo-edge}
\end{figure}
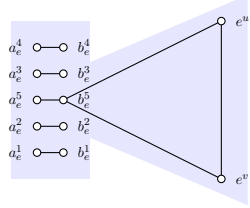

\medskip
\noindent \textbf{MMO-vertex-v.} For a vertex $v$ in $V(H)$, an \textsc{MMO}-vertex-$v$ $G_v$ contains a \emph{representative vertex of $v$} denoted by $w_v$, adjacent to $r$ \emph{target vertices of v} denoted by $x_v^1, x_v^2, \ldots, x_v^{r}$ and one extra vertex~$x_v^{r + 1}$. 
Additionally, for each edge $e \in E(H)$ incident to $v$, the \textsc{MMO}-vertex-$v$ contains $\sigma(e)$ edges with endpoints $y_e^{v(i)}$ and $z_e^{v(i)}$ for $i \in [\sigma(e)]$ such that $y_e^{v(i)}$ is adjacent to $w_v$, the representative vertex of $v$ (see \Cref{fig:mmo-vertex-gadget}). 
We define $X_v = \cup_{i \in [r]} \text{ } \text{ } x_v^i$, $Y^v_e = \cup_{i \in [\sigma(e)]} \text{ } \text{ } y_e^{v(i)}$, $Z^v_e = \cup_{i \in [\sigma(e)]} \text{ } \text{ } z_e^{v(i)}$, $Y^v = \cup_{e \in E(H)} \text{ } \text{ } Y^v_e$, and $Z^v = \cup_{e \in E(H)} \text{ } \text{ } Z^v_e$.

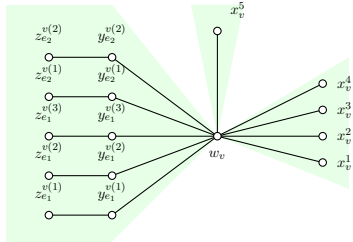
\begin{figure}[h]
    \centering
    \begin{tikzpicture}[scale=0.7]
    \fill[green!10] (0,0) -- (2.5,-1) -- (2.5,1.5) -- (0,0) -- cycle;
    \fill[green!10] (0,0) -- (-0.5,2.5) -- (0.5,2.5) -- (0,0) -- cycle;
    \fill[green!10] (0,0) -- (-2,-2) -- (-2,2.5) -- (0,0) -- cycle;
    \fill[green!10] (-2,2.5) -- (-4,2.5) -- (-4,-2) -- (-2,-2) -- cycle;
    \draw[black] (0,0) -- (2,-0.5);
    \draw[black] (0,0) -- (2,0);
    \draw[black] (0,0) -- (2,0.5);
    \draw[black] (0,0) -- (2,1);  
    \draw[black] (0,0) -- (0,2);   
    \draw[black] (0,0) -- (-2,-1.5);
    \draw[black] (0,0)-- (-2,-0.75);
    \draw[black] (0,0) -- (-2,0);
    \draw[black] (0,0) -- (-2,1.5);
    \draw[black] (0,0) -- (-2,0.75);
    \node[circle, fill=black, inner sep=1pt, draw=black, fill=white, label=below:{\scalebox{0.5}{$w_v$}}] at (0,0) {};
    \node[circle, fill=black, inner sep=1pt, draw=black, fill=white, label=right:{\scalebox{0.5}{$x_v^1$}}] at (2,-0.5) {};
    \node[circle, fill=black, inner sep=1pt, draw=black, fill=white, label=right:{\scalebox{0.5}{$x_v^2$}}] at (2,0) {};
    \node[circle, fill=black, inner sep=1pt, draw=black, fill=white, label=right:{\scalebox{0.5}{$x_v^3$}}] at (2,0.5) {};
    \node[circle, fill=black, inner sep=1pt, draw=black, fill=white, label=right:{\scalebox{0.5}{$x_v^4$}}] at (2,1) {};
    \node[circle, fill=black, inner sep=1pt, draw=black, fill=white, label=above right:{\scalebox{0.5}{$x_v^{5}$}}] at (0,2) {};
    \draw[black] (-3.2,-1.5) -- (-2,-1.5) node[circle, draw=black, fill=white, inner sep=1pt, label=above:{\scalebox{0.5}{$y_{e_1}^{v(1)}$}}, pos=1]{} node[circle, draw=black, fill=white, inner sep=1pt, label=above:{\scalebox{0.5}{$z_{e_1}^{v(1)}$}}, pos=0]{} ;
    \draw[black] (-3.2,-0.75) -- (-2,-0.75) node[circle, draw=black, fill=white, inner sep=1pt, label=above:{\scalebox{0.5}{$y_{e_1}^{v(2)}$}}, pos=1]{} node[circle, draw=black, fill=white, inner sep=1pt, label=above:{\scalebox{0.5}{$z_{e_1}^{v(2)}$}}, pos=0]{};
    \draw[black] (-3.2,0) -- (-2,0) node[circle, draw=black, fill=white, inner sep=1pt, label=above:{\scalebox{0.5}{$y_{e_1}^{v(3)}$}}, pos=1]{} node[circle, draw=black, fill=white, inner sep=1pt, label=above:{\scalebox{0.5}{$z_{e_1}^{v(3)}$}}, pos=0]{};
    \draw[black] (-3.2,1.5) -- (-2,1.5) node[circle, draw=black, fill=white, inner sep=1pt, label=above:{\scalebox{0.5}{$y_{e_2}^{v(2)}$}}, pos=1]{} node[circle, draw=black, fill=white, inner sep=1pt, label=above:{\scalebox{0.5}{$z_{e_2}^{v(2)}$}}, pos=0]{};
    \draw[black] (-3.2,0.75) -- (-2,0.75) node[circle, draw=black, fill=white, inner sep=1pt, label=above:{\scalebox{0.5}{$y_{e_2}^{v(1)}$}}, pos=1]{} node[circle, draw=black, fill=white, inner sep=1pt, label=above:{\scalebox{0.5}{$z_{e_2}^{v(1)}$}}, pos=0]{};
    \end{tikzpicture}
    \caption{\footnotesize An \textsc{MMO}-vertex-$v$ $G_v$ for a vertex $v \in V(H)$ for a graph $H$, edge weight function $\sigma$, and integer $r$ of an \textsc{MMO} instance. The vertex $v$ is incident to edges $e_1, e_2 \in E(H)$, $\sigma(e_1) = 3$, $\sigma(e_2) = 2$, and $r= 4$.}
    \label{fig:mmo-vertex-gadget}
\end{figure}

\noindent\textbf{\boldmath The graph $G_H$.} We let $A = \cup_{e \in E(H)} \text{ } \text{ } A_e$, $A^+ = \cup_{e \in E(H)} \text{ } \text{ } a_e^{\sigma(e)+1}$, $B = \cup_{e \in E(H)} \text{ } \text{ } B_e$, $B^+ = \cup_{e \in E(H)} \text{ } \text{ } b_e^{\sigma(e)+1}$, $X = \cup_{v \in V(H)} \text{ } \text{ } X_v$, $X^+ = \cup_{v \in V(H)} \text{ } \text{ } x_v^{r+1}$, $Y = \cup_{v \in V(H)} \text{ } \text{ } Y^v$, and $Z = \cup_{v \in V(H)} \text{ } \text{ } Z^v$. 
We form $G_H$ by connecting its \textsc{MMO}-edge-gadget vertices to its \textsc{MMO}-vertex-gadget  vertices as follows. 
For a vertex $v \in V(H)$ and edge $e \in E(H)$ incident to $v$, we connect each vertex of $B_e$ to a corresponding distinct vertex in $Z_e^v$ (in other words, each $b_e^i$ to $z_e^{v(i)}$ for $i \in [\sigma(e)]$). Similarly, we connect $e^v$ to each vertex of $Y_e^v$ (see \Cref{fig:combined} for an example). 

\begin{figure}[H]
    \centering
    \begin{tikzpicture}[scale=0.7]
        \fill[blue!10] (-10.5,-1) -- (-9,-1) -- (-9,1) -- (-10.5,1) -- cycle; 
        \fill[blue!10] (-9.5,-0.5) -- (-7,-2) -- (-7,0) -- (-9.5, 0) -- cycle; 
        \fill[blue!10] (-9.5,0.5) -- (-7,2) -- (-7,0) -- (-9.5, 0) -- cycle; 
        \fill[green!10] (0,-2) -- (1.5,-3) -- (1.5,-0.5) -- (0,-2) -- cycle;
        \fill[green!10] (0,-2) -- (-0.5,-0.5) -- (0.5,-0.5) -- (0,-2) -- cycle;
        \fill[green!10] (0,-2) -- (-2,-2.5) -- (-2,0) -- (0,-2) -- cycle;
        \fill[green!10] (-2,0) -- (-4,0) -- (-4,-2.5) -- (-2,-2.5) -- cycle;
        \fill[green!10] (0,2) -- (1.5,1) -- (1.5,3.5) -- (0,2) -- cycle;
        \fill[green!10] (0,2) -- (-0.5,3.5) -- (0.5,3.5) -- (0,2) -- cycle;
        \fill[green!10] (0,2) -- (-2,1) -- (-2,4) -- (1,2) -- cycle;
        \fill[green!10] (-2,4) -- (-4,4) -- (-4,1) -- (-2,1) -- cycle;
        \draw[red] (-9.5,0.5) .. controls (-8,3) .. (-3.2,2);
        \draw[red] (-9.5,0.5) .. controls (-6,-1) .. (-3.2,-1.5);
        \draw[red] (-9.5,-0.5) .. controls (-7,0.5) and (-5,1) .. (-3.2,1.25);
        \draw[red] (-9.5,-0.5) .. controls (-8,-3) .. (-3.2,-2.25);
        \draw[yellow] (-7.5,1.5) .. controls (-5,0.5) .. (-2,2);
        \draw[yellow] (-7.5,1.5) .. controls (-5,0) .. (-2,1.25);
        \draw[yellow] (-7.5,-1.5) .. controls (-5,-1) .. (-2,-1.5);
        \draw[yellow] (-7.5,-1.5) .. controls (-5,-1.75) .. (-2,-2.25);
        \draw[black] (-10,-0.5) -- (-9.5,-0.5) node[circle, draw=black, fill=white, inner sep=1pt, label=left:{\scalebox{0.5}{$a^1_e$}}, pos=0]{} node[circle, draw=black, fill=white, inner sep=1pt, label=right:{\scalebox{0.5}{$b^1_e$}}, pos=1]{};
        \draw[black] (-10,0.5) -- (-9.5,0.5) node[circle, draw=black, fill=white, inner sep=1pt, label=left:{\scalebox{0.5}{$a^2_e$}}, pos=0]{} node[circle, draw=black, fill=white, inner sep=1pt, label=right:{\scalebox{0.5}{$b^2_e$}}, pos=1]{};
        \draw[black] (-7.5,1.5) -- (-7.5, -1.5);
        \draw[black] (-9.5,0) -- (-7.5,1.5) node[circle, draw=black, fill=white, inner sep=1pt, label=above right:{\scalebox{0.5}{$e^u$}}, pos=1]{};
        \draw[black] (-9.5,0) -- (-7.5,-1.5) node[circle, draw=black, fill=white, inner sep=1pt, label=above right:{\scalebox{0.5}{$e^v$}}, pos=1]{};
        \draw[black] (-10,0) -- (-9.5,0) node[circle, draw=black, fill=white, inner sep=1pt, label=left:{\scalebox{0.5}{$a^3_e$}}, pos=0]{} node[circle, draw=black, fill=white, inner sep=1pt, label=right:{\scalebox{0.5}{$b^3_e$}}, pos=1]{};
        \draw[black] (0,2) -- (1,1.5);
        \draw[black] (0,2) -- (1,2);
        \draw[black] (0,2) -- (1,2.5);
        \draw[black] (0,2) -- (1,3);
        \draw[black] (0,2) -- (0,3);
        \draw[black] (0,2)-- (-2,1.25);
        \draw[black] (0,2) -- (-2,2);
        \draw[black] (0,2) -- (-2,3.5);
        \draw[black] (0,2) -- (-2,2.75);
        \node[circle, fill=black, inner sep=1pt, draw=black, fill=white, label=below:{\scalebox{0.5}{$w_u$}}] at (0,2) {};
        \node[circle, fill=black, inner sep=1pt, draw=black, fill=white, label=right:{\scalebox{0.5}{$x_u^1$}}] at (1,1.5) {};
        \node[circle, fill=black, inner sep=1pt, draw=black, fill=white, label=right:{\scalebox{0.5}{$x_u^2$}}] at (1,2) {};
        \node[circle, fill=black, inner sep=1pt, draw=black, fill=white, label=right:{\scalebox{0.5}{$x_u^3$}}] at (1,2.5) {};
        \node[circle, fill=black, inner sep=1pt, draw=black, fill=white, label=right:{\scalebox{0.5}{$x_u^4$}}] at (1,3) {};
        \node[circle, fill=black, inner sep=1pt, draw=black, fill=white, label=above right:{\scalebox{0.5}{$x_u^{5}$}}] at (0,3) {};
        \draw[black] (-3.2,1.25) -- (-2,1.25) node[circle, draw=black, fill=white, inner sep=1pt, label=above:{\scalebox{0.5}{$y_{e}^{u(1)}$}}, pos=1]{} node[circle, draw=black, fill=white, inner sep=1pt, label=above:{\scalebox{0.5}{$z_{e}^{u(1)}$}}, pos=0]{};
        \draw[black] (-3.2,2) -- (-2,2) node[circle, draw=black, fill=white, inner sep=1pt, label=above:{\scalebox{0.5}{$y_{e}^{u(2)}$}}, pos=1]{} node[circle, draw=black, fill=white, inner sep=1pt, label=above:{\scalebox{0.5}{$z_{e}^{u(2)}$}}, pos=0]{};
        \draw[black] (-3.2,3.5) -- (-2,3.5) node[circle, draw=black, fill=white, inner sep=1pt, label=above:{\scalebox{0.5}{$y_{e'}^{u(2)}$}}, pos=1]{} node[circle, draw=black, fill=white, inner sep=1pt, label=above:{\scalebox{0.5}{$z_{e'}^{u(2)}$}}, pos=0]{};
        \draw[black] (-3.2,2.75) -- (-2,2.75) node[circle, draw=black, fill=white, inner sep=1pt, label=above:{\scalebox{0.5}{$y_{e'}^{u(1)}$}}, pos=1]{} node[circle, draw=black, fill=white, inner sep=1pt, label=above:{\scalebox{0.5}{$z_{e'}^{u(1)}$}}, pos=0]{};
        \draw[black] (0,-2) -- (1,-2.5);
        \draw[black] (0,-2) -- (1,-2);
        \draw[black] (0,-2) -- (1,-1.5);
        \draw[black] (0,-2) -- (1,-1);
        \draw[black] (0,-2) -- (0,-1);
        \draw[black] (0,-2)-- (-2,-2.25);
        \draw[black] (0,-2) -- (-2,-1.5);
        \draw[black] (0,-2) -- (-2,-0.75);
        \node[circle, fill=black, inner sep=1pt, draw=black, fill=white, label=below:{\scalebox{0.5}{$w_v$}}] at (0,-2) {};
        \node[circle, fill=black, inner sep=1pt, draw=black, fill=white, label=right:{\scalebox{0.5}{$x_v^1$}}] at (1,-2.5) {};
        \node[circle, fill=black, inner sep=1pt, draw=black, fill=white, label=right:{\scalebox{0.5}{$x_v^2$}}] at (1,-2) {};
        \node[circle, fill=black, inner sep=1pt, draw=black, fill=white, label=right:{\scalebox{0.5}{$x_v^3$}}] at (1,-1.5) {};
        \node[circle, fill=black, inner sep=1pt, draw=black, fill=white, label=right:{\scalebox{0.5}{$x_v^4$}}] at (1,-1) {};
        \node[circle, fill=black, inner sep=1pt, draw=black, fill=white, label=above right:{\scalebox{0.5}{$x_v^{5}$}}] at (0,-1) {};
        \draw[black] (-3.2,-2.25) -- (-2,-2.25) node[circle, draw=black, fill=white, inner sep=1pt, label=above:{\scalebox{0.5}{$y_{e}^{v(1)}$}}, pos=1]{} node[circle, draw=black, fill=white, inner sep=1pt, label=above:{\scalebox{0.5}{$z_{e_1}^{v(1)}$}}, pos=0]{};
        \draw[black] (-3.2,-1.5) -- (-2,-1.5) node[circle, draw=black, fill=white, inner sep=1pt, label=above:{\scalebox{0.5}{$y_{e}^{v(2)}$}}, pos=1]{} node[circle, draw=black, fill=white, inner sep=1pt, label=above:{\scalebox{0.5}{$z_{e}^{v(2)}$}}, pos=0]{};
        \draw[black] (-3.2,-0.75) -- (-2,-0.75) node[circle, draw=black, fill=white, inner sep=1pt, label=above:{\scalebox{0.5}{$y_{e''}^{v(1)}$}}, pos=1]{} node[circle, draw=black, fill=white, inner sep=1pt, label=above:{\scalebox{0.5}{$z_{e''}^{v(1)}$}}, pos=0]{};
    \end{tikzpicture}
    \caption{\footnotesize Edges from one \textsc{MMO}-edge-$e$, for an edge $e = uv$ for a graph $H$, edge weight function $\sigma$, and integer $r$ of an \textsc{MMO} instance, to the \textsc{MMO}-vertex-$u$ and \textsc{MMO}-vertex-$v$ subgraphs in $G_H$. Red is used for edges between vertices in $B$ and $Z$ and yellow is used for edges between vertices in $\{e^u, e^v\}$ and $Y$. $\sigma(e) = 2$ and $r = 4$.}
    \label{fig:combined}
\end{figure}
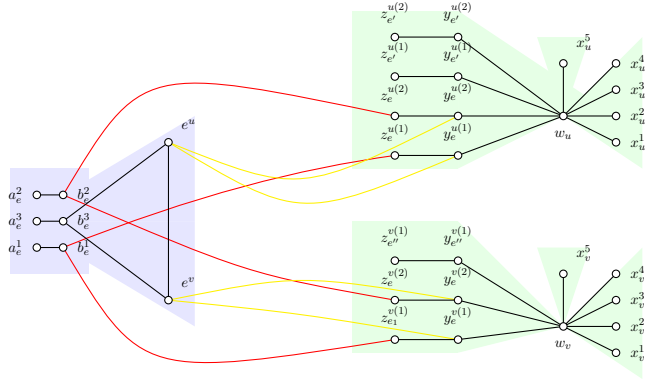

\medskip
\noindent\textbf{\boldmath The supplier gadget and the graph $\Tilde{G}_H$.}
In some of our reductions, we add a new gadget to~$G_H$ and make one of its vertices the \emph{supplier vertex} adjacent to various vertices within~$G_H$. 
We denote the graph thus obtained by $\Tilde{G}_H$ and refer to the gadget containing the supplier vertex as the \emph{supplier gadget}.
We let $G_s$ be the supplier gadget that we connect to~$G_H$, and we let $s$ denote the \emph{supplier vertex}.
In particular, $G_s$ contains a supplier vertex $s$ that is adjacent to \emph{donor vertices} $d_1^i$ of the \emph{donor paths} $D^i =\{d_1^i, d_2^i, d_3^i\}$ for $i \in [rn -  \bm{\sigma}]$ as well as another vertex $d_1^{rn - \bm{\sigma} + 1}$ (see \Cref{fig:G-supplier}). 

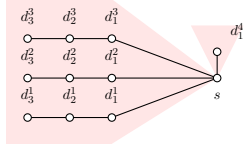
\begin{figure}[h]
    \centering
    \begin{tikzpicture}[scale=0.7]
    \fill[red!10] (5,0) -- (4.5,1) -- (5.5,1) -- (5,0) -- cycle;
    \fill[red!10] (5,0) -- (3,-1.25) -- (3,1.5) -- (5,0) -- cycle;
    \fill[red!10] (3,-1.25) -- (1,-1.25) -- (1,1.5) -- (3,1.5) -- cycle;
    \draw[black] (5,0) -- (5,0.5);
    \draw[black] (5,0)-- (3,-0.75);
    \draw[black] (5,0) -- (3,0);
    \draw[black] (5,0) -- (3,0.75);
    \draw[black] (2.2,-0.75) -- (3,-0.75);
    \draw[black] (2.2,0) -- (3,0);
    \draw[black] (2.2,0.75) -- (3,0.75); 
    \draw[black] (1.4,-0.75) -- (2.2,-0.75);
    \draw[black] (1.4,0) -- (2.2,0);
    \draw[black] (1.4,0.75) -- (2.2,0.75); 
    \node[circle, fill=black, inner sep=1pt, draw=black, fill=white, label=below:{\scalebox{0.5}{$s$}}] at (5,0) {};
    \node[circle, fill=black, inner sep=1pt, draw=black, fill=white, label=above right:{\scalebox{0.5}{$d_1^4$}}] at (5,0.5) {};
    \node[circle, draw=black, fill=white, inner sep=1pt, label=above:{\scalebox{0.5}{$d^1_2$}}] at (2.2,-0.75) {};
    \node[circle, draw=black, fill=white, inner sep=1pt, label=above:{\scalebox{0.5}{$d^1_1$}}] at (3,-0.75) {};
    \node[circle, draw=black, fill=white, inner sep=1pt, label=above:{\scalebox{0.5}{$d^2_2$}}] at (2.2,0) {};
    \node[circle, draw=black, fill=white, inner sep=1pt, label=above:{\scalebox{0.5}{$d^2_1$}}] at (3,0) {};
    \node[circle, draw=black, fill=white, inner sep=1pt, label=above:{\scalebox{0.5}{$d^3_2$}}] at (2.2,0.75) {};
    \node[circle, draw=black, fill=white, inner sep=1pt, label=above:{\scalebox{0.5}{$d^3_1$}}] at (3,0.75) {};

    \node[circle, draw=black, fill=white, inner sep=1pt, label=above:{\scalebox{0.5}{$d^1_3$}}] at (1.4,-0.75) {};
    \node[circle, draw=black, fill=white, inner sep=1pt, label=above:{\scalebox{0.5}{$d^2_3$}}] at (1.4,0) {};
    \node[circle, draw=black, fill=white, inner sep=1pt, label=above:{\scalebox{0.5}{$d^3_3$}}] at (1.4,0.75) {};
    \end{tikzpicture}
    \caption{\footnotesize $G_s$ for a graph $H$, edge weight function $\sigma$, and integer r, such that $rn - \bm{\sigma}$ = $3$.}
    \label{fig:G-supplier}
\end{figure}

\medskip
\noindent\textbf{\boldmath Pathwidth of $G_H$,  $\Tilde{G}_H$, and their subdivisions.}
Our reductions and compositions must use at most $\Oof(h(pw) + \log |x|)$ working space, for an input instance of size $|x|$ and parameter $pw$, and a computable function $h$.
We show that our reductions/compositions can be performed on a log-space transducer and are pl-reductions.
A log-space transducer is a type of Turing machine with a read-only input tape, a read/write work tape of logarithmic size and a write-only, write-once output tape.

\begin{lemma}\label{lem:bounded-pathwidth-G-H}
Let $(H, \mathcal{P}_H, w, r)$ be an instance of \textsc{MMO}. 
Then, there exists a log-space transducer that transforms a path decomposition of $H$ to one of~$G_H$ (resp.~$\Tilde{G}_H$, or any subdivision of~$G_H$, or any subdivision of~$\Tilde{G}_H$) with width at most $pw(H) + 6$ (resp.\ $pw(H) + 7$, or $pw(H) + 6$, or $pw(H) + 7$).
Thus,~$pw(G_H) \le pw(H) + 6$, and~$pw(\Tilde{G}_H) \le pw(H) + 7$, and any subdivision of~$G_H$ or~$\Tilde{G}_H$ results in a graph of bounded pathwidth. 
\end{lemma}

\begin{proof}
The final statement of the lemma follows from the preceding statement. 
Thus, we build log-space transducers for the graphs $G_H$, $\Tilde{G}_H$ and their subdivisions and we start with $G_H$. 

Given the path decomposition of $H$, we first ensure that it is nice (which can be done via a log-space transducer~\cite{kintali2012computing}).
Then, we pass through the bags from left to right. 
For a forget bag in the path decomposition of~$H$, we output a bag containing the representative vertices (of the vertex gadgets) of the vertices in the forget bag. 
For a bag that introduces a vertex $v \in V(H)$, we output the following bags in order:
\begin{enumerate}[itemsep=0pt]
   \item for $j \in [r+ 1]$, we output one bag that introduces the vertex $x_v^{j}$, followed by one that forgets the same vertex (these bags have a size larger than the pathwidth of $H$ by only $1$), 
   \item for each vertex $u$ in the bag such that $u \neq v$ and $uv = e \in E(H)$:
   \begin{enumerate}
       \item we output four bags that introduce the vertices $e^u$, $e^v$, $b_e^{\sigma(e)+1}$, and $a_e^{\sigma(e)+1}$, respectively, followed by two bags that forget the vertices $b_e^{\sigma(e)+1}$, and $a_e^{\sigma(e)+1}$, respectively (these bags have a size larger than the pathwidth of $H$ by only $4$),
       \item for $j \in [\sigma(e)]$, we output bags that introduce the vertices $a_e^{j}$, $b_e^{j}$, $z_e^{u(j)}$, $y_e^{u(j)}$, $z_e^{v(j)}$ and~$y_e^{v(j)}$, followed by bags that forget all those vertices (these bags have a size larger than the pathwidth of $H$ by only $6$), 
       \item then, we output two bags that forget the vertices $e^u$ and $e^v$, respectively.
   \end{enumerate}
\end{enumerate}

It is easy to verify that the result is a nice path decomposition of the graph~$G_H$. The width has increased by at most $6$.

For the graph $\Tilde{G}_H$, we change the log-space transducer for $G_H$ to output at first a bag that introduces the supplier vertex. Additionally, we then let the same log-space transducer output for each $i \in [rn - \bm{\sigma}]$, bags that represent the path decompositions of each of the donor paths (augmented by the supplier vertex). 
The log-space transducer then outputs a bag containing the supplier vertex and the vertex $d^{rn-\bm{\sigma}+1}$ followed by a bag that forgets $d^{rn-\bm{\sigma}+1}$. 
The log-space transducer finally behaves similarly to that of $G_H$ but augments each of the then outputted bags by the supplier vertex. 
It is easy to verify that the result is a nice path decomposition of the graph~$\Tilde{G}_H$. The width has increased by at most $7$.

The following claim finalizes the proof of the lemma.

\begin{claim}\label{cla:log-space-transducer}
A log-space transducer that takes as input a nice path decomposition of a graph $G$ and outputs a path decomposition of a graph $G'$ can be adapted to output a path decomposition of any subdivision of $G'$ with width $pw(G')$.   
\end{claim}
\begin{claimproof}
Note that for any edge subdivision of edges $uv_1, \ldots, uv_q$ incident to a vertex $u$ in the graph $G$, that introduce vertices $w_1, \ldots, w_q$ to the graph, the log-space transducer can be adapted to output directly before the bag that introduces $v_i$ for $i \in [q]$, a bag that introduces $w_i$ and before any bag that introduces a neighbor of $v_i$ one bag that forgets the vertex $w_i$. 
If some vertices $v_i, \ldots, v_{i'}$ have no neighbors, the log-space transducer can be adapted to output bags that introduce $w_i, \ldots, w_{i'}$ just before the vertex $u$ is forgotten and only output the bags that introduce the vertices $v_i, \ldots, v_{i'}$ after the vertex $u$ is forgotten.     
\end{claimproof}
It is easy to see then that~\Cref{cla:log-space-transducer} can be used to prove the existence of a log-space transducer that outputs a path decomposition of any subdivision of $G_H$ (resp.~$\Tilde{G}_H$) of the same width of a path decomposition of $G_H$ (resp.~$\Tilde{G}_H$).
\end{proof}

We note here that in the \textsc{DS-D} reduction and composition, we may augment subdivisions of~$G_H$ (resp.~$\Tilde{G}_H$), by an edge $dd'$ ($d$ and $d'$ are new vertices, and we refer to $d$ as the \emph{dominator vertex}) where $d$ is adjacent to various vertices in the subdivisions of~$G_H$ (resp.~$\Tilde{G}_H$). 
We denote the resulting graphs by \emph{augmented subdivisions of}~$G_H$ (resp.~$\Tilde{G}_H$). 
By Observation~\ref{obs:treewidth}, this modification can increase the pathwidth of those graphs by at most $2$.  
A log-space transducer for such modified graphs can be built by adapting one of the log-space transducers from~\Cref{lem:bounded-pathwidth-G-H} to first output bags that introduce the vertices $d$ and $d'$ and only forget them at the end of the path decomposition. 

\begin{corollary}\label{cor:aug-G-H-pathwidth}
Let $(H, \mathcal{P}_H, w, r)$ be an instance of \textsc{MMO}. 
Then, there exists a log-space transducer that transforms a path decomposition of $H$ to one of an augmented subdivision of~$G_H$ (resp.~$\Tilde{G}_H$) with width at most $pw(H) + 8$ (resp.\ $pw(H) + 9$).
\end{corollary}

\medskip
\noindent 
\textbf{MMO-instance-selector.} In our or-cross-compositions, we assume that we are given as input a family of $t$ \textsc{MMO} instances $(H_j, \mathcal{P}_{H_j}, \sigma_j, r_j)$, where for each $j \in [t]$, $H_j$ is a bounded-pathwidth graph with path decomposition $\mathcal{P}_{H_j}$, $|V(H_j)| = n$, $|E(H_j)| = m$, $\sigma_j: E(H_j) \rightarrow \mathbb{Z}_{+}$ such that $\sum_{e_j \in E(H_j)} \sigma_j(e_j) = \bm{\sigma}$ and $r_j = r \in \mathbb{Z}_{+}$ (integers are given in unary). 
It is not hard to see that these instances belong to the same equivalence class of a polynomial equivalence relation $\mathcal{R}$ (\Cref{def:equivalence-relation}) whose polynomial-time algorithm decides that two instances are equivalent if they have the same number of vertices, number of edges, and total weight on the edges.
$\mathcal{R}$ also has at most $\max_{x \in S} |x|^{\Oof(1)}$ equivalence classes, where $S$ is a set of \textsc{MMO} instances of the form $(H, \mathcal{P}_H, \sigma, r)$, where $H$ is of bounded pathwidth. In particular it has at most $m \cdot n \cdot (\max_{x \in S} |x|)$ equivalence classes.

Some of our or-cross-compositions will encode, in a graph $G_t$, all $t$ input instances of \textsc{MMO} in~$(y^*, k^*)$ (\Cref{def:or-cross-composition}) using the multiple induced subgraphs $G_{H_j}$ for $j \in [t]$.
We must also encode the OR behavior.
An \textit{instance selector} is a gadget with $t$ possible states, each corresponding to a distinct instance $x_j$ for $j \in [t]$ and compelling us to select $x_j$ so that $(x^*, k^*)$ is solved. 
We form an instance selector by constructing a new gadget, called a \textit{MMO-instance-selector}. 
In $G_t$, we make some of the \textsc{MMO}-instance-selector vertices adjacent to various vertices of $G_{H_j}$ for $j \in [t]$. 

An \textsc{MMO}-instance-selector contains, for each $j \in [t]$, an edge with endpoints $\textsc{\footnotesize Select}_j$ and $\textsc{\footnotesize Unselect}_j$.
It also contains edges $f^1g^1, f^2g^2, \ldots, f^{\bm{\sigma}} g^{\bm{\sigma}}$ and a \textit{weights-hub vertex} $h$ adjacent to each vertex $f^i$ for $i \in [\bm{\sigma}]$.
It also comprises edges $o^1p^1, o^2p^2, \ldots, o^m p^m$ and an \textit{orientations-quay vertex} $q$ adjacent to each vertex $o^i$ for $i \in [m]$.

\begin{figure}[h]
    \centering
    \begin{tikzpicture}[scale=0.7]
    \fill[yellow!10] (-2,0) -- (-4,-2) -- (-4,2) -- (-2,0) -- cycle;
    \fill[yellow!10] (-4,-2) -- (-6,-2) -- (-6,2) -- (-4,2) -- cycle;
    \fill[yellow!10] (3.2,0) -- (1.2,-2) -- (1.2,2) -- (3.2,0) -- cycle;
    \fill[yellow!10] (1.2,-2) -- (-0.7,-2) -- (-0.7,2) -- (1.2,2) -- cycle;
    \fill[yellow!10] (4.5,-0.5) -- (4.5,1.5) -- (5.5,1.5) -- (5.5,-0.5) -- cycle;
    \fill[yellow!10] (6,-0.5) -- (6,1.5) -- (7,1.5) -- (7,-0.5) -- cycle;
    \fill[yellow!10] (7.5,-0.5) -- (7.5,1.5) -- (8.5,1.5) -- (8.5,-0.5) -- cycle;
    \draw[black] (-2,0) -- (-4,-1.5);
    \draw[black] (-2,0)-- (-4,-0.75);
    \draw[black] (-2,0) -- (-4,0);
    \draw[black] (-2,0) -- (-4,1.5);
    \draw[black] (-2,0) -- (-4,0.75);
    \draw[black] (3.2,0) -- (1.2,-1.5);
    \draw[black] (3.2,0)-- (1.2,-0.5);
    \draw[black] (3.2,0) -- (1.2,0.5);
    \draw[black] (3.2,0) -- (1.2,1.5);
    \node[circle, fill=black, inner sep=1pt, draw=black, fill=white, label=below:{\scalebox{0.5}{$h$}}] at (-2,0) {};
    \node[circle, fill=black, inner sep=1pt, draw=black, fill=white, label=below:{\scalebox{0.5}{$q$}}] at (3.2,0) {};
    \draw[black] (1.2,-1.5) -- (0,-1.5) node[circle, draw=black, fill=white, inner sep=1pt, label=above:{\scalebox{0.5}{$p^1$}}, pos=1]{} node[circle, draw=black, fill=white, inner sep=1pt, label=above:{\scalebox{0.5}{$o^1$}}, pos=0]{} ;
    \draw[black] (1.2,-0.5) -- (0,-0.5) node[circle, draw=black, fill=white, inner sep=1pt, label=above:{\scalebox{0.5}{$p^2$}}, pos=1]{} node[circle, draw=black, fill=white, inner sep=1pt, label=above:{\scalebox{0.5}{$o^2$}}, pos=0]{};
    \draw[black] (1.2,0.5) -- (0,0.5) node[circle, draw=black, fill=white, inner sep=1pt, label=above:{\scalebox{0.5}{$p^3$}}, pos=1]{} node[circle, draw=black, fill=white, inner sep=1pt, label=above:{\scalebox{0.5}{$o^3$}}, pos=0]{};
    \draw[black] (1.2,1.5) -- (0,1.5) node[circle, draw=black, fill=white, inner sep=1pt, label=above:{\scalebox{0.5}{$p^4$}}, pos=1]{} node[circle, draw=black, fill=white, inner sep=1pt, label=above:{\scalebox{0.5}{$o^4$}}, pos=0]{};
    \draw[black] (5,0) -- (5,1) node[circle, draw=black, fill=white, inner sep=1pt, label=above:{\scalebox{0.5}{\textsc{Unselect}$_1$}}, pos=1]{} node[circle, draw=black, fill=white, inner sep=1pt, label=below:{\scalebox{0.5}{\textsc{Select}$_1$}}, pos=0]{} ;
    \draw[black] (6.5,0) -- (6.5,1) node[circle, draw=black, fill=white, inner sep=1pt, label=above:{\scalebox{0.5}{\textsc{Unselect}$_2$}}, pos=1]{} node[circle, draw=black, fill=white, inner sep=1pt, label=below:{\scalebox{0.5}{\textsc{Select}$_2$}}, pos=0]{} ;
    \draw[black] (8,0) -- (8,1) node[circle, draw=black, fill=white, inner sep=1pt, label=above:{\scalebox{0.5}{\textsc{Unselect}$_3$}}, pos=1]{} node[circle, draw=black, fill=white, inner sep=1pt, label=below:{\scalebox{0.5}{\textsc{Select}$_3$}}, pos=0]{} ;
    \draw[black] (-5.2,-1.5) -- (-4,-1.5) node[circle, draw=black, fill=white, inner sep=1pt, label=above:{\scalebox{0.5}{$f^1$}}, pos=1]{} node[circle, draw=black, fill=white, inner sep=1pt, label=above:{\scalebox{0.5}{$g^1$}}, pos=0]{} ;
    \draw[black] (-5.2,-0.75) -- (-4,-0.75) node[circle, draw=black, fill=white, inner sep=1pt, label=above:{\scalebox{0.5}{$f^2$}}, pos=1]{} node[circle, draw=black, fill=white, inner sep=1pt, label=above:{\scalebox{0.5}{$g^2$}}, pos=0]{};
    \draw[black] (-5.2,0) -- (-4,0) node[circle, draw=black, fill=white, inner sep=1pt, label=above:{\scalebox{0.5}{$f^3$}}, pos=1]{} node[circle, draw=black, fill=white, inner sep=1pt, label=above:{\scalebox{0.5}{$g^3$}}, pos=0]{};
    \draw[black] (-5.2,1.5) -- (-4,1.5) node[circle, draw=black, fill=white, inner sep=1pt, label=above:{\scalebox{0.5}{$f^5$}}, pos=1]{} node[circle, draw=black, fill=white, inner sep=1pt, label=above:{\scalebox{0.5}{$g^5$}}, pos=0]{};
    \draw[black] (-5.2,0.75) -- (-4,0.75) node[circle, draw=black, fill=white, inner sep=1pt, label=above:{\scalebox{0.5}{$f^4$}}, pos=1]{} node[circle, draw=black, fill=white, inner sep=1pt, label=above:{\scalebox{0.5}{$g^4$}}, pos=0]{};
    \end{tikzpicture}
    \caption{\footnotesize An \textsc{MMO}-instance-selector for $t = 3$ \textsc{MMO} instances $(H_j, \mathcal{P}_{H_j}, \sigma_j, r_j)$ with $|E(H_j)| = m = 4$, edge weight function $\sigma_j$ such that $\sum_{e_j \in E(H_j)} \sigma_j(e_j) = \bm{\sigma} = 5$, and integers $r_j = r \in \mathbb{Z}_{+}$.}
    \label{fig:mmo-instance-selector}
\end{figure}
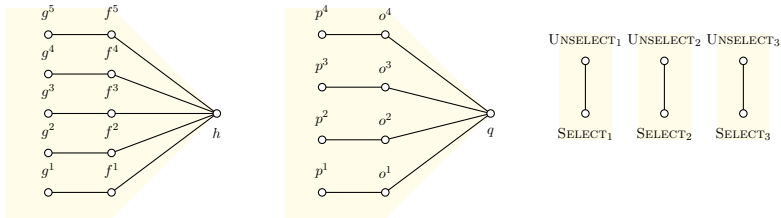
In $G_t$, we will make the vertices $h$ and $q$ adjacent to different vertices in the $t$ induced subgraphs~$G_{H_j}$ for $j \in [t]$. 
Additionally, the vertices $\textsc{\footnotesize Select}_j$ or $\textsc{\footnotesize Unselect}_j$ for $j \in [t]$ will also be adjacent to different vertices within their corresponding induced subgraph $G_{H_j}$.
For some source problems, we may additionally augment subdivisions of $G_t$ by attaching to $h$ and $q$ a number of pendant vertices. 
We denote the resulting graphs by \emph{augmented subdivisions of} $G_t$.

\begin{lemma}\label{lem:bounded-pathwidth-G}
There exists a log-space transducer that given $t$ \textsc{MMO} instances $(H_j, \mathcal{P}_{H_j}, \sigma_j, r_j)$, where for each $j \in [t]$, $|V(H_j)| = n$, $|E(H_j)| = m$, $\sigma_j: E(H_j) \rightarrow \mathbb{Z}_{>0}$ is such that $\sum_{e_j \in E(H_j)} \sigma_j(e_j) = \bm{\sigma}$, and $r_j = r \in \mathbb{Z}_{+}$ (integers are given in unary), outputs a path decomposition of the graph~$G_t$ (resp.\ any augmented subdivision of~$G_t$) with width at most $\max_{j \in [t]} pw(H_j) + 10$.
Thus, the graph~$G_t$ and any augmented subdivision of it are graphs of bounded pathwidth.
\end{lemma}

\begin{proof}
The last statement of the lemma follows from the preceding statement. 
Thus, we build log-space transducers for the graphs $G_t$ and its augmented subdivisions and we start with a log-space transducer for (an augmented) $G_t$.
An augmented $G_t$ is $G_t$ with a number of pendant vertices attached to $h$ and $q$.
We first ensure that the path decompositions of $H_j$ for each $j \in [t]$ are nice (which can be done via a log-space transducer). 
Afterwards, the log-space transducer outputs bags that introduce the vertices $h$ and $q$, introduce and forget the pendant vertices attached to~$h$ or $q$ in the case of an augmented $G_t$ one-by one, and then outputs bags that represent the path decompositions of the edges $f^1g^1, f^2g^2, \ldots, f^{\bm{\sigma}} g^{\bm{\sigma}}$, $o^1p^1, o^2p^2, \ldots, o^m p^m$, augmented by the vertices~$h$ and $q$. 
Next, the log-space transducer outputs for each $j \in t$, the bags in the path decomposition of the graph $G_{H_j}$ (\ie, behaves as the log-space transducer of~\Cref{lem:bounded-pathwidth-G-H} that outputs a path decomposition of the graph $G_H$) but augmented by the vertices $h$, $q$, $\textsc{\footnotesize Select}_j$ and $\textsc{\footnotesize Unselect}_j$, followed by bags that forget the vertices $\textsc{\footnotesize Select}_j$ and $\textsc{\footnotesize Unselect}_j$.
It is easy to see that the result is a path decomposition of (an augmented) $G_t$ with width $\max_{j \in [t]} pw(H_j) + 10$ (as the path decomposition of $G_{H_j}$ for any $j \in [t]$ has width at most $\max_{j \in [t]} pw(H_j) + 6$). 

Using~\Cref{cla:log-space-transducer}, we can build a log-space transducer for any augmented subdivision of $G_t$ that outputs a path decomposition of the subdivision with width $\max_{j \in [t]} pw(H_j) + 10$.
\end{proof}

In the \textsc{DS-D} composition, we form the graph $\hat{G}_t$ in a manner akin to a subdivision of $G_t$ except, we encode each of the $t$ input instances of \textsc{MMO} using the multiple induced subgraphs that are augmented subdivisions of $G_{H_j}$ for $j \in [t]$.
Using \Cref{cor:aug-G-H-pathwidth} instead of \Cref{lem:bounded-pathwidth-G-H} in the proof of \Cref{lem:bounded-pathwidth-G}, the log-space transducer can output for each $j \in [t]$, the bags in the path decomposition of an augmented subdivision of $G_{H_j}$, and we get the following.

\begin{corollary}\label{cor:cor-bounded-pathwidth-G'-t}
There exists a log-space transducer that given $t$ \textsc{MMO} instances $(H_j, \mathcal{P}_{H_j}, \sigma_j, r_j)$, where for each $j \in [t]$, $|V(H_j)| = n$, $|E(H_j)| = m$, $\sigma_j: E(H_j) \rightarrow \mathbb{Z}_{+}$ is such that $\sum_{e_j \in E(H_j)} \sigma_j(e_j) =~\bm{\sigma}$, and $r_j = r \in \mathbb{Z}_{+}$ (integers are given in unary), outputs a path decomposition of the graph~$\hat{G}_t$ (resp.\ any subdivision of $\hat{G}_t)$ with width at most $\max_{j \in [t]} pw(H_j) + 12$.
Thus, the graph~$\hat{G}_t$ and any subdivision of it are graphs of bounded pathwidth.    
\end{corollary}

\begin{corollary}\label{cor:reduction-space}
Given an \textsc{MMO} instance $(H, \mathcal{P}_H, \sigma, r)$, one can build a log-space transducer that outputs a path decomposition of (any subdivision of) $G_H$, (resp.\ (any subdivision of)~$\Tilde{G}_H$, or an augmented subdivision of $G_H$, or an augmented subdivision of~$\Tilde{G}_H$) with width at most $pw(H) + 6$ (resp.\ $pw(H) + 7$, or $pw(H) + 8$, or $pw(H) + 9$), along with a representation of the graph, any subset of its vertices, and an integer with at most a polynomial (in the input size) number of bits. 
\end{corollary}

\begin{corollary}\label{cor:composition-space}
Given $t$ \textsc{MMO} instances $(H_j, \mathcal{P}_{H_j}, \sigma_j, r_j)$, where for each $j \in [t]$, $|V(H_j)| = n$, $|E(H_j)| = m$, $\sigma_j$ is such that $\sum_{e_j \in E(H_j)} \sigma_j(e_j) = \bm{\sigma}$ and $r_j = r \in \mathbb{Z}_{+}$ (integers are given in unary), one can build a log-space transducer that outputs a path decomposition (of an augmented subdivision) of the graph~$G_t$ (resp.\ a path decomposition of the graph~$\hat{G}_t$), with width at most $\max_{j \in [t]} pw(H_j) + 10$ (resp.\ $\max_{j \in [t]} pw(H_j) + 12$), along with a representation of the graph, any subset of its vertices, and an integer with at most a polynomial (in the input size) number of bits. 
\end{corollary}

\section{Vertex Cover Discovery}
\label{sec:vc}
Fellows et al.~\cite{fellows2023solution} showed that \textsc{VC-D} is in $\FPT$ with respect to parameter $k$ on general graphs and in $\FPT$ with respect to parameter $\budget$ on nowhere dense classes of graphs.
We show in this section that the problem has a polynomial kernel with respect to parameter $k$.
With respect to the parameter~$\budget + pw$, where $pw$ is the pathwidth of the input graph, we show that the problem does not have a polynomial kernel unless $\NP \subseteq \cp$.

\begin{theorem}\label{thm:VC-parameter-k}
\textsc{VC-D} has a kernel of size $\Oof(k^2)$. 
\end{theorem}
\begin{proof}
Let $(G, S, \budget)$ be an instance of \textsc{VC-D}.
Let $G'$ be the graph obtained from $G$ by deleting the vertices of degree greater than $k$. If $G'$ has more than $k^2$ edges or more than $2k^2$ non-isolated vertices we can reject the instance. 
The vertices of degree greater than $k$ must be in any vertex cover of size at most $k$. 
The remaining vertices of the vertex cover can cover at most $k$ edges each, as we have at most $k$ vertices for the vertex cover, there can be at most $k^2$ edges left. These have at most $2k^2$ endpoints. 

We now construct the kernel $(H,S,b)$. 
We define $H$ to be the graph obtained from $G$ as follows: We keep all non-isolated vertices of $G'$ as well as all vertices that contain a token. Furthermore, we keep all vertices of $G$ with degree greater than $k$ (but not all their neighbors). 
For all $u,v$ with degree greater than $k$ in $G$, if $N_G(u)\cap N_G(v)$ contains only isolated vertices in $G'$, then we keep one arbitrary vertex of this intersection and name it $x_{uv}$. 
These vertices need to be kept to ensure that all discovery sequences survive in $H$. 
Finally, for every vertex $u$ of degree greater than $k$ in $G$, if $u$ has degree $d<k+1$ in $H$, we add arbitrary $k+1-d$ isolated neighbors to $u$. 

We claim that $(H,S,\budget)$ is equivalent to $(G,S,\budget)$ and has at most $3k^2+2k$ vertices: 
at most~$k$ vertices with degree at least $k+1$ (yielding $k(k+1)$ vertices), at most 
$2k^2$ non-isolated vertices from $G'$ and at most~$k$ isolated vertices with tokens on them. 

It remains to show that the instances  are equivalent. 
Assume $(G,S,\budget)$ is a yes-instance. 
Consider a shortest discovery sequence $S=C_0\vdash C_1\vdash\ldots \vdash C_\ell$ for $\ell\leq \budget$. 
We claim that there exists a discovery sequence $S=C_0\vdash C_1'\vdash \ldots \vdash C_{\ell-1}'\vdash C_\ell$ in $H$ of the same length that ends in the same configuration $C_\ell$ and that also constitutes a vertex cover in $H$. 
First observe that because we consider a shortest discovery sequence all $C_i$ do not contain isolated vertices of $G'$ unless they belong to $S$, or to $N_G(u)\cap N_G(v)$ for vertices $u,v$ with degree greater than $k$ in $G$. Furthermore,~$C_\ell$ contains no isolated vertices of $G'$ unless they belong to $S$. 
Now every slide along a vertex $x$ of $N(u)\cap N(v)$ that does not belong to $H$ can be replaced by the slide along $x_{uv}$, that is,  either $C_i'=C_i$ or $C_i'=(C_i\setminus \{x\})\cup \{x_{uv}\}$. As $C_\ell$ is a vertex of $G$ and $H$ is a subgraph of $G$, also $C_\ell$ is a vertex cover of $H$. 

Conversely, assume $(H,S,\budget)$ is a yes-instance. 
As $H$ is a subgraph of $G$, every discovery sequence is also a discovery sequence in $G$. 
It remains to show that every vertex cover of size $k$ of~$H$ is also a vertex cover of $G$. 
This easily follows from the fact that every vertex of degree greater than $k$ in $G$ is also a vertex of degree greater than $k$ in $H$ and every vertex cover of $H$ must contain all vertices of degree greater than $k$. 
This implies that all edges between high degree vertices and isolated vertices in $G'$ are covered in $G$. 
All other edges appear also in $H$ and are hence covered in~$H$ as well as in~$G$. 
\end{proof}

\begin{theorem}\label{thm:VC-D-pathwidth}
\textsc{VC-D} is \XNLP-hard parameterized by pathwidth.
\end{theorem}

As stated in \Cref{sec:foundational}, we present a pl-reduction from \textsc{MMO}. 
Let $(H, \mathcal{P}_H, \sigma, r)$ be an instance of \textsc{MMO} where $H$ is a bounded pathwidth graph, $|V(H)| = n$, $|E(H)| = m$, $\sigma: E(H) \rightarrow \mathbb{Z}_+$ such that $\sum_{e \in E(H)} \sigma(e) = \bm{\sigma}$ and $r \in \mathbb{Z}_+$ (integers are given in unary).
We construct an instance $(\Tilde{G}_H, \mathcal{P}_{\Tilde{G}_H}, S, \budget)$ of \textsc{VC-D} as follows.
We form the graph $\Tilde{G}_H$ as outlined below (see \Cref{fig:VC-pathwidth-reduction}):

\begin{enumerate}[itemsep=0pt, label=(\alph*)]
    \item We subdivide each edge $a^ib^i$ for each $i \in [\sigma(e)]$ for each edge $e \in E(H)$, of a subgraph $G_e$ (which is the \textsc{MMO}-edge-$e$ described in \Cref{sec:foundational}), and add it to $\Tilde{G}_H$.
    We denote the introduced vertex (from a subdivision of an edge $a^ib^i$) by $c^i$. 
    We let $C_e = \cup_{i \in [w(e)]} \text{ } \text{ } c^i$, $C = \cup_{e \in E(H)} C_e$.
    \item We subdivide each edge $w_v x^{v(i)}$ for $i \in [r]$, of a subgraph $G_v$ (which is the \textsc{MMO}-vertex-$v$ described in \Cref{sec:foundational}), for each vertex $v \in V(H)$, and add it to $\Tilde{G}_H$.
    We denote the introduced vertex (from a subdivision of an edge $w_v x^{v(i)}$) by $c(x^{v(i)})$. 
    We let $C(X_v) = \cup_{i \in [r]} \text{ } \text{ } c(x^{v(i)})$, $C(X) = \cup_{v \in V(H)} C(X_v)$.
    \item As described in \Cref{sec:foundational} (under the graph $G_H$ heading), we make each vertex $b^i_e$, for each edge $uv=e \in E(H)$, adjacent to the vertices in $Z_e^{v}$ and $Z_e^u$, and each vertex $e^v$ for each edge $uv=e \in E(H)$ adjacent to all vertices in $Y_e^v$.
    \item We let the supplier gadget be $G_s$ as described in \Cref{sec:foundational}, add it to $\Tilde{G}_H$, and make the vertex $s$ adjacent to all vertices in $X$.
\end{enumerate}

\begin{figure}
    \centering
    \begin{tikzpicture}

    \fill[blue!10] (-0.75,2.75) -- (0.75,2.75) -- (0.75,-0.25) -- (-0.75,-0.25) -- cycle; 
    \fill[blue!10] (0.5,0.1) -- (1.25,0.9) -- (1.25,0) -- (0.5,0) -- cycle; 
    \fill[blue!10] (0.5,-0.1) -- (1.25,-0.9) -- (1.25,0) -- (0.5, 0) -- cycle; 
    \fill[blue!10] (-0.75,-2) -- (0.75,-2) -- (0.75,-3.25) -- (-0.75,-3.25) -- cycle; 
    \fill[blue!10] (0.5,-2.9) -- (1.25,-2) -- (1.25,-3) -- (0.5,-3) -- cycle; 
    \fill[blue!10] (0.5,-3.1) -- (1.25,-4) -- (1.25,-3) -- (0.5, -3) -- cycle; 
    \fill[green!10] (8,1) -- (7.5,2.5) -- (8.5,2.5) -- (8,1) -- cycle;
    \fill[green!10] (8,1) -- (6,0) -- (6,3) -- (8,1) -- cycle;
    \fill[green!10] (8,1) -- (10,-0.25) -- (10,2.5) -- (8,1) -- cycle;
    \fill[green!10] (6,0) -- (6,3) -- (4.75,3) -- (4.75,0) -- cycle;
    \fill[green!10] (8,-1.5) -- (7.5,0) -- (8.5,0) -- (8,-1.5) -- cycle;
    \fill[green!10] (8,-1.5) -- (6,-1.75) -- (6,-1.25) -- (8,-1.5) -- cycle;
    \fill[green!10] (8,-1.5) -- (10,0) -- (10,-3) -- (8,-1.5) -- cycle;
    \fill[green!10] (6,-1.75) -- (6,-1.25) -- (4.75,-1.25) -- (4.75,-1.75) -- cycle;
    \fill[green!10] (8,-4.7) -- (7.5,-3.2) -- (8.5,-3.2) -- (8,-4.7) -- cycle;
    \fill[green!10] (8,-4.7) -- (6,-6.25) -- (6,-2.75) -- (8,-4.7) -- cycle;
    \fill[green!10] (8,-4.7) -- (10,-2.5) -- (10,-5.75) -- (8,-4.7) -- cycle;
    \fill[green!10] (6,-6.25) -- (6,-2.75) -- (4.75,-2.75) -- (4.75,-6.25) -- cycle;
    \draw[black] (8,1) -- (9.5,0.25);
    \draw[black] (8,1) -- (9.5,1);
    \draw[black] (8,1) -- (9.5,1.75);
    \draw[black] (8,1) -- (8,2);
    \draw[black] (8,1) -- (6,0.25);
    \draw[black] (8,1) -- (6,2.25);
    \draw[black] (8,1) -- (6,1.25);
    \draw[black] (8,-1.5) -- (9.5,-2.5);
    \draw[black] (8,-1.5) -- (9.5,-1.5);
    \draw[black] (8,-1.5) -- (9.5,-0.5);
    \draw[black] (8,-1.5) -- (6,-1.5);
    \draw[black] (8,-1.5) -- (8,-0.5);
    \draw[black] (8,-4.7) -- (9.5,-5.2);
    \draw[black] (8,-4.7) -- (9.5,-4.2);
    \draw[black] (8,-4.7) -- (9.5,-3.2);
    \draw[black] (8,-4.7) -- (8,-3.75);
    \draw[black] (8,-4.7)-- (6,-6);
    \draw[black] (8,-4.7) -- (6,-5.2);
    \draw[black] (8,-4.7)-- (6,-4.2);
    \draw[black] (8,-4.7) -- (6,-3.4);
    \draw[yellow] (1,-2.5) .. controls (2,-2.25) .. (6,-1.5);
    \draw[yellow] (1,-3.5) .. controls (2.5,-5) and  (5,-5) .. (6,-6);    
    \draw[red] (0.5,-2.25) .. controls (1.5,-1.25) .. (5,-1.5);
    \draw[red] (0.5,-2.25) .. controls (2,-4.7) and  (5,-6) .. (5,-6);  
    \draw[red] (0.5,0.75) .. controls (1,1) .. (5,0.25);
    \draw[red] (0.5,1.5) .. controls (1,2.5) and  (2,2) .. (5,1.25);
    \draw[red] (0.5,2.25) .. controls (1,3) and  (2,3.5) .. (5,2.25);  
    \draw[blue] (5,-10) .. controls (8,-9) .. (9.5,-5.2);
    \draw[blue] (5,-10) .. controls (10,-9.25) and (11,-5) .. (9.5,-4.2);
    \draw[blue] (5,-10) .. controls (10.25,-10.25) and (11.25,-5) .. (9.5,-3.2);
    \fill[red!10] (5,-10) -- (4.5,-9) -- (5.5,-9) -- (5,-10) -- cycle;
    \fill[red!10] (5,-10) -- (3,-8) -- (3,-12) -- (5,-10) -- cycle;
    \fill[red!10] (3,-8) -- (1,-8) -- (1,-12) -- (3,-12) -- cycle;
    \node[circle, fill=black, inner sep=1pt, draw=black, fill=black, label=below:{\scalebox{0.5}{$w_u$}}] at (8,1) {};
    \node[circle, fill=black, inner sep=1pt, draw=black, fill=black, label=below:{\scalebox{0.5}{$w_w$}}] at (8,-1.5) {};
    \node[circle, fill=black, inner sep=1pt, draw=black, fill=black, label=below:{\scalebox{0.5}{$w_v$}}] at (8,-4.7) {};
    \node[circle, fill=black, inner sep=1pt, draw=black, fill=white, label=right:{\scalebox{0.5}{$x_u^1$}}] at (9.5,0.25) {};
    \node[circle, fill=black, inner sep=1pt, draw=black, fill=white, label=above:{\scalebox{0.5}{$c(x_u^1)$}}] at (9,0.5) {};
    \node[circle, fill=black, inner sep=1pt, draw=black, fill=white, label=right:{\scalebox{0.5}{$x_u^2$}}] at (9.5,1) {};
    \node[circle, fill=black, inner sep=1pt, draw=black, fill=white, label=above:{\scalebox{0.5}{$c(x_u^2)$}}] at (9,1) {};
    \node[circle, fill=black, inner sep=1pt, draw=black, fill=white, label=right:{\scalebox{0.5}{$x_u^3$}}] at (9.5,1.75) {};
    \node[circle, fill=black, inner sep=1pt, draw=black, fill=white, label=above:{\scalebox{0.5}{$c(x_u^3)$}}] at (9,1.5) {};
    \node[circle, fill=black, inner sep=1pt, draw=black, fill=white, label=right:{\scalebox{0.5}{$x_w^1$}}] at (9.5,-2.5) {};
    \node[circle, fill=black, inner sep=1pt, draw=black, fill=white, label=above:{\scalebox{0.5}{$c(x_w^1)$}}] at (9,-2.15) {};
    \node[circle, fill=black, inner sep=1pt, draw=black, fill=white, label=right:{\scalebox{0.5}{$x_w^2$}}] at (9.5,-1.5) {};
    \node[circle, fill=black, inner sep=1pt, draw=black, fill=white, label=above:{\scalebox{0.5}{$c(x_w^2)$}}] at (9,-1.5) {};
    \node[circle, fill=black, inner sep=1pt, draw=black, fill=white, label=right:{\scalebox{0.5}{$x_w^3$}}] at (9.5,-0.5) {};
    \node[circle, fill=black, inner sep=1pt, draw=black, fill=white, label=above:{\scalebox{0.5}{$c(x_w^3)$}}] at (9,-0.8) {};
    \node[circle, fill=black, inner sep=1pt, draw=black, fill=white, label=right:{\scalebox{0.5}{$x_v^1$}}] at (9.5,-5.2) {};
    \node[circle, fill=black, inner sep=1pt, draw=black, fill=white, label=above:{\scalebox{0.5}{$c(x_v^1)$}}] at (8.9,-5) {};
    \node[circle, fill=black, inner sep=1pt, draw=black, fill=white, label=right:{\scalebox{0.5}{$x_v^2$}}] at (9.5,-4.2) {};
    \node[circle, fill=black, inner sep=1pt, draw=black, fill=white, label=above:{\scalebox{0.5}{$c(x_v^2)$}}] at (8.9,-4.4) {};
    \node[circle, fill=black, inner sep=1pt, draw=black, fill=white, label=right:{\scalebox{0.5}{$x_v^3$}}] at (9.5,-3.2) {};
    \node[circle, fill=black, inner sep=1pt, draw=black, fill=white, label=above:{\scalebox{0.5}{$c(x_v^3)$}}] at (8.9,-3.8) {};
    \node[circle, fill=black, inner sep=1pt, draw=black, fill=white, label=above right:{\scalebox{0.5}{$x_u^{4}$}}] at (8,2) {};
    \node[circle, fill=black, inner sep=1pt, draw=black, fill=white, label=above right:{\scalebox{0.5}{$x_w^{4}$}}] at (8,-0.5) {};
    \node[circle, fill=black, inner sep=1pt, draw=black, fill=white, label=above right:{\scalebox{0.5}{$x_v^{4}$}}] at (8,-3.75) {};
    \draw[black] (5,-1.5) -- (6,-1.5) node[circle, draw=black, fill=black, inner sep=1pt, label=above:{\scalebox{0.5}{$y_{e_2}^{w(1)}$}}, pos=1]{} node[circle, draw=black, fill=white, inner sep=1pt, label=above:{\scalebox{0.5}{$z_{e_2}^{w(1)}$}}, pos=0]{};
    \draw[black] (5,1.25) -- (6,1.25) node[circle, draw=black, fill=black, inner sep=1pt, label=above:{\scalebox{0.5}{$y_{e_1}^{u(2)}$}}, pos=1]{} node[circle, draw=black, fill=white, inner sep=1pt, label=above:{\scalebox{0.5}{$z_{e_1}^{u(2)}$}}, pos=0]{};
    \draw[black] (5,0.25) -- (6,0.25) node[circle, draw=black, fill=black, inner sep=1pt, label=above:{\scalebox{0.5}{$y_{e_1}^{u(1)}$}}, pos=1]{}  node[circle, draw=black, fill=white, inner sep=1pt, label=above:{\scalebox{0.5}{$z_{e_1}^{u(1)}$}}, pos=0]{};
    \draw[black] (5,2.25) -- (6,2.25) node[circle, draw=black, fill=black, inner sep=1pt, label=above:{\scalebox{0.5}{$y_{e_1}^{u(3)}$}}, pos=1]{} node[circle, draw=black, fill=white, inner sep=1pt, label=above:{\scalebox{0.5}{$z_{e_1}^{v(3)}$}}, pos=0]{};
    \draw[black] (5,-3.4) -- (6,-3.4) node[circle, draw=black, fill=black, inner sep=1pt, label=above:{\scalebox{0.5}{$y_{e_1}^{v(2)}$}}, pos=1]{} node[circle, draw=black, fill=white, inner sep=1pt, label=above:{\scalebox{0.5}{$z_{e_1}^{v(2)}$}}, pos=0]{};
    \draw[black] (5,-4.2) -- (6,-4.2) node[circle, draw=black, fill=black, inner sep=1pt, label=above:{\scalebox{0.5}{$y_{e_1}^{v(3)}$}}, pos=1]{} node[circle, draw=black, fill=white, inner sep=1pt, label=above:{\scalebox{0.5}{$z_{e_1}^{v(3)}$}}, pos=0]{};
    \draw[black] (5,-5.2) -- (6,-5.2) node[circle, draw=black, fill=black, inner sep=1pt, label=above:{\scalebox{0.5}{$y_{e_1}^{v(1)}$}}, pos=1]{} node[circle, draw=black, fill=white, inner sep=1pt, label=above:{\scalebox{0.5}{$z_{e_1}^{v(1)}$}}, pos=0]{};
    \draw[black] (5,-6) -- (6,-6) node[circle, draw=black, fill=black, inner sep=1pt, label=above:{\scalebox{0.5}{$y_{e_2}^{v(1)}$}}, pos=1]{} node[circle, draw=black, fill=white, inner sep=1pt, label=above:{\scalebox{0.5}{$z_{e_2}^{v(1)}$}}, pos=0]{};
    \draw[black] (-0.5,0.75) -- (0.5,0.75) node[circle, draw=black, fill=white, inner sep=1pt, label=above:{\scalebox{0.5}{$a^1_{e_1}$}}, pos=0]{} node[circle, draw=black, fill=black, inner sep=1pt, label=above:{\scalebox{0.5}{$b^1_{e_1}$}}, pos=1]{};
    \node[draw, circle, fill=black, inner sep=1pt, label=above:{\scalebox{0.5}{$c^1_{e_1}$}}] at (0,0.75) {};
    \draw[black] (-0.5,1.5) -- (0.5,1.5) node[circle, draw=black, fill=white, inner sep=1pt, label=above:{\scalebox{0.5}{$a^2_{e_1}$}}, pos=0]{} node[circle, draw=black, fill=black, inner sep=1pt, label=above:{\scalebox{0.5}{$b^2_{e_1}$}}, pos=1]{};
    \node[draw, circle, fill=black, inner sep=1pt, label=above:{\scalebox{0.5}{$c^2_{e_1}$}}] at (0,1.5) {};
    \draw[black] (-0.5,2.25) -- (0.5,2.25) node[circle, draw=black, fill=white, inner sep=1pt, label=above:{\scalebox{0.5}{$a^3_{e_1}$}}, pos=0]{} node[circle, draw=black, fill=black, inner sep=1pt, label=above:{\scalebox{0.5}{$b^3_{e_1}$}}, pos=1]{};
    \node[draw, circle, fill=black, inner sep=1pt, label=above:{\scalebox{0.5}{$c^3_{e_1}$}}] at (0,2.25) {};
    \draw[black] (1,0.5) -- (1, -0.5);
    \draw[black] (0.5,0) -- (1,0.5) node[circle, draw=black, fill=black, inner sep=1pt, label=right:{\scalebox{0.5}{$e_1^v$}}, pos=1]{};
    \draw[black] (0.5,0) -- (1,-0.5) node[circle, draw=black, fill=black, inner sep=1pt, label=right:{\scalebox{0.5}{$e_1^u$}}, pos=1]{};
    \draw[black] (-0.5,0) -- (0.5,0) node[circle, draw=black, fill=white, inner sep=1pt, label=above:{\scalebox{0.5}{$a^4_{e_1}$}}, pos=0]{} node[circle, draw=black, fill=white, inner sep=1pt, label=above:{\scalebox{0.5}{$b^4_{e_1}$}}, pos=1]{};
    \draw[black] (-0.5,-2.25) -- (0.5,-2.25) node[circle, draw=black, fill=white, inner sep=1pt, label=above:{\scalebox{0.5}{$a^1_{e_2}$}}, pos=0]{} node[circle, draw=black, fill=black, inner sep=1pt, label=above:{\scalebox{0.5}{$b^1_{e_2}$}}, pos=1]{};
    \node[draw, circle, fill=black, inner sep=1pt, label=above:{\scalebox{0.5}{$c^1_{e_2}$}}] at (0,-2.25) {};
    \draw[black] (1,-2.5) -- (1, -3.5);
    \draw[black] (0.5,-3) -- (1,-2.5) node[circle, draw=black, fill=black, inner sep=1pt, label=right:{\scalebox{0.5}{$e_2^w$}}, pos=1]{};
    \draw[black] (0.5,-3) -- (1,-3.5) node[circle, draw=black, fill=black, inner sep=1pt, label=right:{\scalebox{0.5}{$e_2^u$}}, pos=1]{};
    \draw[black] (-0.5,-3) -- (0.5,-3) node[circle, draw=black, fill=white, inner sep=1pt, label={above:{\scalebox{0.5}{$a^2_{e_2}$}}}, pos=0]{} node[circle, draw=black, fill=white, inner sep=1pt, label=above:{\scalebox{0.5}{$b^2_{e_2}$}}, pos=1]{};
    \draw[black] (5,-10) -- (5,-9.5);
    \draw[black] (5,-10)-- (3,-11.5);
    \draw[black] (5,-10)-- (3,-10.75);
    \draw[black] (5,-10) -- (3,-10);
    \draw[black] (5,-10) -- (3,-9.25);
    \draw[black] (5,-10) -- (3,-8.5);
    \draw[black] (2.2,-11.5) -- (3,-11.5);
    \draw[black] (2.2,-10.75) -- (3,-10.75);
    \draw[black] (2.2,-10) -- (3,-10);
    \draw[black] (2.2,-9.25) -- (3,-9.25); 
    \draw[black] (2.2,-8.5) -- (3,-8.5); 
    \draw[black] (1.4,-11.5) -- (2.2,-11.5);
    \draw[black] (1.4,-10.75) -- (2.2,-10.75);
    \draw[black] (1.4,-10) -- (2.2,-10);
    \draw[black] (1.4,-9.25) -- (2.2,-9.25); 
    \draw[black] (1.4,-8.5) -- (2.2,-8.5);
    \node[circle, fill=black, inner sep=1pt, draw=black, fill=black, label=below right:{\scalebox{0.5}{$s$}}] at (5,-10) {};
    \node[circle, fill=black, inner sep=1pt, draw=black, fill=white, label=above right:{\scalebox{0.5}{$d_1^6$}}] at (5,-9.5) {}; 
    \node[circle, draw=black, fill=black, inner sep=1pt, label=above:{\scalebox{0.5}{$d^1_1$}}] at (3,-11.5) {};
    \node[circle, draw=black, fill=white, inner sep=1pt, label=above:{\scalebox{0.5}{$d^1_2$}}] at (2.2,-11.5) {};
    \node[circle, draw=black, fill=black, inner sep=1pt, label=above:{\scalebox{0.5}{$d^1_3$}}] at (1.4,-11.5) {};
    \node[circle, draw=black, fill=white, inner sep=1pt, label=above:{\scalebox{0.5}{$d^2_2$}}] at (2.2,-10.75) {};
    \node[circle, draw=black, fill=black, inner sep=1pt, label=above:{\scalebox{0.5}{$d^2_1$}}] at (3,-10.75) {};
    \node[circle, draw=black, fill=white, inner sep=1pt, label=above:{\scalebox{0.5}{$d^3_2$}}] at (2.2,-10) {};
    \node[circle, draw=black, fill=black, inner sep=1pt, label=above:{\scalebox{0.5}{$d^3_1$}}] at (3,-10) {};
    \node[circle, draw=black, fill=white, inner sep=1pt, label=above:{\scalebox{0.5}{$d^4_2$}}] at (2.2,-9.25) {};
    \node[circle, draw=black, fill=black, inner sep=1pt, label=above:{\scalebox{0.5}{$d^4_1$}}] at (3,-9.25) {};
    \node[circle, draw=black, fill=black, inner sep=1pt, label=above:{\scalebox{0.5}{$d^2_3$}}] at (1.4,-10.75) {};
    \node[circle, draw=black, fill=black, inner sep=1pt, label=above:{\scalebox{0.5}{$d^3_3$}}] at (1.4,-10) {};
    \node[circle, draw=black, fill=black, inner sep=1pt, label=above:{\scalebox{0.5}{$d^4_3$}}] at (1.4,-9.25) {};

    \node[circle, draw=black, fill=black, inner sep=1pt, label=above:{\scalebox{0.5}{$d^5_1$}}] at (3,-8.5) {};
    \node[circle, draw=black, fill=white, inner sep=1pt, label=above:{\scalebox{0.5}{$d^5_2$}}] at (2.2,-8.5) {};
    \node[circle, draw=black, fill=black, inner sep=1pt, label=above:{\scalebox{0.5}{$d^5_3$}}] at (1.4,-8.5) {};
    \end{tikzpicture}
    \caption{\footnotesize Parts of the graph $\Tilde{G}_H$ constructed by the reduction of Theorem~\ref{thm:VC-D-pathwidth} given an instance $(H, \mathcal{P}_H, \sigma, r)$, where~$H$ has three vertices $u, v$, and $w$, and two edges $e_1=uv$ and $e_2=uw$, and $r = 3$. Additionally, $\sigma(e_1) = 3$ and $\sigma(e_2) = 1$. For clarity, the edges between the vertices in $B_{e_1}$ and $Z_{e_1}^u$ are missing. The same applies for the edges between $e_1^v$ and the vertices of $Y_{e_1}^v$, the edges between $e_1^u$ and the vertices of $Y_{e_1}^u$ as well as some of the edges between~$s$ and the vertices in $X$. Red, yellow, and blue edges are used to highlight the different types of edges used to connect the subgraphs $G_{e_1}, G_{e_2}, G_{u}, G_{v}, G_{w}$, and $G_s$ of $\Tilde{G}_H$, vertices in black are in $S$ and those in white are not.}
    \label{fig:VC-pathwidth-reduction}
\end{figure}
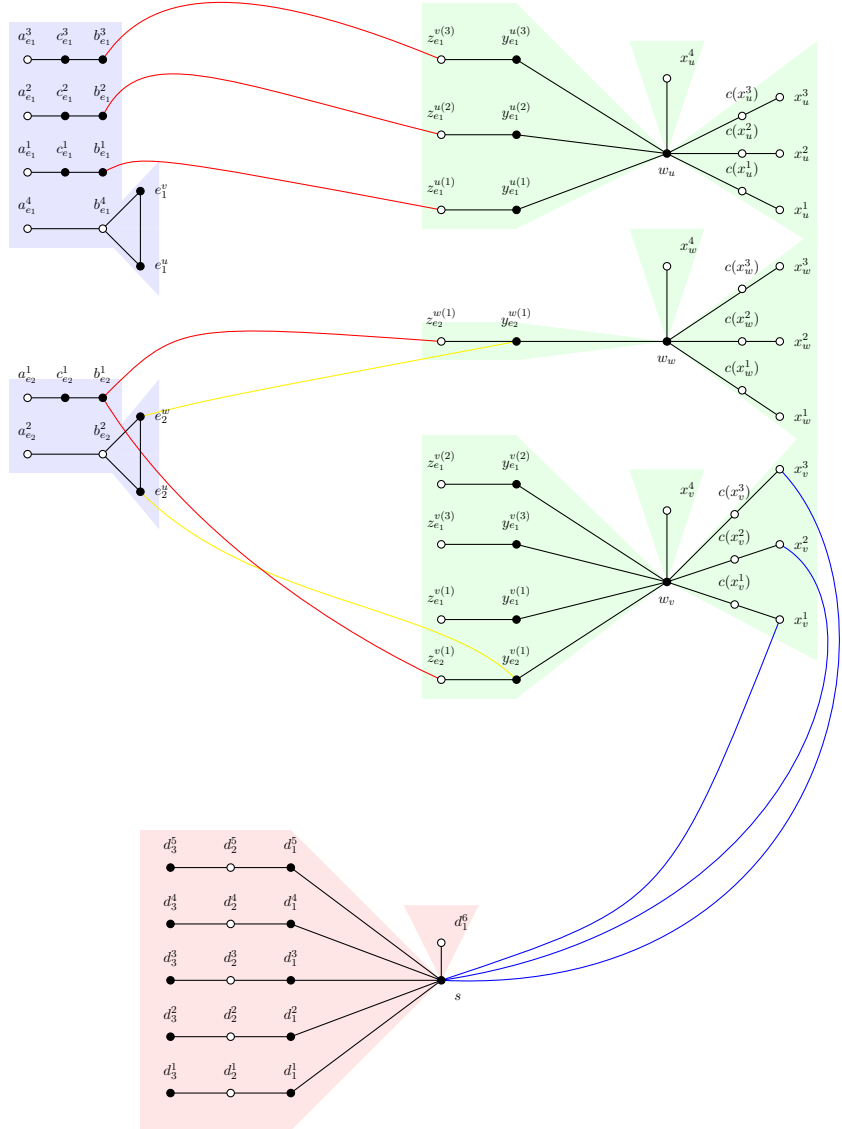
By \Cref{lem:bounded-pathwidth-G-H}, $\Tilde{G}_H$ is of bounded pathwidth (it is a subdivision of the original graph $\Tilde{G}_H$ constructed in \Cref{sec:foundational}).

We set $S = C \text{ } \cup \text{ } B \text{ } \cup \text{ } Y \text{ } \cup \text{ } \bigcup_{uv=e \in E(H)} (e^u \cup e^v) \text{ } \cup \text{ } \bigcup_{v \in V(H)} w_v \text{ } \cup \text{ } s  \text{ } \cup \text{ } \bigcup_{i \in [rn -  \bm{\sigma}]} (d_1^i \cup d_3^i)$ and $\budget = m + 3rn$.
Given that all integers are given in unary, the construction of the graph $\Tilde{G}_H$, or its path decomposition (as described in \Cref{lem:bounded-pathwidth-G-H}), and as a consequence the reduction, take time polynomial in the size of the input instance.
Additionally, by \Cref{cor:reduction-space}, this reduction is a pl-reduction.
We claim that $(H, \mathcal{P}_H, \sigma, r)$ is a yes-instance of \textsc{MMO} if and only if $(\Tilde{G}_H, \mathcal{P}_{\Tilde{G}_H}, S, \budget)$ is a yes-instance of \textsc{VC-D}.

\begin{lemma}\label{lem:hardness-VS-pathwidth-forward}
If $(H, \mathcal{P}_H, \sigma , r)$ is a yes-instance of \textsc{MMO}, then $(\Tilde{G}_H, \mathcal{P}_{\Tilde{G}_H}, S, \budget)$ is a yes-instance of \textsc{VC-D}.
\end{lemma}
\begin{proof}
Let $\lambda: E(H) \rightarrow V(H) \times V(H)$ be an orientation of the graph $H$ such that for each $v \in V(H)$, the total weight of the edges directed out of $v$ is at most $r$. 
In $\Tilde{G}_H$, the edges between $c(X)$ and $X$ are not covered.
The same applies for the edges between $A^+$ and $B^+$.
To fix that, for each edge $uv=e \in E(H)$ such that $\lambda(e) = (v, u)$:
\begin{itemize}[itemsep=0pt]
    \item we move, for each $i \in [\sigma(e)]$, the token on $y_e^{v(i)}$ to any free vertex of $c(X_v)$ and the token on $b_e^{i}$ to $z_e^{v(i)}$ (this consumes $3\sigma(e)$ slides),
    \item we slide the token on $e^u$ to $b_e^{\sigma(e)+1}$, hence covering $a_e^{\sigma(e)+1}b_e^{\sigma(e)+1}$ (this consumes $1$ slide).
\end{itemize}
This constitutes $3\bm{\sigma} + m$ slides.
We cover the $rn - \bm{\sigma}$ remaining uncovered edges between $c(X)$ and~$X$ using three slides per $D^i$ path for $i \in [rn - \bm{\sigma}]$ (by sliding the token on $d^i_3$ to $d^i_{2}$, and moving the token on $d^i_1$ to a token-free vertex in $X$).
\end{proof}

\begin{lemma}\label{lem:hardness-VC-pathwidth-backward}
If  $(\Tilde{G}_H, \mathcal{P}_{\Tilde{G}_H}, S, \budget)$ is a yes-instance of  \textsc{VC-D}, then $(H, \mathcal{P}_H, \sigma, r)$ is a yes-instance of \textsc{MMO}.
\end{lemma}

\begin{proof}
First, note that for an edge $a_e^{\sigma(e)+1}b_e^{\sigma(e)+1}$, where $uv=e \in E(H)$, to be covered with a minimal number of slides, the token on the vertex $e^u$ or the token on the vertex $e^v$ must move to~$b_e^{\sigma(e)+1}$ (note that any other token on the vertices of the graph must pass through either $e^u$ or $e^v$ to get to $b_e^{\sigma(e)+1}$, thus we can safely assume that the token already on either of $e^u$ or $e^v$ is the token that slides to $b_e^{\sigma(e)+1}$). This consumes at least $m$ slides, leaving $3rn$ slides.

No vertex cover formed with a minimal number of slides would need to make the token on the vertex $s$ or the token on a vertex $w_v$ for a vertex $v \in V(H)$ slide (as this token must always be replaced by another to cover the incident edges $sd^{rn-\bm{\sigma}+1}$ or $w_vx^{r+1}$ for a vertex $v \in V(H)$, respectively, with a minimal number of slides, thus we can always assume that the token has not been moved).
Thus, an edge between a pair of vertices in $X_v$ and $C(X_v)$ can be covered by either moving the token on a vertex $d^i_1$ for an integer $i \in [rn - \bm{\sigma}]$ towards the vertex in $X_v$, or moving a token from a vertex $y_e^{v(i_1)}$, for an edge $uv=e \in E(H)$ and an integer $i_1 \in [\sigma(e)]$, towards the vertex in $C(X_v)$. 
If the token on $d^i_1$ moves towards the vertex in $X_v$, the token on $d^i_3$ must slide to $d^i_2$. 
If a token on $y_e^{v(i_1)}$ moves towards the vertex in $C(X_v)$, it must be the case that another token has moved to either the vertex $z_e^{v(i_1)}$ or to $y_e^{v(i_1)}$. 
This however requires at least one slide per such a token.

Thus, if an edge between a pair of vertices in $X$ and $C(X)$ is covered by moving a token on one vertex $d_1^i$ for an integer $i \in [rn-\bm{\sigma}]$ towards the vertex in $X$, it does not consume more slides than moving a token from a vertex in $Y$ towards the vertex in $C(X)$.
Given that at most $rn - \bm{\sigma}$ edges can be covered using tokens from the donor paths (as $rn - \bm{\sigma} + 1$ tokens are needed to cover~$G_s$), each of the at least $\bm{\sigma}$ remaining uncovered edges must be covered by moving a token from a vertex in $Y$ towards a vertex in $c(X)$. 
Additionally, each of the remaining uncovered edges between~$c(X)$ and $X$ will require at least one additional slide (besides the two slides needed to move a token from~$Y$) and thus, tokens on distinct vertices in $Y$ must be used to cover the edges, as the remaining at most $\bm{\sigma}$ slides do not allow to get any token not initially on a vertex in $Y$ to a vertex in $Y$. 
This is true because, each $G_{e_1}^{sel}$ component for an edge $u_1v_1=e_1 \in E(H)$ cannot have less than two tokens, thus if, w.l.o.g., we move a token from $e_1^{u_1}$ to $Y^{u_1}$, then we have to move another token onto the component.
This also implies that if the token on $y_e^{v(i_1)}$ moves towards a vertex in $c(X_v)$, and consequently the token on the vertex $e^v$ slides to $y_e^{v(i_1)}$, it will require at least one more slide as~$e^ue^v$ will not be covered.
Given that we have only $\bm{\sigma}$ remaining slides, and at least one slide for each of the remaining at least $\bm{\sigma}$ remaining uncovered edges is needed, the token on the vertex $b_e^{(i_1)}$ must slide to $z_e^{v(i_1)}$.

This totals $\budget$ slides. 
Each vertex that is token-free in~$Y$ after the $\budget$ slides are consumed must be adjacent to a vertex of the form $e_2^{u_2}$ with a token, for an edge $u_2v_2=e_2 \in E(H)$ (so that the edges between $e_2^{u_2}$ and the vertices in $Y$ are covered). 
This implies that for each edge $u_2v_2= e_2 \in E(H)$, at most $\sigma(e_2)$ tokens can move to $c(X)$ from tokens on the vertices of the sets $Y^{v_2}_{e_2}$ and $Y^{u_2}_{e_2}$, and from only one of those sets, as only one of $e_2^{u_2}$ and $e_2^{v_2}$ has a token. 
To cover the $\bm{\sigma}$ remaining uncovered edges, each edge $u_2v_2=e_2 \in E(H)$ must allow $\sigma(e_2)$ tokens to move from either vertices in $Y^{v_2}_{e_2}$ and $Y^{u_2}_{e_2}$ and from at most one.
This gives a feasible orientation for the instance $(H, \mathcal{P}_H, \sigma, r)$ as any of $c(X_u)$ or $c(X_v)$ can receive at most $r$ tokens.
\end{proof}
The proofs of Lemmas~\ref{lem:hardness-VS-pathwidth-forward} and \ref{lem:hardness-VC-pathwidth-backward} complete the proof of Theorem~\ref{thm:VC-D-pathwidth}.  

\begin{theorem}\label{thm:cross_composition-VC-bpw}
There exists an or-cross-composition from \textsc{MMO} into \textsc{VC-D} on bounded pathwidth graphs with respect to parameter $\budget$. Consequently, \textsc{VC-D} does not admit a polynomial kernel with respect to $\budget + pw$, where $pw$ denotes the pathwidth of the input graphs, unless $\NP \subseteq \cp$.
\end{theorem}

\begin{proof}
As stated in \Cref{sec:foundational}, we can assume that we are given a family of $t$ \textsc{MMO} instances $(H_j, \mathcal{P}_{H_j}, \sigma_j, r_j)$, where $H_j$ is a bounded pathwidth graph with path decomposition $\mathcal{P}_{H_j}$, $|V(H_j)| = n$, $|E(H_j)| = m$, $\sigma_j: E_j \rightarrow \mathbb{Z}_+$ is a weight function such that $\sum_{e_j \in E(H_j)} \sigma_j(e_j) = \bm{\sigma}$ and $r_j = r \in \mathbb{Z}_+$ (integers are given in unary).
The construction of the instance $(G_t, \mathcal{P}_{G_t}, S, \budget)$ of \textsc{VC-D} is twofold. 

For each instance $H_j$ for $j \in [t]$, we add to $G_t$ the graph $G_{H_j}$ formed as per the construction in \Cref{thm:VC-D-pathwidth}, but without the supplier gadget.
We refer to the sets $A$, $B$, $X$, $X^+$, $C$, $Y$, and $c(X)$, subsets of vertices of a subgraph $G_{H_j}$ of $G_t$, by $A_j$, $B_j$, $X_j$, $X_j^+$, $C_j$, $Y_j$, and $c(X_j)$, respectively. 
Subsequently, we let $A = \cup_{j \in [t]} A_j$, $B = \cup_{j \in [t]} B_j$, $X = \cup_{j \in [t]} X_j$, and so on.
We attach to each of the weights-hub vertex $h$ and the orientations-quay vertex $q$ of the \textsc{MMO}-instance-selector described in \Cref{sec:foundational}, $2m + 5\bm{\sigma} + 2$ pendant vertices.
We add the \textsc{MMO}-instance selector and connect it to the rest of $G_t$ as follows (see \Cref{fig:VC-pathwidth-composition}). 
We make for each $j \in [t]$, the vertex $\textsc{\footnotesize Select}_j$ adjacent to the vertices in $V(G_{H_j}) \cap S$, where $S$ is as defined later. 
We make the vertex $h$ adjacent to each vertex in $X$ and the vertex $q$ adjacent to each vertex of $e^u$ and $e^v$ for each edge $uv=e \in E(H_j)$ for each $j \in [t]$.
By \Cref{lem:bounded-pathwidth-G}, $G_t$ is of bounded pathwidth (it is an augmented subdivision of the original graph $G_t$ appearing in \Cref{sec:foundational}).
Now, we set
\begin{equation*}
\begin{split}
    S = C \text{ } \cup \text{ } 
    B  \text{ } \cup \text{ } 
    B^+ \text{ } \cup \text{ } 
    Y \text{ } \cup \text{ } 
    X \text{ } \cup  
    \bigcup_{\substack{j \in [t] \\ uv=e \in E(H_j)}} (e^u \cup e^v) 
    \text{ } \cup 
    \bigcup_{\substack{j \in [t] \\ v \in V(H_j)}} w_v 
    \text{ } \cup 
    \bigcup_{j \in [t]} \textsc{\footnotesize Unselect}_j 
    \text{ } \cup \text{ } h 
    \text{ } \cup \text{ } q
\end{split}
\end{equation*}
and $\budget = 2m + 5\bm{\sigma} + 1$. 

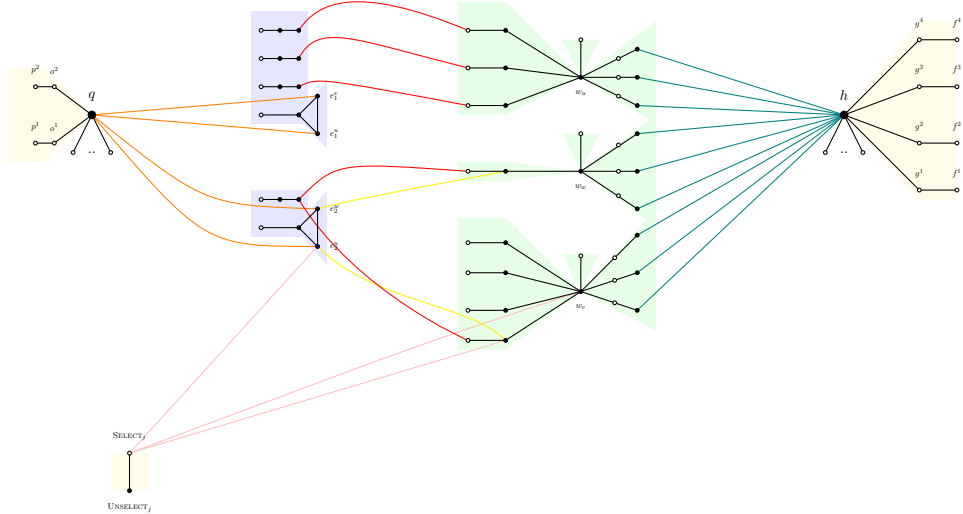
\begin{figure}[H]
    \centering
    \begin{tikzpicture}[scale=0.5]
    \fill[blue!10] (-0.75,2.75) -- (0.75,2.75) -- (0.75,-0.25) -- (-0.75,-0.25) -- cycle; 
    \fill[blue!10] (0.5,0.1) -- (1.25,0.9) -- (1.25,0) -- (0.5,0) -- cycle; 
    \fill[blue!10] (0.5,-0.1) -- (1.25,-0.9) -- (1.25,0) -- (0.5, 0) -- cycle; 
    \fill[blue!10] (-0.75,-2) -- (0.75,-2) -- (0.75,-3.25) -- (-0.75,-3.25) -- cycle; 
    \fill[blue!10] (0.5,-2.9) -- (1.25,-2) -- (1.25,-3) -- (0.5,-3) -- cycle; 
    \fill[blue!10] (0.5,-3.1) -- (1.25,-4) -- (1.25,-3) -- (0.5, -3) -- cycle;
    \fill[green!10] (8,1) -- (7.5,2) -- (8.5,2) -- (8,1) -- cycle;
    \fill[green!10] (8,1) -- (6,0) -- (6,3) -- (8,1) -- cycle;
    \fill[green!10] (8,1) -- (10,-0.25) -- (10,2.5) -- (8,1) -- cycle;
    \fill[green!10] (6,0) -- (6,3) -- (4.75,3) -- (4.75,0) -- cycle;
    \fill[green!10] (8,-1.5) -- (7.5,-0.5) -- (8.5,-0.5) -- (8,-1.5) -- cycle;
    \fill[green!10] (8,-1.5) -- (6,-1.75) -- (6,-1.25) -- (8,-1.5) -- cycle;
    \fill[green!10] (8,-1.5) -- (10,0) -- (10,-3) -- (8,-1.5) -- cycle;
    \fill[green!10] (6,-1.75) -- (6,-1.25) -- (4.75,-1.25) -- (4.75,-1.75) -- cycle;
    \fill[green!10] (8,-4.7) -- (7.5,-3.7) -- (8.5,-3.7) -- (8,-4.7) -- cycle;
    \fill[green!10] (8,-4.7) -- (6,-6.25) -- (6,-2.75) -- (8,-4.7) -- cycle;
    \fill[green!10] (8,-4.7) -- (10,-2.5) -- (10,-5.75) -- (8,-4.7) -- cycle;
    \fill[green!10] (6,-6.25) -- (6,-2.75) -- (4.75,-2.75) -- (4.75,-6.25) -- cycle;
    \fill[yellow!10] (-3.5,-10) -- (-4.5,-10) -- (-4.5,-9) -- (-3.5,-9) -- cycle;
    \fill[yellow!10] (-5,0) -- (-6.25,-1.25) -- (-6.25,1.25) -- cycle;
    \fill[yellow!10] (-6.25,-1.25) -- (-7.25,-1.25) -- (-7.25, 1.25) -- (-6.25,1.25) -- cycle;
    \fill[yellow!10] (15,0) --  (17,2.5) -- (17,-2.25) -- cycle;
    \fill[yellow!10] (17,2.5) -- (18,2.5) -- (18,-2.25) -- (17,-2.25) -- cycle; 
    \draw[teal] (9.5,0.25) -- (15, 0);
    \draw[teal] (9.5,1) -- (15, 0);
    \draw[teal] (9.5,1.75) -- (15, 0);
    \draw[teal] (9.5,-2.5) -- (15, 0);
    \draw[teal] (9.5,-1.5) -- (15, 0);
    \draw[teal] (9.5,-0.5) -- (15, 0);
    \draw[teal] (9.5,-5.2) -- (15, 0);
    \draw[teal] (9.5,-4.2) -- (15, 0);
    \draw[teal] (9.5,-3.2) -- (15, 0);
    \draw[orange] (1,0.5) -- (-5,0);
    \draw[orange] (1,-0.5) -- (-5,0);
    \draw[orange] (1,-2.5) .. controls (-2,-2.4) .. (-5,0);
    \draw[orange] (1,-3.5) .. controls (-2,-3.5) .. (-5,0);
    \draw[pink] (-4,-9) -- (1,-3.5);
    \draw[pink] (-4,-9) .. controls (4,-6) .. (8,-4.7);
    \draw[pink] (-4,-9) -- (6,-6);
    \draw[black] (8,1) -- (9.5,0.25);
    \draw[black] (8,1) -- (9.5,1);
    \draw[black] (8,1) -- (9.5,1.75);
    \draw[black] (8,1) -- (8,2);
    \draw[black] (8,1) -- (6,0.25);
    \draw[black] (8,1) -- (6,2.25);
    \draw[black] (8,1) -- (6,1.25);
    \draw[black] (8,-1.5) -- (9.5,-2.5);
    \draw[black] (8,-1.5) -- (9.5,-1.5);
    \draw[black] (8,-1.5) -- (9.5,-0.5);
    \draw[black] (8,-1.5) -- (6,-1.5);
    \draw[black] (8,-1.5) -- (8,-0.5);
    \draw[black] (8,-4.7) -- (9.5,-5.2);
    \draw[black] (8,-4.7) -- (9.5,-4.2);
    \draw[black] (8,-4.7) -- (9.5,-3.2);
    \draw[black] (8,-4.7) -- (8,-3.75);
    \draw[black] (8,-4.7)-- (6,-6);
    \draw[black] (8,-4.7) -- (6,-5.2);
    \draw[black] (8,-4.7)-- (6,-4.2);
    \draw[black] (8,-4.7) -- (6,-3.4);
    \draw[black] (-6,0.75) -- (-5,0);
    \draw[black] (-6,-0.75) -- (-5, 0);
    \draw[black] (17,0.75) -- (15,0);
    \draw[black] (17,-0.75) -- (15,0);
    \draw[black] (17,2) -- (15,0);
    \draw[black] (17,-2) -- (15,0);
    \draw[yellow] (1,-2.5) .. controls (2,-2.25) .. (6,-1.5);
    \draw[yellow] (1,-3.5) .. controls (2.5,-5) and  (5,-5) .. (6,-6);    
    \draw[red] (0.5,-2.25) .. controls (1.5,-1.25) .. (5,-1.5);
    \draw[red] (0.5,-2.25) .. controls (2,-4.7) and  (5,-6) .. (5,-6);  
    \draw[red] (0.5,0.75) .. controls (1,1) .. (5,0.25);
    \draw[red] (0.5,1.5) .. controls (1,2.5) and  (2,2) .. (5,1.25);
    \draw[red] (0.5,2.25) .. controls (1,3) and  (2,3.5) .. (5,2.25);  
    \node[circle, fill=black, inner sep=0.5pt, draw=black, fill=black, label=below:{\scalebox{0.3}{$w_u$}}] at (8,1) {};
    \node[circle, fill=black, inner sep=0.5pt, draw=black, fill=black, label=below:{\scalebox{0.3}{$w_w$}}] at (8,-1.5) {};
    \node[circle, fill=black, inner sep=0.5pt, draw=black, fill=black, label=below:{\scalebox{0.3}{$w_v$}}] at (8,-4.7) {};
    \node[circle, fill=black, inner sep=0.5pt, draw=black, fill=black] at (9.5,0.25) {};
    \node[circle, fill=black, inner sep=0.5pt, draw=black, fill=white] at (9,0.5) {};
    \node[circle, fill=black, inner sep=0.5pt, draw=black, fill=black] at (9.5,1) {};
    \node[circle, fill=black, inner sep=0.5pt, draw=black, fill=white] at (9,1) {};
    \node[circle, fill=black, inner sep=0.5pt, draw=black, fill=black] at (9.5,1.75) {};
    \node[circle, fill=black, inner sep=0.5pt, draw=black, fill=white] at (9,1.5) {};
    \node[circle, fill=black, inner sep=0.5pt, draw=black, fill=black] at (9.5,-2.5) {};
    \node[circle, fill=black, inner sep=0.5pt, draw=black, fill=white] at (9,-2.15) {};
    \node[circle, fill=black, inner sep=0.5pt, draw=black, fill=black] at (9.5,-1.5) {};
    \node[circle, fill=black, inner sep=0.5pt, draw=black, fill=white] at (9,-1.5) {};
    \node[circle, fill=black, inner sep=0.5pt, draw=black, fill=black] at (9.5,-0.5) {};
    \node[circle, fill=black, inner sep=0.5pt, draw=black, fill=white] at (9,-0.8) {};
    \node[circle, fill=black, inner sep=0.5pt, draw=black, fill=black] at (9.5,-5.2) {};
    \node[circle, fill=black, inner sep=0.5pt, draw=black, fill=white] at (8.9,-5) {};
    \node[circle, fill=black, inner sep=0.5pt, draw=black, fill=black] at (9.5,-4.2) {};
    \node[circle, fill=black, inner sep=0.5pt, draw=black, fill=white] at (8.9,-4.4) {};
    \node[circle, fill=black, inner sep=0.5pt, draw=black, fill=black] at (9.5,-3.2) {};
    \node[circle, fill=black, inner sep=0.5pt, draw=black, fill=white] at (8.9,-3.8) {};
    \node[circle, fill=black, inner sep=0.5pt, draw=black, fill=white] at (8,2) {};
    \node[circle, fill=black, inner sep=0.5pt, draw=black, fill=white] at (8,-0.5) {};
    \node[circle, fill=black, inner sep=0.5pt, draw=black, fill=white] at (8,-3.75) {};
    \draw[black] (5,-1.5) -- (6,-1.5) node[circle, draw=black, fill=black, inner sep=0.5pt, pos=1]{} node[circle, draw=black, fill=white, inner sep=0.5pt, pos=0]{};    
    \draw[black] (5,1.25) -- (6,1.25) node[circle, draw=black, fill=black, inner sep=0.5pt, pos=1]{} node[circle, draw=black, fill=white, inner sep=0.5pt, pos=0]{};
    \draw[black] (5,0.25) -- (6,0.25) node[circle, draw=black, fill=black, inner sep=0.5pt, pos=1]{} node[circle, draw=black, fill=white, inner sep=0.5pt, pos=0]{};
    \draw[black] (5,2.25) -- (6,2.25) node[circle, draw=black, fill=black, inner sep=0.5pt, pos=1]{} node[circle, draw=black, fill=white, inner sep=0.5pt, pos=0]{};
    \draw[black] (-6.5,0.75) -- (-6,0.75) node[circle, draw=black, fill=white, inner sep=0.5pt, pos=1, label=above:\scalebox{0.3}{$o^2$}]{} node[circle, draw=black, fill=white, inner sep=0.5pt, pos=0, label=above:\scalebox{0.3}{$p^2$}]{};
    \draw[black] (-6.5,-0.75) -- (-6, -0.75) node[circle, draw=black, fill=white, inner sep=0.5pt, pos=1, label=above:\scalebox{0.3}{$o^1$}]{} node[circle, draw=black, fill=white, inner sep=0.5pt, pos=0, label=above:\scalebox{0.3}{$p^1$}]{};
    \draw[black] (17,0.75) -- (18,0.75) node[circle, draw=black, fill=white, inner sep=0.5pt, pos=1, label=above:\scalebox{0.3}{$f^3$}]{} node[circle, draw=black, fill=white, inner sep=0.5pt, pos=0, label=above:\scalebox{0.3}{$g^3$}]{};
    \draw[black] (17,-0.75) -- (18,-0.75) node[circle, draw=black, fill=white, inner sep=0.5pt, pos=1, label=above:\scalebox{0.3}{$f^2$}]{} node[circle, draw=black, fill=white, inner sep=0.5pt, pos=0, label=above:\scalebox{0.3}{$g^2$}]{};
    \draw[black] (17,2) -- (18,2) node[circle, draw=black, fill=white, inner sep=0.5pt, pos=1, label=above:\scalebox{0.3}{$f^4$}]{} node[circle, draw=black, fill=white, inner sep=0.5pt, pos=0, label=above:\scalebox{0.3}{$g^4$}]{};
    \draw[black] (17,-2) -- (18,-2) node[circle, draw=black, fill=white, inner sep=0.5pt, pos=1, label=above:\scalebox{0.3}{$f^1$}]{} node[circle, draw=black, fill=white, inner sep=0.5pt, pos=0, label=above:\scalebox{0.3}{$g^1$}]{};
    \draw[black] (-4,-10) -- (-4,-9) node[circle, draw=black, fill=white, inner sep=0.5pt, pos=1, label=above:\scalebox{0.3}{$\textsc{Select}_j$}]{} node[circle, draw=black, fill=black, inner sep=0.5pt, pos=0, label=below:\scalebox{0.3}{$\textsc{Unselect}_j$}]{};
    \draw[black] (5,-3.4) -- (6,-3.4) node[circle, draw=black, fill=black, inner sep=0.5pt, pos=1]{} node[circle, draw=black, fill=white, inner sep=0.5pt, pos=0]{};
    \draw[black] (5,-4.2) -- (6,-4.2) node[circle, draw=black, fill=black, inner sep=0.5pt, pos=1]{} node[circle, draw=black, fill=white, inner sep=0.5pt, pos=0]{};
    \draw[black] (5,-5.2) -- (6,-5.2) node[circle, draw=black, fill=black, inner sep=0.5pt, pos=1]{} node[circle, draw=black, fill=white, inner sep=0.5pt, pos=0]{};
    \draw[black] (5,-6) -- (6,-6) node[circle, draw=black, fill=black, inner sep=0.5pt, pos=1]{} node[circle, draw=black, fill=white, inner sep=0.5pt, pos=0]{};
    \draw[black] (-0.5,0.75) -- (0.5,0.75) node[circle, draw=black, fill=white, inner sep=0.5pt, pos=0]{} node[circle, draw=black, fill=black, inner sep=0.5pt, pos=1]{};
    \node[draw, circle, fill=black, inner sep=0.5pt] at (0,0.75) {};
    \draw[black] (-0.5,1.5) -- (0.5,1.5) node[circle, draw=black, fill=white, inner sep=0.5pt, pos=0]{} node[circle, draw=black, fill=black, inner sep=0.5pt, pos=1]{};
    \node[draw, circle, fill=black, inner sep=0.5pt] at (0,1.5) {};
    \draw[black] (-0.5,2.25) -- (0.5,2.25) node[circle, draw=black, fill=white, inner sep=0.5pt, pos=0]{} node[circle, draw=black, fill=black, inner sep=0.5pt, pos=1]{};
    \node[draw, circle, fill=black, inner sep=0.5pt] at (0,2.25) {};
    \draw[black] (1,0.5) -- (1, -0.5);
    \draw[black] (0.5,0) -- (1,0.5) node[circle, draw=black, fill=black, inner sep=0.5pt, label=right:{\scalebox{0.3}{$e_1^v$}}, pos=1]{};
    \draw[black] (0.5,0) -- (1,-0.5) node[circle, draw=black, fill=black, inner sep=0.5pt, label=right:{\scalebox{0.3}{$e_1^u$}}, pos=1]{};
    \draw[black] (-0.5,0) -- (0.5,0) node[circle, draw=black, fill=white, inner sep=0.5pt, pos=0]{} node[circle, draw=black, fill=black, inner sep=0.5pt, pos=1]{};
    \draw[black] (-0.5,-2.25) -- (0.5,-2.25) node[circle, draw=black, fill=white, inner sep=0.5pt, pos=0]{} node[circle, draw=black, fill=black, inner sep=0.5pt, pos=1]{};
    \node[draw, circle, fill=black, inner sep=0.5pt] at (0,-2.25) {};
    \draw[black] (1,-2.5) -- (1, -3.5);
    \draw[black] (0.5,-3) -- (1,-2.5) node[circle, draw=black, fill=black, inner sep=0.5pt, label=right:{\scalebox{0.3}{$e_2^w$}}, pos=1]{};
    \draw[black] (0.5,-3) -- (1,-3.5) node[circle, draw=black, fill=black, inner sep=0.5pt, label=right:{\scalebox{0.3}{$e_2^u$}}, pos=1]{};
    \draw[black] (-0.5,-3) -- (0.5,-3) node[circle, draw=black, fill=white, inner sep=0.5pt, pos=0]{} node[circle, draw=black, fill=black, inner sep=0.5pt, pos=1]{};
    \draw[black] (-5, 0) -- (-5.5, -1);
    \draw[black] (-5,0) -- (-4.5, -1);
    \node at (-5,-1) {\scalebox{0.5}{..}};
    \draw[black] (15,0) -- (15.5, -1);
    \draw[black] (15,0) -- (14.5,-1);
    \node at (15,-1) {\scalebox{0.5}{..}};
    \node[circle, draw=black, fill=black, inner sep=1pt, label=above:\scalebox{0.5}{$h$}] at (15, 0) {};
    \node[circle, draw=black, fill=black, inner sep=1pt, label=above:\scalebox{0.5}{$q$}] at (-5, 0) {};
    \node[circle, draw=black, fill=white, inner sep=0.5pt] at (-5.5, -1) {};
    \node[circle, draw=black, fill=white, inner sep=0.5pt] at (-4.5, -1) {};
    \node[circle, draw=black, fill=white, inner sep=0.5pt] at (14.5, -1) {};
    \node[circle, draw=black, fill=white, inner sep=0.5pt] at (15.5, -1) {};
    
    \end{tikzpicture}
    \caption{\footnotesize Orange, pink, and green edges highlighting the different types of edges between an \textsc{MMO}-instance-selector of the composition in \Cref{thm:cross_composition-VC-bpw} and a subgraph $G_{H_j}$ for a $j \in [t]$ of the same. $H_j$ has $3$ vertices $u, v$, and $w$, and two edges $e_1=uv$ and $e_2=uw$, and $r = 3$. Additionally, $\sigma_j(e_1) = 3$ and $\sigma_j(e_2) = 1$. For clarity, not all edges inside~$G_{H_j}$ are drawn, nor are all the pink edges depicted. Vertices in black are in $S$ and those in white are not.}
    \label{fig:VC-pathwidth-composition}
\end{figure}

Given that all integers are given in unary, the construction of the graph $G_t$, or its path decomposition (as described in \Cref{lem:bounded-pathwidth-G}), and as a consequence the reduction take time polynomial in the size of the input instances. 
Additionally, by \Cref{cor:composition-space}, this composition is a pl-reduction.
We claim that $(G_t, \mathcal{P}_{G_t}, S, \budget)$ is a yes-instance of \textsc{VC-D} if and only if for some $\mathfrak{j} \in [t]$, $(H_\mathfrak{j}, \mathcal{P}_{H_\mathfrak{j}}, \sigma_\mathfrak{j}, r_\mathfrak{j})$ is a yes-instance of \textsc{MMO}.
\\
\begin{claim}
If for some $\mathfrak{j} \in [t]$, $(H_\mathfrak{j}, \mathcal{P}_{H_\mathfrak{j}}, \sigma_\mathfrak{j}, r_\mathfrak{j})$ is a yes-instance of \textsc{MMO}, then $(G_t, \mathcal{P}_{G_t}, S, b)$ is a yes-instance of \textsc{VC-D}.
\end{claim}

\begin{claimproof}
Let $(H_\mathfrak{j}, \mathcal{P}_{H_\mathfrak{j}}, \sigma_\mathfrak{j}, r_\mathfrak{j})$ be a yes-instance of \textsc{MMO} and let $\lambda$ be a feasible orientation of $H_\mathfrak{j}$ such that for each $v \in V(H_\mathfrak{j})$, the total weight of the edges directed out of $v$ is at most $r$.
In $G_t$, the edges $f^1g^1, \ldots, f^{\bm{\sigma}}g^{\bm{\sigma}}, o^1p^1,$ $\ldots, o^mp^m$ are not covered.
First, we slide the token on $\textsc{\footnotesize Unselect}_\mathfrak{j}$ to $\textsc{\footnotesize Select}_\mathfrak{j}$. 
Using $2m$ slides, we move for each edge $e \in E(H_\mathfrak{j})$ the token on $e^u$ ($e^v$), if $\lambda(e) = (v, u)$ ($\lambda(e) = (u, v)$), towards a token-free vertex in $o^1, \ldots, o^m$.
We additionally slide each token on a vertex $b^i_e$ for $i \in [\sigma_\mathfrak{j}(e)]$ to the vertex $z_e^{v(i)}$, move the token on $y_e^{v(i)}$ towards a token-free vertex in~$c(X_v)$ and consequently, move the token on the adjacent vertex in $X_v$ towards a token-free vertex in $f^1, \ldots, f^{\bm{\sigma}}$.
The total number of slides performed is $\budget$ and they achieve a configuration for the tokens that covers all of $G_t$.
\end{claimproof}

\begin{claim}
If $(G_t, \mathcal{P}_{G_t}, S, b)$ is a yes-instance of \textsc{VC-D}, then there exists an integer $\mathfrak{j} \in [t]$, such that $(H_\mathfrak{j}, \mathcal{P}_{H_\mathfrak{j}}, \sigma_\mathfrak{j}, r_\mathfrak{j})$ is a yes-instance of \textsc{MMO}. 
\end{claim}

\begin{claimproof}
In any solution that uses the minimal number of slides, the tokens on the vertices $h$, $q$, $w_v$ for each $v \in V(H_j)$ and integer $j \in [t]$, and on the vertices in $C$ do not need to be moved (as these tokens must be replaced by others to cover the edges incident to the pendant vertices and $h$, or the pendant vertices and $q$, or to cover the edges $w_v x^{r+1}$, or the edges incident to the vertices in $A$, thus we can assume these tokens remain stationary).
In the same solution, we can similarly assume that a token on one of $\textsc{\footnotesize Unselect}_j$ and $\textsc{\footnotesize Select}_j$ for each $j \in [t]$ remains on either one of those vertices.
To cover the edges $o^1p^1, \ldots, o^mp^m$ and $f^1g^1, \ldots, f^{\bm{\sigma}}g^{\bm{\sigma}}$, at least $2m + 2\bm{\sigma}$ slides are needed to get tokens from one or more of the vertices in $X$ onto the vertices $f^1, \ldots, f^{\bm{\sigma}}$, and from one or more of the vertices of the form $e^u$ for an edge $e \in E(H_j)$ incident to a vertex $u \in V(H_j)$ for an integer $j \in [t]$, onto the vertices $o^1, \ldots, o^m$.
If a token is moved out of a subgraph $G_{H_j}$ (for an integer $j \in [t]$) of $G_t$, which is bound to happen to get tokens onto the vertices $f^1, \ldots, f^{\bm{\sigma}}, o^1, \ldots, o^m$, at least one slide is needed to cover the edge between now token-free vertices in $G_{H_j}$ and $\textsc{\footnotesize Select}_j$ and exactly one slide can only be achieved by sliding the token on $\textsc{\footnotesize Unselect}_j$ to $\textsc{\footnotesize Select}_j$ (since otherwise a token has to move from one of the vertices of a subgraph $G_{H_{j'}}$ for $j' \neq j \in [t]$ into $G_{H_j}$, and this requires more than one slide). 

W.l.o.g. assume a token on a vertex, denoted $x^i_v$, in $X$, for a vertex $v \in V(H_\mathfrak{j})$, and integers $\mathfrak{j} \in [t]$ and $i \in [r]$, is moved to one of the vertices $f^1, \ldots, f^{\bm{\sigma}}$, then at least $2$ slides are needed to move a token into either $x^i_v$ or $c(x^i_v)$ (since the tokens on $w_v$ and $h$ are assumed to be stationary).
Since a token moving from any other vertex, except $x^i_v$, in $X$ into $x^i_v$ can replace the token on $x^i_v$ in moving into one of the vertices $f^1, \ldots, f^{\bm{\sigma}}$, in a solution that uses the minimal number of slides $2$ slides can only be achieved by moving a token on a vertex, denoted $y_e^{v(i_1)}$, in $Y^v$, for some edge $e \in E(H_\mathfrak{j})$ adjacent to $v$ and some integer $i_1 \in [\sigma_\mathfrak{j}(e)]$, to $c(x^i_v)$.
Additionally, $3$ slides can only be achieved by moving a token on the same vertex to $x^i_v$.
Since in a solution that uses the minimal number of slides a token on one of $\textsc{\footnotesize Unselect}_\mathfrak{j}$ and $\textsc{\footnotesize Select}_\mathfrak{j}$ is assumed to remain on either of those vertices, and a token in $B$ can slide at most one slide to a vertex in $Z$, if a token on $y_e^{v(i_1)}$ is moved to a vertex in $c(X)$ (or $X$), either a token has to move to the vertex $z_e^{v(i_1)}$, or a token has to slide from the vertex $e^v$ to $y_e^{v(i_1)}$. 

Given the budget and the fact that $\bm{\sigma}$ tokens in any solution must move from $X$ onto the vertices $f^1, \ldots, f^{\bm{\sigma}}$, tokens must move from distinct vertices in $X$ onto the vertices $f^1, \ldots, f^{\bm{\sigma}}$ and from distinct vertices of the form $e^u$ for an edge $e \in E(H_j)$ incident to a vertex $u \in V(H_j)$ for an integer $j \in [t]$, onto the vertices $o^1, \ldots, o^m$. 
Additionally, given the budget, tokens in the same solution must move onto $f^1, \ldots, f^{\bm{\sigma}}$ from only the vertices in $X_{\mathfrak{j}}$ and onto $o^1, \ldots, o^m$ from only the vertices of the form $e_1^{u}$, for an edge $uw=e_1 \in E(H_\mathfrak{j})$ (note that one token sliding to $\textsc{\footnotesize Select}_{\mathfrak{j}}$ will cover the edge between $\textsc{\footnotesize Select}_{\mathfrak{j}}$ and $e_1^{u}$ and the edge between $\textsc{\footnotesize Select}_{\mathfrak{j}}$ and a vertex in $X_{\mathfrak{j}}$).
In the same solution, if a token moves from the vertex $e_1^{u}$ onto one of the vertices $o^1, \ldots, o^m$, the token on $e_1^{w}$ remains stationary as the budget does not allow for another token to move into either one of the adjacent vertices $e_1^{u}$ and $e_1^{w}$.
To fill all of  $o^1, \ldots, o^m$ with tokens, exactly one token must move from~$G_{e_2}^{sel}$ onto the vertices $o^1, \ldots, o^m$ for each edge $e_2 \in E(H_\mathfrak{j})$.
The latter implies that the token on~$e^v$ does not move to $y_e^{v(i_1)}$ and given the budget that the token on $b_e^{i_1}$ slides to $z_e^{v(i_1)}$.

W.l.o.g. assume that the token on $e^v$ does not move to one of $o^1, \ldots, o^m$, then at most the $\sigma(e)$ tokens on the vertices of $Y_e^v$ can be sent to $c(X_v)$.
This implies that for each edge in $H$, at most its weight in tokens can move to $c(X)$ from and to exactly one of the vertex gadgets corresponding to the vertices incident to that edge in $H$.
Given that $\bm{\sigma}$ tokens are needed on the vertices of $c(X)$, it must be the case that for each edge, all its weight in tokens must move to $c(X)$.
This gives a feasible orientation for $(H_\mathfrak{j}, \mathcal{P}_{H_\mathfrak{j}}, \sigma_\mathfrak{j}, r_\mathfrak{j})$, since for each $v \in V(H_\mathfrak{j})$, we have at most $r$ vertices in $c(X_v)$.
\end{claimproof}
This concludes the proof of the theorem.
\end{proof}

Next we consider the feedback vertex set number (\emph{fvs}) parameterization of the VC-D problem. In Theorem~\ref{thm:VC-D-pathwidth}, we proved that the VC-D problem is \textsc{XNLP}-hard for the parameter pathwidth of the input graph. 
The feedback vertex set number (\emph{fvs}) and pathwidth are upper bounds of treewidth, but are incomparable. 
We will show that the VC-D problem is $\W[1]$-hard for the parameter \emph{fvs}. 
\begin{theorem}
    \label{thm:VC-D-fvs}
    The VC-D problem is $\W[1]$-hard when parameterized by the feedback vertex set number, i.e., \emph{fvs}, of the input graph. 
\end{theorem}
We present a parameterized reduction from the \mcc~problem. 
We utilize the reduction presented in Theorem~\ref{thm:DS-D-fvs}, and apply some changes over the constructed graph to obtain a reduced instance of  the VC-D problem. 
Consider the graph $H$ constructed in the proof of Theorem~\ref{thm:DS-D-fvs}. 
For each $i \in [\kappa]$, we add a neighbor $\Tilde{t}_i$ to $t_i$ in the vertex-block $H_i$. 
For each $\iljk$, we add a neighbor $\Tilde{t}_{i,j}$ to $t_{i,j}$ in the edge-block $H_{i,j}$. 
For each $\iljk$ and $l \in \{i,j\}$, we do the following changes in the connector $\conn_{i,j}^l$:
\begin{itemize}
    \item Add four new vertices $\Tilde{s}_{i,j}^l, \hat{s}_{i,j}^l, \Tilde{r}_{i,j}^l$ and $\hat{r}_{i,j}^l$. 
    \item Connect $\Tilde{s}_{i,j}^l, \hat{s}_{i,j}^l$ with $s_{i,j}^l$, and $\Tilde{r}_{i,j}^l, \hat{r}_{i,j}^l$ with $r_{i,j}^l$. 
    \item For each neighboring vertex $v$ of $s_{i,j}^l$ from the vertex-blocks, remove the edge $s_{i,j}^lv$ and add the edge $\Tilde{s}_{i,j}^lv$.
    \item For each neighboring vertex $v$ of $s_{i,j}^l$ from the edge-blocks, remove the edge $s_{i,j}^lv$ and add the edge $\hat{s}_{i,j}^lv$.
    \item For each neighboring vertex $v$ of $r_{i,j}^l$ from the vertex-blocks, remove the edge $r_{i,j}^lv$ and add the edge $\Tilde{r}_{i,j}^lv$.
    \item For each neighboring vertex $v$ of $r_{i,j}^l$ from the vertex-blocks, remove the edge $r_{i,j}^lv$ and add the edge $\hat{r}_{i,j}^lv$.
\end{itemize}
An illustration of a connector connecting a vertex-block and an edge-block is given in Figure~\ref{fig:VC-D-fvs-illustration}. 
This completes the construction of graph $H$ for the VC-D instance. 
Next we describe the initial configuration $\varS$ as follows:
\[\varS = \bigcup_{i \in [\kappa], {x \in [n]}} Q_{i,x} \cup \bigcup_{e \in E} Q_e \cup \bigcup_{i \in [\kappa]}\{t_i , \Tilde{t}_i\}  \cup \bigcup_{\iljk}\{t_{i,j},\Tilde{t}_{i,j}\} \cup \bigcup_{\iljk, l\in \{i,j\}} \{\Tilde{s}_{i,j}^l, \hat{s}_{i,j}^l, \Tilde{r}_{i,j}^l, \hat{r}_{i,j}^l\}.\]
Finally, we set $\budget = (12n+2)\kctwo+2\kappa$ and the reduced VC-D instance is $(H, \varS, \budget)$. 
\begin{figure}
    \tikzset{decorate sep/.style 2 args=
{decorate,decoration={shape backgrounds,shape=circle,shape size=#1,shape sep=#2}}}
\centering
    \begin{tikzpicture}
    \coordinate (ti) at (0,0);
    \coordinate (pix) at (1,0);
    \coordinate (qix1) at (2,1.5);
    \coordinate (qixx) at (2,0.5);
    \coordinate (qixx1) at (2,-0.5);
    \coordinate (qixn) at (2,-1.5);
    \coordinate (siij) at (6,2);
    \coordinate (stiij) at (5,2.5);
    \coordinate (shiij) at (7,2.5);
    \coordinate (riij) at (6, -2);
    \coordinate (rtiij) at (5,-2.5);
    \coordinate (rhiij) at (7,-2.5);
    \coordinate (aiji1) at (5, 1.5);
    \coordinate (aijin) at (7, 1.5);
    \coordinate (biji1) at (5, 0.5);
    \coordinate (bijin) at (7, 0.5);
    \coordinate (ciji1) at (5, -0.5);
    \coordinate (cijin) at (7, -0.5);
    \coordinate (diji1) at (5, -1.5);
    \coordinate (dijin) at (7, -1.5);
    \coordinate (tij) at (12, 0);
    \coordinate (pe) at (11, 0);
    \coordinate (qez1) at (10,1.5);
    \coordinate (qezz) at (10, 0.5);
    \coordinate (qezz1) at (10, -0.5);
    \coordinate (qezn) at (10, -1.5);
    \coordinate (qew1) at (10, -2.5);
    \coordinate (qewn) at (10, -3.5);
    \draw[black, thick] (ti) -- ($(ti) + (0,-1)$);
    \draw[black, thick] (ti) -- (pix);
    \draw[black, thick] (pix) -- (qix1);
    \draw[black, thick] (pix) -- (qixx);
    \draw[black, thick] (pix) -- (qixx1);
    \draw[black, thick] (pix) -- (qixn);
    \draw[black, thick] (tij) -- ($(tij) + (0,-1)$);
    \draw[black, thick] (tij) -- (pe);
    \draw[black, thick] (pe) -- (qez1);
    \draw[black, thick] (pe) -- (qezz);
    \draw[black, thick] (pe) -- (qezz1);
    \draw[black, thick] (pe) -- (qezn);
    \draw[black, thick] (pe) -- (qew1);
    \draw[black, thick] (pe) -- (qewn);
    \draw[black, thick] (shiij) -- ($(shiij) + (0,0.5)$);
    \draw[black, thick] (stiij) -- ($(stiij) + (0,0.5)$);
    \draw[black, thick] (rhiij) -- ($(rhiij) + (0,-0.5)$);
    \draw[black, thick] (rtiij) -- ($(rtiij) + (0,-0.5)$);
    \draw[black, thick] (siij) -- (stiij);
    \draw[black, thick] (siij) -- (shiij);
    \draw[black, thick] (riij) -- (rtiij);
    \draw[black, thick] (riij) -- (rhiij);
    \draw[black, thick] (siij) -- (aiji1);
    \draw[black, thick] (siij) -- (aijin);
    \draw[black, thick] (riij) -- (diji1);
    \draw[black, thick] (riij) -- (dijin);
    \draw[black, thick] (aiji1) -- (biji1);
    \draw[black, thick] (aijin) -- (bijin);
    \draw[black, thick] (ciji1) -- (diji1);
    \draw[black, thick] (cijin) -- (dijin);
    \draw[black, thick] (stiij) to[bend right=30] (qix1);
    \draw[black, thick] (stiij) to[bend right=30] (qixx);
    \draw[black, thick] (rtiij) to[bend left=30] (qixx1);
    \draw[black, thick] (rtiij) to[bend left=30] (qixn);
    \draw[black, thick] (shiij) to[bend left=30] (qez1);
    \draw[black, thick] (shiij) to[bend left=30] (qezz);
    \draw[black, thick] (rhiij) to[bend right=30] (qezz1);
    \draw[black, thick] (rhiij) to[bend right=30] (qezn);
    \fill[white, draw=black, thick] (ti) circle (0.1cm) node[left,blue] {$t_i$};
    \fill[white, draw=black, thick] ($(ti) + (0,-1)$) circle (0.1cm);
    \fill[white, draw=black, thick] (tij) circle (0.1cm) node[right,blue] {$t_{i,j}$};
    \fill[white, draw=black, thick] ($(tij) + (0,-1)$) circle (0.1cm);
    \fill[white, draw=black, thick] (siij) circle (0.1cm) node[above,blue] {$s_{i,j}^i$};
    \fill[white, draw=black, thick] (stiij) circle (0.1cm); 
    \fill[white, draw=black, thick] (shiij) circle (0.1cm);
    \fill[white, draw=black, thick] ($(stiij) + (0,0.5)$) circle (0.1cm);
    \fill[white, draw=black, thick] ($(shiij) + (0,0.5)$) circle (0.1cm);
    \fill[white, draw=black, thick] (riij) circle (0.1cm) node[below,blue] {$r_{i,j}^i$};
    \fill[white, draw=black, thick] (rtiij) circle (0.1cm); 
    \fill[white, draw=black, thick] (rhiij) circle (0.1cm); 
    \fill[white, draw=black, thick] ($(rtiij) + (0,-0.5)$) circle (0.1cm);
    \fill[white, draw=black, thick] ($(rhiij) + (0,-0.5)$) circle (0.1cm);
    \fill[white, draw=black, thick] (pix) circle (0.1cm) node[below,blue] {$p_{i,x}$};
    \fill[white, draw=black, thick] (qix1) circle (0.1cm) node[above,blue] {$q_{i,x}^{j,1}$};
    \fill[white, draw=black, thick] (qixx) circle (0.1cm) node[below,blue] {$q_{i,x}^{j,x}$};
    \fill[white, draw=black, thick] (qixx1) circle (0.1cm);
    \fill[white, draw=black, thick] (qixn) circle (0.1cm) node[below,blue] {$q_{i,x}^{j,n}$};
    \fill[white, draw=black, thick] (pe) circle (0.1cm) node[above,blue] {$p_e$};    
    \fill[white, draw=black, thick] (qez1) circle (0.1cm) node[above,blue] {$q_e^1$};
    \fill[white, draw=black, thick] (qezz) circle (0.1cm) node[below,blue] {$q_e^{n-z}$};
    \fill[white, draw=black, thick] (qezz1) circle (0.1cm);
    \fill[white, draw=black, thick] (qezn) circle (0.1cm) node[below,blue] {$q_e^n$};
    \fill[white, draw=black, thick] (qew1) circle (0.1cm); 
    \fill[white, draw=black, thick] (qewn) circle (0.1cm); 
    \fill[white, draw=black, thick] (aiji1) circle (0.1cm);
    \fill[white, draw=black, thick] (aijin) circle (0.1cm);
    \fill[white, draw=black, thick] (biji1) circle (0.1cm);
    \fill[white, draw=black, thick] (bijin) circle (0.1cm);
    \fill[white, draw=black, thick] (ciji1) circle (0.1cm);
    \fill[white, draw=black, thick] (cijin) circle (0.1cm);
    \fill[white, draw=black, thick] (diji1) circle (0.1cm);
    \fill[white, draw=black, thick] (dijin) circle (0.1cm);
    \fill[red] (ti) circle (0.05cm);
    \fill[red] ($(ti) + (0,-1)$) circle (0.05cm);
    \fill[red] (tij) circle (0.05cm);
    \fill[red] ($(tij) + (0,-1)$) circle (0.05cm);
    \fill[red] (qix1) circle (0.05cm);
    \fill[red] (qixx) circle (0.05cm);
    \fill[red] (qixx1) circle (0.05cm);
    \fill[red] (qixn) circle (0.05cm);
    \fill[red] (qez1) circle (0.05cm);
    \fill[red] (qezz) circle (0.05cm);
    \fill[red] (qezz1) circle (0.05cm);
    \fill[red] (qezn) circle (0.05cm);
    \fill[red] (qew1) circle (0.05cm);
    \fill[red] (qewn) circle (0.05cm);
    \fill[red] (shiij) circle (0.05cm);
    \fill[red] (stiij) circle (0.05cm);
    \fill[red] (rhiij) circle (0.05cm);
    \fill[red] (rtiij) circle (0.05cm);
    \draw[black,rounded corners] (4.75, 1.25) rectangle (7.25, 1.75) node[midway] {$A_{i,j}^i$};
    \draw[black,rounded corners] (4.75, 0.25) rectangle (7.25, 0.75) node[midway] {$B_{i,j}^i$};
    \draw[black,rounded corners] (4.75, -0.75) rectangle (7.25,-0.25) node[midway] {$C_{i,j}^i$};
    \draw[black,rounded corners] (4.75, -1.75) rectangle (7.25, -1.25) node[midway] {$D_{i,j}^i$};
    \draw[gray, thick, dashed, rounded corners=15pt] (1.5, -2.5) rectangle (2.5, 2.5);
    \draw[gray, thick, dashed, rounded corners=15pt] (0.6, -3.15) rectangle (2.75, 3.15);
    \draw[gray, thick, dashed, rounded corners=15pt] (-0.5, -3.5) rectangle (3.5, 3.5) node[midway, black, above=3.5cm] {$H_i$};
    \draw[gray, thick, dashed, rounded corners=15pt] (9.5, 2.5) rectangle (11.25, -4);
    \draw[gray, thick, dashed, rounded corners=15pt] (8.5, 3.5) rectangle (12.5, -5) node[midway, black, above=4.25cm] {$H_{i,j}$};
    \draw[gray, thick, dashed, rounded corners=15pt] (4.5, -3.5) rectangle (7.5, 3.5) node[black, midway, above=3.5cm] {$\conn_{i,j}^i$};
    \draw[decorate sep={0.5mm}{1.75mm},fill] (2,1.25) -- (2,0.7);
    \draw[decorate sep={0.5mm}{1.75mm},fill] (2,-0.75) -- (2,-1.3);
    \draw[decorate sep={0.5mm}{1.75mm},fill] (10,1.25) -- (10,0.7);
    \draw[decorate sep={0.5mm}{1.75mm},fill] (10,-0.75) -- (10,-1.3);
    \draw[decorate sep={0.5mm}{1.75mm},fill] (10,-2.75) -- (10,-3.3);
    \draw[decorate sep={0.5mm}{1.75mm},fill] (2,2.7) -- (2,3.1);
    \draw[decorate sep={0.5mm}{1.75mm},fill] (2,-2.7) -- (2,-3.1);
    \end{tikzpicture}
    \caption{\footnotesize An illustration of the reduction of Theorem~\ref{thm:VC-D-fvs}. For $\iljk$, the vertex-block $H_i$ and the edge-block $H_{i,j}$ are connected to the connector $\conn_{i,j}^i$. The initial configuration is denoted by vertices with red circle. }
    \label{fig:VC-D-fvs-illustration}
\end{figure}
\begin{lemma}
    \label{lem:VC-D-fvs-bound}
    The \emph{fvs} of the graph $H$ is at most $8\kctwo$. 
\end{lemma}
\begin{proof}
    Let $F = \{\Tilde{s}_{i,j}^l, \hat{s}_{i,j}^l, \Tilde{r}_{i,j}^l, \hat{r}_{i,j}^l \mid \iljk, l\in \{i,j\}\}$. 
    The removal of $F$ from $H$ results in a forest. 
    Therefore, the \emph{fvs} of $H$ is at most $|F| = 8\binom{\kappa}{2}$. 
\end{proof}
Next we prove the correctness of the reduction.
\begin{lemma}
\label{lem:VC-D-fvs-forward}
    If $(G, \kappa)$ is a yes-instance of the \mcc~problem, then $(H, \varS, \budget)$ is a yes-instance of the \textsc{VC-D} problem. 
\end{lemma}
\begin{proof}
    Let $C = \subseteq V(G)$ be a $\kappa$-clique in $G$. 
    For each $i \in [\kappa]$, let $u_{i,x_i}$ be the vertex in $C \cap V_i$ for some $x_i \in [n]$. 
    For each $\iljk$, let $e_{i,j} = u_{i,x_i}u_{j,x_j}$.
    For each $i \in [\kappa]$, we slide the token on~$\Tilde{t}_i$ to $p_{i,x_i}$. 
    Then, for each $j \not= i \in [\kappa]$, we slide $x_i$-tokens in $Q_{i,x_i}^j$ towards $\Tilde{s}^i_{i,j}$ and $n-x_i$-tokens in~$Q_{i,x_i}^j$ towards $\Tilde{r}_{i,j}^i$. 
    For each $\iljk$, we slide the token on $\Tilde{t}_{i,j}$ to $p_{e_{i,j}}$. 
    Then, we slide 
    \begin{itemize}
        \item $n-x_i$ tokens in $Q_{e_{i,j}}$ to $\hat{s}_{i,j}^i$,
        \item $x_i$ tokens in $Q_{e_{i,j}}$ to $\hat{r}_{i,j}^i$,
        \item $n-x_j$ tokens in $Q_{e_{i,j}}$ to $\hat{s}_{i,j}^j$, and
        \item $x_j$ tokens in $Q_{e_{i,j}}$ to $\hat{r}_{i,j}^j$. 
    \end{itemize}
    
    The tokens received at the vertices $\Tilde{s}_{*,*}^*$ and $\hat{s}_{*,*}^*$ are pushed to $s_{*,*}^*$. 
    Similarly, the tokens received at the vertices $\Tilde{r}_{*,*}^*$ and $\hat{r}_{*,*}^*$ are pushed to $r_{*,*}^*$. 
    For each $\iljk$, and for each $l \in {i,j}$, $s_{i,j}^l$ receives $x_l$-tokens from $H_l$ and $n-x_l$-tokens from $H_{i,j}$.
    Similarly, $r_{i,j}^l$ receives $n-x_l$-tokens from~$H_l$ and $x_l$-tokens from $H_{i,j}$. 
    Further, we push the $n$-tokens received by $s_{i,j}^l$ to $A_{i,j}^l$ and $n$-tokens received by $r_{i,j}^l$ to $D_{i,j}^l$. 
    The resulting configuration is a valid vertex cover. 
    Finally, let $S' \subseteq V(H)$ be the solution obtained from the above token sliding steps. 
    More precisely,
    \begin{eqnarray*}
        S' &=& \bigcup_{i \in [\kappa]}\big\{\{t_i, p_{i, x_i}\} \cup (Q_{i} \setminus Q_{i, x_i})\big\} \cup \bigcup_{\iljk} \big\{\{t_{i,j}, p_{e_{i,j}}\} \cup  (Q_{i,j} \setminus Q_{e_{i,j}})\big\}\\
        && \cup \bigcup_{\iljk, l \in \{i,j\}} (\{\Tilde{s}_{i,j}^l, \hat{s}_{i,j}^l, \Tilde{r}_{i,j}^l, \hat{r}_{i,j}^l\}\cup A_{i,j}^l \cup D_{i,j}^l) .
    \end{eqnarray*}
    
    It is clear that the set $S'$ is a vertex cover in $H$. 
    Next we count the number of token steps used to obtain $S'$ from $\varS$. 
    In each vertex-block, we spend $2(\kappa-1)n+2$ steps to push tokens towards the connectors. 
    Similarly, at each edge-block, we spend $(4n+2)$ steps. 
    At each connectors, we spend $2n$ steps. 
    Therefore, we spend $\kappa\cdot\big(2(\kappa-1)n+2\big) + \binom{\kappa}{2}\cdot(4n+2) + 2\binom{\kappa}{2} \cdot 2n = (12n+2)\binom{\kappa}{2} + 2\kappa = \budget$. 
    Hence, $(H, \varS, \budget)$ is a yes-instance of VC-D problem.
\end{proof}

\begin{lemma}
\label{lem:VC-D-fvs-reverse}
    If $(H, \varS, \budget)$ is a yes-instance of the VC-D problem, then $(G, \kappa)$ is a yes-instance of the \mcc~problem. 
\end{lemma}
\begin{proof}
    Let $S^*$ be a feasible solution for the instance $(H, \varS, \budget)$ of the VC-D problem. 
    At each connector $\conn_{i,j}^l$ for $\iljk$ and $l \in \{i,j\}$, at least $2n$ tokens need to be slid from either vertex-blocks or edge blocks.
    Because, each $H[\conn_{i,j}^l]$ contains a matching of size $2n+4$, but it has only four tokens in the initial configuration. 
    It is clear that every token must move at least 3 steps to reach the sets $A_{*,*}^*$ and $D_{*,*}^*$ in order to cover the uncovered edges. 
    This saturates a budget of $4n\cdot 2\kctwo = 12n\kctwo$. 
    Therefore, we left with exactly $2\kappa+2\kctwo$ budget to adjust the tokens on the vertex blocks and edge blocks. 
    For any $\iljk$, let $q_{i,x}^{j,z}$ for some integers $x,z \in [n]$ be a vertex that looses the token where the token is moved to some vertex in a connector. 
    Since the edge $p_{i,x}q_{i,x}^{j,z}$ become uncovered, we need to slide a token to $q_{i,x}^{j,z}$ or $p_{i,x}$. 
    This cost at least two token step. 
    By construction of the vertex-block $H_i$, by sliding a token (for two steps) to the vertex $p_{i,x}$ for some $x \in [n]$, one can release at most $n(\kappa-1)$ tokens from the neighboring set $Q_{i,x}$. 
    Similarly, on a edge-block $H_{i,j}$ for some $\iljk$, by sliding a token (for two steps) to the vertex $p_e$ for some $e \in E_{i,j}$, one can release at most $2n$ tokens from the neighboring set $Q_e$. 
    This implies that by sliding at most $2\kappa$ tokens on the vertex-blocks, one can release at most $\kappa\cdot n(\kappa-1) = 2n\kctwo$ token from the vertex-blocks. 
    Similarly, by sliding at most $2\kctwo$ tokens on the edge-blocks, one can release at most $2n\kctwo$ tokens from the edge-blocks. 
    Therefore, we need to slide exactly one token (for two steps) in each vertex-block and each edge-block. 
    
    For each $i\in [\kappa]$, let $p_{i,x_i}$ for some $x_i \in [n]$ be the vertex in $H_i$ that gets token in $S^*$ and releases all the tokens in $Q_{i,x_i}$. 
    Similarly, for each $\iljk$, let $p_e$ for some $e=u_{i,z_i}u_{i,z_j} \in E_{i,j}$ with $z_i,z_j \in [n]$ be the vertex in $H_{i,j}$ that gets token in $S^*$ and releases all tokens in $Q_{e}$. 
    Consider the connector $\conn_{i,j}^i$. 
    The set $Q_{i,x_i}^j$ pushes $x_i$ tokens to $s_{i,j}^i$ and $n-x_i$ tokens to $r_{i,j}^i$. 
    The set $Q_e$ pushes $z_i$ tokens to $r_{i,j}^i$ and $n-z_i$ tokens to $s_{i,j}^i$. 
    The number of tokens passed through $s_{i,j}^i$ to $A_{i,j}^i$ is $x_i + (n - z_i)$. 
    Since $A_{i,j}^i$ need $n$ tokens, it is mandatory that $x_i=z_i$. 
    This equality should hold for every $i$. 
    Therefore, for each $\iljk$, there exist an edge $u_{i,x_i}u_{j,x_j}$. 
    Hence $(G,\kappa)$ is an yes-instance of the \mcc~problem.     
\end{proof}
The proofs of Lemmas~\ref{lem:VC-D-fvs-bound}, \ref{lem:VC-D-fvs-forward} and \ref{lem:VC-D-fvs-reverse} complete the proof of Theorem~\ref{thm:VC-D-fvs}. 

\section{Independent Set Discovery}
\label{sec:is}
Fellows et al.~\cite{fellows2023solution} showed that \textsc{IS-D} is in $\FPT$ with respect to parameter $k + b$ on nowhere dense classes of graphs.
We show in this section that \textsc{IS-D} has a polynomial kernel with respect to parameter $k$ on nowhere dense classes of graphs and does not admit a polynomial kernel with respect to the parameter $\budget + pw$, where $pw$ is the pathwidth of the input graph, unless $\NP \subseteq \cp$. 

\begin{definition} 
For any instance $\Imc = (G, S, \budget$) of a \textsc{$\Qmc$-Discovery} problem for some vertex (resp.\ edge) selection problem $\Qmc$, we call a vertex $v \in V(G) \setminus S$ (resp.\ $e \in E(G) \setminus S$) \emph{irrelevant} with respect to $s \in S$ if there exists a configuration $C_\ell$ such that $\ell \leq b$, $C_\ell$ is a solution for $\Qmc$, and the token on $s$ is not on~$v$ (resp.\ $e$) in $C_\ell$.
\end{definition}

The kernelization algorithm for nowhere dense graphs uses \Cref{thm:quasi-wide-bounds}, along with other structural properties of the input graph, to form a ``sunflower'' and find an irrelevant vertex. It then removes from the graph some of the vertices that are irrelevant with respect to every token. 
A \emph{sunflower} with $p$ \emph{petals} and a \emph{core} $Y$ is a family of sets $F_1$, $\ldots$, $F_p$ such that $F_i \cap F_j = Y$ for all $i \neq j$; the sets $F_i \setminus Y$ are petals and we require none of them to be empty~\cite{erdos1960intersection}. 

\begin{lemma}\label{lem:IS-always-at-distance-3k}
Let $(G, S, \budget)$ be an instance of \textsc{IS-D} where $|S|=k$, and let $G'$ be the subgraph of~$G$ induced by the vertices of $\bigcup_{s \in S,i \in [3k]} V(s,i) \text{ } \cup \text{ } S$. Then $(G', S, \budget)$ is equivalent to $(G, S, \budget)$.
\end{lemma}

\begin{proof}
We show that in any solution to $(G, S, \budget)$, if a token on a vertex $s \in S$ moves to a vertex $u \in C_\ell$ such that $d(s, u) > 3k$, it can instead move towards a vertex $v \in V(H)$ such that $d(s, v) \le 3k$, while keeping the rest of the solution unchanged. 
First, the vertices in $C_\ell \setminus \{u\}$ can appear in at most $k - 1$ of the $3k$ sets $V(s,i)$ for $i \in [3k]$ and every such vertex that appears in a set $V(s,\mathfrak{i})$ for a specific $\mathfrak{i} \in [3k]$ can be a neighbor of at most all the vertices in $V(s,\mathfrak{i}-1)$ and $V(s,\mathfrak{i}+1)$. 
This implies that the token on $s$ cannot move towards any vertex of at most $3k - 3$ of the $3k$ sets $V(s,i)$ for $i \in [3k]$ (as these contain tokens and thus might result in adjacent tokens) and thus there exists a vertex $v$ which the token on $s$ can move to while maintaining an independent set in $C_\ell \setminus \{u\} \cup \{v\}$. 
\end{proof}

\begin{lemma}\label{lem:one-sunflower-petal}
Let $(G, S, \budget)$ be an instance of \textsc{IS-D} where $|S|=k$, and let $\mathcal{V} = \{v_1, v_2, \ldots, v_t\}$ be a set of vertices of $G \setminus S$ such that for a given token on a vertex $s \in S$, $d(s, v_i) = d(s, v_j)$ for $i \neq j \in [t]$. If $\mathcal{A} = \{N[v_1], \ldots, N[v_t]\}$ contains a sunflower with $k + 1$ petals, then any vertex whose closed neighborhood corresponds to one of those petals is irrelevant with respect to $s$. 
\end{lemma}

\begin{proof}
Let $v_{del}$ be one such vertex whose closed neighborhood corresponds to one of the $k + 1$ petals, and consider a solution to $(G, S, \budget)$ in which the token on $s$ is on the vertex $v_{del}$ in $C_\ell$.
First, note that no vertex of $C_\ell \setminus \{v_{del}\}$ is in the core of the sunflower since that would contradict the fact that $v_{del} \in C_\ell$ as $v_{del}$ is a neighbor of every vertex in the core.
Second, since the sunflower has~$k$ remaining petals (besides the one corresponding to the closed neighborhood of $v_{del}$) and $|C_\ell \setminus \{v_{del}\}| = k - 1$, there must remain one vertex (denoted $v_{rep}$) whose closed neighborhood corresponds to one of the remaining $k$ petals and such that $v_{rep}$ is not in the closed neighborhood of any of the vertices in $C_\ell \setminus \{v_{del}\}$.
Thus, $C_\ell \setminus \{v_{del}\} \cup \{v_{rep}\}$ forms an independent set. 
Additionally, since $d(s, v_{del}) = d(s, v_{rep})$, the number of slides in the solution remains constant.    
As a result, $v_{del}$ is irrelevant with respect to $s$.
\end{proof}

\begin{theorem}\label{thm:nowheredense-IS-k}
\textsc{IS-D} has a polynomial kernel with respect to parameter $k$ on nowhere dense graphs. 
\end{theorem}

\begin{proof}
Let $(G, S, \budget)$ be an instance of \textsc{IS-D} where $G$ is nowhere dense. 
Without loss of generality, we assume the graph $G$ to be connected.
For each vertex $s \in S$ and integer $i \in [3k]$, we compute~$V(s,i)$.
We maintain the invariant that we remove from $V(s,i)$ for each $s \in S$ and $i \in [3k]$, irrelevant vertices with respect to $s$ (note that a vertex can appear in multiple sets $V(s,i)$).

We remove an irrelevant vertex with respect to a vertex $s \in S$ from $V(s,i)$ for an integer $i \in [3k]$ as follows. 
If $|V(s,i)| > N_2(2^{x_2} \cdot (k + 1))$, where $N_2$ and $x_2$ are as per \Cref{thm:quasi-wide-bounds} (here $V(s,i)$ plays the role of the set $A$), we can compute sets $X, B \subseteq V(s,i)$ such that $|X| \le x_2$, $|B| \ge 2^{x_2} \cdot (k + 1)$ and $B$ is $2$-independent in $G - X$. 
Let $\mathcal{B}' = \{B'_1, B'_2, \ldots \}$ be a family of sets that partitions the vertices in $B$ such that for any two vertices $u, v \in B$, $u, v \in B'_j$ if and only if $N[u] \cap X = N[v] \cap X$.
Since $|B| \ge 2^{x_2} \cdot (k + 1)$ and $|X| \le x_2$, at least one set $B_{\mathfrak{j}} \in \mathcal{B}$, for a specific~$\mathfrak{j}$, contains at least $k + 1$ vertices of $B$. 
All vertices in $B_\mathfrak{j}$ have the same neighborhood in $X$ and they are $2$-independent~$G-X$ (\ie, no vertex from outside of $X$ can be in the closed neighborhood of two vertices in~$B_{\mathfrak{j}}$); thus their closed neighborhoods form a sunflower with at least $k + 1$ petals and a core that is contained in $X$ (\Cref{fig:IS-nowhere-dense}).
By \Cref{lem:one-sunflower-petal}, one vertex of $B_\mathfrak{j}$ is irrelevant with respect to $s$ and can be removed from~$V(s,i)$.
We can repeatedly apply \Cref{thm:quasi-wide-bounds} on the set 
$V(s,i)$ until $|V(s,i)| \le N_2(2^{x_2} \cdot (k + 1))$.

We form the kernel $(G', S, \budget)$ of the original instance $(G, S, \budget)$ as follows. 
We set $V(G') = \bigcup_{s \in S,i \in [3k]} V(s,i) \text{ } \cup \text{ } S$.
By \Cref{lem:IS-always-at-distance-3k}, any vertex $v \in V(G)$ such $d(s,v) > 3k$ for every $s \in S$ is irrelevant with respect to every $s \in S$ and not required in the kernel $(G', S, \budget)$.
For each vertex $v \in V(s,i)$, for $s \in S$ and $i \in [3k]$, we add to $V(G')$ at most $i$ vertices that are on the shortest path from $s$ to $v$, if such vertices are not already present in $V(G')$. 
$G'$ is the subgraph of $G$ induced by the vertices in $V(G')$.
By the end of this process, $|V(G')| \le k + [9k^3 \cdot N_2(2^{x_2} \cdot (k + 1))]$, as for each $s \in S$ and $i \in [3k]$, $V(s,i) \le N_2(2^{x_2} \cdot (k + 1))$ and for each vertex in the latter sets, we added to~$V(G')$ at most $3k - 1$ vertices that are on a shortest path from that vertex to the vertex $s$.
$(G', S, \budget)$ is a kernel as only vertices that are irrelevant with respect to every token in $S$ might not be in $V(G')$ and all vertices needed to move tokens from vertices in $S$ towards an independent set using only $b$ slides are present in $V(G')$.
\end{proof}

\begin{figure}[H]
\centering
\begin{tikzpicture}[scale=0.5]

% Draw the sets as ovals
\draw[thick, draw=black, fill=teal!20, opacity=0.5, rotate around={30:(0,0)}] (0,0) ellipse (2 and 0.9);
\draw[thick, draw=black, fill=red!20, opacity=0.5, rotate around={-30:(2,0)}] (2,0) ellipse (2 and 0.9);
\draw[thick, draw=black, fill=orange!20, opacity=0.5] (1,1.6) ellipse (0.9 and 2);

% Draw the intersection area
\begin{scope}
\clip[rotate around={-30:(2,0)}] (2,0) ellipse (2 and 0.9);
\clip (1,1.6) ellipse (0.9 and 2);
\fill[pattern=north east lines] (-1,-1) rectangle (4,4);
\end{scope}

% Draw the intersection area
\begin{scope}
\clip[rotate around={30:(0,0)}] (0,0) ellipse (2 and 0.9);
\clip (1,1.6) ellipse (0.9 and 2);
\fill[pattern=north east lines] (-1,-1) rectangle (4,4);
\end{scope}

% Draw vertices outside the intersection
\node[draw, circle, fill=black, inner sep=1pt] at (-0.7, 0) {};
\node[draw, circle, fill=black, inner sep=1pt] at (-1, -1) {};
\node[draw, circle, fill=black, inner sep=1pt] at (2, 0.5) {};
\node[draw, circle, fill=black, inner sep=1pt] at (3, -1) {};
\node[draw, circle, fill=black, inner sep=1pt] at (1, 3) {};
\node[draw, circle, fill=black, inner sep=1pt] at (0.4, 2.2) {};
\node[draw, circle, fill=teal, inner sep=1pt, opacity=0.6] at (0.7, 0.8) {};
\node[draw, circle, fill=teal, inner sep=1pt, opacity=0.6] at (1.2, 0.4) {};
\node[draw, circle, fill=teal, inner sep=1pt, opacity=0.6] at (1, 0) {};
\end{tikzpicture}
    \caption{\footnotesize An example of a sunflower (with pink, beige and blue petals) formed by the closed neighborhoods of the vertices in $B_\mathfrak{j}$ of \Cref{thm:nowheredense-IS-k}. The vertices in $B_\mathfrak{j}$ are $2$-independent in $G-X$ and they have the same closed neighborhood in $X$ (the blue colored vertices).}
    \label{fig:IS-nowhere-dense}
\end{figure}
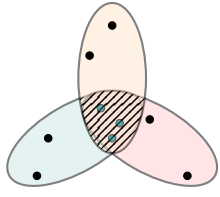

\begin{theorem}\label{thm:IS-D-pathwidth}
\textsc{IS-D} is \XNLP-hard with respect to parameter pathwidth. 
\end{theorem}
\noindent 

As stated in \Cref{sec:foundational}, we present an \textsf{FPT}-reduction from \textsc{MMO}. 
Let $(H, \mathcal{P}_H, \sigma, r)$ be an instance of \textsc{MMO} where $H$ is a bounded pathwidth graph, $|V(H)| = n$, $|E(H)| = m$, $\sigma: E(H) \rightarrow \mathbb{Z}_+$ such that $\sum_{e \in E(H)} \sigma(e) = \bm{\sigma}$ and $r \in \mathbb{Z}_{+}$ (integers are given in unary).
We construct an instance $(G_H, \mathcal{P}_{G_H}, S, \budget)$ of \textsc{IS-D} where $G_H$ is exactly as described in \Cref{sec:foundational} (under the graph $G_H$ heading). 
See \Cref{fig:IS-pathwidth-reduction}.
From \Cref{lem:bounded-pathwidth-G-H}, $G_H$ is of bounded pathwidth.
We set $S = A \text{ } \cup \text{ } A^+ \text{ } \cup \text{ } B \text{ } \cup \text{ } B^+ \text{ } \cup \text{ } Y \text{ } \cup \text{ } X^+$ and $\budget = m + 3 \bm{\sigma}$.
Given that all integers are given in unary, the construction of the graph $G_H$, or its path decomposition (as described in \Cref{lem:bounded-pathwidth-G-H}), and as a consequence the reduction, take time polynomial in the size of the input instance. 
Additionally, by \Cref{cor:reduction-space}, this reduction is a pl-reduction.
We claim that $(H, \mathcal{P}_H, \sigma, r)$ is a yes-instance of \textsc{MMO} if and only if $(G_H, \mathcal{P}_{G_H}, S, \budget)$ is a yes-instance of \textsc{IS-D}.
\begin{figure}[ht]
    \centering
    \begin{tikzpicture}
    \fill[blue!10] (-0.75,2.75) -- (0.75,2.75) -- (0.75,-0.25) -- (-0.75,-0.25) -- cycle; 
    \fill[blue!10] (0.5,0.1) -- (1.25,0.9) -- (1.25,0) -- (0.5,0) -- cycle; 
    \fill[blue!10] (0.5,-0.1) -- (1.25,-0.9) -- (1.25,0) -- (0.5, 0) -- cycle; 
    \fill[blue!10] (-0.75,-2) -- (0.75,-2) -- (0.75,-3.25) -- (-0.75,-3.25) -- cycle; 
    \fill[blue!10] (0.5,-2.9) -- (1.25,-2) -- (1.25,-3) -- (0.5,-3) -- cycle; 
    \fill[blue!10] (0.5,-3.1) -- (1.25,-4) -- (1.25,-3) -- (0.5, -3) -- cycle; 
    \fill[green!10] (8,1) -- (7.5,2.5) -- (8.5,2.5) -- (8,1) -- cycle;
    \fill[green!10] (8,1) -- (6,0) -- (6,3) -- (8,1) -- cycle;
    \fill[green!10] (8,1) -- (10,-0.25) -- (10,2.5) -- (8,1) -- cycle;
    \fill[green!10] (6,0) -- (6,3) -- (4.75,3) -- (4.75,0) -- cycle;
    \fill[green!10] (8,-1.5) -- (7.5,0) -- (8.5,0) -- (8,-1.5) -- cycle;
    \fill[green!10] (8,-1.5) -- (6,-1.75) -- (6,-1.25) -- (8,-1.5) -- cycle;
    \fill[green!10] (8,-1.5) -- (10,0) -- (10,-3) -- (8,-1.5) -- cycle;
    \fill[green!10] (6,-1.75) -- (6,-1.25) -- (4.75,-1.25) -- (4.75,-1.75) -- cycle;
    \fill[green!10] (8,-4.7) -- (7.5,-3.2) -- (8.5,-3.2) -- (8,-4.7) -- cycle;
    \fill[green!10] (8,-4.7) -- (6,-6.25) -- (6,-2.75) -- (8,-4.7) -- cycle;
    \fill[green!10] (8,-4.7) -- (10,-2.5) -- (10,-5.75) -- (8,-4.7) -- cycle;
    \fill[green!10] (6,-6.25) -- (6,-2.75) -- (4.75,-2.75) -- (4.75,-6.25) -- cycle;
    \draw[black] (8,1) -- (9.5,0.25);
    \draw[black] (8,1) -- (9.5,1);
    \draw[black] (8,1) -- (9.5,1.75);
    \draw[black] (8,1) -- (8,2);
    \draw[black] (8,1) -- (6,0.25);
    \draw[black] (8,1) -- (6,2.25);
    \draw[black] (8,1) -- (6,1.25);
    \draw[black] (8,-1.5) -- (9.5,-2.5);
    \draw[black] (8,-1.5) -- (9.5,-1.5);
    \draw[black] (8,-1.5) -- (9.5,-0.5);
    \draw[black] (8,-1.5) -- (6,-1.5);
    \draw[black] (8,-1.5) -- (8,-0.5);
    \draw[black] (8,-4.7) -- (9.5,-5.2);
    \draw[black] (8,-4.7) -- (9.5,-4.2);
    \draw[black] (8,-4.7) -- (9.5,-3.2);
    \draw[black] (8,-4.7) -- (8,-3.75);
    \draw[black] (8,-4.7)-- (6,-6);
    \draw[black] (8,-4.7) -- (6,-5.2);
    \draw[black] (8,-4.7)-- (6,-4.2);
    \draw[black] (8,-4.7) -- (6,-3.4);
    \draw[yellow] (1,-2.5) .. controls (2,-2.25) .. (6,-1.5);
    \draw[yellow] (1,-3.5) .. controls (2.5,-5) and  (5,-5) .. (6,-6);    
    \draw[red] (0.5,-2.25) .. controls (1.5,-1.25) .. (5,-1.5);
    \draw[red] (0.5,-2.25) .. controls (2,-4.7) and  (5,-6) .. (5,-6);  
    \draw[red] (0.5,0.75) .. controls (1,1) .. (5,0.25);
    \draw[red] (0.5,1.5) .. controls (1,2.5) and  (2,2) .. (5,1.25);
    \draw[red] (0.5,2.25) .. controls (1,3) and  (2,3.5) .. (5,2.25);  
    \node[circle, fill=black, inner sep=1pt, draw=black, fill=white, label=below:{\scalebox{0.5}{$w_u$}}] at (8,1) {};
    \node[circle, fill=black, inner sep=1pt, draw=black, fill=white, label=below:{\scalebox{0.5}{$w_w$}}] at (8,-1.5) {};
    \node[circle, fill=black, inner sep=1pt, draw=black, fill=white, label=below:{\scalebox{0.5}{$w_v$}}] at (8,-4.7) {};
    \node[circle, fill=black, inner sep=1pt, draw=black, fill=white, label=right:{\scalebox{0.5}{$x_u^1$}}] at (9.5,0.25) {};
    \node[circle, fill=black, inner sep=1pt, draw=black, fill=white, label=right:{\scalebox{0.5}{$x_u^2$}}] at (9.5,1) {};
    \node[circle, fill=black, inner sep=1pt, draw=black, fill=white, label=right:{\scalebox{0.5}{$x_u^3$}}] at (9.5,1.75) {};
    \node[circle, fill=black, inner sep=1pt, draw=black, fill=white, label=right:{\scalebox{0.5}{$x_w^1$}}] at (9.5,-2.5) {};
    \node[circle, fill=black, inner sep=1pt, draw=black, fill=white, label=right:{\scalebox{0.5}{$x_w^2$}}] at (9.5,-1.5) {};    
    \node[circle, fill=black, inner sep=1pt, draw=black, fill=white, label=right:{\scalebox{0.5}{$x_w^3$}}] at (9.5,-0.5) {};
    \node[circle, fill=black, inner sep=1pt, draw=black, fill=white, label=right:{\scalebox{0.5}{$x_v^1$}}] at (9.5,-5.2) {};
    \node[circle, fill=black, inner sep=1pt, draw=black, fill=white, label=right:{\scalebox{0.5}{$x_v^2$}}] at (9.5,-4.2) {};
    \node[circle, fill=black, inner sep=1pt, draw=black, fill=white, label=right:{\scalebox{0.5}{$x_v^3$}}] at (9.5,-3.2) {};
    \node[circle, fill=black, inner sep=1pt, draw=black, fill=black, label=above right:{\scalebox{0.5}{$x_u^{4}$}}] at (8,2) {};
    \node[circle, fill=black, inner sep=1pt, draw=black, fill=black, label=above right:{\scalebox{0.5}{$x_w^{4}$}}] at (8,-0.5) {};
    \node[circle, fill=black, inner sep=1pt, draw=black, fill=black, label=above right:{\scalebox{0.5}{$x_v^{4}$}}] at (8,-3.75) {};
    \draw[black] (5,-1.5) -- (6,-1.5) node[circle, draw=black, fill=black, inner sep=1pt, label=above:{\scalebox{0.5}{$y_{e_2}^{w(1)}$}}, pos=1]{} node[circle, draw=black, fill=white, inner sep=1pt, label=above:{\scalebox{0.5}{$z_{e_2}^{w(1)}$}}, pos=0]{};  
    \draw[black] (5,1.25) -- (6,1.25) node[circle, draw=black, fill=black, inner sep=1pt, label=above:{\scalebox{0.5}{$y_{e_1}^{u(2)}$}}, pos=1]{} node[circle, draw=black, fill=white, inner sep=1pt, label=above:{\scalebox{0.5}{$z_{e_1}^{u(2)}$}}, pos=0]{};
    \draw[black] (5,0.25) -- (6,0.25) node[circle, draw=black, fill=black, inner sep=1pt, label=above:{\scalebox{0.5}{$y_{e_1}^{u(1)}$}}, pos=1]{} node[circle, draw=black, fill=white, inner sep=1pt, label=above:{\scalebox{0.5}{$z_{e_1}^{u(1)}$}}, pos=0]{};
    \draw[black] (5,2.25) -- (6,2.25) node[circle, draw=black, fill=black, inner sep=1pt, label=above:{\scalebox{0.5}{$y_{e_1}^{u(3)}$}}, pos=1]{} node[circle, draw=black, fill=white, inner sep=1pt, label=above:{\scalebox{0.5}{$z_{e_1}^{v(3)}$}}, pos=0]{};   
    \draw[black] (5,-3.4) -- (6,-3.4) node[circle, draw=black, fill=black, inner sep=1pt, label=above:{\scalebox{0.5}{$y_{e_1}^{v(2)}$}}, pos=1]{} node[circle, draw=black, fill=white, inner sep=1pt, label=above:{\scalebox{0.5}{$z_{e_1}^{v(2)}$}}, pos=0]{};
    \draw[black] (5,-4.2) -- (6,-4.2) node[circle, draw=black, fill=black, inner sep=1pt, label=above:{\scalebox{0.5}{$y_{e_1}^{v(3)}$}}, pos=1]{} node[circle, draw=black, fill=white, inner sep=1pt, label=above:{\scalebox{0.5}{$z_{e_1}^{v(3)}$}}, pos=0]{};
    \draw[black] (5,-5.2) -- (6,-5.2) node[circle, draw=black, fill=black, inner sep=1pt, label=above:{\scalebox{0.5}{$y_{e_1}^{v(1)}$}}, pos=1]{} node[circle, draw=black, fill=white, inner sep=1pt, label=above:{\scalebox{0.5}{$z_{e_1}^{v(1)}$}}, pos=0]{};
    \draw[black] (5,-6) -- (6,-6) node[circle, draw=black, fill=black, inner sep=1pt, label=above:{\scalebox{0.5}{$y_{e_2}^{v(1)}$}}, pos=1]{} node[circle, draw=black, fill=white, inner sep=1pt, label=above:{\scalebox{0.5}{$z_{e_2}^{v(1)}$}}, pos=0]{};
    \draw[black] (-0.5,0.75) -- (0.5,0.75) node[circle, draw=black, fill=black, inner sep=1pt, label=above:{\scalebox{0.5}{$a^1_{e_1}$}}, pos=0]{} node[circle, draw=black, fill=black, inner sep=1pt, label=above:{\scalebox{0.5}{$b^1_{e_1}$}}, pos=1]{};
    \draw[black] (-0.5,1.5) -- (0.5,1.5) node[circle, draw=black, fill=black, inner sep=1pt, label=above:{\scalebox{0.5}{$a^2_{e_1}$}}, pos=0]{} node[circle, draw=black, fill=black, inner sep=1pt, label=above:{\scalebox{0.5}{$b^2_{e_1}$}}, pos=1]{};
    \draw[black] (-0.5,2.25) -- (0.5,2.25) node[circle, draw=black, fill=black, inner sep=1pt, label=above:{\scalebox{0.5}{$a^3_{e_1}$}}, pos=0]{} node[circle, draw=black, fill=black, inner sep=1pt, label=above:{\scalebox{0.5}{$b^3_{e_1}$}}, pos=1]{};
    \draw[black] (1,0.5) -- (1, -0.5);
    \draw[black] (0.5,0) -- (1,0.5) node[circle, draw=black, fill=white, inner sep=1pt, label=right:{\scalebox{0.5}{$e_1^v$}}, pos=1]{};
    \draw[black] (0.5,0) -- (1,-0.5) node[circle, draw=black, fill=white, inner sep=1pt, label=right:{\scalebox{0.5}{$e_1^u$}}, pos=1]{};
    \draw[black] (-0.5,0) -- (0.5,0) node[circle, draw=black, fill=black, inner sep=1pt, label=above:{\scalebox{0.5}{$a^4_{e_1}$}}, pos=0]{} node[circle, draw=black, fill=black, inner sep=1pt, label=above:{\scalebox{0.5}{$b^4_{e_1}$}}, pos=1]{};
    \draw[black] (-0.5,-2.25) -- (0.5,-2.25) node[circle, draw=black, fill=black, inner sep=1pt, label=above:{\scalebox{0.5}{$a^1_{e_2}$}}, pos=0]{} node[circle, draw=black, fill=black, inner sep=1pt, label=above:{\scalebox{0.5}{$b^1_{e_2}$}}, pos=1]{};
    \draw[black] (1,-2.5) -- (1, -3.5);
    \draw[black] (0.5,-3) -- (1,-2.5) node[circle, draw=black, fill=white, inner sep=1pt, label=right:{\scalebox{0.5}{$e_2^w$}}, pos=1]{};
    \draw[black] (0.5,-3) -- (1,-3.5) node[circle, draw=black, fill=white, inner sep=1pt, label=right:{\scalebox{0.5}{$e_2^u$}}, pos=1]{};
    \draw[black] (-0.5,-3) -- (0.5,-3) node[circle, draw=black, fill=black, inner sep=1pt, label={above:{\scalebox{0.5}{$a^2_{e_2}$}}}, pos=0]{} node[circle, draw=black, fill=black, inner sep=1pt, label=above:{\scalebox{0.5}{$b^2_{e_2}$}}, pos=1]{};
    \end{tikzpicture}
    \caption{\footnotesize Parts of the graph $G_H$ constructed by the reduction of \Cref{thm:IS-D-pathwidth} given an instance $(H, \mathcal{P}_H, \sigma, r)$, where $H$ has three vertices $u, v$ and $w$, and two edges $e_1=uv$ and $e_2=uw$, and $r = 3$. Additionally, $\sigma(e_1) = 3$ and $\sigma(e_2) = 1$. For clarity, the edges between the vertices in $B_{e_1}$ and $Z_{e_1}^u$ are missing. The same applies for the edges between $e_1^v$ and the vertices of $Y_{e_1}^v$ and the edges between $e_1^u$ and the vertices of $Y_{e_1}^u$. 
    Red and yellow edges are used to highlight the different types of edges used to connect the subgraphs $G_{e_1}, G_{e_2}, G_{u}, G_{v}$ and $G_{w}$ of $G_H$, vertices in black are in $S$ and those in white are not.}
    \label{fig:IS-pathwidth-reduction}
\end{figure}
\begin{lemma}\label{lem:hardness-IS-pathwidth-forward}
If $(H, \mathcal{P}_H, \sigma, r)$ is a yes-instance of \textsc{MMO}, then $(G_H, \mathcal{P}_{G_H}, S, \budget)$ is a yes-instance of \textsc{IS-D}.
\end{lemma}
\begin{proof}
Let $\lambda: E(H) \rightarrow V(H) \times V(H)$ be an orientation of the graph $H$ such that for each $v \in V(H)$, the total weight of the edges directed out of $v$ is at most $r$.
In $(H, \mathcal{P}_H, \sigma, r)$, the vertices in~$A$ and~$B$ contain tokens. 
The same applies for the vertices in $A^+$ and $B^+$.
To fix that, for each edge $e \in E(H)$ such that $\lambda(e) = (v, u)$:
\begin{enumerate}[itemsep=0pt]
    \item we slide, for each $i \in [\sigma(e)]$, the token on $b_e^i$ to $z_e^{v(i)}$ (this consumes $\sigma(e)$ slides),
    \item we move, for each $i \in [\sigma(e)]$, the token on $y_e^{v(i)}$ to any free vertex of $X_v$ (this consumes $2\sigma(e)$ slides),
    \item we slide the token on $b_e^{\sigma(e)+1}$ to $e^v$ (this consumes 1 slide).
\end{enumerate}

This constitutes $m + 3\bm{\sigma}$ slides and we get an independent set in~$G_H$. 
Step 2 above is possible (i.e.\ a token-free vertex exists in $X^v$) since $\lambda$ is an orientation of the graph~$H$ such that for each $v \in V(H)$, the total weight of the edges directed out of $v$ is at most $r$.
Step 3 is possible for each edge $e \in E(H)$ since in Step 2 all tokens were removed from the vertices in $Y_e^v$.
\end{proof}

\begin{lemma}\label{lem:hardness-IS-pathwidth-backward}
If  $(G_H, \mathcal{P}_{G_H}, S, \budget)$ is a yes-instance of  \textsc{IS-D}, then $(H, \mathcal{P}_H, \sigma, r)$ is a yes-instance of \textsc{MMO}.
\end{lemma}

\begin{proof}
The minimum number of slides used inside any induced subgraph $G_e$ for an edge $uv = e \in E(H)$ is one and it can only be achieved by sliding the token on $b_e^{\sigma(e)+1}$ to one of either $e^u$ or~$e^v$.
Thus, at least $m$ slides are required inside the \textsc{MMO}-edge-gadgets and the budget remaining is~$3\bm{\sigma}$. 
Additionally, each token on a vertex $b^i_e$ in~$B_e$, for an edge $uv=e \in E(H)$ and an integer $i \in [\sigma(e)]$ must slide to either $z^{u(i)}_e$ or $z^{v(i)}_e$, consuming $\bm{\sigma}$ slides. 
Since a solution that moves the token on~$a_e^{\sigma(e)+1}$ but not the token on $b_e^{\sigma(e)+1}$ is not minimal, we can safely assume that the described $m + \bm{\sigma}$ slides are executed in any minimal solution.

In the same solutions, each token on a vertex $z^{u(i)}_e$ for an edge $uv= e \in E(H)$ and an integer $i \in [\sigma(e)]$ requires the token on $y^{u(i)}_e$ to slide to either $e^u$ or $w_u$, utilizing as a result $\bm{\sigma}$ other slides. 
A token that slides from $y^{u(i)}_e$ to the vertex $w_u$ must slide again at least once, since any independent set that is achieved through the minimal number of slides would never require the sliding of the tokens on the vertices in $X^+$ (the token that moves to the vertex $w_u$ can be moved, using one less slide, to the vertex the token on $x_u^{r+1}$ moves to).
Since $G_e^{sel}$ can contain at most $2$ tokens, a token on $y^{u(i)}_e$ that slides to the vertex $e^u$ must either slide again at least once to a vertex, denoted $y_e^{u(i_1)}$ (for an integer $i_1 \in [\sigma(e)]$) in $Y_e^u$, or require another token on a vertex in $G_e^{sel}$ to slide at least once to either a vertex, denoted $y_e^{u(i_2)}$ (for an integer $i_2 \in [\sigma(e)]$) in $Y_e^u$, or a vertex, denoted $y_e^{v(i_2)}$ (for an integer $i_2 \in [\sigma(e)]$) in $Y_e^v$, while the token initially on $y^{u(i)}_e$ stays on $e^u$.
Given that at most $\bm{\sigma}$ slides remain in any minimal solution, and that each of the $\bm{\sigma}$ tokens initially on vertices in $Y$ that moved to either vertices of the form ${e_1}^{u_1}$ or $w_{u_1}$, for an edge $e_1 \in E(H)$ incident to a vertex $u_1 \in V(H)$, uses or requires at least one additional slide, each one such token can use or require exactly one additional slide.
If the token on $y^{u(i)}_e$ slides to $w_u$, then either in exactly one more slide it can move to a free vertex in $X_u$, or it can slide back to a vertex, denoted $y_{e_2}^{u(i_3)}$ (for an edge $e_2$ adjacent to $u$ in $H$ and an integer $i_3 \in [\sigma(e_1)]$) in $Y^u$.
However, either $y_{e_2}^{u(i_3)}$ (resp.\ $y_e^{u(i_1)}$, $y_e^{u(i_2)}$, or $y_e^{v(i_2)}$) or its adjacent vertex, denoted $z_{e_2}^{u(i_3)}$ in $Z^u$ (resp.\ $z_e^{u(i_1)}$ in $Z_e^u$, $z_e^{u(i_2)}$ in $Z_e^u$, or $z_e^{v(i_2)}$ in $Z_e^v$), contains a token, thus requiring at least one other additional slide, which is impossible. 
As a result, it can only be the case that a token on $y^{u(i)}_e$ slides to $w_u$ and then in exactly one more slide it moves to a free vertex in $X_u$.

For any edge $uv = e \in E(H)$, if $e^v \in C_\ell$ (resp.\ $e^u \in C_\ell$), then no vertex of $Y_e^v$ (resp.\ $Y_e^u$) appears in $C_\ell$ and the tokens on the vertices of $Y_e^v$ (resp.\ $Y_e^u$) have been moved to some of the free vertices of $X_v$ (resp.\ $X_u$). 
Given the latter, we produce an orientation $\lambda$ to $H$, where $\lambda(e) = (v, u)$ (resp.\ $\lambda(e) = (u, v)$) if $e^v \in C_\ell$ (resp.\ $e^u \in C_\ell$).
Since $|X_v| = |X_u| \le r$, $\lambda$ is such that the total weight directed out of any vertex $v \in V(H)$ is at most $r$.
\end{proof}
The proofs of Lemmas~\ref{lem:hardness-IS-pathwidth-forward} and \ref{lem:hardness-IS-pathwidth-backward} complete the proof of Theorem~\ref{thm:IS-D-pathwidth}. 
\\

By making some of the vertices of our \textsc{MMO}-instance-selector adjacent to some of the vertices in $G_{H_1},  \ldots, G_{H_t}$ constructed following the reduction of \Cref{thm:IS-D-pathwidth} for $t$ \textsc{MMO} input instances $H_1, \ldots, H_t$, we prove the following.
\begin{theorem}\label{thm:cross_composition-IS-bpw}
There exists an or-cross-composition from \textsc{MMO} into \textsc{IS-D} on bounded pathwidth graphs with respect to $\budget$. Consequently, \textsc{IS-D} does not admit a polynomial kernel with respect to $\budget + pw$, where $pw$ denotes the pathwidth of the input graphs, unless $\NP \subseteq \cp$.
\end{theorem}

\begin{proof}
As stated in \Cref{sec:foundational}, we can assume that we are given a family of $t$ \textsc{MMO} instances $(H_j, \mathcal{P}_{H_j}, \sigma_j, r_j)$, where $H_j$ is a bounded pathwidth graph with path decomposition $\mathcal{P}_{H_j}$, $|V(H_j)| = n$, $|E(H_j)| = m$, $\sigma_j: E_j \rightarrow \mathbb{Z}_+$ is a weight function such that $\sum_{e_j \in E(H_j)} \sigma_j(e_j) = \bm{\sigma}$ and $r_j = r \in \mathbb{Z}_+$ (integers are given in unary).
The construction of the instance $(G_t, \mathcal{P}_{G_t}, S, \budget)$ of \textsc{IS-D} is twofold. 

For each instance $H_j$ for $j \in [t]$, we add to $G_t$ the graph $G_{H_j}$ as per the reduction in \Cref{thm:IS-D-pathwidth}. 
We refer to the sets $A$, $B$, $X$, and $X^+$, subsets of vertices of a subgraph $G_{H_j}$ of $G_t$, by $A_j$, $B_j$, $X_j$, and $X_j^+$, respectively. 
Subsequently, we let $A = \cup_{j \in [t]} A_j$, $B = \cup_{j \in [t]} B_j$, $X = \cup_{j \in [t]} X_j$, and $X^+ = \cup_{j \in [t]} X^+_j$.
We add the \textsc{MMO}-instance-selector described in \Cref{sec:foundational} (under the \textsc{MMO}-instance-selector heading) and connect it to the rest of $G_t$ as follows (see \Cref{fig:IS-pathwidth-composition}). 
We make for each $j \in [t]$, the vertex $\textsc{\footnotesize Unselect}_j$ adjacent to the vertices in $V(G_{H_j}) \setminus S$, where $S$ is as defined later.
We make the vertex $h$ adjacent to the vertices in $A$ and the vertex $q$ adjacent to the vertices in $A^+$.
The result is the original graph $G_t$ appearing in \Cref{sec:foundational}.
By \Cref{lem:bounded-pathwidth-G}, $G_t$ is of bounded pathwidth.
Now, we set
\begin{equation*}
S =  B \text{ } \cup \text{ } B^+ \text{ } \cup \text{ } X^+ \text{ } \cup \text{ } Y \text{ } \cup 
\bigcup_{j \in [t]} \textsc{\footnotesize Unselect}_j
\text{ } \cup \bigcup_{i \in [\bm{\sigma}]} \text{ } (f^i \cup g^i) 
\text{ } \cup \bigcup_{i \in [m]} \text{ } (o^i \cup p^i) 
\text{ } \cup \text{ } h  \text{ } \cup \text{ } q
\end{equation*}
and $\budget = 3m + 5\bm{\sigma} + 1$.
\afterpage{
\begin{figure}[H]
    \centering
    \begin{tikzpicture}[scale=0.5]
    \fill[blue!10] (-0.75,2.75) -- (0.75,2.75) -- (0.75,-0.25) -- (-0.75,-0.25) -- cycle; 
    \fill[blue!10] (0.5,0.1) -- (1.25,0.9) -- (1.25,0) -- (0.5,0) -- cycle; 
    \fill[blue!10] (0.5,-0.1) -- (1.25,-0.9) -- (1.25,0) -- (0.5, 0) -- cycle; 
    \fill[blue!10] (-0.75,-2) -- (0.75,-2) -- (0.75,-3.25) -- (-0.75,-3.25) -- cycle; 
    \fill[blue!10] (0.5,-2.9) -- (1.25,-2) -- (1.25,-3) -- (0.5,-3) -- cycle; 
    \fill[blue!10] (0.5,-3.1) -- (1.25,-4) -- (1.25,-3) -- (0.5, -3) -- cycle; 
    \fill[green!10] (8,1) -- (7.5,2) -- (8.5,2) -- (8,1) -- cycle;
    \fill[green!10] (8,1) -- (6,0) -- (6,3) -- (8,1) -- cycle;
    \fill[green!10] (8,1) -- (10,-0.25) -- (10,2.5) -- (8,1) -- cycle;
    \fill[green!10] (6,0) -- (6,3) -- (4.75,3) -- (4.75,0) -- cycle;
    \fill[green!10] (8,-1.5) -- (7.5,-0.5) -- (8.5,-0.5) -- (8,-1.5) -- cycle;
    \fill[green!10] (8,-1.5) -- (6,-1.75) -- (6,-1.25) -- (8,-1.5) -- cycle;
    \fill[green!10] (8,-1.5) -- (10,0) -- (10,-3) -- (8,-1.5) -- cycle;
    \fill[green!10] (6,-1.75) -- (6,-1.25) -- (4.75,-1.25) -- (4.75,-1.75) -- cycle;
    \fill[green!10] (8,-4.7) -- (7.5,-3.7) -- (8.5,-3.7) -- (8,-4.7) -- cycle;
    \fill[green!10] (8,-4.7) -- (6,-6.25) -- (6,-2.75) -- (8,-4.7) -- cycle;
    \fill[green!10] (8,-4.7) -- (10,-2.5) -- (10,-5.75) -- (8,-4.7) -- cycle;
    \fill[green!10] (6,-6.25) -- (6,-2.75) -- (4.75,-2.75) -- (4.75,-6.25) -- cycle;
    \fill[yellow!10] (-3.5,-10) -- (-4.5,-10) -- (-4.5,-9) -- (-3.5,-9) -- cycle;
    \fill[yellow!10] (-5,-5) -- (-6.25,-6.25) -- (-6.25,-4.25) -- cycle;
    \fill[yellow!10] (-6.25,-6.25) -- (-7.25,-6.25) -- (-7.25, -4.25) -- (-6.25,-4.25) -- cycle;
    \fill[yellow!10] (-5,0) --  (-7,2.5) -- (-7,-2.25) -- cycle;
    \fill[yellow!10] (-7,2.5) -- (-8,2.5) -- (-8,-2.25) -- (-7,-2.25) -- cycle; 
    \draw[teal] (-0.5,0.75) -- (-5, 0);
    \draw[teal] (-0.5,1.5) -- (-5, 0);
    \draw[teal] (-0.5,2.25) -- (-5, 0);
    \draw[teal] (-0.5,-2.25) -- (-5, 0);  
    \draw[orange] (-0.5,0) .. controls (-1.5, -0.5) .. (-5,-5);
    \draw[orange] (-0.5,-3) .. controls (-1,-4) .. (-5,-5);
    \draw[pink] (-4,-9) -- (1,-3.5);
    \draw[pink] (-4,-9) .. controls (4,-6) .. (8,-4.7);
    \draw[pink] (-4,-9) -- (5,-5.2);
    \draw[black] (8,1) -- (9.5,0.25);
    \draw[black] (8,1) -- (9.5,1);
    \draw[black] (8,1) -- (9.5,1.75);
    \draw[black] (8,1) -- (8,2);
    \draw[black] (8,1) -- (6,0.25);
    \draw[black] (8,1) -- (6,2.25);
    \draw[black] (8,1) -- (6,1.25);
    \draw[black] (8,-1.5) -- (9.5,-2.5);
    \draw[black] (8,-1.5) -- (9.5,-1.5);
    \draw[black] (8,-1.5) -- (9.5,-0.5);
    \draw[black] (8,-1.5) -- (6,-1.5);
    \draw[black] (8,-1.5) -- (8,-0.5);
    \draw[black] (8,-4.7) -- (9.5,-5.2);
    \draw[black] (8,-4.7) -- (9.5,-4.2);
    \draw[black] (8,-4.7) -- (9.5,-3.2);
    \draw[black] (8,-4.7) -- (8,-3.75);
    \draw[black] (8,-4.7)-- (6,-6);
    \draw[black] (8,-4.7) -- (6,-5.2);
    \draw[black] (8,-4.7)-- (6,-4.2);
    \draw[black] (8,-4.7) -- (6,-3.4);
    \draw[yellow] (1,-2.5) .. controls (2,-2.25) .. (6,-1.5);
    \draw[yellow] (1,-3.5) .. controls (2.5,-5) and  (5,-5) .. (6,-6);    
    \draw[red] (0.5,-2.25) .. controls (1.5,-1.25) .. (5,-1.5);
    \draw[red] (0.5,-2.25) .. controls (2,-4.7) and  (5,-6) .. (5,-6);  
    \draw[red] (0.5,0.75) .. controls (1,1) .. (5,0.25);
    \draw[red] (0.5,1.5) .. controls (1,2.5) and  (2,2) .. (5,1.25);
    \draw[red] (0.5,2.25) .. controls (1,3) and  (2,3.5) .. (5,2.25);  
    \draw[black] (-6,-4.25) -- (-5,-5);
    \draw[black] (-6,-5.75) -- (-5, -5);
    \draw[black] (-7,0.75) -- (-5,0);
    \draw[black] (-7,-0.75) -- (-5,0);
    \draw[black] (-7,2) -- (-5,0);
    \draw[black] (-7,-2) -- (-5,0);
    \node[circle, fill=black, inner sep=0.5pt, draw=black, fill=white, label=below:{\scalebox{0.3}{$w_u$}}] at (8,1) {};
    \node[circle, fill=black, inner sep=0.5pt, draw=black, fill=white, label=below:{\scalebox{0.3}{$w_w$}}] at (8,-1.5) {};
    \node[circle, fill=black, inner sep=0.5pt, draw=black, fill=white, label=below:{\scalebox{0.3}{$w_v$}}] at (8,-4.7) {};
    \node[circle, fill=black, inner sep=0.5pt, draw=black, fill=white] at (9.5,0.25) {};
    \node[circle, fill=black, inner sep=0.5pt, draw=black, fill=white] at (9.5,1) {};
    \node[circle, fill=black, inner sep=0.5pt, draw=black, fill=white] at (9.5,1.75) {};
    \node[circle, fill=black, inner sep=0.5pt, draw=black, fill=white] at (9.5,-2.5) {};
    \node[circle, fill=black, inner sep=0.5pt, draw=black, fill=white] at (9.5,-1.5) {};    
    \node[circle, fill=black, inner sep=0.5pt, draw=black, fill=white] at (9.5,-0.5) {};
    \node[circle, fill=black, inner sep=0.5pt, draw=black, fill=white] at (9.5,-5.2) {};
    \node[circle, fill=black, inner sep=0.5pt, draw=black, fill=white] at (9.5,-4.2) {};
    \node[circle, fill=black, inner sep=0.5pt, draw=black, fill=white] at (9.5,-3.2) {};
    \node[circle, fill=black, inner sep=0.5pt, draw=black, fill=black] at (8,2) {};
    \node[circle, fill=black, inner sep=0.5pt, draw=black, fill=black] at (8,-0.5) {};
    \node[circle, fill=black, inner sep=0.5pt, draw=black, fill=black] at (8,-3.75) {};
    \draw[black] (5,-1.5) -- (6,-1.5) node[circle, draw=black, fill=black, inner sep=0.5pt, pos=1]{} node[circle, draw=black, fill=white, inner sep=0.5pt, pos=0]{};  
    \draw[black] (5,1.25) -- (6,1.25) node[circle, draw=black, fill=black, inner sep=0.5pt, pos=1]{} node[circle, draw=black, fill=white, inner sep=0.5pt, pos=0]{};
    \draw[black] (5,0.25) -- (6,0.25) node[circle, draw=black, fill=black, inner sep=0.5pt, pos=1]{} node[circle, draw=black, fill=white, inner sep=0.5pt, pos=0]{};
    \draw[black] (5,2.25) -- (6,2.25) node[circle, draw=black, fill=black, inner sep=0.5pt, pos=1]{} node[circle, draw=black, fill=white, inner sep=0.5pt, pos=0]{};   
    \draw[black] (5,-3.4) -- (6,-3.4) node[circle, draw=black, fill=black, inner sep=0.5pt, pos=1]{} node[circle, draw=black, fill=white, inner sep=0.5pt, pos=0]{};
    \draw[black] (5,-4.2) -- (6,-4.2) node[circle, draw=black, fill=black, inner sep=0.5pt, pos=1]{} node[circle, draw=black, fill=white, inner sep=0.5pt, pos=0]{};
    \draw[black] (5,-5.2) -- (6,-5.2) node[circle, draw=black, fill=black, inner sep=0.5pt, pos=1]{} node[circle, draw=black, fill=white, inner sep=0.5pt, pos=0]{};
    \draw[black] (5,-6) -- (6,-6) node[circle, draw=black, fill=black, inner sep=0.5pt, pos=1]{} node[circle, draw=black, fill=white, inner sep=0.5pt, pos=0]{};
    \draw[black] (-0.5,0.75) -- (0.5,0.75) node[circle, draw=black, fill=white, inner sep=0.5pt, pos=0]{} node[circle, draw=black, fill=black, inner sep=0.5pt, pos=1]{};
    \draw[black] (-0.5,1.5) -- (0.5,1.5) node[circle, draw=black, fill=white, inner sep=0.5pt, pos=0]{} node[circle, draw=black, fill=black, inner sep=0.5pt, pos=1]{};
    \draw[black] (-0.5,2.25) -- (0.5,2.25) node[circle, draw=black, fill=white, inner sep=0.5pt, pos=0]{} node[circle, draw=black, fill=black, inner sep=0.5pt, pos=1]{};
    \draw[black] (1,0.5) -- (1, -0.5);
    \draw[black] (0.5,0) -- (1,0.5) node[circle, draw=black, fill=white, inner sep=0.5pt, label=right:{\scalebox{0.3}{$e_1^v$}}, pos=1]{};
    \draw[black] (0.5,0) -- (1,-0.5) node[circle, draw=black, fill=white, inner sep=0.5pt, label=right:{\scalebox{0.3}{$e_1^u$}}, pos=1]{};
    \draw[black] (-0.5,0) -- (0.5,0) node[circle, draw=black, fill=white, inner sep=0.5pt, pos=0]{} node[circle, draw=black, fill=black, inner sep=0.5pt, pos=1]{};
    \draw[black] (-0.5,-2.25) -- (0.5,-2.25) node[circle, draw=black, fill=white, inner sep=0.5pt, pos=0]{} node[circle, draw=black, fill=black, inner sep=0.5pt, pos=1]{};
    \draw[black] (1,-2.5) -- (1, -3.5);
    \draw[black] (0.5,-3) -- (1,-2.5) node[circle, draw=black, fill=white, inner sep=0.5pt, label=right:{\scalebox{0.3}{$e_2^w$}}, pos=1]{};
    \draw[black] (0.5,-3) -- (1,-3.5) node[circle, draw=black, fill=white, inner sep=0.5pt, label=right:{\scalebox{0.3}{$e_2^u$}}, pos=1]{};
    \draw[black] (-0.5,-3) -- (0.5,-3) node[circle, draw=black, fill=white, inner sep=0.5pt, pos=0]{} node[circle, draw=black, fill=black, inner sep=0.5pt, pos=1]{};
    \node[circle, draw=black, fill=black, inner sep=1pt, label=above:\scalebox{0.5}{$h$}] at (-5, 0) {};
    \node[circle, draw=black, fill=black, inner sep=1pt, label=above:\scalebox{0.5}{$q$}] at (-5, -5) {};
    \draw[black] (-6.5,-4.25) -- (-6,-4.25) node[circle, draw=black, fill=black, inner sep=0.5pt, pos=1, label=above:\scalebox{0.3}{$o^2$}]{} node[circle, draw=black, fill=black, inner sep=0.5pt, pos=0, label=above:\scalebox{0.3}{$p^2$}]{};
    \draw[black] (-6.5,-5.75) -- (-6, -5.75) node[circle, draw=black, fill=black, inner sep=0.5pt, pos=1, label=above:\scalebox{0.3}{$o^1$}]{} node[circle, draw=black, fill=black, inner sep=0.5pt, pos=0, label=above:\scalebox{0.3}{$p^1$}]{};
    \draw[black] (-7,0.75) -- (-8,0.75) node[circle, draw=black, fill=black, inner sep=0.5pt, pos=1, label=above:\scalebox{0.3}{$f^3$}]{} node[circle, draw=black, fill=black, inner sep=0.5pt, pos=0, label=above:\scalebox{0.3}{$g^3$}]{};
    \draw[black] (-7,-0.75) -- (-8,-0.75) node[circle, draw=black, fill=black, inner sep=0.5pt, pos=1, label=above:\scalebox{0.3}{$f^2$}]{} node[circle, draw=black, fill=black, inner sep=0.5pt, pos=0, label=above:\scalebox{0.3}{$g^2$}]{};
    \draw[black] (-7,2) -- (-8,2) node[circle, draw=black, fill=black, inner sep=0.5pt, pos=1, label=above:\scalebox{0.3}{$f^4$}]{} node[circle, draw=black, fill=black, inner sep=0.5pt, pos=0, label=above:\scalebox{0.3}{$g^4$}]{};
    \draw[black] (-7,-2) -- (-8,-2) node[circle, draw=black, fill=black, inner sep=0.5pt, pos=1, label=above:\scalebox{0.3}{$f^1$}]{} node[circle, draw=black, fill=black, inner sep=0.5pt, pos=0, label=above:\scalebox{0.3}{$g^1$}]{};
    \draw[black] (-4,-10) -- (-4,-9) node[circle, draw=black, fill=black, inner sep=0.5pt, pos=1, label=above:\scalebox{0.3}{$\textsc{Unselect}_j$}]{} node[circle, draw=black, fill=white, inner sep=0.5pt, pos=0, label=below:\scalebox{0.3}{$\textsc{Select}_j$}]{};
    \end{tikzpicture}
    \caption{\footnotesize Orange, pink, and green edges highlighting the different types of edges between an \textsc{MMO}-instance-selector of the composition in \Cref{thm:cross_composition-IS-bpw} and a subgraph $G_{H_j}$ for a $j \in [t]$ of the same. $H_j$ has three vertices $u$, $v$, and $w$ and two edges $e_1=uv$ and $e_2=uw$, and $r=3$. Additionally, $\sigma_j(e_1) = 3$ and $\sigma_j(e_2) = 1$. For clarity, not all edges inside $G_{H_j}$ are drawn, nor are all the pink edges depicted. Vertices in black are in $S$ and those in white are not.}
    \label{fig:IS-pathwidth-composition}
\end{figure}}
Given that all integers are given in unary, the construction of the graph $G_t$, or its path decomposition (as described in \Cref{lem:bounded-pathwidth-G}), and as a consequence the reduction take time polynomial in the size of the input instances. 
Additionally, by \Cref{cor:composition-space}, this composition is a pl-reduction.
We claim that $(G_t, \mathcal{P}_{G_t}, S, \budget)$ is a yes-instance of \textsc{IS-D} if and only if for some integer $\mathfrak{j} \in [t]$, $(H_\mathfrak{j}, \mathcal{P}_{H_\mathfrak{j}}, \sigma_\mathfrak{j}, r_\mathfrak{j})$ is a yes-instance of \textsc{MMO}.
\newline

\begin{claim}
If for some $\mathfrak{j} \in [t]$, $(H_\mathfrak{j}, \mathcal{P}_{H_\mathfrak{j}}, \sigma_\mathfrak{j}, r_\mathfrak{j})$ is a yes-instance of \textsc{MMO}, then $(G_t, \mathcal{P}_{G_t}, S, \budget)$ is a yes-instance of \textsc{IS-D}.
\end{claim}
\begin{claimproof}
Let $(H_\mathfrak{j}, \mathcal{P}_{H_\mathfrak{j}}, \sigma_\mathfrak{j}, r_\mathfrak{j})$ be a yes-instance of \textsc{MMO} and let $\lambda$ be a feasible orientation of $H_\mathfrak{j}$ such that for each $v \in V(H_\mathfrak{j})$, the total weight of the edges directed out of $v$ is at most $r$.
In $G_t$, the tokens on $f^i$ and $g^i$ are adjacent for each $i \in \bm{\sigma}$, and the tokens on $o^i$ and $p^i$ are adjacent for each $i \in [m]$.
First, we slide the token on $\textsc{\footnotesize Unselect}_\mathfrak{j}$ onto the vertex $\textsc{\footnotesize Select}_\mathfrak{j}$. 
Using $2\bm{\sigma} + 2m$ slides, we move the token on the vertex $f_i$ for each $i \in [\bm{\sigma}]$ towards a token-free vertex in $A_\mathfrak{j}$ and we move the token on the vertex $o_i$ for each $i \in [m]$ towards a token-free vertex in $A^+_\mathfrak{j}$.
This achieves a configuration of the tokens on $G_{H_\mathfrak{j}}$ that is similar to the starting configuration of the tokens on the graph $G_H$ of the reduction of \Cref{thm:IS-D-pathwidth}.
From \Cref{lem:hardness-IS-pathwidth-forward}, given a feasible orientation of the \textsc{MMO} instance $(H_\mathfrak{j}, \mathcal{P}_{H_\mathfrak{j}}, \sigma_\mathfrak{j}, r_\mathfrak{j})$, we can achieve a configuration of the tokens that constitutes an independent set in $G_{H_\mathfrak{j}}$ in $m + 3\bm{\sigma}$ slides. 
This totals $\budget$ slides and achieves a configuration of the tokens that constitutes an independent set in $G_t$.    
\end{claimproof}

\begin{claim}
If $(G_t, \mathcal{P}_{G_t}, S, \budget)$ is a yes-instance of \textsc{IS-D}, then there exists an integer $\mathfrak{j} \in [t]$, such that $(H_\mathfrak{j}, \mathcal{P}_{H_\mathfrak{j}}, \sigma_\mathfrak{j}, r_\mathfrak{j})$ is a yes-instance of \textsc{MMO}.
\end{claim}

\begin{claimproof}
In any solution that uses the minimal number of slides, at least $2m + 2\sigma$ slides are needed to get the tokens on $f^1, \ldots, f^{\bm{\sigma}}$ and $o^1, \ldots, o^m$ to vertices in $A$ and $A^+$, respectively. 
Note that a solution that moves the token on $g^i$ for some integer $i \in [\bm{\sigma}]$ but not the token on $f^i$ or that moves the token on $p^i$ for some integer $i \in [m]$ but not the token on $o^i$ is not minimal. 
When a token is moved from one of the vertices $f^1, \ldots, f^{\bm{\sigma}}, o^1, \ldots, o^m$ to one of the vertices in $V(G_{H_j})$ for an integer $j \in [t]$, which is bound to happen, at least one slide is needed so that the token on $\textsc{\footnotesize Unselect}_j$ is not adjacent to any other token and one slide can be achieved by sliding the token on $\textsc{\footnotesize Unselect}_j$ to $\textsc{\footnotesize Select}_j$.
Note that a solution that uses the minimal number of slides and slides the token on $\textsc{\footnotesize Unselect}_j$ for any $j \in [t]$ to any vertex $v \in V(G_t) \setminus \{\textsc{\footnotesize Select}_j\}$ (and possibly slides another token to $\textsc{\footnotesize Select}_j$) can safely be replaced by a solution that performs the same number of slides and the same slides except that it slides the token on $\textsc{\footnotesize Unselect}_j$ to $\textsc{\footnotesize Select}_j$ (and possibly the other token to $v$).
Thus, we consider minimal solutions that slide $\textsc{\footnotesize Unselect}_j$ for any $j \in [t]$ only to $\textsc{\footnotesize Select}_j$.

Assume that for some integer $j \in [t]$, $\textsc{\footnotesize Unselect}_j$ has moved to $\textsc{\footnotesize Select}_j$.
In a solution that uses the minimal number of slides, if a token is the first to move from one of the vertices $o^1, \ldots, o^m$ to a vertex $a_e^{\sigma_j(e)+1}$ for an edge $uv=e \in E(H_j)$, it will require the token on $b_e^{\sigma_j(e)+1}$ to slide to one of $e^{u}$ or $e^{v}$.
Each other token that moves from one of the vertices $o^1, \ldots, o^m$ to $a_e^{\sigma(e)+1}$ will require in the same solution at least $3$ additional slides as one token must exit the induced subgraph $G_e^{sel}$ and the two remaining tokens must not be on adjacent vertices. 
Thus, if we assume that the $m$ tokens on $o^1, \ldots, o^m$ move to distinct vertices in $A^+$, they require at least $m$ additional slides.

Assume that for some integer $j_1 \in [t]$, $\textsc{\footnotesize Unselect}_{j_1}$ has moved to $\textsc{\footnotesize Select}_{j_1}$.
In a solution that uses the minimal number of slides, if a token is the first to move from one of the vertices $f^1, \ldots, f^{\bm{\sigma}}$ to a vertex $a^i_{e_1}$ for an edge $u_1v_1= e_1 \in E(H_{j_1})$ and $i \in [\sigma_{j_1}(e_1)]$, it will require the token on $b^i_{e_1}$ to slide to one of $z^{u_1(i)}_{e_1}$, or $z^{v_1(i)}_{e_1}$.
W.l.o.g., assume the token on $b^i_{e_1}$ slides to $z^{u_1(i)}_{e_1}$.
In the same solution, the token on $z^{u_1(i)}_{e_1}$ will in turn require the token on $y^{u_1(i)}_{e_1}$ to slide to either $w_{u_1}$ or $e_1^{u_1}$.
Each other token that moves from one of the vertices $f^1, \ldots, f^{\bm{\sigma}}$ to $a^i_{e_1}$ will require in the same solution at least $4$ additional slides to leave the path $a^i_{e_1} b^i_{e_1} z_{e_1}^{u_1(i)} y_{e_1}^{u_1(i)}$ which cannot accommodate more tokens.
Given that the remaining budget is at most $3\bm{\sigma}$, a second token sliding to the vertex $a^i_{e_1}$ is only possible if for some token initially on a vertex in $Y$ that moved 
to a vertex $w_{u_2}$ or a vertex $e_2^{u_2}$, for a vertex $u_2 \in V(H_{j_2})$, edge $u_2v_2=e_2 \in E(H_{j_2})$, and integer $j_2 \in [t]$, does not slide or require another token to slide.
In any solution that uses the minimal number of slides, the tokens on the vertices in $X^+$ do not need to be moved (since we can move any token that moves into the adjacent representative vertices to the vertices the tokens on the vertices of $X^+$ were being moved to).
Thus, a token on $w_{u_2}$ must slide again.
Similarly, a token on $e_2^{u_2}$ must either slide again, or require another token (on the vertex in $B^+$) to slide.
Thus, with a remaining budget of at most $3\bm{\sigma}$, no token can afford to move from $f^1, \ldots, f^{\bm{\sigma}}$ to $a^i_{e_1}$ once another token has already moved to the same vertex.
In other words, given $\budget$, it must be the case that the $\bm{\sigma}$ tokens on $f^1, \ldots, f^{\bm{\sigma}}$ slide towards distinct vertices in $A$.

Additionally, each such token will require at least $3$ slides.
Consequently, the $m$ tokens on $o^1, \ldots, o^m$ must also slide to distinct vertices in $A^+$.
Given that one slide remains for moving the token on one vertex $\textsc{\footnotesize Unselect}_\mathfrak{j}$ to $\textsc{\footnotesize Select}_\mathfrak{j}$ for an integer $\mathfrak{j} \in [t]$, it must be the case that the tokens on $f^1, \ldots, f^{\bm{\sigma}}$ (resp.\ $o^1, \ldots, o^m$) move to distinct vertices in $A_\mathfrak{j}$ (resp.\ $A_\mathfrak{j}^+$).
This achieves a configuration of the tokens on $G_{H_\mathfrak{j}}$ that is similar to the starting configuration of the tokens inside $G_H$ of the reduction of \Cref{thm:IS-D-pathwidth}. From \Cref{lem:hardness-IS-pathwidth-backward}, with a remaining budget of $m + 3\bm{\sigma}$, we get that $(H_\mathfrak{j}, \mathcal{P}_{H_\mathfrak{j}}, \sigma_\mathfrak{j}, r_\mathfrak{j})$ is a yes-instance of \textsc{MMO}. 
\end{claimproof}
This concludes the proof the theorem.
\end{proof}
Next we consider the \emph{fvs} parameterization of the IS-D problem. 
\begin{theorem}
    \label{thm:IS-D-fvs}
    The IS-D problem is W[1]-hard for the parameter \emph{fvs} of the input graph. 
\end{theorem}
We present a parameterized reduction from the \mcc~problem. 
We utilize the reduction given in Theorem~\ref{thm:DS-D-fvs}, and apply some changes over the constructed graph to obtain reduced instance for the IS-D instance. 
Consider the graph $H$ constructed in the proof of Theorem~\ref{thm:DS-D-fvs}. 
For each $i \in [\kappa]$, we add a vertex $\Tilde{t}_i$ with $n(n-1)(\kappa-1)$ pendent neighbors (call the set as $T_i$) in the vertex-block $H_i$. 
For each vertex $v$ in $Q_i$, we add a pendent neighbor $b(v)$. 
For any set $B \subseteq Q_i$, we used to refer $b(B)$ as the set of all pendent neighbors of the vertices in $B$. 
The vertex $\Tilde{t}_{i}$ is made adjacent to all the vertices in $b(Q_i)$. 
We add a pendent neighbor $\hat{t}_i$ to $t_i$. 
For each $\iljk$, we add a vertex $\Tilde{t}_{i,j}$ with $2n|E_{i,j}|-2n$ pendent neighbors (call the set as $T_{i,j}$) in the edge-block~$H_{i,j}$. 
For each vertex $v$ in $Q_{i,j}$, we add a pendent neighbor $c(v)$. 
For any set $C \subseteq Q_{i,j}$, we used to refer~$c(C)$ as the set of all pendent neighbors of the vertices in $C$. 
The vertex $\Tilde{t}_{i,j}$ is made adjacent to all the vertices in $c(Q_{i,j})$. 
We add a pendent neighbor $\hat{t}_{i,j}$ to $t_{i,j}$. 
For each $\iljk$ and $l \in \{i,j\}$, remove the vertex subsets $B_{i,j}^l$ and $C_{i,j}^l$, and the edges incident on them in the connector~$\conn_{i,j}^l$. 
We add a pendent neighbor $\Tilde{s}_{i,j}^l$ to $s_{i,j}^l$ and $\Tilde{r}_{i,j}^l$ to $r_{i,j}^l$. 
An illustration of a connector connecting a vertex-block and an edge-block is given in Figure~\ref{fig:IS-D-fvs-illustration}. 
This completes the construction of graph $H$ for the IS-D instance. 
Next we describe the initial configuration $\varS$ as follows:
\[\varS = \bigcup_{i \in [\kappa], {x \in [n]}} (Q_{i,x}\cup b(Q_{i,x}) \cup \bigcup_{e \in E} (Q_e \cup c(Q_e)) \cup \bigcup_{i \in [\kappa]}\{t_i, \hat{t}_{i}\}  \cup \bigcup_{\iljk}\{t_{i,j}, \hat{t}_{i,j}\}.\]
Finally, we set $\budget = (4n^2+1)\kctwo+4nm+\kappa$ and the reduced IS-D instance is $(H, \varS, \budget)$. 

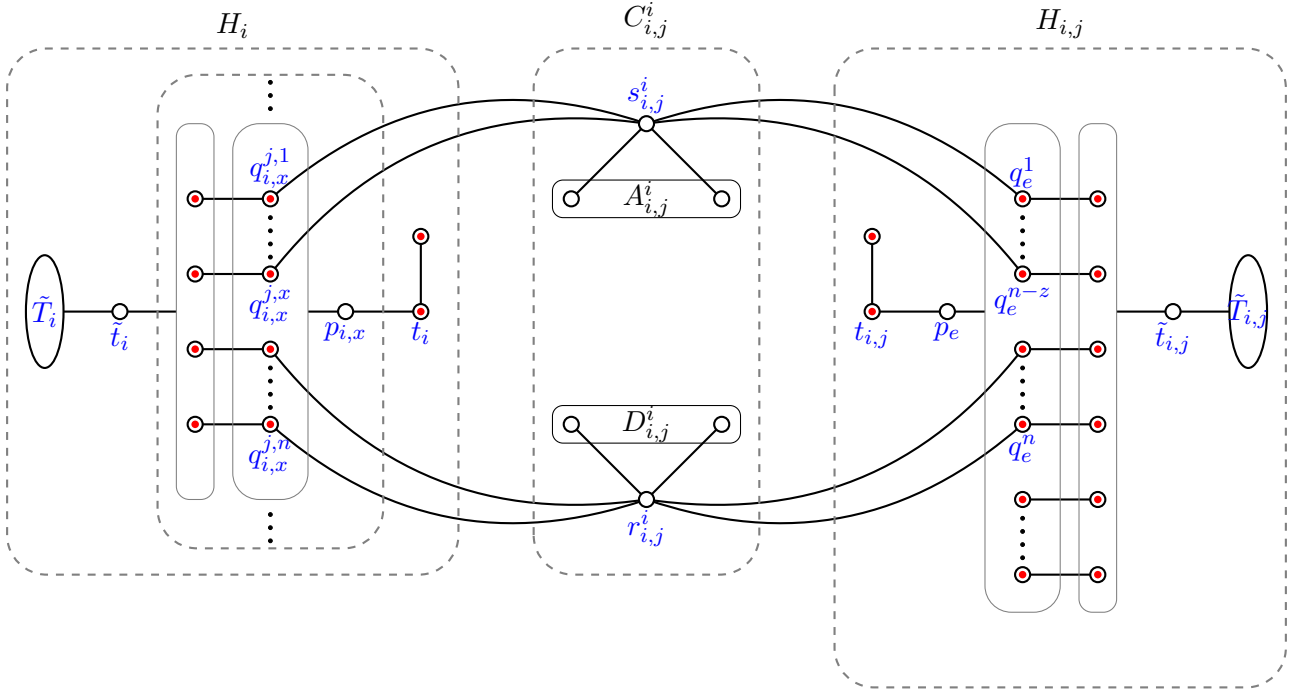
\begin{figure}
    \tikzset{decorate sep/.style 2 args=
{decorate,decoration={shape backgrounds,shape=circle,shape size=#1,shape sep=#2}}}
\centering
    \begin{tikzpicture}
    \coordinate (ti) at (3,0);
    \coordinate (tti) at (-1,0);
    \coordinate (pix) at (2,0);
    \coordinate (qix1) at (1,1.5);
    \coordinate (qixx) at (1,0.5);
    \coordinate (qixx1) at (1,-0.5);
    \coordinate (qixn) at (1,-1.5);
    \coordinate (siij) at (6,2.5);
    \coordinate (riij) at (6, -2.5);
    \coordinate (aiji1) at (5, 1.5);
    \coordinate (aijin) at (7, 1.5);
    \coordinate (diji1) at (5, -1.5);
    \coordinate (dijin) at (7, -1.5);
    \coordinate (tij) at (9, 0);
    \coordinate (ttij) at (13, 0);
    \coordinate (pe) at (10, 0);
    \coordinate (qez1) at (11,1.5);
    \coordinate (qezz) at (11, 0.5);
    \coordinate (qezz1) at (11, -0.5);
    \coordinate (qezn) at (11, -1.5);
    \coordinate (qew1) at (11, -2.5);
    \coordinate (qewn) at (11, -3.5);
    \draw[black, thick] (ti) -- (pix);
    \draw[black, thick] (ti) -- ($(ti)+(0, 1)$);
    \draw[black, thick] (tti) -- ($(tti)+(0.75, 0)$);
    \draw[black, thick] (tti) -- ($(tti)+(-0.75, 0)$);
    \draw[black, thick] (pix) -- ($(pix)+(-0.5, 0)$);
    \draw[black, thick] (tij) -- (pe);
    \draw[black, thick] (tij) -- ($(tij)+(0, 1)$);
    \draw[black, thick] (ttij) -- ($(ttij)+(-0.75, 0)$);
    \draw[black, thick] (ttij) -- ($(ttij)+(0.75, 0)$);
    \draw[black, thick] (pe) -- ($(pe)+(0.5, 0)$);
    \draw[black, thick] (siij) -- (aiji1);
    \draw[black, thick] (siij) -- (aijin);
    \draw[black, thick] (riij) -- (diji1);
    \draw[black, thick] (riij) -- (dijin);
    \draw[black, thick] (siij) to[bend right=30] (qix1);
    \draw[black, thick] (siij) to[bend right=30] (qixx);
    \draw[black, thick] (riij) to[bend left=30] (qixx1);
    \draw[black, thick] (riij) to[bend left=30] (qixn);
    \draw[black, thick] (siij) to[bend left=30] (qez1);
    \draw[black, thick] (siij) to[bend left=30] (qezz);
    \draw[black, thick] (riij) to[bend right=30] (qezz1);
    \draw[black, thick] (riij) to[bend right=30] (qezn);
    \draw[black, thick] ($(qix1)+(-1,0)$) -- (qix1);
    \draw[black, thick] ($(qixx)+(-1,0)$) -- (qixx);
    \draw[black, thick] ($(qixx1)+(-1,0)$) -- (qixx1);
    \draw[black, thick] ($(qixn)+(-1,0)$) -- (qixn);
    \draw[black, thick] ($(qez1)+(1,0)$) -- (qez1);
    \draw[black, thick] ($(qezz)+(1,0)$) -- (qezz);
    \draw[black, thick] ($(qezz1)+(1,0)$) -- (qezz1);
    \draw[black, thick] ($(qezn)+(1,0)$) -- (qezn);
    \draw[black, thick] ($(qew1)+(1,0)$) -- (qew1);
    \draw[black, thick] ($(qewn)+(1,0)$) -- (qewn);
    \fill[white, draw=black, thick] (ti) circle (0.1cm) node[below,blue] {$t_i$};
    \fill[white, draw=black, thick] ($(ti) + (0, 1)$) circle (0.1cm);
    \fill[white, draw=black, thick] (tti) circle (0.1cm) node[below,blue] {$\Tilde{t}_i$};
    \fill[white, draw=black, thick] (tij) circle (0.1cm) node[below,blue] {$t_{i,j}$};
    \fill[white, draw=black, thick] ($(tij) + (0, 1)$) circle (0.1cm);
    \fill[white, draw=black, thick] (ttij) circle (0.1cm) node[below,blue] {$\Tilde{t}_{i,j}$};
    \fill[white, draw=black, thick] (siij) circle (0.1cm) node[above,blue] {$s_{i,j}^i$};
    \fill[white, draw=black, thick] (riij) circle (0.1cm) node[below,blue] {$r_{i,j}^i$};
    \fill[white, draw=black, thick] (pix) circle (0.1cm) node[below,blue] {$p_{i,x}$};
    \fill[white, draw=black, thick] (qix1) circle (0.1cm) node[above,blue] {$q_{i,x}^{j,1}$};
    \fill[white, draw=black, thick] (qixx) circle (0.1cm) node[below,blue] {$q_{i,x}^{j,x}$};
    \fill[white, draw=black, thick] (qixx1) circle (0.1cm);
    \fill[white, draw=black, thick] (qixn) circle (0.1cm) node[below,blue] {$q_{i,x}^{j,n}$};
    \fill[white, draw=black, thick] (pe) circle (0.1cm) node[below,blue] {$p_e$};    
    \fill[white, draw=black, thick] (qez1) circle (0.1cm) node[above,blue] {$q_e^1$};
    \fill[white, draw=black, thick] (qezz) circle (0.1cm) node[below,blue] {$q_e^{n-z}$};
    \fill[white, draw=black, thick] (qezz1) circle (0.1cm);
    \fill[white, draw=black, thick] (qezn) circle (0.1cm) node[below,blue] {$q_e^n$};
    \fill[white, draw=black, thick] (qew1) circle (0.1cm); 
    \fill[white, draw=black, thick] (qewn) circle (0.1cm); 
    \fill[white, draw=black, thick] (aiji1) circle (0.1cm);
    \fill[white, draw=black, thick] (aijin) circle (0.1cm);
    \fill[white, draw=black, thick] (diji1) circle (0.1cm);
    \fill[white, draw=black, thick] (dijin) circle (0.1cm);
    \draw[black, thick] ($(tti) + (-1,0)$) ellipse (0.25 and 0.75) node[blue] {$\Tilde{T}_i$};
    \draw[black, thick] ($(ttij) + (1,0)$) ellipse (0.25 and 0.75) node[blue] {$\Tilde{T}_{i,j}$};
    \fill[white, draw=black, thick] ($(qix1)+(-1,0)$) circle (0.1cm);
    \fill[white, draw=black, thick] ($(qixx)+(-1,0)$) circle (0.1cm);
    \fill[white, draw=black, thick] ($(qixx1)+(-1,0)$) circle (0.1cm);
    \fill[white, draw=black, thick] ($(qixn)+(-1,0)$) circle (0.1cm);
    \fill[white, draw=black, thick] ($(qez1)+(1,0)$) circle (0.1cm);
    \fill[white, draw=black, thick] ($(qezz)+(1,0)$) circle (0.1cm);
    \fill[white, draw=black, thick] ($(qezz1)+(1,0)$) circle (0.1cm);
    \fill[white, draw=black, thick] ($(qezn)+(1,0)$) circle (0.1cm);
    \fill[white, draw=black, thick] ($(qew1)+(1,0)$) circle (0.1cm);
    \fill[white, draw=black, thick] ($(qewn)+(1,0)$) circle (0.1cm);
    \fill[red] (ti) circle (0.05cm);
    \fill[red] ($(ti) + (0, 1)$) circle (0.05cm);
    \fill[red] (tij) circle (0.05cm);
    \fill[red] ($(tij) + (0, 1)$) circle (0.05cm);
    \fill[red] (qix1) circle (0.05cm);
    \fill[red] (qixx) circle (0.05cm);
    \fill[red] (qixx1) circle (0.05cm);
    \fill[red] (qixn) circle (0.05cm);
    \fill[red] (qez1) circle (0.05cm);
    \fill[red] (qezz) circle (0.05cm);
    \fill[red] (qezz1) circle (0.05cm);
    \fill[red] (qezn) circle (0.05cm);
    \fill[red] (qew1) circle (0.05cm);
    \fill[red] (qewn) circle (0.05cm);
    \fill[red] ($(qix1)+(-1,0)$) circle (0.05cm);
    \fill[red] ($(qixx)+(-1,0)$) circle (0.05cm);
    \fill[red] ($(qixx1)+(-1,0)$) circle (0.05cm);
    \fill[red] ($(qixn)+(-1,0)$) circle (0.05cm);
    \fill[red] ($(qez1)+(1,0)$) circle (0.05cm);
    \fill[red] ($(qezz)+(1,0)$) circle (0.05cm);
    \fill[red] ($(qezz1)+(1,0)$) circle (0.05cm);
    \fill[red] ($(qezn)+(1,0)$) circle (0.05cm);
    \fill[red] ($(qew1)+(1,0)$) circle (0.05cm);
    \fill[red] ($(qewn)+(1,0)$) circle (0.05cm);
    \draw[black,rounded corners] (4.75, 1.25) rectangle (7.25, 1.75) node[midway] {$A_{i,j}^i$};
    \draw[black,rounded corners] (4.75, -1.75) rectangle (7.25, -1.25) node[midway] {$D_{i,j}^i$};
    \draw[gray, rounded corners=10pt] (0.5, -2.5) rectangle (1.5, 2.5);
    \draw[gray, rounded corners=5pt] (-0.25, -2.5) rectangle (0.25, 2.5);
    \draw[gray, thick, dashed, rounded corners=15pt] (-0.5, -3.15) rectangle (2.5, 3.15);
    \draw[gray, thick, dashed, rounded corners=15pt] (-2.5, -3.5) rectangle (3.5, 3.5) node[midway, black, above=3.5cm] {$H_i$};
    \draw[gray, rounded corners=10pt] (10.5, 2.5) rectangle (11.5, -4);
    \draw[gray, rounded corners=5pt] (11.75, 2.5) rectangle (12.25, -4);
    \draw[gray, thick, dashed, rounded corners=15pt] (8.5, 3.5) rectangle (14.5, -5) node[midway, black, above=4.25cm] {$H_{i,j}$};
    \draw[gray, thick, dashed, rounded corners=15pt] (4.5, -3.5) rectangle (7.5, 3.5) node[black, midway, above=3.5cm] {$C_{i,j}^i$};
    \draw[decorate sep={0.5mm}{1.75mm},fill] (1,1.25) -- (1,0.7);
    \draw[decorate sep={0.5mm}{1.75mm},fill] (1,-0.75) -- (1,-1.3);
    \draw[decorate sep={0.5mm}{1.75mm},fill] (11,1.25) -- (11,0.7);
    \draw[decorate sep={0.5mm}{1.75mm},fill] (11,-0.75) -- (11,-1.3);
    \draw[decorate sep={0.5mm}{1.75mm},fill] (11,-2.75) -- (11,-3.3);
    \draw[decorate sep={0.5mm}{1.75mm},fill] (1,2.7) -- (1,3.1);
    \draw[decorate sep={0.5mm}{1.75mm},fill] (1,-2.7) -- (1,-3.1);
    \end{tikzpicture}
    \caption{\footnotesize An illustration of the reduction of Theorem~\ref{thm:IS-D-fvs}. For $\iljk$, the vertex-block $H_i$ and the edge-block $H_{i,j}$ are connected to the connector $C_{i,j}^i$. The initial configuration is denoted by vertices with red circle. }
    \label{fig:IS-D-fvs-illustration}
\end{figure}

\begin{lemma}
    \label{lem:IS-D-fvs-bound}
    The \emph{fvs} of the graph $H$ is at most $5\kctwo+\kappa$. 
\end{lemma}
\begin{proof}
    Let $F = \{s_{i,j}^l, r_{i,j}^l \mid \iljk, l\in\{i,j\}\} \cup \{ \Tilde{t}_{i,j} \mid \iljk\} \cup \{\Tilde{t}_i \mid i \in [\kappa]\}$. 
    Removal of $F$ from $H$ results a forest. 
    Therefore, the \emph{fvs} of $H$ is at most $|F| = 5\binom{\kappa}{2} + \kappa$. 
\end{proof}
Next we prove the correctness of the reduction.
\begin{lemma}
\label{lem:IS-D-fvs-forward}
    If $(G, \kappa)$ is a yes-instance of the \mcc~problem, then $(H, \varS, \budget)$ is a yes-instance of the \textsc{IS-D} problem. 
\end{lemma}
\begin{proof}
    Let $C = \subseteq V(G)$ be a $\kappa$-clique in $G$. 
    For each $i \in [\kappa]$, let $u_{i,x_i}$ be the vertex in $C \cap V_i$ for some $x_i \in [n]$. 
    For each $\iljk$, let $e_{i,j} = u_{i,x_i}u_{j,x_j}$.
    For each $i \in [\kappa]$, we slide the token on $t_i$ to $p_{i, x_i}$. 

    \pagebreak
    For each $x \in [n]$, 
    \begin{itemize}
        \item if $x = x_i$, then for each $j \not= i \in [\kappa]$, we slide $x_i$-tokens in $Q_{i,x_i}^j$ towards $s^i_{i,j}$ and $n-x_i$-tokens in $Q_{i,x_i}^j$ towards $r_{i,j}^i$.
        \item if $x \not= x_i$, then we slide all $n(\kappa-1)$ tokens in $b(Q_{i,x})$ to $\Tilde{T}_i$.
    \end{itemize}
    
    For each $\iljk$, we slide the token on $t_{i,j}$ to $p_{e_{i,j}}$. 
    For each $e \in E_{i,j}$, if $e=u_{i,x_i}u_{j, x_j}$, then we slide 
    \begin{itemize}
        \item $n-x_i$ tokens in $Q_{e_{i,j}}$ to $s_{i,j}^i$,
        \item $x_i$ tokens in $Q_{e_{i,j}}$ to $r_{i,j}^i$,
        \item $n-x_j$ tokens in $Q_{e_{i,j}}$ to $s_{i,j}^j$, and
        \item $x_j$ tokens in $Q_{e_{i,j}}$ to $r_{i,j}^j$. 
    \end{itemize}
    
    Otherwise, we slide all $2n$ tokens in $c(Q_e)$ to $\Tilde{T}_{i,j}$. 
    For each $\iljk$, and for each $l \in {i,j}$, $s_{i,j}^l$ receives $x_l$-tokens from $H_l$ and $n-x_l$-tokens from $H_{i,j}$.
    Similarly, $r_{i,j}^l$ receives $n-x_l$-tokens from $H_l$ and $x_l$-tokens from $H_{i,j}$. 
    Further, we push the $n$-tokens received by $s_{i,j}^l$ to $A_{i,j}^l$ and $n$-tokens received by $r_{i,j}^l$ to $D_{i,j}^l$. 
    The resulting configuration is a valid independent set. 
    Finally, let $S' \subseteq V(H)$ be the solution obtained from the above token sliding steps. 
    It is clear that the set $S'$ is an  independent set in~$H$. 
    Next we count the number of token steps used to obtain~$S'$ from $\varS$. 
    In each vertex-block, we spend $2(\kappa-1)n^2+1$ steps to push tokens towards the connectors and $\Tilde{T}_i$. 
    Similarly, at each edge-block, we spend $4n|E_{i,j}|+1$ steps. 
    Therefore, we spend $\kappa\cdot\big(2(\kappa-1)n^2+1\big) + 4nm+\kctwo = (4n^2+1)\kctwo + 4nm + \kappa= \budget$. 
    Hence, $(H, \varS, \budget)$ is a yes-instance of IS-D problem.
\end{proof}

\begin{lemma}
\label{lem:IS-D-fvs-reverse}
    If $(H, \varS, \budget)$ is a yes-instance of the IS-D problem, then $(G, \kappa)$ is a yes-instance of the \mcc~problem. 
\end{lemma}
\begin{proof}
    Let $S^*$ be a feasible solution for the instance $(H, \varS, \budget)$ of the IS-D problem. 
    In each vertex-block $H_i$, we need to slide at least $n^2(\kappa-1)$ tokens as the vertices in both sets $Q_i$ and $b(Q_i)$ have tokens in the initial configuration. 
    We can accommodate at most $n(n-1)(\kappa-1)$ tokens at the vertices in $\Tilde{T}_i$. 
    Therefore, at least $n(\kappa-1)$ token should be pushed towards connectors. 
    Each such tokens must slide at least two steps to find a free vertex. 
    Therefore, we need at least $2n^2(\kappa-1)$ token steps to settle the tokens on $Q_i$ and $b(Q_i)$.     
    Similarly, at each edge-block $H_{i,j}$, we need to slide at least $2n|E_{i,j}|$ tokens as the vertices in the sets $Q_{i,j}$ and $c(Q_{i,j})$ have tokens in the initial configuration. 
    We can accommodate at most $2n(|E_{i,j}| - 1)$ tokens at the vertices in $\Tilde{T}_{i,j}$. 
    Therefore, at least $2n$ tokens should be pushed towards connectors. 
    Each such tokens must slide at least two steps to find a free vertex. 
    Therefore, we need at least $4n|E_{i,j}|$ token steps to settle the tokens on~$Q_{i,j}$ and $c(Q_{i,j})$.     
    This saturates a budget of $2n^2(\kappa-1)\kappa + 4nm = 4n^2\kctwo + 4nm$. 
    We left with a budget of at most $\kctwo + k$. 
    In each vertex-block $H_i$, we still need to fix the tokens on $t_i$ and $\hat{t}_i$. 
    We can use at most one token step to fix this. 
    Therefore, the token on $t_i$ should move to a neighbor $p_{i,x}$ for some $x \in [n]$. 
    Similarly, the token on $t_{i,j}$ for each $\iljk$, should move to a neighbor $p_e$ for some $e \in E_{i,j}$. 
    
    For each $i\in [\kappa]$, let $p_{i,x_i}$ for some $x_i \in [n]$ be the vertex in $H_i$ that gets token in $S^*$ and releases all the tokens in $Q_{i,x_i}$. 
    Similarly, for each $\iljk$, let $p_e$ for some $e=u_{i,z_i}u_{i,z_j} \in E_{i,j}$ with $z_i,z_j \in [n]$ be the vertex in $H_{i,j}$ that gets token in $S^*$ and releases all tokens in $Q_{e}$. 
    Consider the connector $\conn_{i,j}^i$. 
    The set $Q_{i,x_i}^j$ pushes $x_i$ tokens to $s_{i,j}^i$ and $n-x_i$ tokens to $r_{i,j}^i$. 
    The set $Q_e$ pushes $z_i$ tokens to $r_{i,j}^i$ and $n-z_i$ tokens to $s_{i,j}^i$. 
    The number of tokens passed through $s_{i,j}^i$ to $A_{i,j}^i$ is $x_i + (n - z_i)$. 
    Since $A_{i,j}^i$ need $n$ tokens, it is mandatory that $x_i=z_i$. 
    This equality should hold for every $i$. 
    Therefore, for each $\iljk$, there exist an edge $u_{i,x_i}u_{j,x_j}$. 
    Hence $(G,\kappa)$ is an yes-instance of the \mcc~problem.     
\end{proof}
The proofs of Lemmas~\ref{lem:IS-D-fvs-bound}, \ref{lem:IS-D-fvs-forward} and \ref{lem:IS-D-fvs-reverse} complete the proof of Theorem~\ref{thm:IS-D-fvs}.

\section{Dominating Set Discovery}
\label{sec:ds}
\textsc{DS-D} was shown to be $\W[1]$-hard with respect to parameter $k + \budget$ on general graphs and with respect to parameter $\budget$ on 2-degenerate graphs. 
On the positive side, however, it is in $\FPT$ for parameter $k$ on biclique-free graphs as well as with respect to parameter $\budget$ on nowhere dense classes of graphs~\cite{fellows2023solution}. 
We show in this section that the problem has polynomial kernels with respect to parameter $k$ on biclique-free classes.
Additionally, via a slight modification to the proofs of \Cref{thm:VC-D-pathwidth,thm:cross_composition-VC-bpw}, we show that \textsc{DS-D} is \XNLP-hard with respect to parameter $pw$ and does not have a polynomial kernel with respect to the parameter $\budget + pw$ where $pw$ is the pathwidth of the input graph unless $\NP \subseteq \cp$.

\begin{theorem}\label{thm:DS-K-biclique-semi-ladder}
Let $\Cc$ be a biclique-free class of graphs. Then \textsc{DS-D} has a polynomial kernel on~$\Cc$ with respect to parameter $k$. 
\end{theorem}

For the proof of \Cref{thm:DS-K-biclique-semi-ladder} we use the concept of \emph{$k$-domination cores}, which were introduced by Dawar
and Kreutzer to approach domination type problems~\cite{DawarK09}. 

\begin{definition}
    Let $G$ be a graph and $k\geq 1$. A set $C\subseteq V(G)$ is a \emph{$k$-domination core} if every set of size at most $k$ that dominates $C$ also dominates $G$.
\end{definition}

Bounded size domination cores do not exist for general graphs, however, they do exist for
many important graph classes, see e.g.~\cite{kreutzer2018polynomial,philip2009solving}, most generally for semi-ladder free graphs~\cite{FabianskiPST19}. 
For our construction of polynomial kernels it is important that biclique-free classes admit polynomial domination cores. For semi-ladder free graphs no polynomial cores are known and the proof of existence only yields an fpt and no polynomial-time algorithm. 
Note that the notion of cores does not appear explicitly in~\cite{philip2009solving}, however, it is easily observed that the set of black vertices in the auxiliary RWB-dominating set problem considered in that work yields a $k$-domination core. 

\begin{lemma}[follows from \cite{philip2009solving}]\label{lem:core}
    Let $\Cc$ be a biclique-free class of graphs. Then there exists a polynomial time algorithm that given $G\in \Cc$ and $k\in \N$ decides that $G$ cannot be dominated by~$k$~vertices or computes a $k$-comination core $C\subseteq V(G)$ of size polynomial in $k$. 
\end{lemma}

\begin{proof}[Proof of \Cref{thm:DS-K-biclique-semi-ladder}]
Let $(G, S, \budget)$ be an instance of \textsc{DS-D}, where $G\in \Cc$ and $|S|=k$. 
We first compute a domination core $C\subseteq V(G)$ 
of size polynomial in $k$, which is possible by \Cref{lem:core}. 
We then compute the projection classes of all vertices towards $C$, where we classify two vertices $u,v\in V(G)$ as equivalent if and only if $N(u)\cap C=N(v)\cap C$. 
The number of
projection classes is polynomially bounded in $|C|$, as biclique-free classes have bounded VC-dimension. 
As $|C|$ is polynomially bounded we derive a polynomial bound also for the number of projection classes. 
\newcommand{\dist}{\mathrm{dist}}
For a set $M\subseteq V(G)$ and a vertex $v\in V(G)$ we define $d(v,M)=\min_{w\in M}d(v,w)$. 
For every projection class $X$ we now fix a minimal set $R_X$ of representative vertices such that for each token~$t$ on a vertex $v_t$ the set $R_X$ contains a vertex $v_{t,X}$ such that $d(v_t,v_{t,X})=d(t,X)$. 
Note that $R_X$ contains at most $k$ vertices and that such a set can be computed in polynomial time by simple breadth-first searches. 
We define $W\subseteq V(G)$ as the union of the vertices of $C$, the vertices of~$S$, and the vertices of a shortest path between $v_t$ and the vertex $v_{t,X}$ for each $v_t\in S$ and projection class $X$. We define the kernel as $(G[M], S,b)$. 

First we prove that $(G[M],S,\budget)$ is an equivalent instance. 
First assume that $(G,S,\budget)$ is a positive instance. 
Let $C_0\vdash C_1\vdash \ldots \vdash C_\ell$ for $\ell\leq C_\budget$ be a discovery sequence. 
We may assume that in each step a token moves on a shortest path to its final destination in $C_\budget$. 
As $C_\budget$ is a dominating set, it dominates in particular the core $C$, say token $t$ is moved to a vertex of projection class $X_t$. 
Then we obtain an equivalent discovery sequence where the token $t$ is moved to $v_{t,X}$ instead. 
The same sequence exists in $G[M]$ and ends in a set of size at most $k$ that dominates $C$. 
Hence, it also dominates $G[M]$, which shows that also $(G[M],S,\budget)$ is a positive instance. 
Conversely, a discovery sequence in $G[M]$ to a dominating set of $G[M]$ leads to a dominating set of $C$ in $G[M]$, which exists exactly like this in $G$. By definition of a $k$-domination core, we also discover a dominating set in $G$. 

Finally, it remains to show that $G[M]$ has size bounded by a polynomial in $k$. 
As we argued already, we have a polynomial size core $C$ and at most polynomially many projection classes. 
From each class we keep at most $k$ representative vertices. 
It remains to show that $\budget$ can be upper bounded by a polynomial in $k$. 
This is easily derived from the fact that a graph with a dominating set of size~$k$ can have diameter at most $3k+2$, as a shortest path with $3k+3$ vertices cannot be dominated by~$k$~vertices. 
Hence, every token arrives in its final position after at most $3k+2$ steps and we can assume that $b\leq 3k^2+2k$. 
\end{proof}

\begin{theorem}\label{thm:DS-D-pathwidth}
\textsc{DS-D} is \XNLP-hard with respect to parameter pathwidth.   
\end{theorem}

As stated in \Cref{sec:foundational}, we present a pl-reduction from \textsc{MMO}. 
Let $(H, \mathcal{P}_H, \sigma, r)$ be an instance of \textsc{MMO} where $H$ is a bounded pathwidth graph with path decomposition $\mathcal{P}_H$, $|V(H)| = n$, $|E(H)| = m$, $\sigma: E(H) \rightarrow \mathbb{Z}_+$ such that $\sum_{e \in E(H)} \sigma(e) = \bm{\sigma}$ and $r \in \mathbb{Z}_{+}$ (integers are given in unary).
We construct an instance $(\Tilde{G}_H, \mathcal{P}_{\Tilde{G}_H}, S, \budget)$ of \textsc{DS-D} as follows (see \Cref{fig:DS-pathwidth-reduction}).
\newline
\newline
We form the graph $\Tilde{G}_H$ as outlined below:
\begin{enumerate}[itemsep=0pt, label=(\alph*)]
    \item We subdivide twice each edge $a^i_eb^i_e$ for each $i \in [\sigma(e)]$ for each edge $e \in E(H)$, and once each edge $a^{\sigma(e)+1}_eb^{\sigma(e)+1}_e$, of a subgraph $G_e$ (which is the \textsc{MMO}-edge-$e$ described in \Cref{sec:foundational}), and add it to $\Tilde{G}_H$.
    We denote the introduced vertices between $a^i_e$ and $b^i_e$ by, in order, $c^i_e$ and $c'^i_e$.
    We denote the introduced vertex from a subdivision of an edge $a^{\sigma(e)+1}_e b^{\sigma(e)+1}_e$ by~$c^{\sigma(e)+1}_e$.
    We let $C_e = \cup_{i \in [\sigma(e)]} \text{ } \text{ } c_e^i$, $C = \cup_{e \in E(H)} C_e$, $C'_e = \cup_{i \in [\sigma(e)]} \text{ } \text{ } c'^i_e$, $C' = \cup_{e \in E(H)} C'_e$, and $C^+ = \cup_{e \in E(H)} c_e^{\sigma(e)+1}$.
    \item We subdivide each edge $w_v x^{v(i)}$ for $i \in [r+1]$ for each vertex $v \in V(H)$, and each edge $w_v y_e^{v(i)}$ for $uv=e \in E(H)$ and $i \in [\sigma(e)]$, of a subgraph $G_v$ (which is the \textsc{MMO}-vertex-$v$ described in \Cref{sec:foundational}), and add it to $\Tilde{G}_H$.
    We denote the introduced vertex from a subdivision of an edge~$w_v x^{v(i)}$ by $c(x^{v(i)})$ and the introduced vertex from a subdivision of an edge~$w_v y_e^{v(i)}$ by~$c(y_e^{v(i)})$.
    We let $c(X_v) = \cup_{i \in [r]} \text{ } \text{ } c(x^{v(i)})$,
    $c(X) = \cup_{v \in V(H)} \text{ } \text{ } c(X_v)$,
    $c(X^+) = \cup_{v \in V(H)} \text{ } \text{ } c(x^{v(r+1)})$, 
    $c(Y_e^v) = \cup_{i \in [\sigma(e)]} \text{ } \text{ } c(y_e^{v(i)})$, $c(Y^v) = \cup_{e \in E(H)} \text{ } \text{ } c(Y_e^v)$, and 
    $c(Y) = \cup_{v \in V(H)} \text{ } \text{ } c(Y_e^v)$.
    \item We make each vertex $b^i_e$ adjacent to the vertices $z_e^{v(i)}$ and $z_e^{u(i)}$, for each $uv=e \in E(H)$ and $i \in [\sigma(e)]$.
    \item We connect each vertex $e^v$ and vertex in $Y_e^v$, for each edge $uv=e \in E(H)$, via paths of length $2$.
    We denote the vertex between $e^v$ and $y_e^{v(i)}$ for $i \in [\sigma(e)]$ by $c'(y_e^{v(i)})$.
    We let $c'(Y_e^v) = \cup_{i \in [\sigma(e)]} \text{ } \text{ } c'(y_e^{v(i)})$, $c'(Y^v) = \cup_{e \in E(H)} \text{ } \text{ } c'(Y_e^v)$, and $c'(Y) = \cup_{v \in V(H)} c'(Y^v)$.
    \item We subdivide the edges $d^i_1d^i_2$, and $d^i_2d^i_3$ for $i \in [rn - \bm{\sigma}]$ of the supplier gadget $G_s$ (described under the supplier gadget and the graph $\Tilde{G}_H$ heading in \Cref{sec:foundational}) twice, and denote the introduced vertices by $d^i_{1^+}$ (for the vertex adjacent to $d^i_1$), $d^i_{2^-}$ (for the vertex adjacent to $d^i_2$ and $d^i_{1^+}$), $d^i_{2^+}$ (for the other vertex adjacent to $d^i_2$), and $d^i_{3^-}$ (for the vertex adjacent to $d^i_3$). 
    We denote the subgraph resulting from subdividing the edges of $G_s$ by the \emph{subdivision of }$G_s$.
    \item We add the subdivision of $G_s$ to $\Tilde{G}_H$ and make the vertex $s$ adjacent to all vertices in $X$.
    \item We add the edge $dd'$ to $\Tilde{G}_H$ and make the dominator vertex $d$ adjacent to all vertices in $c(Y)$.
\end{enumerate}
\begin{figure}
    \centering
    \begin{tikzpicture}
    \fill[green!10] (8,1) -- (7.5,3) -- (8.5,3) -- (8,1) -- cycle;
    \fill[green!10] (8,1) -- (6,0) -- (6,3) -- (8,1) -- cycle;
    \fill[green!10] (8,1) -- (10,-0.25) -- (10,2.5) -- (8,1) -- cycle;
    \fill[green!10] (6,0) -- (6,3) -- (4.75,3) -- (4.75,0) -- cycle;
    \fill[green!10] (8,-1.5) -- (7.5,0.5) -- (8.5,0.5) -- (8,-1.5) -- cycle;
    \fill[green!10] (8,-1.5) -- (6,-1.75) -- (6,-1.25) -- (8,-1.5) -- cycle;
    \fill[green!10] (8,-1.5) -- (10,0) -- (10,-3) -- (8,-1.5) -- cycle;
    \fill[green!10] (6,-1.75) -- (6,-1.25) -- (4.75,-1.25) -- (4.75,-1.75) -- cycle;
    \fill[green!10] (8,-4.7) -- (7.5,-2.7) -- (8.5,-2.7) -- (8,-4.7) -- cycle;
    \fill[green!10] (8,-4.7) -- (6,-6.25) -- (6,-2.75) -- (8,-4.7) -- cycle;
    \fill[green!10] (8,-4.7) -- (10,-2.5) -- (10,-5.75) -- (8,-4.7) -- cycle;
    \fill[green!10] (6,-6.25) -- (6,-2.75) -- (4.75,-2.75) -- (4.75,-6.25) -- cycle;   
    \fill[red!10] (5,-10) -- (4.5,-9) -- (5.5,-9) -- (5,-10) -- cycle;
    \fill[red!10] (5,-10) -- (3,-8) -- (3,-12) -- (5,-10) -- cycle;
    \fill[red!10] (3,-8) -- (-2,-8) -- (-2,-12) -- (3,-12) -- cycle;
    \fill[blue!10] (-0.75,3.75) -- (0.75,3.75) -- (0.75,-0.25) -- (-0.75,-0.25) -- cycle; % Left vertical area
    \fill[blue!10] (0.5,0.1) -- (1.25,0.9) -- (1.25,0) -- (0.5,0) -- cycle; % top right triangle
    \fill[blue!10] (0.5,-0.1) -- (1.25,-0.9) -- (1.25,0) -- (0.5, 0) -- cycle; % bottom right triangle
    \fill[blue!10] (-0.75,-1.25) -- (0.75,-1.25) -- (0.75,-3.25) -- (-0.75,-3.25) -- cycle; % Left vertical area
    \fill[blue!10] (0.5,-2.9) -- (1.25,-2) -- (1.25,-3) -- (0.5,-3) -- cycle; % top right triangle
    \fill[blue!10] (0.5,-3.1) -- (1.25,-4) -- (1.25,-3) -- (0.5, -3) -- cycle; % bottom right triangle
    \draw[black] (8,1) -- (9.5,0.25);
    \draw[black] (8,1) -- (9.5,1);
    \draw[black] (8,1) -- (9.5,1.75);  
    \draw[black] (8,1) -- (8,2.5);
    \draw[black] (8,1) -- (6,0.25);
    \draw[black] (8,1) -- (6,2.25);
    \draw[black] (8,1) -- (6,1.25);
    \draw[black] (8,-1.5) -- (9.5,-2.5);
    \draw[black] (8,-1.5) -- (9.5,-1.5);
    \draw[black] (8,-1.5) -- (9.5,-0.5);
    \draw[black] (8,-1.5) -- (8,0);
    \draw[black] (8,-1.5) -- (6,-1.5);
    \draw[black] (8,-4.7) -- (9.5,-5.2);
    \draw[black] (8,-4.7) -- (9.5,-4.2);
    \draw[black] (8,-4.7) -- (9.5,-3.2);
    \draw[black] (8,-4.7) -- (8,-3.2);
    \draw[black] (8,-4.7)-- (6,-6);
    \draw[black] (8,-4.7) -- (6,-5.2);
    \draw[black] (8,-4.7)-- (6,-4.2);
    \draw[black] (8,-4.7) -- (6,-3.4);
    \draw[orange] (1,-2.5) .. controls (1.5,-2) .. (6,-1.5);
    \draw[orange] (1,-3.5) .. controls (2.5,-5) and  (5,-5) .. (6,-6);    
    \draw[red] (0.5,-2) .. controls (2,-0.75) .. (5,-1.5);
    \draw[red] (0.5,-2) .. controls (1.7,-4.7) and  (5,-6) .. (5,-6);  
    \draw[red] (0.5,3) -- (5,2.25);
    \draw[red] (0.5,2) -- (5,1.25);
    \draw[red] (0.5,1) -- (5,0.25);  
    \draw[blue] (5,-10) .. controls (8,-9) .. (9.5,-5.2);
    \draw[blue] (5,-10) .. controls (10,-9.25) and (11,-5) .. (9.5,-4.2);
    \draw[blue] (5,-10) .. controls (10.25,-10.25) and (11.25,-5) .. (9.5,-3.2);
    \draw[dotted,teal,line width=0.5mm] (7,1.65) -- (7.5,1.65);
    \draw[dotted,teal,line width=0.5mm] (7,-1.5) -- (7.5,-2);
    \draw[dotted,teal,line width=0.5mm] (7,0.6) -- (7.5,0.6);
    \draw[dotted,teal,line width=0.5mm] (7,1.145) -- (6.5,1);
    \draw[dotted,teal,line width=0.5mm] (7,-4.05) -- (7.5,-4.05);
    \draw[dotted,teal,line width=0.5mm] (7,-4.95) -- (6.5,-4.95);
    \draw[dotted,teal,line width=0.5mm] (7,-4.45) -- (6.5,-4.6);
    \draw[dotted,teal,line width=0.5mm] (7,-5.35) -- (7.5,-5.35);
    \node[circle, fill=black, inner sep=1pt, draw=black, fill=white, label=above right:{\scalebox{0.5}{$c'(y_{e_2}^{w(1)})$}}] at (3,-1.83) {};
    \node[circle, fill=black, inner sep=1pt, draw=black, fill=white, label=above right:{\scalebox{0.5}{$c'(y_{e_1}^{u(1)})$}}] at (3,-4.7) {};
    \node[circle, fill=black, inner sep=1pt, draw=black, fill=white, label=below:{\scalebox{0.5}{$w_v$}}] at (8,1) {};
    \node[circle, fill=black, inner sep=1pt, draw=black, fill=white, label=below:{\scalebox{0.5}{$w_w$}}] at (8,-1.5) {};
    \node[circle, fill=black, inner sep=1pt, draw=black, fill=white, label=below:{\scalebox{0.5}{$w_u$}}] at (8,-4.7) {};
    \node[circle, fill=black, inner sep=1pt, draw=black, fill=white, label=right:{\scalebox{0.5}{$x_v^1$}}] at (9.5,0.25) {};
    \node[circle, fill=black, inner sep=1pt, draw=black, fill=white, label=above:{\scalebox{0.5}{$c(x_v^1)$}}] at (9,0.5) {};
    \node[circle, fill=black, inner sep=1pt, draw=black, fill=white, label=right:{\scalebox{0.5}{$x_v^2$}}] at (9.5,1) {};
    \node[circle, fill=black, inner sep=1pt, draw=black, fill=white, label=above:{\scalebox{0.5}{$c(x_v^2)$}}] at (9,1) {};
    \node[circle, fill=black, inner sep=1pt, draw=black, fill=white, label=right:{\scalebox{0.5}{$x_v^3$}}] at (9.5,1.75) {};
    \node[circle, fill=black, inner sep=1pt, draw=black, fill=white, label=above:{\scalebox{0.5}{$c(x_v^3)$}}] at (9,1.5) {};
    \node[circle, fill=black, inner sep=1pt, draw=black, fill=white, label=right:{\scalebox{0.5}{$x_w^1$}}] at (9.5,-2.5) {};
    \node[circle, fill=black, inner sep=1pt, draw=black, fill=white, label=above:{\scalebox{0.5}{$c(x_w^1)$}}] at (9,-2.15) {};
    \node[circle, fill=black, inner sep=1pt, draw=black, fill=white, label=right:{\scalebox{0.5}{$x_w^2$}}] at (9.5,-1.5) {};
    \node[circle, fill=black, inner sep=1pt, draw=black, fill=white, label=above:{\scalebox{0.5}{$c(x_w^2)$}}] at (9,-1.5) {};
    \node[circle, fill=black, inner sep=1pt, draw=black, fill=white, label=right:{\scalebox{0.5}{$x_w^3$}}] at (9.5,-0.5) {};
    \node[circle, fill=black, inner sep=1pt, draw=black, fill=white, label=above:{\scalebox{0.5}{$c(x_w^3)$}}] at (9,-0.8) {};
    \node[circle, fill=black, inner sep=1pt, draw=black, fill=white, label=right:{\scalebox{0.5}{$x_u^1$}}] at (9.5,-5.2) {};
    \node[circle, fill=black, inner sep=1pt, draw=black, fill=white, label=above:{\scalebox{0.5}{$c(x_u^1)$}}] at (8.9,-5) {};
    \node[circle, fill=black, inner sep=1pt, draw=black, fill=white, label=right:{\scalebox{0.5}{$x_u^2$}}] at (9.5,-4.2) {};
    \node[circle, fill=black, inner sep=1pt, draw=black, fill=white, label=above:{\scalebox{0.5}{$c(x_u^2)$}}] at (8.9,-4.4) {};
    \node[circle, fill=black, inner sep=1pt, draw=black, fill=white, label=right:{\scalebox{0.5}{$x_u^3$}}] at (9.5,-3.2) {};
    \node[circle, fill=black, inner sep=1pt, draw=black, fill=white, label=above:{\scalebox{0.5}{$c(x_u^3)$}}] at (8.9,-3.8) {};
    \node[circle, fill=black, inner sep=1pt, draw=black, fill=white, label=above right:{\scalebox{0.5}{$x_v^{4}$}}] at (8,2.5) {};
    \node[circle, fill=black, inner sep=1pt, draw=black, fill=black, label=right:{\scalebox{0.5}{$c(x_v^{4})$}}] at (8,1.75) {};   
    \node[circle, fill=black, inner sep=1pt, draw=black, fill=white, label=above right:{\scalebox{0.5}{$x_w^{4}$}}] at (8,0) {};
    \node[circle, fill=black, inner sep=1pt, draw=black, fill=black, label=right:{\scalebox{0.5}{$c(x_w^{4})$}}] at (8,-0.75) {};
    \node[circle, fill=black, inner sep=1pt, draw=black, fill=white, label=above right:{\scalebox{0.5}{$x_u^{4}$}}] at (8,-3.2) {};
    \node[circle, fill=black, inner sep=1pt, draw=black, fill=black, label=right:{\scalebox{0.5}{$c(x_u^4)$}}] at (8,-3.95) {};
    \draw[black] (5,-1.5) -- (6,-1.5) node[circle, draw=black, fill=black, inner sep=1pt, label=above:{\scalebox{0.5}{$y_{e_2}^{w(1)}$}}, pos=1]{} node[circle, draw=black, fill=white, inner sep=1pt, label=above:{\scalebox{0.5}{$z_{e_2}^{w(1)}$}}, pos=0]{};
    \node[circle, fill=black, inner sep=1pt, draw=black, fill=white, label=above:{\scalebox{0.5}{$c(y_{e_2}^{w(1)})$}}] at (7,-1.5) {};
    \draw[black] (5,1.25) -- (6,1.25) node[circle, draw=black, fill=black, inner sep=1pt, label=above:{\scalebox{0.5}{$y_{e_1}^{v(2)}$}}, pos=1]{} node[circle, draw=black, fill=white, inner sep=1pt, label=above:{\scalebox{0.5}{$z_{e_1}^{v(2)}$}}, pos=0]{};
    \node[circle, fill=black, inner sep=1pt, draw=black, fill=white] () at (7,1.145) {};
    \draw[black] (5,0.25) -- (6,0.25) node[circle, draw=black, fill=black, inner sep=1pt, label=above:{\scalebox{0.5}{$y_{e_1}^{v(1)}$}}, pos=1]{} node[circle, draw=black, fill=white, inner sep=1pt, label=above:{\scalebox{0.5}{$z_{e_1}^{v(1)}$}}, pos=0]{};
    \node[circle, fill=black, inner sep=1pt, draw=black, fill=white, label=below:{\scalebox{0.5}{$c(y_{e_1}^{v(1)})$}}] at (7,0.6) {};
    \draw[black] (5,2.25) -- (6,2.25) node[circle, draw=black, fill=black, inner sep=1pt, label=above:{\scalebox{0.5}{$y_{e_1}^{v(3)}$}}, pos=1]{} node[circle, draw=black, fill=white, inner sep=1pt, label=above:{\scalebox{0.5}{$z_{e_1}^{v(3)}$}}, pos=0]{};
    \node[circle, fill=black, inner sep=1pt, draw=black, fill=white, label=above:{\scalebox{0.5}{$c(y_{e_1}^{v(3)})$}}]at (7,1.65) {};
    \draw[black] (5,-3.4) -- (6,-3.4) node[circle, draw=black, fill=black, inner sep=1pt, label=above:{\scalebox{0.5}{$y_{e_1}^{u(2)}$}}, pos=1]{} node[circle, draw=black, fill=white, inner sep=1pt, label=above:{\scalebox{0.5}{$z_{e_1}^{u(2)}$}}, pos=0]{};
    \node[circle, fill=black, inner sep=1pt, draw=black, fill=white, label=above:{\scalebox{0.5}{$c(y_{e_1}^{u(2)})$}}]at (7,-4.05) {};
    \draw[black] (5,-4.2) -- (6,-4.2) node[circle, draw=black, fill=black, inner sep=1pt, label=above:{\scalebox{0.5}{$y_{e_1}^{u(3)}$}}, pos=1]{} node[circle, draw=black, fill=white, inner sep=1pt, label=above:{\scalebox{0.5}{$z_{e_1}^{u(3)}$}}, pos=0]{};
    \node[circle, fill=black, inner sep=1pt, draw=black, fill=white]at (7,-4.45) {};
    \draw[black] (5,-5.2) -- (6,-5.2) node[circle, draw=black, fill=black, inner sep=1pt, label=above:{\scalebox{0.5}{$y_{e_1}^{u(1)}$}}, pos=1]{} node[circle, draw=black, fill=white, inner sep=1pt, label=above:{\scalebox{0.5}{$z_{e_1}^{u(1)}$}}, pos=0]{};
    \node[circle, fill=black, inner sep=1pt, draw=black, fill=white] at (7,-4.95) {};
    \draw[black] (5,-6) -- (6,-6) node[circle, draw=black, fill=black, inner sep=1pt, label=above:{\scalebox{0.5}{$y_{e_2}^{u(1)}$}}, pos=1]{} node[circle, draw=black, fill=white, inner sep=1pt, label=above:{\scalebox{0.5}{$z_{e_2}^{u(1)}$}}, pos=0]{};
    \node[circle, fill=black, inner sep=1pt, draw=black, fill=white, label=below:{\scalebox{0.5}{$c(y_{e_2}^{u(1)})$}}] at (7,-5.35) {};
    \draw[black] (-0.5,3) -- (0.5,3) node[circle, draw=black, fill=white, inner sep=1pt, label=above:{\scalebox{0.5}{$a^1_{e_1}$}}, pos=0]{} node[circle, draw=black, fill=black, inner sep=1pt, label=above:{\scalebox{0.5}{$b^1_{e_1}$}}, pos=1]{};
    \node[draw, circle, fill=black, inner sep=1pt, label=above:{\scalebox{0.5}{$c^1_{e_1}$}}] at (-0.2,3) {};
    \node[draw, circle, fill=white, inner sep=1pt, label=above:{\scalebox{0.5}{$c'^1_{e_1}$}}] at (0.2,3) {};
    \draw[black] (-0.5,2) -- (0.5,2) node[circle, draw=black, fill=white, inner sep=1pt, label=above:{\scalebox{0.5}{$a^2_{e_1}$}}, pos=0]{} node[circle, draw=black, fill=black, inner sep=1pt, label=above:{\scalebox{0.5}{$b^2_{e_1}$}}, pos=1]{};
    \node[draw, circle, fill=black, inner sep=1pt, label=above:{\scalebox{0.5}{$c^2_{e_1}$}}] at (-0.2,2) {};
    \node[draw, circle, fill=white, inner sep=1pt, label=above:{\scalebox{0.5}{$c'^2_{e_1}$}}] at (0.2,2) {};
    \draw[black] (-0.5,1) -- (0.5,1) node[circle, draw=black, fill=white, inner sep=1pt, label=above:{\scalebox{0.5}{$a^3_{e_1}$}}, pos=0]{} node[circle, draw=black, fill=black, inner sep=1pt, label=above:{\scalebox{0.5}{$b^3_{e_1}$}}, pos=1]{};
    \node[draw, circle, fill=black, inner sep=1pt, label=above:{\scalebox{0.5}{$c^3_{e_1}$}}] at (-0.2,1) {};
    \node[draw, circle, fill=white, inner sep=1pt, label=above:{\scalebox{0.5}{$c'^3_{e_1}$}}] at (0.2,1) {};
    \draw[black] (1,0.5) -- (1, -0.5);
    \draw[black] (0.5,0) -- (1,0.5) node[circle, draw=black, fill=black, inner sep=1pt, label=right:{\scalebox{0.5}{$e_1^v$}}, pos=1]{};
    \draw[black] (0.5,0) -- (1,-0.5) node[circle, draw=black, fill=black, inner sep=1pt, label=right:{\scalebox{0.5}{$e_1^u$}}, pos=1]{};
    \draw[black] (-0.5,0) -- (0.5,0) node[circle, draw=black, fill=white, inner sep=1pt, label=above:{\scalebox{0.5}{$a^4_{e_1}$}}, pos=0]{} node[circle, draw=black, fill=white, inner sep=1pt, label=above:{\scalebox{0.5}{$b^4_{e_1}$}}, pos=1]{};
    \node[draw, circle, fill=white, inner sep=1pt, label=above:{\scalebox{0.5}{$c^4_{e_1}$}}] at (0,0) {};
    \draw[black] (-0.5,-2) -- (0.5,-2) node[circle, draw=black, fill=white, inner sep=1pt, label=above:{\scalebox{0.5}{$a^1_{e_2}$}}, pos=0]{} node[circle, draw=black, fill=black, inner sep=1pt, label=above:{\scalebox{0.5}{$b^1_{e_2}$}}, pos=1]{};
    \node[draw, circle, fill=black, inner sep=1pt, label=above:{\scalebox{0.5}{$c^1_{e_2}$}}] at (-0.2,-2) {};
    \node[draw, circle, fill=white, inner sep=1pt, label=above:{\scalebox{0.5}{$c'^1_{e_2}$}}] at (0.2,-2) {};
    \draw[black] (1,-2.5) -- (1, -3.5);
    \draw[black] (0.5,-3) -- (1,-2.5) node[circle, draw=black, fill=black, inner sep=1pt, label=right:{\scalebox{0.5}{$e_2^w$}}, pos=1]{};
    \draw[black] (0.5,-3) -- (1,-3.5) node[circle, draw=black, fill=black, inner sep=1pt, label=right:{\scalebox{0.5}{$e_2^u$}}, pos=1]{};
    \draw[black] (-0.5,-3) -- (0.5,-3) node[circle, draw=black, fill=white, inner sep=1pt, label=above:{\scalebox{0.5}{$a^2_{e_2}$}}, pos=0]{} node[circle, draw=black, fill=white, inner sep=1pt, label=above:{\scalebox{0.5}{$b^2_{e_2}$}}, pos=1]{};
    \node[draw, circle, fill=white, inner sep=1pt, label=above:{\scalebox{0.5}{$c^2_{e_2}$}}] at (0,-3) {};
    \draw[black] (5,-10) -- (5,-9.5);
    \draw[black] (5,-10)-- (3,-11.5);
    \draw[black] (5,-10)-- (3,-10.75);
    \draw[black] (5,-10) -- (3,-10);
    \draw[black] (5,-10) -- (3,-9.25);
    \draw[black] (5,-10) -- (3,-8.5);
    \draw[black] (2.2,-11.5) -- (3,-11.5);
    \draw[black] (2.2,-10.75) -- (3,-10.75);
    \draw[black] (2.2,-10) -- (3,-10);
    \draw[black] (2.2,-9.25) -- (3,-9.25); 
    \draw[black] (2.2,-8.5) -- (3,-8.5); 
    \draw[black] (1.4,-11.5) -- (2.2,-11.5);
    \draw[black] (1.4,-10.75) -- (2.2,-10.75);
    \draw[black] (1.4,-10) -- (2.2,-10);
    \draw[black] (1.4,-9.25) -- (2.2,-9.25); 
    \draw[black] (1.4,-8.5) -- (2.2,-8.5);
    \draw[black] (0.6,-11.5) -- (1.4,-11.5);
    \draw[black] (0.6,-10.75) -- (1.4,-10.75);
    \draw[black] (0.6,-10) -- (1.4,-10);
    \draw[black] (0.6,-9.25) -- (1.4,-9.25); 
    \draw[black] (0.6,-8.5) -- (1.4,-8.5); 
    \draw[black] (-0.2,-11.5) -- (0.6,-11.5);
    \draw[black] (-0.2,-10.75) -- (0.6,-10.75);
    \draw[black] (-0.2,-10) -- (0.6,-10);
    \draw[black] (-0.2,-9.25) -- (0.6,-9.25); 
    \draw[black] (-0.2,-8.5) -- (0.6,-8.5); 
    \draw[black] (-1,-11.5) -- (-0.2,-11.5);
    \draw[black] (-1,-10.75) -- (-0.2,-10.75);
    \draw[black] (-1,-10) -- (-0.2,-10);
    \draw[black] (-1,-9.25) -- (-0.2,-9.25); 
    \draw[black] (-1,-8.5) -- (-0.2,-8.5);
    \draw[black] (-1.8,-11.5) -- (-1,-11.5);
    \draw[black] (-1.8,-10.75) -- (-1,-10.75);
    \draw[black] (-1.8,-10) -- (-1,-10);
    \draw[black] (-1.8,-9.25) -- (-1,-9.25); 
    \draw[black] (-1.8,-8.5) -- (-1,-8.5);
    \node[circle, fill=black, inner sep=1pt, draw=black, fill=black, label=below right:{\scalebox{0.5}{$s$}}] at (5,-10) {};
    \node[circle, fill=black, inner sep=1pt, draw=black, fill=white, label=above right:{\scalebox{0.5}{$d_1^6$}}] at (5,-9.5) {}; 
    \node[circle, draw=black, fill=white, inner sep=1pt, label=above:{\scalebox{0.5}{$d^1_{1^+}$}}] at (2.2,-11.5) {};
    \node[circle, draw=black, fill=black, inner sep=1pt, label=above:{\scalebox{0.5}{$d^1_1$}}] at (3,-11.5) {};
    \node[circle, draw=black, fill=white, inner sep=1pt, label=above:{\scalebox{0.5}{$d^2_{1^+}$}}] at (2.2,-10.75) {};
    \node[circle, draw=black, fill=black, inner sep=1pt, label=above:{\scalebox{0.5}{$d^2_1$}}] at (3,-10.75) {};
    \node[circle, draw=black, fill=white, inner sep=1pt, label=above:{\scalebox{0.5}{$d^3_{1^+}$}}] at (2.2,-10) {};
    \node[circle, draw=black, fill=black, inner sep=1pt, label=above:{\scalebox{0.5}{$d^3_1$}}] at (3,-10) {};
    \node[circle, draw=black, fill=white, inner sep=1pt, label=above:{\scalebox{0.5}{$d^4_{1^+}$}}] at (2.2,-9.25) {};
    \node[circle, draw=black, fill=black, inner sep=1pt, label=above:{\scalebox{0.5}{$d^4_1$}}] at (3,-9.25) {};
    \node[circle, draw=black, fill=white, inner sep=1pt, label=above:{\scalebox{0.5}{$d^5_{1^+}$}}] at (2.2,-8.5) {};
    \node[circle, draw=black, fill=black, inner sep=1pt, label=above:{\scalebox{0.5}{$d^5_1$}}] at (3,-8.5) {};
    \node[circle, draw=black, fill=white, inner sep=1pt, label=above:{\scalebox{0.5}{$d^1_{2^-}$}}] at (1.4,-11.5) {};
    \node[circle, draw=black, fill=white, inner sep=1pt, label=above:{\scalebox{0.5}{$d^2_{2^-}$}}] at (1.4,-10.75) {};
    \node[circle, draw=black, fill=white, inner sep=1pt, label=above:{\scalebox{0.5}{$d^3_{2^-}$}}] at (1.4,-10) {};
    \node[circle, draw=black, fill=white, inner sep=1pt, label=above:{\scalebox{0.5}{$d^4_{2^-}$}}] at (1.4,-9.25) {};
    \node[circle, draw=black, fill=white, inner sep=1pt, label=above:{\scalebox{0.5}{$d^5_{2^-}$}}] at (1.4,-8.5) {};
    \node[circle, draw=black, fill=black, inner sep=1pt, label=above:{\scalebox{0.5}{$d^1_2$}}] at (0.6,-11.5) {};
    \node[circle, draw=black, fill=black, inner sep=1pt, label=above:{\scalebox{0.5}{$d^2_2$}}] at (0.6,-10.75) {};
    \node[circle, draw=black, fill=black, inner sep=1pt, label=above:{\scalebox{0.5}{$d^3_2$}}] at (0.6,-10) {};
    \node[circle, draw=black, fill=black, inner sep=1pt, label=above:{\scalebox{0.5}{$d^4_2$}}] at (0.6,-9.25) {};
    \node[circle, draw=black, fill=black, inner sep=1pt, label=above:{\scalebox{0.5}{$d^5_2$}}] at (0.6,-8.5) {};
    \node[circle, draw=black, fill=white, inner sep=1pt, label=above:{\scalebox{0.5}{$d^1_{2^+}$}}] at (-0.2,-11.5) {};
    \node[circle, draw=black, fill=white, inner sep=1pt, label=above:{\scalebox{0.5}{$d^2_{2^+}$}}] at (-0.2,-10.75) {};
    \node[circle, draw=black, fill=white, inner sep=1pt, label=above:{\scalebox{0.5}{$d^3_{2^+}$}}] at (-0.2,-10) {};
    \node[circle, draw=black, fill=white, inner sep=1pt, label=above:{\scalebox{0.5}{$d^4_{2^+}$}}] at (-0.2,-9.25) {};
    \node[circle, draw=black, fill=white, inner sep=1pt, label=above:{\scalebox{0.5}{$d^5_{2^+}$}}] at (-0.2,-8.5) {};
    \node[circle, draw=black, fill=white, inner sep=1pt, label=above:{\scalebox{0.5}{$d^1_{3^-}$}}] at (-1,-11.5) {};
    \node[circle, draw=black, fill=white, inner sep=1pt, label=above:{\scalebox{0.5}{$d^2_{3^-}$}}] at (-1,-10.75) {};
    \node[circle, draw=black, fill=white, inner sep=1pt, label=above:{\scalebox{0.5}{$d^3_{3^-}$}}] at (-1,-10) {};
    \node[circle, draw=black, fill=white, inner sep=1pt, label=above:{\scalebox{0.5}{$d^4_{3^-}$}}] at (-1,-9.25) {};
    \node[circle, draw=black, fill=white, inner sep=1pt, label=above:{\scalebox{0.5}{$d^5_{3^-}$}}] at (-1,-8.5) {};
    \node[circle, draw=black, fill=black, inner sep=1pt, label=above:{\scalebox{0.5}{$d^1_{3}$}}] at (-1.8,-11.5) {};
    \node[circle, draw=black, fill=black, inner sep=1pt, label=above:{\scalebox{0.5}{$d^2_{3}$}}] at (-1.8,-10.75) {};
    \node[circle, draw=black, fill=black, inner sep=1pt, label=above:{\scalebox{0.5}{$d^3_3$}}] at (-1.8,-10) {};
    \node[circle, draw=black, fill=black, inner sep=1pt, label=above:{\scalebox{0.5}{$d^4_3$}}] at (-1.8,-9.25) {};
    \node[circle, draw=black, fill=black, inner sep=1pt, label=above:{\scalebox{0.5}{$d^5_3$}}] at (-1.8,-8.5) {};
    
    \end{tikzpicture}
    \caption{\footnotesize Parts of the graph $\Tilde{G}_H$ constructed by the reduction of Theorem~\ref{thm:DS-D-pathwidth} given an instance $(H, \mathcal{P}_H, \sigma, r)$, where $H$ has three vertices $u, v$, and $w$, and two edges $e_1=uv$ and $e_2=uw$, and $r = 3$. Additionally, $\sigma(e_1) = 3$ and $\sigma(e_2) = 1$. For clarity, the edges between the vertices in $B_{e_1}$ and $Z_{e_1}^u$ are missing. The same applies for the paths between $e_1^v$ and the vertices of $Y_{e_1}^v$, the paths between $e_1^u$ and the vertices of $Y_{e_1}^u$ as well as some of the edges between~$s$ and the vertices in $X$. Dotted lines incident to the vertices in $c(Y)$ represent edges to the vertex $d$ which is not shown in the figure. Red, yellow, and blue edges are used to highlight the different types of edges used to connect the subgraphs $G_{e_1}$, $G_{e_2}$, $G_u$, $G_v$, $G_w$ and $G_s$ of $\Tilde{G}_H$, vertices in black are in S and those in white are not.}
    \label{fig:DS-pathwidth-reduction}
\end{figure}
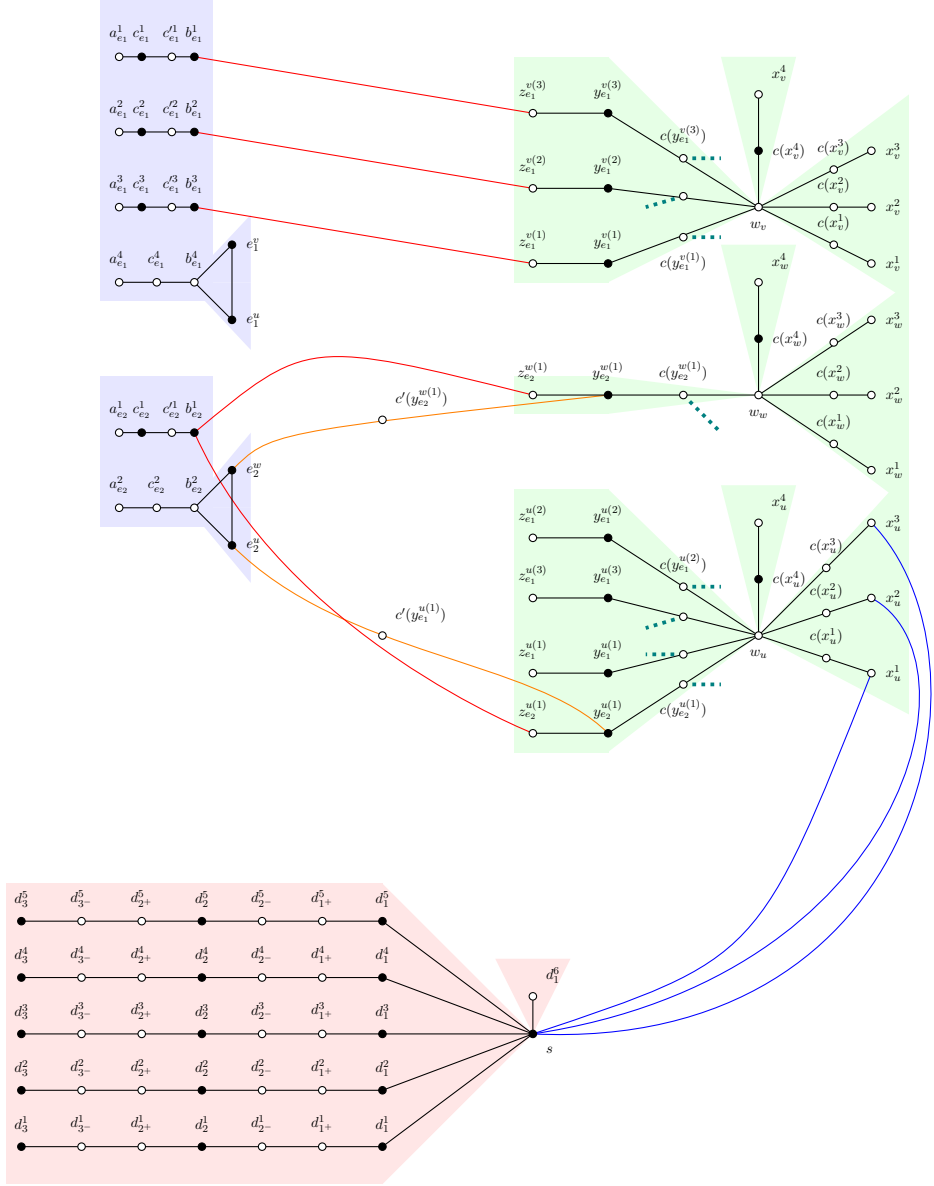

By \Cref{cor:aug-G-H-pathwidth}, $\Tilde{G}_H$ has bounded pathwidth (as an augmented subdivision of the original graph~$\Tilde{G}_H$ constructed in \Cref{sec:foundational}).
We set $S = C \text{ } \cup \text{ } B \text{ } \cup \text{ } Y \text{ } \cup \text{ } c(X^+)  \text{ } \cup \text{ } \bigcup_{uv=e \in E(H)} (e^u \cup e^v) \text{ } \cup $ $\bigcup_{i \in [rn -  \bm{\sigma}]} (d_1^i \cup d_2^i \cup d_3^i) \text{ } \cup \text{ } s  \text{ } \cup \text{ } d$ and $\budget = 2m + 4rn$.
Given that all integers are given in unary, the construction of the graph $\Tilde{G}_H$, or its path decomposition (as described in the discussion for \Cref{cor:aug-G-H-pathwidth}), and as a consequence the reduction, take time polynomial in the size of the input instance. Additionally, by \Cref{cor:reduction-space}, this reduction is a pl-reduction. 
We claim that $(H, \mathcal{P}_H, w, r)$ is a yes-instance of \textsc{MMO} if and only if $(\Tilde{G}_H, \mathcal{P}_{\Tilde{G}_H}, S, \budget)$ is a yes-instance of \textsc{DS-D}.

\begin{lemma}\label{lem:hardness-DS-pathwidth-forward}
If $(H, \mathcal{P}_H, \sigma, r)$ is a yes-instance of \textsc{MMO}, then $(\Tilde{G}_H, \mathcal{P}_{\Tilde{G}_H}, S, \budget)$ is a yes-instance of \textsc{DS-D}.
\end{lemma}
\begin{proof}
Let $\lambda: E(H) \rightarrow V(H) \times V(H)$ be an orientation of the graph $H$ such that for each $v \in V(H)$, the total weight of the edges directed out of $v$ is at most $r$.
In $\Tilde{G}_H$, the vertices in $c(X)$, $A^+$, and~$C^+$ are not dominated.
To fix that, for each edge $uv=e \in E(H)$ such that $\lambda(e) = (v, u)$:
\begin{itemize}[itemsep=0pt]
    \item we move, for each $i \in [\sigma(e)]$, the token on $y_e^{v(i)}$ to any free vertex of $c(X_v)$ and the token on $b_e^{i}$ to $z_e^{v(i)}$ (this consume $4\sigma(e)$ slides),
    \item we slide the token on $e^u$ to $c_e^{\sigma(e)+1}$, hence dominating $a_e^{\sigma(e)+1}$ and $c_e^{\sigma(e)+1}$ (this consumes $2$ slides).
\end{itemize}

This constitutes $4\bm{\sigma} + 2m$ slides.
We dominate the $rn - \bm{\sigma}$ remaining non-dominated vertices in~$c(X)$, using $4$ slides per $D^i$ path for $i \in [rn - \bm{\sigma}]$ (by sliding the token on $d^i_3$ to $d^i_{3^-}$, the token on~$d^i_2$ to~$d^i_{2^-}$ and moving the token on $d^i_1$ to a token-free vertex in $X$ that neighbors a non-dominated vertex in $c(X)$).
\end{proof}

\begin{lemma}\label{lem:hardness-DS-pathwidth-backward}
If $(\Tilde{G}_H, \mathcal{P}_{\Tilde{G}_H}, S, \budget)$ is a yes-instance of \textsc{DS-D}, then $(H, \mathcal{P}_H, \sigma, r)$ is a yes-instance of \textsc{MMO}.
\end{lemma}

\begin{proof}
First, note that for a vertex $a_e^{\sigma(e)+1}$, where $uv=e \in E(H)$ to be dominated with a minimal number of slides, the token on the vertex $e^u$ or the token on the vertex $e^v$ must move to $c_e^{\sigma(e)+1}$ (note that any other token on the vertices of the graph must pass through either $e^u$ or $e^v$ to get to~$c_e^{\sigma(e)+1}$, thus we can safely assume that the token already on either of $e^u$ or $e^v$ is the token that slides to $c_e^{\sigma(e)+1}$).
This consumes at least $2m$ slides, leaving $4rn$ slides.

No dominating set formed with a minimal number of slides would need to make the token on a vertex in $c(X^+)$ or the tokens on either of the vertices $s$ or $d$ slide (as this token must always be replaced by another to dominate the vertices in $X^+$, or the vertex $d_1^{rn-\sigma+1}$, or the vertex $d'$, respectively, with a minimal number of slides, thus we can always assume that the token has not been moved).
Thus, a pair of vertices in $c(X_v)$ and $X_v$ for a vertex $v \in V(H)$ can be dominated by either moving the token on a vertex $d_1^i$ for an integer $i \in [rn - \bm{\sigma}]$ towards the vertex in $X^v$, or moving a token from a vertex $y_e^{v(i_1)}$ for an edge $uv=e \in E(H)$ and an integer $i_1 \in [\sigma(e)]$, towards the non-dominated vertex in $c(X_v)$.
If the token on $d_1^i$ moves towards the vertex in $X_v$, the token on~$d_2^i$ must slide to $d_{2^-}^i$ and, the token on~$d_3^i$ must slide to $d_{3^-}^i$, so that $d^i_{1^+}$ is dominated. 
If a token on~$y_e^{v(i_1)}$ moves towards the vertex in $c(X_v)$, it must be the case that another token has moved to either the vertex $z_e^{v(i_1)}$ or, the vertex $c(y_e^{v(i_1)})$ or, the vertex $c'(y_e^{v(i_1)})$ or, to $y_e^{v(i_1)}$ itself (to dominate~$y_e^{v(i_1)}$).
This however requires at least one slide per such a token (as no vertex that dominates more than one vertex in $Y$ exists).

Thus, if a vertex in $c(X)$ is dominated by moving a token from one vertex~$d_1^i$ for an integer $i \in [rn - \bm{\sigma}]$ towards the vertex in $X$, it does not consume more slides than moving a token from a vertex in~$Y$ towards the vertex in $c(X)$.
Given that at most $rn - \bm{\sigma}$ vertices can be dominated using tokens from the donor paths (as $rn - \bm{\sigma} +1$ tokens are needed to dominate the vertices in $G_s$), each of the at least $\bm{\sigma}$ remaining vertices in $c(X)$ must be dominated by moving a token from a vertex in~$Y$ towards a vertex in $c(X)$.
Additionally, each of the remaining vertices in $c(X)$ will require at least one additional slide (besides the two slides needed to move a token from $Y$) and thus, tokens on distinct vertices in~$Y$ must be used to dominate the vertices in $c(X)$, as the remaining at most $\bm{\sigma}$ slides do not allow to get any token not initially on a vertex in $Y$ to a vertex in $Y$.
If the token on the vertex $e^v$, slides to $c'(y_e^{v(i_1)})$, it will require at least one more slide as $e^u$ will not be dominated.
Thus, the token on the vertex $b_e^{i_1}$ slides to $z_e^{v(i_1)}$ when the token on $y_e^{v(i_1)}$ moves towards a vertex in~$c(X^v)$.

This totals $\budget$ slides.
For each vertex that is token-free in $Y$ after the $\budget$ slides are consumed, the adjacent vertices in $c'(Y)$ must be adjacent to a vertex of the form $e_2^{u_2}$ with a token, for an edge $u_2v_2 = e_2 \in E(H)$ (so that they are dominated). 
This implies that for each edge $u_2v_2= e_2 \in E(H)$, at most $\sigma(e_2)$ tokens can move to $c(X)$ from tokens on the vertices of the sets $Y^{v_2}_{e_2}$ and $Y^{u_2}_{e_2}$, and from only one of those sets, as only one of $e_2^{u_2}$ and $e_2^{v_2}$ has a token. 
To dominate the $\bm{\sigma}$ remaining non-dominated vertices in $c(X)$, each edge $u_2v_2=e_2 \in E(H)$ must allow $\sigma(e_2)$ tokens to move from either vertices in $Y^{v_2}_{e_2}$ and $Y^{u_2}_{e_2}$ and from at most one.
This gives a feasible orientation for the instance $(H, \mathcal{P}_H, \sigma, r)$ as any of $c(X_u)$ or $c(X_v)$ can receive at most $r$ tokens.
\end{proof}
The proofs of Lemmas~\ref{lem:hardness-DS-pathwidth-forward} and \ref{lem:hardness-DS-pathwidth-backward} complete the proof of Theorem~\ref{thm:DS-D-pathwidth}.

\begin{theorem}\label{thm:cross_composition-DS-bpw}
There exists an or-cross-composition from \textsc{MMO} into \textsc{DS-D} on bounded pathwidth graphs and where the parameter is $\budget$. Consequently, \textsc{DS-D} does not admit a polynomial kernel with respect to $\budget + pw$, where $pw$ denotes the pathwidth of the input graphs, unless $\NP \subseteq \cp$. 
\end{theorem}

\begin{proof}
As stated in \Cref{sec:foundational}, we can assume that we are given a family of $t$ \textsc{MMO} instances $(H_j, \mathcal{P}_{H_j}, \sigma_j, r_j)$, where $H_j$ is a bounded pathwidth graph with path decomposition $\mathcal{P}_{H_j}$, $|V(H_j)| = n$, $|E(H_j)| = m$, $\sigma_j: E_j \rightarrow \mathbb{Z}_+$ is a weight function such that $\sum_{e_j \in E(H_j)} \sigma_j(e_j) = \bm{\sigma}$ and $r_j = r \in \mathbb{Z}_+$ (integers are given in unary).
The construction of the instance $(\hat{G}_t, \mathcal{P}_{G_t}, S, \budget)$ of \textsc{DS-D} is twofold. 

For each instance $H_j$ for $j \in [t]$, we add to $\hat{G}_t$ the graph $G_{H_j}$ formed as per the construction in \Cref{thm:DS-D-pathwidth}, but without the supplier gadget.
We refer to the sets $A$, $B$, $X$, $X^+$, $C$, $C'$, $C^+$, $Y$, $c(X)$, $c(X^+)$, $c(Y)$, and $c'(Y)$, subsets of vertices of a subgraph $G_{H_j}$ of $\hat{G}_t$, by $A_j$, $B_j$, $X_j$, $X_j^+$, $C_j$, $C_j'$, $C_j^+$, $Y_j$, $c(X_j)$, $c(X_j^+)$, $c(Y_j)$, and $c'(Y_j)$, respectively. 
Similarly, we refer to the vertices $d$ and $d'$ of a subgraph $G_{H_j}$ of $G$, by $d_j$ and $d'_j$, respectively.
Subsequently, we let $A = \cup_{j \in [t]} A_j$, $B = \cup_{j \in [t]} B_j$, $X = \cup_{j \in [t]} X_j$, and so on.
We add the \textsc{MMO}-instance-selector (described in \Cref{sec:foundational}) and connect it to the rest of $\hat{G}_t$ as follows (see \Cref{fig:DS-pathwidth-composition}). 
We connect, for each $j \in [t]$, the vertex $\textsc{\footnotesize Select}_j$ to the vertices in $V(G_{H_j}) \cap S$, where $S$ is as defined later, via paths of length $2$.
We make the vertex $h$ adjacent to each vertex in $X$ and the vertex $q$ adjacent to each vertex of $e^u$ and $e^v$ for each edge $uv=e \in E(H_j)$ for each $j \in [t]$.
By \Cref{cor:cor-bounded-pathwidth-G'-t}, $\hat{G}_t$ is of bounded pathwidth.
Now, we set
\begin{equation*}
\begin{split}
    S = C \text{ } \cup \text{ } B \text{ } \cup \text{ } A^+ \text{ } \cup \text{ } Y \text{ } \cup \text{ } X \text{ } \cup \text{ } c(X^+)  \text{ } \cup \text{ } d \text{ } \cup \text{ } \bigcup_{j \in [t]} \textsc{\footnotesize Unselect}_j \text{ } \cup  \bigcup_{\substack{j \in [t] \\ uv=e \in E(H_j)}} (e^u \text{ } \cup \text{ } e^v)
\end{split}
\end{equation*}
and $\budget = 2m + 6\bm{\sigma} + 1$. 
\begin{figure}[H]
    \centering
    \begin{tikzpicture}[scale=0.4]
    \fill[green!10] (8,1) -- (7.5,3) -- (8.5,3) -- (8,1) -- cycle;
    \fill[green!10] (8,1) -- (6,0) -- (6,3) -- (8,1) -- cycle;
    \fill[green!10] (8,1) -- (10,-0.25) -- (10,2.5) -- (8,1) -- cycle;
    \fill[green!10] (6,0) -- (6,3) -- (4.75,3) -- (4.75,0) -- cycle;
    \fill[green!10] (8,-1.5) -- (7.5,0.5) -- (8.5,0.5) -- (8,-1.5) -- cycle;
    \fill[green!10] (8,-1.5) -- (6,-1.75) -- (6,-1.25) -- (8,-1.5) -- cycle;
    \fill[green!10] (8,-1.5) -- (10,0) -- (10,-3) -- (8,-1.5) -- cycle;
    \fill[green!10] (6,-1.75) -- (6,-1.25) -- (4.75,-1.25) -- (4.75,-1.75) -- cycle;
    \fill[green!10] (8,-4.7) -- (7.5,-2.7) -- (8.5,-2.7) -- (8,-4.7) -- cycle;
    \fill[green!10] (8,-4.7) -- (6,-6.25) -- (6,-2.75) -- (8,-4.7) -- cycle;
    \fill[green!10] (8,-4.7) -- (10,-2.5) -- (10,-5.75) -- (8,-4.7) -- cycle;
    \fill[green!10] (6,-6.25) -- (6,-2.75) -- (4.75,-2.75) -- (4.75,-6.25) -- cycle;
    \fill[blue!10] (-0.75,3.75) -- (0.75,3.75) -- (0.75,-0.25) -- (-0.75,-0.25) -- cycle; 
    \fill[blue!10] (0.5,0.1) -- (1.25,0.9) -- (1.25,0) -- (0.5,0) -- cycle; 
    \fill[blue!10] (0.5,-0.1) -- (1.25,-0.9) -- (1.25,0) -- (0.5, 0) -- cycle;
    \fill[blue!10] (-0.75,-1.25) -- (0.75,-1.25) -- (0.75,-3.25) -- (-0.75,-3.25) -- cycle; 
    \fill[blue!10] (0.5,-2.9) -- (1.25,-2) -- (1.25,-3) -- (0.5,-3) -- cycle;
    \fill[blue!10] (0.5,-3.1) -- (1.25,-4) -- (1.25,-3) -- (0.5, -3) -- cycle; 
    \draw[black] (8,1) -- (9.5,0.25);
    \draw[black] (8,1) -- (9.5,1);
    \draw[black] (8,1) -- (9.5,1.75);
    \draw[black] (8,1) -- (8,2.5);
    \draw[black] (8,1) -- (6,0.25);
    \draw[black] (8,1) -- (6,2.25);
    \draw[black] (8,1) -- (6,1.25);
    \draw[black] (8,-1.5) -- (9.5,-2.5);
    \draw[black] (8,-1.5) -- (9.5,-1.5);
    \draw[black] (8,-1.5) -- (9.5,-0.5);
    \draw[black] (8,-1.5) -- (8,0);
    \draw[black] (8,-1.5) -- (6,-1.5);
    \draw[black] (8,-4.7) -- (9.5,-5.2);
    \draw[black] (8,-4.7) -- (9.5,-4.2);
    \draw[black] (8,-4.7) -- (9.5,-3.2);
    \draw[black] (8,-4.7) -- (8,-3.2);
    \draw[black] (8,-4.7)-- (6,-6);
    \draw[black] (8,-4.7) -- (6,-5.2);
    \draw[black] (8,-4.7)-- (6,-4.2);
    \draw[black] (8,-4.7) -- (6,-3.4);
    \draw[orange] (1,-2.5) .. controls (1.5,-2) .. (6,-1.5);
    \draw[orange] (1,-3.5) .. controls (2.5,-5) and  (5,-5) .. (6,-6);    
    \draw[red] (0.5,-2) .. controls (2,-0.75) .. (5,-1.5);
    \draw[red] (0.5,-2) .. controls (1.3,-4.7) and  (5,-6) .. (5,-6);  
    \draw[red] (0.5,3) -- (5,2.25);
    \draw[red] (0.5,2) -- (5,1.25);
    \draw[red] (0.5,1) -- (5,0.25);  
    \draw[dotted,teal,line width=0.25mm] (7,1.65) -- (6.5,1.65);
    \draw[dotted,teal,line width=0.25mm] (7,-1.5) -- (6.5,-1.75);
    \draw[dotted,teal,line width=0.25mm] (7,0.6) -- (6.5,0.6);
    \draw[dotted,teal,line width=0.25mm] (7,1.145) -- (6.5,1);
    \draw[dotted,teal,line width=0.25mm] (7,-4.05) -- (6.5,-4.05);
    \draw[dotted,teal,line width=0.25mm] (7,-4.95) -- (6.5,-4.95);
    \draw[dotted,teal,line width=0.25mm] (7,-4.45) -- (6.5,-4.6);
    \draw[dotted,teal,line width=0.25mm] (7,-5.35) -- (6.5,-5.35);
    \fill[yellow!10] (-3.5,-10) -- (-4.5,-10) -- (-4.5,-9) -- (-3.5,-9) -- cycle;
    \fill[yellow!10] (-5,0) -- (-6.25,-1.25) -- (-6.25,1.25) -- cycle;
    \fill[yellow!10] (-6.25,-1.25) -- (-7.25,-1.25) -- (-7.25, 1.25) -- (-6.25,1.25) -- cycle;
    \fill[yellow!10] (15,0) --  (17,2.5) -- (17,-2.25) -- cycle;
    \fill[yellow!10] (17,2.5) -- (18,2.5) -- (18,-2.25) -- (17,-2.25) -- cycle; 
    \draw[teal] (9.5,0.25) -- (15, 0);
    \draw[teal] (9.5,1) -- (15, 0);
    \draw[teal] (9.5,1.75) -- (15, 0);
    \draw[teal] (9.5,-2.5) -- (15, 0);
    \draw[teal] (9.5,-1.5) -- (15, 0);
    \draw[teal] (9.5,-0.5) -- (15, 0);
    \draw[teal] (9.5,-5.2) -- (15, 0);
    \draw[teal] (9.5,-4.2) -- (15, 0);
    \draw[teal] (9.5,-3.2) -- (15, 0);
    \draw[orange] (1,0.5) -- (-5,0);
    \draw[orange] (1,-0.5) -- (-5,0);
    \draw[orange] (1,-2.5) .. controls (-2,-2.4) .. (-5,0);
    \draw[orange] (1,-3.5) .. controls (-2,-3.5) .. (-5,0);
    \draw[pink] (-4,-9) -- (1,-3.5);
    \draw[pink] (-4,-9) -- (6,-5.2);
    \draw[pink] (-4,-9) -- (6,-6);
    \draw[black] (-6,0.75) -- (-5,0);
    \draw[black] (-6,-0.75) -- (-5, 0);
    \draw[black] (17,0.75) -- (15,0);
    \draw[black] (17,-0.75) -- (15,0);
    \draw[black] (17,2) -- (15,0);
    \draw[black] (17,-2) -- (15,0);
    \node[circle, fill=black, inner sep=0.5pt, draw=black, fill=white] at (3,-1.83) {};
    \node[circle, fill=black, inner sep=0.5pt, draw=black, fill=white] at (3,-4.7) {};
    \node[circle, fill=black, inner sep=0.5pt, draw=black, fill=white, label=below:{\scalebox{0.5}{$w_v$}}] (center1) at (8,1) {};
    \node[circle, fill=black, inner sep=0.5pt, draw=black, fill=white, label=below:{\scalebox{0.5}{$w_w$}}] (center2) at (8,-1.5) {};
    \node[circle, fill=black, inner sep=0.5pt, draw=black, fill=white, label=below:{\scalebox{0.5}{$w_u$}}] (center3) at (8,-4.7) {};
    \node[circle, fill=black, inner sep=0.5pt, draw=black, fill=black] at (9.5,0.25) {};
    \node[circle, fill=black, inner sep=0.5pt, draw=black, fill=white] at (9,0.5) {};
    \node[circle, fill=black, inner sep=0.5pt, draw=black, fill=black] at (9.5,1) {};
    \node[circle, fill=black, inner sep=0.5pt, draw=black, fill=white] at (9,1) {};
    \node[circle, fill=black, inner sep=0.5pt, draw=black, fill=black] at (9.5,1.75) {};
    \node[circle, fill=black, inner sep=0.5pt, draw=black, fill=white] at (9,1.5) {};
    \node[circle, fill=black, inner sep=0.5pt, draw=black, fill=black] at (9.5,-2.5) {};
    \node[circle, fill=black, inner sep=0.5pt, draw=black, fill=white] at (9,-2.15) {};
    \node[circle, fill=black, inner sep=0.5pt, draw=black, fill=black] at (9.5,-1.5) {};
    \node[circle, fill=black, inner sep=0.5pt, draw=black, fill=white] at (9,-1.5) {};
    \node[circle, fill=black, inner sep=0.5pt, draw=black, fill=black] at (9.5,-0.5) {};
    \node[circle, fill=black, inner sep=0.5pt, draw=black, fill=white] at (9,-0.8) {};
    \node[circle, fill=black, inner sep=0.5pt, draw=black, fill=black] at (9.5,-5.2) {};
    \node[circle, fill=black, inner sep=0.5pt, draw=black, fill=white] at (8.9,-5) {};
    \node[circle, fill=black, inner sep=0.5pt, draw=black, fill=black] at (9.5,-4.2) {};
    \node[circle, fill=black, inner sep=0.5pt, draw=black, fill=white] at (8.9,-4.4) {};
    \node[circle, fill=black, inner sep=0.5pt, draw=black, fill=black] at (9.5,-3.2) {};
    \node[circle, fill=black, inner sep=0.5pt, draw=black, fill=white] at (8.9,-3.8) {};
    \node[circle, fill=black, inner sep=0.5pt, draw=black, fill=white] at (8,2.5) {};
    \node[circle, fill=black, inner sep=0.5pt, draw=black, fill=black] at (8,1.75) {};
    \node[circle, fill=black, inner sep=0.5pt, draw=black, fill=white] at (8,0) {};
    \node[circle, fill=black, inner sep=0.5pt, draw=black, fill=black] at (8,-0.75) {};
    \node[circle, fill=black, inner sep=0.5pt, draw=black, fill=white] at (8,-3.2) {};
    \node[circle, fill=black, inner sep=0.5pt, draw=black, fill=black] at (8,-3.95) {};
    \draw[black] (5,-1.5) -- (6,-1.5) node[circle, draw=black, fill=black, inner sep=0.5pt, pos=1]{} node[circle, draw=black, fill=white, inner sep=0.5pt, pos=0]{};
    \node[circle, fill=black, inner sep=0.5pt, draw=black, fill=white] at (7,-1.5) {};
    \draw[black] (5,1.25) -- (6,1.25) node[circle, draw=black, fill=black, inner sep=0.5pt, pos=1]{} node[circle, draw=black, fill=white, inner sep=0.5pt, pos=0]{};
    \node[circle, fill=black, inner sep=0.5pt, draw=black, fill=white] at (7,1.145) {};
    \draw[black] (5,0.25) -- (6,0.25) node[circle, draw=black, fill=black, inner sep=0.5pt, pos=1]{} node[circle, draw=black, fill=white, inner sep=0.5pt, pos=0]{};
    \node[circle, fill=black, inner sep=0.5pt, draw=black, fill=white] at (7,0.6) {};
    \draw[black] (5,2.25) -- (6,2.25) node[circle, draw=black, fill=black, inner sep=0.5pt, pos=1]{} node[circle, draw=black, fill=white, inner sep=0.5pt, pos=0]{};
    \node[circle, fill=black, inner sep=0.5pt, draw=black, fill=white] at (7,1.65) {};
    \draw[black] (5,-3.4) -- (6,-3.4) node[circle, draw=black, fill=black, inner sep=0.5pt, pos=1]{} node[circle, draw=black, fill=white, inner sep=0.5pt, pos=0]{};
    \node[circle, fill=black, inner sep=0.5pt, draw=black, fill=white] at (7,-4.05) {};
    \draw[black] (5,-4.2) -- (6,-4.2) node[circle, draw=black, fill=black, inner sep=0.5pt, pos=1]{} node[circle, draw=black, fill=white, inner sep=0.5pt, pos=0]{};
    \node[circle, fill=black, inner sep=0.5pt, draw=black, fill=white] at (7,-4.45) {};
    \draw[black] (5,-5.2) -- (6,-5.2) node[circle, draw=black, fill=black, inner sep=0.5pt, pos=1]{} node[circle, draw=black, fill=white, inner sep=0.5pt, pos=0]{};
    \node[circle, fill=black, inner sep=0.5pt, draw=black, fill=white] at (7,-4.95) {};
    \draw[black] (5,-6) -- (6,-6) node[circle, draw=black, fill=black, inner sep=0.5pt, pos=1]{} node[circle, draw=black, fill=white, inner sep=0.5pt, pos=0]{};
    \node[circle, fill=black, inner sep=0.5pt, draw=black, fill=white] at (7,-5.35) {};
    \draw[black] (-0.5,3) -- (0.5,3) node[circle, draw=black, fill=white, inner sep=0.5pt, pos=0]{} node[circle, draw=black, fill=black, inner sep=0.5pt, pos=1]{};
    \node[draw, circle, fill=black, inner sep=0.5pt] at (-0.2,3) {};
    \node[draw, circle, fill=white, inner sep=0.5pt] at (0.2,3) {};
    \draw[black] (-0.5,2) -- (0.5,2) node[circle, draw=black, fill=white, inner sep=0.5pt, pos=0]{} node[circle, draw=black, fill=black, inner sep=0.5pt, pos=1]{};
    \node[draw, circle, fill=black, inner sep=0.5pt] at (-0.2,2) {};
    \node[draw, circle, fill=white, inner sep=0.5pt] at (0.2,2) {};
    \draw[black] (-0.5,1) -- (0.5,1) node[circle, draw=black, fill=white, inner sep=0.5pt, pos=0]{} node[circle, draw=black, fill=black, inner sep=0.5pt, pos=1]{};
    \node[draw, circle, fill=black, inner sep=0.5pt] at (-0.2,1) {};
    \node[draw, circle, fill=white, inner sep=0.5pt] at (0.2,1) {};
    \draw[black] (1,0.5) -- (1, -0.5);
    \draw[black] (0.5,0) -- (1,0.5) node[circle, draw=black, fill=black, inner sep=0.5pt, label=right:{\scalebox{0.3}{$e_1^v$}}, pos=1]{};
    \draw[black] (0.5,0) -- (1,-0.5) node[circle, draw=black, fill=black, inner sep=0.5pt, label=right:{\scalebox{0.3}{$e_1^u$}}, pos=1]{};
    \draw[black] (-0.5,0) -- (0.5,0) node[circle, draw=black, fill=black, inner sep=0.5pt, pos=0]{} node[circle, draw=black, fill=white, inner sep=0.5pt, pos=1]{};
    \node[draw, circle, fill=white, inner sep=0.5pt] at (0,0) {};
    \draw[black] (-0.5,-2) -- (0.5,-2) node[circle, draw=black, fill=white, inner sep=0.5pt, pos=0]{} node[circle, draw=black, fill=black, inner sep=0.5pt, pos=1]{};
    \node[draw, circle, fill=black, inner sep=0.5pt] at (-0.2,-2) {};
    \node[draw, circle, fill=white, inner sep=0.5pt] at (0.2,-2) {};
    \draw[black] (1,-2.5) -- (1, -3.5);
    \draw[black] (0.5,-3) -- (1,-2.5) node[circle, draw=black, fill=black, inner sep=0.5pt, label=right:{\scalebox{0.3}{$e_2^w$}}, pos=1]{};
    \draw[black] (0.5,-3) -- (1,-3.5) node[circle, draw=black, fill=black, inner sep=0.5pt, label=right:{\scalebox{0.3}{$e_2^u$}}, pos=1]{};
    \draw[black] (-0.5,-3) -- (0.5,-3) node[circle, draw=black, fill=black, inner sep=0.5pt, pos=0]{} node[circle, draw=black, fill=white, inner sep=0.5pt, pos=1]{};
    \node[draw, circle, fill=white, inner sep=0.5pt] at (0,-3) {};
    \draw[black] (-6.5,0.75) -- (-6,0.75) node[circle, draw=black, fill=white, inner sep=0.5pt, pos=1, label=above:\scalebox{0.3}{$o^2$}]{} node[circle, draw=black, fill=white, inner sep=0.5pt, pos=0, label=above:\scalebox{0.3}{$p^2$}]{};
    \draw[black] (-6.5,-0.75) -- (-6, -0.75) node[circle, draw=black, fill=white, inner sep=0.5pt, pos=1, label=above:\scalebox{0.3}{$o^1$}]{} node[circle, draw=black, fill=white, inner sep=0.5pt, pos=0, label=above:\scalebox{0.3}{$p^1$}]{};
    \draw[black] (17,0.75) -- (18,0.75) node[circle, draw=black, fill=white, inner sep=0.5pt, pos=0, label=above:\scalebox{0.3}{$f^3$}]{} node[circle, draw=black, fill=white, inner sep=0.5pt, pos=1, label=above:\scalebox{0.3}{$g^3$}]{};
    \draw[black] (17,-0.75) -- (18,-0.75) node[circle, draw=black, fill=white, inner sep=0.5pt, pos=0, label=above:\scalebox{0.3}{$f^2$}]{} node[circle, draw=black, fill=white, inner sep=0.5pt, pos=1, label=above:\scalebox{0.3}{$g^2$}]{};
    \draw[black] (17,2) -- (18,2) node[circle, draw=black, fill=white, inner sep=0.5pt, pos=0, label=above:\scalebox{0.3}{$f^4$}]{} node[circle, draw=black, fill=white, inner sep=0.5pt, pos=1, label=above:\scalebox{0.3}{$g^4$}]{};
    \draw[black] (17,-2) -- (18,-2) node[circle, draw=black, fill=white, inner sep=0.5pt, pos=0, label=above:\scalebox{0.3}{$f^1$}]{} node[circle, draw=black, fill=white, inner sep=0.5pt, pos=1, label=above:\scalebox{0.3}{$g^1$}]{};
    \draw[black] (-4,-9) -- (-4,-10) node[circle, draw=black, fill=white, inner sep=0.5pt, pos=0, label=above:\scalebox{0.3}{$\textsc{Select}_j$}]{} node[circle, draw=black, fill=black, inner sep=0.5pt, pos=1, label=below:\scalebox{0.3}{$\textsc{Unselect}_j$}]{};
    \node[draw, circle, fill=white, inner sep=0.5pt] at (1,-7.1) {};
    \node[draw, circle, fill=white, inner sep=0.5pt] at (-1.35,-6.1) {};
    \node[draw, circle, fill=white, inner sep=0.5pt] at (0.7,-7.6) {};
    \node[circle, draw=black, fill=white, inner sep=1pt, label=above:\scalebox{0.5}{$h$}] at (15, 0) {};
    \node[circle, draw=black, fill=white, inner sep=1pt, label=above:\scalebox{0.5}{$q$}] at (-5, 0) {};
    \end{tikzpicture}
    \caption{\footnotesize Orange, pink, and green edges highlighting the different types of edges between an \textsc{MMO}-instance-selector of the composition in \Cref{thm:cross_composition-DS-bpw} and a subgraph $G_{H_j}$ for a $j \in [t]$ of the same. $H_j$ has three vertices $u$, $v$, and $w$ and two edges $e_1=uv$ and $e_2=uw$, and $r=3$. Additionally, $\sigma_j(e_1) = 3$ and $\sigma_j(e_2) = 1$. For clarity, not all edges inside $G_{H_j}$ are drawn, nor are all the pink edges depicted. Vertices in black are in $S$ and those in white are not. Dotted lines incident to the vertices in $c(Y)$ represent edges to the vertex $d_j$ which is not shown in the figure.}
    \label{fig:DS-pathwidth-composition}
\end{figure}
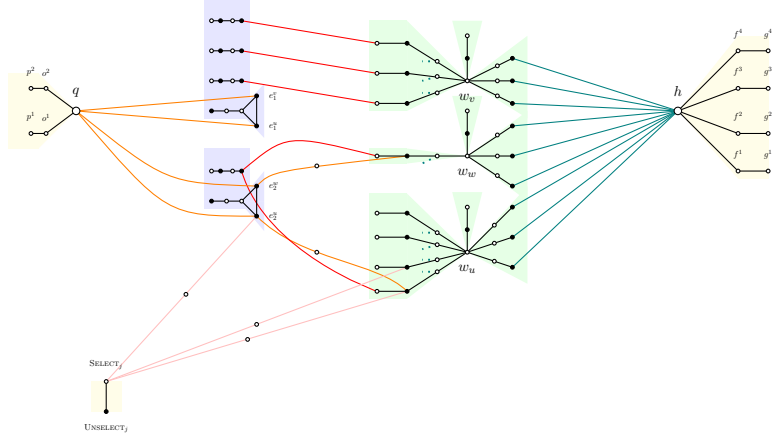
Given that all integers are given in unary, the construction of the graph $\hat{G}_t$ or its path decomposition (as described in the discussion for \Cref{cor:cor-bounded-pathwidth-G'-t}), and as a consequence the reduction take time polynomial in the size of the input instances. Additionally, by \Cref{cor:composition-space}, this composition is a pl-reduction. We claim that $(\hat{G}_t, \mathcal{P}_{\hat{G}_t}, S, \budget)$ is a yes-instance of \textsc{DS-D} if and only if for some integer $\mathfrak{j} \in [t]$, $(H_\mathfrak{j}, \mathcal{P}_{H_\mathfrak{j}}, \sigma_\mathfrak{j}, r_\mathfrak{j})$ is a yes-instance of \textsc{MMO}.
\\
\begin{claim}
If for some $\mathfrak{j} \in [t]$, $(H_\mathfrak{j}, \mathcal{P}_{H_\mathfrak{j}}, \sigma_\mathfrak{j}, r_\mathfrak{j})$ is a yes-instance of \textsc{MMO}, then $(\hat{G}_t, \mathcal{P}_{\hat{G}_t}, S, \budget)$ is a yes-instance of \textsc{DS-D}.
\end{claim}

\begin{claimproof}
Let $(H_\mathfrak{j}, \mathcal{P}_{H_\mathfrak{j}}, \sigma_\mathfrak{j}, r_\mathfrak{j})$ be a yes-instance of \textsc{MMO} and let $\lambda$ be a feasible orientation of $H_\mathfrak{j}$ such that for each $v \in V(H_\mathfrak{j})$, the total weight of the edges directed out of $v$ is at most $r$.
In $\hat{G}_t$, the vertices $f^1, \ldots, f^{\bm{\sigma}}, o^1, \ldots, o^m$ and their neighbors are non-dominated.
First, we slide the token on $\textsc{\footnotesize Unselect}_\mathfrak{j}$ to $\textsc{\footnotesize Select}_\mathfrak{j}$. 
Using $2m$ slides, we move for each edge $e \in E(H_\mathfrak{j})$ the token on $e^u$ (resp. $e^v$) if $\lambda(e) = (v, u)$ (resp. $\lambda(e) = (v, u)$), towards a token-free vertex in $o^1, \ldots, o^m$.
We additionally slide each token on a vertex $b^i_e$ for $i \in [\sigma_\mathfrak{j}(e)]$ to the vertex $z_e^{v(i)}$ (resp. $z_e^{u(i)}$), move the token on $y_e^{v(i)}$ (resp. $y_e^{u(i)}$) towards a token-free vertex in $c(X_v)$ (resp. $c(X_u)$) and consequently, move the token on the adjacent vertex in $X_v$ (resp. $X_u$) towards a token-free vertex in $f^1, \ldots, f^{\bm{\sigma}}$.
The total number of slides performed is $\budget$ and they achieve a configuration for the tokens that dominates all of $\hat{G}_t$.    
\end{claimproof}

\begin{claim}
If $(\hat{G}_t, \mathcal{P}_{\hat{G}_t}, S, \budget)$ is a yes-instance of \textsc{DS-D}, then there exists an integer $\mathfrak{j} \in [t]$, such that $(H_\mathfrak{j}, \mathcal{P}_{H_\mathfrak{j}}, \sigma_\mathfrak{j}, r_\mathfrak{j})$ is a yes-instance of \textsc{MMO}.
\end{claim}

\begin{claimproof}
In any solution that uses the minimal number of slides, the tokens on the vertices $d_j$ for each $j \in [t]$, and on the vertices in $c(X^+)$ do not need to be moved (as these
tokens must be replaced by others to dominate the vertices $d'_j$ for each $j \in [t]$, or the vertices in $X^+$, thus we can assume these tokens remain stationary).
In the same solution, we can similarly assume that a token on one of $\textsc{Unselect}_j$ and $\textsc{Select}_j$ for each $j \in [t]$ remains on either one of those vertices.
To dominate $g^1, \ldots, g^{\bm{\sigma}}, p^1, \ldots, p^m$, at least $2m + 2\bm{\sigma}$ slides are needed to get tokens from one or more of the vertices in $X$ onto the vertices $f^1, \ldots, f^{\bm{\sigma}}$, and from one of more of the vertices of the form $e^u$ for an edge $e \in E(H_j)$ incident to a vertex $u \in V(H_j)$ for an integer $j \in [t]$, onto the vertices $o^1, \ldots, o^m$.
If a token is moved out of a subgraph $G_{H_j}$ (for an integer $j \in [t]$) of $\hat{G}_t$, which is bound to happen to get tokens onto the vertices $f^1, \ldots, f^{\bm{\sigma}}, o^1, \ldots, o^m$, at least one slide is needed to dominate the vertex between now token-free vertices in $G_{H_j}$ and $\textsc{\footnotesize Select}_j$ and exactly one slide can only be achieved by sliding the token on $\textsc{\footnotesize Unselect}_j$ to $\textsc{\footnotesize Select}_j$ (since otherwise a token has to move from one of the vertices of a subgraph $G_{H_{j'}}$ for $j' \neq j \in [t]$ into $G_{H_j}$, and this requires more than one slide).

W.l.o.g. assume a token on a vertex, denoted $x^i_v$, in $X$, for a vertex $v \in V(H_\mathfrak{j})$, and integers $\mathfrak{j} \in [t]$ and $i \in [r]$, is moved to one of the vertices $f^1, \ldots, f^{\bm{\sigma}}$, then at least $3$ slides are needed to move a token into either $x^i_v$ or $c(x^i_v)$ (since the tokens on $c(X^+)$ are assumed to be stationary and a token moving from any other vertex, except $x^i_v$, in $X$ into $x^i_v$ can replace the token on $x^i_v$ in moving into one of the vertices  $f^1, \ldots, f^{\bm{\sigma}}$).
In a solution that uses the minimal number of slides, $3$ slides can only be achieved by moving a token on a vertex, denoted $y_e^{v(i_1)}$, in $Y^v$, for some edge $e \in E(H_\mathfrak{j})$ adjacent to $v$ and some integer $i_1 \in [\sigma_\mathfrak{j}(e)]$, to $c(x^i_v)$.
Additionally, $4$ slides can only be achieved by moving a token on the same vertex to $x^i_v$.
Since in a solution that uses the minimal number of slides a token on one of $\textsc{Unselect}_\mathfrak{j}$ and $\textsc{Select}_\mathfrak{j}$ is assumed to remain on either of those vertices,
and a token in $B$ can slide at most one slide to a vertex in $Z$, if a token on $y_e^{v(i_1)}$ is moved to a vertex in $c(X)$ (or $X$), either a token has to move to the vertex $z_e^{v(i_1)}$ (while a token has to be on the vertex~$e^v$), or a token has to slide from the vertex $e^v$ to $c'(y_e^{v(i_1)})$ or to $y_e^{v(i_1)}$ itself.
Moving the token on $e^v$ to $y_e^{v(i_1)}$ requires two slides.

Given the budget and the fact that $\bm{\sigma}$ tokens in any solution must move from $X$ onto the vertices $f^1, \ldots, f^{\bm{\sigma}}$, tokens must move from distinct vertices in $X$ onto the vertices $f^1, \ldots, f^{\bm{\sigma}}$ and from distinct vertices of the form $e^u$ for an edge $e \in E(H_j)$ incident to a vertex $u \in V(H_j)$ for an integer $j \in [t]$, onto the vertices $o^1, \ldots, o^m$. 
Additionally, given the budget, tokens in the same solution must move onto $f^1, \ldots, f^{\bm{\sigma}}$ from only the vertices in $X_{\mathfrak{j}}$ and onto $o^1, \ldots, o^m$ from only the vertices of the form $e_1^{u}$, for an edge $uw=e_1 \in E(H_\mathfrak{j})$ (note that one token sliding to $\textsc{\footnotesize Select}_{\mathfrak{j}}$ from $\textsc{\footnotesize Unselect}_{\mathfrak{j}}$ will dominate all vertices on the paths between $\textsc{\footnotesize Select}_{\mathfrak{j}}$ and $V(G_{H_\mathfrak{j}})$).
In the same solution, if a token moves from the vertex $e_1^{u}$ onto one of the vertices $o^1, \ldots, o^m$, the token on $e_1^{w}$ remains stationary as the budget does not allow for another token to move into either one of the vertices $e_1^{u}$ and $e_1^{w}$ (to additionally dominate $b_{e_1}^{\sigma_\mathfrak{j}(e_1)+1}$).
To fill all of $o^1, \ldots, o^m$ with tokens, exactly one token must move from $G_{e_2}^{sel}$ onto the vertices $o^1, \ldots, o^m$ for each edge $e_2 \in E(H_\mathfrak{j})$.
The latter implies that the token on~$e^v$ does not move to $c'(y_e^{v(i_1)})$ and given the budget that the token on $b_e^{i_1}$ slides to $z_e^{v(i_1)}$.

W.l.o.g. assume that the token on $e^v$ does not move to one of $o^1, \ldots, o^m$, then at most the $\sigma(e)$ tokens on the vertices of $Y_e^v$ can be sent to $c(X_v)$.
This implies that for each edge in $H$, at most its weight in tokens can move to $c(X)$ from and to exactly one of the vertex gadgets corresponding to the vertices incident to that edge in $H$.
Given that $\bm{\sigma}$ tokens are needed on the vertices of $c(X)$, it must be the case that for each edge, all its weight in tokens must move to $c(X)$.
This gives a feasible orientation for $(H_\mathfrak{j}, \mathcal{P}_{H_\mathfrak{j}}, \sigma_\mathfrak{j}, r_\mathfrak{j})$, since for each $v \in V(H_\mathfrak{j})$, we have at most $r$ vertices in $c(X_v)$.  
\end{claimproof}
This concludes the proof of the theorem.
\end{proof}
Next, we consider the DS-D problem with respect to the parameter \emph{fvs}. 
\begin{theorem}
\label{thm:DS-D-fvs}
The \textsc{DS-D} problem is $\W[1]$-hard for the parameter \emph{fvs} of the input graph. 
\end{theorem}
We present a parameterized reduction from the \mcc~problem, which is known to be \wone-hard~\cite{cygan2015parameterized} with respect to the solution size $\kappa$. 
In the \mcc~problem, we are given a graph $G$, and an integer $\kappa$, where $V(G)$ is partitioned into $\kappa$ independent sets $\{V_1,\ldots V_{\kappa}\}$. 
The goal is to decide whether there exists a clique of size $\kappa$.

Let $(G, \kappa)$ be an instance of the \mcc~problem. 
The edge set $E(G)$ can be partitioned into $\kctwo$ sets $\{E_{i,j} = \{uv : u \in V_i, v \in V_j\} \mid \iljk\}$. 
Without loss of generality, we assume that for each $i \in [\kappa]$, $|V_i| = n$. 
Otherwise, add isolated vertices in respective subsets. 
We usually use $n$ to denote the number of vertices in the input graph. 
However, we use $n$ to denote the number of vertices in each color class. 
For each $i \in [\kappa]$, let $V_i = \{u_{i,\ell} \mid \ell \in [n]\}$. 

For an instance $(G, \kappa)$ of the \mcc~problem, the reduction outputs an instance $(H, \varS, \budget)$ of the DS-D problem. 
The graph $H$ has an induced subgraph $H_i$ for each $i \in [\kappa]$ and an induced subgraph $H_{i,j}$ for each $\iljk$. 
We refer to these induced subgraphs as {\em edge-blocks} and {\em vertex-blocks}, respectively. 
Finally, the vertex-blocks and edge-blocks are connected by connectors. 

\smallskip
\noindent{\bf Vertex-block.} For each $i \in [\kappa]$, we construct a vertex-block $H_i$ as follows.
We start by adding a vertex $t_i$. 
For each $x \in [n]$, we add a star-tree rooted at $p_{i,x}$ with $n(\kappa-1)$ leaves $\{q_{i, x}^{j,1}, \ldots, q_{i, x}^{j,n} \mid j \not= i\}$. 
Each vertex $p_{i, x}$ is connected with $t_i$ by an edge. 
For each $x \in [n]$ and $j \not= i \in [\kappa]$, let $Q_{i, x}^j = \{q_{i, x}^{j,\ell} \mid \ell \in [n]\}$. 
For each $x \in [n]$, $\displaystyle Q_{i,x} = \bigcup_{j\not=i \in [\kappa]} Q_{i,x}^j$. 
Further, for each $i \in [\kappa]$, let $\displaystyle Q_i = \bigcup_{x \in [n]} Q_{i,x}$. 

\medskip
\noindent{\bf Edge-block.} For each $\iljk$, we construct a edge-block $H_{i,j}$ as follows.
We start by adding a vertex $t_{i,j}$. 
For each edge $e \in E_{i,j}$,%$e=u_{i,x}u_{j,y} \in E_{i,j}$ for some $x,y \in [n]$, 
we add a star-tree rooted at $p_e$ with $2n$ leaves $q_e^1, \ldots, q_e^{2n}$. 
Each vertex $p_e$ is connected with $t_i$ by an edge. 
For each $e \in E_{i,j}$, let $Q_e = \{q_e^\ell \mid \ell \in [2n]\}$. 
Further, for each $\iljk$, $\displaystyle Q_{i,j} = \bigcup_{e \in E_{i,j}} Q_e$.

\smallskip
\noindent{\bf Connector.} For each $\iljk$ and for each $l \in \{i, j\}$, we construct a connector $\conn^l_{i,j}$ as follows. 
Let $A_{i,j}^l = \{a_{i,j}^{l,1}, \ldots, a_{i,j}^{l, n}\}$, $B_{i,j}^l = \{b_{i,j}^{l,1}, \ldots, b_{i,j}^{l, n}\}$, $C_{i,j}^l = \{c_{i,j}^{l,1}, \ldots, c_{i,j}^{l, n}\}$ and $D_{i,j}^l = \{d_{i,j}^{l,1}, \ldots, d_{i,j}^{l, n}\}$. 
We add $4n+2$ vertices $(\{s_{i,j}^l, r_{i,j}^l\} \cup A_{i,j}^l \cup B_{i,j}^l \cup C_{i,j}^l \cup D_{i,j}^l)$. 
For each $x \in [n]$, we add the edges $s_{i,j}^l a_{i,j}^{l,x}$, $a_{i,j}^{l,x}b_{i,j}^{l,x}$, $r_{i,j}^l c_{i,j}^{l,x}$ and $c_{i,j}^{l,x}d_{i,j}^{l,x}$. 

For each $\iljk$ and $l \in \{i,j\}$, the vertex-block $H_l$ is connected with the connector $\conn_{i,j}^l$ as follows. 
Let $l' \not= l \in \{i,j\}$. 
For each $x, z \in [n]$, 
\begin{itemize}
    \item if $z \leq x$, then add an edge $q_{l,x}^{l', z} s_{i,j}^l$, and
    \item if $z > x$, then add an edge $q_{l,x}^{l',z} r_{i,j}^l$. 
\end{itemize}

For each $\iljk$, the edge-block $H_{i,j}$ is connected with the connectors $\conn_{i,j}^i$ and $\conn_{i,j}^j$ as follows. 
For each $e = u_{i,x} u_{j,y} \in E_{i,j}$ for some $x,y \in [n]$, and for each $z,w \in [n]$, 
\begin{itemize}
    \item if $z \leq x$, then add an edge $q_e^z r_{i,j}^i$, 
    \item if $z > x$, then add an edge $q_e^z s_{i,j}^i$, 
    \item if $w \leq y$, then add an edge $q_w^{n+w} r_{i,j}^j$, and
    \item if $w > y$, then add an edge $q_{l,x}^{n+w} s_{i,j}^j$. 
\end{itemize}
For a pair $(i,j)$ with $\iljk$, an illustration of a connector $\conn_{i,j}^i$ that connects $V_{H_{i}}$ and~$H_{i,j}$ is given in Figure~\ref{fig:DS-D-fvs-illustration}. 
This completes the construction of the graph $H$. 
Further, we set $\budget =(8n+1)\kctwo+\kappa$, and
we define the initial configuration $\varS$ as follows:
\[\varS = \bigcup_{i \in [\kappa], {x \in [n]}} Q_{i,x} \cup \bigcup_{e \in E} Q_e \cup \{t_i \mid i \in [\kappa]\}  \cup \{t_{i,j} \mid \iljk \}.\]

\begin{figure}
    \tikzset{decorate sep/.style 2 args=
{decorate,decoration={shape backgrounds,shape=circle,shape size=#1,shape sep=#2}}}
\centering
    \begin{tikzpicture}
    \coordinate (ti) at (0,0);
    \coordinate (pix) at (1,0);
    \coordinate (qix1) at (2,1.5);
    \coordinate (qixx) at (2,0.5);
    \coordinate (qixx1) at (2,-0.5);
    \coordinate (qixn) at (2,-1.5);
    \coordinate (siij) at (6,2.5);
    \coordinate (riij) at (6, -2.5);
    \coordinate (aiji1) at (5, 1.5);
    \coordinate (aijin) at (7, 1.5);
    \coordinate (biji1) at (5, 0.5);
    \coordinate (bijin) at (7, 0.5);
    \coordinate (ciji1) at (5, -0.5);
    \coordinate (cijin) at (7, -0.5);
    \coordinate (diji1) at (5, -1.5);
    \coordinate (dijin) at (7, -1.5);
    \coordinate (tij) at (12, 0);
    \coordinate (pe) at (11, 0);
    \coordinate (qez1) at (10,1.5);
    \coordinate (qezz) at (10, 0.5);
    \coordinate (qezz1) at (10, -0.5);
    \coordinate (qezn) at (10, -1.5);
    \coordinate (qew1) at (10, -2.5);
    \coordinate (qewn) at (10, -3.5);
    \draw[black, thick] (ti) -- (pix);
    \draw[black, thick] (pix) -- (qix1);
    \draw[black, thick] (pix) -- (qixx);
    \draw[black, thick] (pix) -- (qixx1);
    \draw[black, thick] (pix) -- (qixn);
    \draw[black, thick] (tij) -- (pe);
    \draw[black, thick] (pe) -- (qez1);
    \draw[black, thick] (pe) -- (qezz);
    \draw[black, thick] (pe) -- (qezz1);
    \draw[black, thick] (pe) -- (qezn);
    \draw[black, thick] (pe) -- (qew1);
    \draw[black, thick] (pe) -- (qewn);
    \draw[black, thick] (siij) -- (aiji1);
    \draw[black, thick] (siij) -- (aijin);
    \draw[black, thick] (riij) -- (diji1);
    \draw[black, thick] (riij) -- (dijin);
    \draw[black, thick] (aiji1) -- (biji1);
    \draw[black, thick] (aijin) -- (bijin);
    \draw[black, thick] (ciji1) -- (diji1);
    \draw[black, thick] (cijin) -- (dijin);
    \draw[black, thick] (siij) to[bend right=30] (qix1);
    \draw[black, thick] (siij) to[bend right=30] (qixx);
    \draw[black, thick] (riij) to[bend left=30] (qixx1);
    \draw[black, thick] (riij) to[bend left=30] (qixn);
    \draw[black, thick] (siij) to[bend left=30] (qez1);
    \draw[black, thick] (siij) to[bend left=30] (qezz);
    \draw[black, thick] (riij) to[bend right=30] (qezz1);
    \draw[black, thick] (riij) to[bend right=30] (qezn);
    \fill[white, draw=black, thick] (ti) circle (0.1cm) node[left,blue] {$t_i$};
    \fill[white, draw=black, thick] (tij) circle (0.1cm) node[right,blue] {$t_{i,j}$};
    \fill[white, draw=black, thick] (siij) circle (0.1cm) node[above,blue] {$s_{i,j}^i$};
    \fill[white, draw=black, thick] (riij) circle (0.1cm) node[below,blue] {$r_{i,j}^i$};
    \fill[white, draw=black, thick] (pix) circle (0.1cm) node[below,blue] {$p_{i,x}$};
    \fill[white, draw=black, thick] (qix1) circle (0.1cm) node[above,blue] {$q_{i,x}^{j,1}$};
    \fill[white, draw=black, thick] (qixx) circle (0.1cm) node[below,blue] {$q_{i,x}^{j,x}$};
    \fill[white, draw=black, thick] (qixx1) circle (0.1cm);
    \fill[white, draw=black, thick] (qixn) circle (0.1cm) node[below,blue] {$q_{i,x}^{j,n}$};
    \fill[white, draw=black, thick] (pe) circle (0.1cm) node[above,blue] {$p_e$};    
    \fill[white, draw=black, thick] (qez1) circle (0.1cm) node[above,blue] {$q_e^1$};
    \fill[white, draw=black, thick] (qezz) circle (0.1cm) node[below,blue] {$q_e^{n-z}$};
    \fill[white, draw=black, thick] (qezz1) circle (0.1cm);
    \fill[white, draw=black, thick] (qezn) circle (0.1cm) node[below,blue] {$q_e^n$};
    \fill[white, draw=black, thick] (qew1) circle (0.1cm); 
    \fill[white, draw=black, thick] (qewn) circle (0.1cm); 
    \fill[white, draw=black, thick] (aiji1) circle (0.1cm);
    \fill[white, draw=black, thick] (aijin) circle (0.1cm);
    \fill[white, draw=black, thick] (biji1) circle (0.1cm);
    \fill[white, draw=black, thick] (bijin) circle (0.1cm);
    \fill[white, draw=black, thick] (ciji1) circle (0.1cm);
    \fill[white, draw=black, thick] (cijin) circle (0.1cm);
    \fill[white, draw=black, thick] (diji1) circle (0.1cm);
    \fill[white, draw=black, thick] (dijin) circle (0.1cm);
    \fill[red] (ti) circle (0.05cm);
    \fill[red] (tij) circle (0.05cm);
    \fill[red] (qix1) circle (0.05cm);
    \fill[red] (qixx) circle (0.05cm);
    \fill[red] (qixx1) circle (0.05cm);
    \fill[red] (qixn) circle (0.05cm);
    \fill[red] (qez1) circle (0.05cm);
    \fill[red] (qezz) circle (0.05cm);
    \fill[red] (qezz1) circle (0.05cm);
    \fill[red] (qezn) circle (0.05cm);
    \fill[red] (qew1) circle (0.05cm);
    \fill[red] (qewn) circle (0.05cm);
    \draw[black,rounded corners] (4.75, 1.25) rectangle (7.25, 1.75) node[midway] {$A_{i,j}^i$};
    \draw[black,rounded corners] (4.75, 0.25) rectangle (7.25, 0.75) node[midway] {$B_{i,j}^i$};
    \draw[black,rounded corners] (4.75, -0.75) rectangle (7.25,-0.25) node[midway] {$C_{i,j}^i$};
    \draw[black,rounded corners] (4.75, -1.75) rectangle (7.25, -1.25) node[midway] {$D_{i,j}^i$};
    \draw[gray, thick, dashed, rounded corners=15pt] (1.5, -2.5) rectangle (2.5, 2.5);
    \draw[gray, thick, dashed, rounded corners=15pt] (0.6, -3.15) rectangle (2.75, 3.15);
    \draw[gray, thick, dashed, rounded corners=15pt] (-0.5, -3.5) rectangle (3.5, 3.5) node[midway, black, above=3.5cm] {$H_i$};
    \draw[gray, thick, dashed, rounded corners=15pt] (9.5, 2.5) rectangle (11.25, -4);
    \draw[gray, thick, dashed, rounded corners=15pt] (8.5, 3.5) rectangle (12.5, -5) node[midway, black, above=4.25cm] {$H_{i,j}$};
    \draw[gray, thick, dashed, rounded corners=15pt] (4.5, -3.5) rectangle (7.5, 3.5) node[black, midway, above=3.5cm] {$\conn_{i,j}^i$};
    \draw[decorate sep={0.5mm}{1.75mm},fill] (2,1.25) -- (2,0.7);
    \draw[decorate sep={0.5mm}{1.75mm},fill] (2,-0.75) -- (2,-1.3);
    \draw[decorate sep={0.5mm}{1.75mm},fill] (10,1.25) -- (10,0.7);
    \draw[decorate sep={0.5mm}{1.75mm},fill] (10,-0.75) -- (10,-1.3);
    \draw[decorate sep={0.5mm}{1.75mm},fill] (10,-2.75) -- (10,-3.3);
    \draw[decorate sep={0.5mm}{1.75mm},fill] (2,2.7) -- (2,3.1);
    \draw[decorate sep={0.5mm}{1.75mm},fill] (2,-2.7) -- (2,-3.1);
    \end{tikzpicture}
    \caption{\footnotesize An illustration of the reduction of Theorem~\ref{thm:DS-D-fvs}. For $\iljk$, the vertex-block $H_i$ and the edge-block $H_{i,j}$ are connected to the connector $\conn_{i,j}^i$. The initial configuration is denoted by vertices with red circle. }
    \label{fig:DS-D-fvs-illustration}
\end{figure}
\begin{lemma}
\label{lem:DS-D-fvs-bound}
    The \emph{fvs} of the graph $H$ is at most $4\binom{\kappa}{2}$.
\end{lemma}
\begin{proof}
    Let $F = \{s_{i,j}^i, r_{i,j}^i, s_{i,j}^j, r_{i,j}^j \mid \iljk\}$. Removal of $F$ from $H$ results a forest. 
    Therefore, the \emph{fvs} of $H$ is at most $|F| = 4\binom{\kappa}{2}$. 
\end{proof}
\begin{lemma}
\label{lem:DS-D-fvs-forward}
    If $(G, \kappa)$ is a yes-instance of the \mcc~problem, then $(H, \varS, \budget)$ is a yes-instance of the \textsc{DS-D} problem. 
\end{lemma}
\begin{proof}
    Let $C = \subseteq V(G)$ be a $\kappa$-clique in $G$. 
    For each $i \in [\kappa]$, let $u_{i,x_i}$ be the vertex in $C \cap V_i$ for some $x_i \in [n]$. 
    For each $\iljk$, let $e_{i,j} = u_{i,x_i}u_{j,x_j}$.
    For each $i \in [\kappa]$, we slide the token on~$t_i$ to $p_{i,x_i}$. 
    Then, for each $j \not= i \in [\kappa]$, we slide $x_i$-tokens in $Q_{i,x_i}^j$ towards $s^i_{i,j}$ and $n-x_i$-tokens in~$Q_{i,x_i}^j$ towards $r_{i,j}^i$. 
    For each $\iljk$, we slide the token on $t_{i,j}$ to $p_{e_{i,j}}$. 
    
    \pagebreak
    Then, we slide 
    \begin{itemize}
        \item $n-x_i$ tokens in $Q_{e_{i,j}}$ to $s_{i,j}^i$,
        \item $x_i$ tokens in $Q_{e_{i,j}}$ to $r_{i,j}^i$,
        \item $n-x_j$ tokens in $Q_{e_{i,j}}$ to $s_{i,j}^j$, and
        \item $x_j$ tokens in $Q_{e_{i,j}}$ to $r_{i,j}^j$. 
    \end{itemize}
    
    For each $\iljk$, and for each $l \in {i,j}$, $s_{i,j}^l$ receives $x_l$-tokens from $H_l$ and $n-x_l$-tokens from~$H_{i,j}$.
    Similarly, $r_{i,j}^l$ receives $n-x_l$-tokens from $H_l$ and $x_l$-tokens from $H_{i,j}$. 
    Further, we push the $n$-tokens received by $s_{i,j}^l$ to $A_{i,j}^l$ and $n$-tokens received by $r_{i,j}^l$ to $D_{i,j}^l$. 
    The above token slides result the following. 
    For each $i \in [\kappa]$,
    \begin{itemize}
        \item $t_i$ is dominated by $p_{i,x_i}$,
        \item for each $j \not= i \in [\kappa]$, the vertices in $Q_{i, x_i}^{j}$ are dominated by $p_{i, x_i}$, and 
        \item for each $\ell \not= x_i \in [n]$, $p_{i, \ell}$ is dominated by $Q_{i,\ell}^j$ for any $j \not= i$. 
    \end{itemize}
    for each $\iljk$, 
    \begin{itemize}
        \item $t_{i,j}$ is dominated by $p_{e_{i,j}}$,
        \item the vertices in $Q_{e_{i,j}}$ are dominated by $p_{e_{i,j}}$, and 
        \item for each $e \not= e_{i,j} \in E_{i,j}$, $p_{e}$ is dominated by the vertices in $Q_{e}$.
    \end{itemize}
    Finally, let $S' \subseteq V(H)$ be the solution obtained from the above token sliding steps. 
    More precisely,
    $$S' = \bigcup_{i \in [\kappa]}\big\{\{p_{i, x_i}\} \cup (Q_{i} \setminus Q_{i, x_i})\big\} \cup \bigcup_{\iljk} \big\{\{p_{e_{i,j}}\} \cup  (Q_{i,j} \setminus Q_{e_{i,j}})\big\} \cup \bigcup_{\iljk, l \in \{i,j\}} (A_{i,j}^l \cup D_{i,j}^l).$$

    It is clear that the set $S'$ is a dominating set in $H$. 
    Next we count the number of token steps used to obtain $S'$ from $\varS$. 
    In each vertex-block, we spend $(\kappa-1)n+1$ steps to push tokens towards the connectors. 
    Similarly, at each edge-block, we spend $(2n+1)$ steps. 
    At each connectors, we spend $2n$ steps. 
    Therefore, we spend $\kappa\cdot\big((\kappa-1)n+1\big) + \binom{\kappa}{2}\cdot(2n+1) + 2\binom{\kappa}{2} \cdot 2n = (8n+1)\binom{\kappa}{2} + \kappa = \budget$. 
    Hence, $(H, \varS, \budget)$ is a yes-instance of DS-D problem. 
\end{proof}

\begin{lemma}
\label{lem:DS-D-fvs-reverse}
    If $(H, \varS, \budget)$ is a yes-instance of the DS-D problem, then $(G, \kappa)$ is a yes-instance of the \mcc~problem. 
\end{lemma}
\begin{proof}
    Let $S^*$ be a feasible solution for the instance $(H, \varS, \budget)$ of the DS-D problem. 
    At each connector~$\conn_{i,j}^l$ for $\iljk$ and $l \in \{i,j\}$, at least $2n$ tokens need to be slid from either vertex-blocks or edge blocks.
    It is clear that every token must move at least 2 steps to reach the sets $A_{*,*}^*$ and $D_{*,*}^*$ in order to dominate the vertices in the set $B_{*,*}^*$ and $C_{*,*}^*$, respectively. 
    This saturates a budget of $4n\cdot 2\kctwo = 8n\kctwo$. 
    Therefore, we left with exactly $\kappa+\kctwo$ budget to adjust the tokens on the vertex blocks and edge blocks. 
    For any $\iljk$, let $q_{i,x}^{j,z}$ for some integers $x,z \in [n]$ be a vertex that looses the token where the token is moved to some vertex in a connector. 
    Since none of its neighbors have token, we need to slide a token to the vertex or to it's neighbors. 
    This cost at least one token step. 
    By construction of the vertex-block $H_i$, by sliding a token to the vertex $p_{i,x}$ for some $x \in [n]$, one can release at most $n(\kappa-1)$ tokens from the neighboring set~$Q_{i,x}$. 
    Similarly, on an edge-block $H_{i,j}$ for some $\iljk$, by sliding a token to the vertex $p_e$ for some $e \in E_{i,j}$, one can release at most $2n$ tokens from the neighboring set $Q_e$. 
    This implies that by sliding at most $\kappa$ tokens on the vertex-blocks, one can release at most $\kappa\cdot n(\kappa-1) = 2n\kctwo$ token from the vertex-blocks. 
    Similarly, by sliding at most $\kctwo$ tokens on the edge-blocks, one can release at most $2n\kctwo$ tokens from the edge-blocks. 
    Therefore, we need to slide exactly one token in each vertex-block and each edge-block. 
    
    For each $i\in [\kappa]$, let $p_{i,x_i}$ for some $x_i \in [n]$ be the vertex in $H_i$ that gets token in $S^*$ and releases all the tokens in $Q_{i,x_i}$. 
    Similarly, for each $\iljk$, let $p_e$ for some $e=u_{i,z_i}u_{i,z_j} \in E_{i,j}$ with $z_i,z_j \in [n]$ be the vertex in $H_{i,j}$ that gets token in $S^*$ and releases all tokens in $Q_{e}$. 
    Consider the connector $\conn_{i,j}^i$. 
    The set $Q_{i,x_i}^j$ pushes $x_i$ tokens to $s_{i,j}^i$ and $n-x_i$ tokens to $r_{i,j}^i$. 
    The set $Q_e$ pushes $z_i$ tokens to $r_{i,j}^i$ and $n-z_i$ tokens to $s_{i,j}^i$. 
    The number of tokens passed through $s_{i,j}^i$ to $A_{i,j}^i$ is $x_i + (n - z_i)$. 
    Since $A_{i,j}^i$ need $n$ tokens, it is mandatory that $x_i=z_i$. 
    This equality should hold for every $i$. 
    Therefore, for each $\iljk$, there exist an edge $u_{i,x_i}u_{j,x_j}$. 
    Hence $(G,\kappa)$ is an yes-instance of the \mcc~problem.     
\end{proof}
The proofs of Lemmas~\ref{lem:DS-D-fvs-bound}, \ref{lem:DS-D-fvs-forward} and \ref{lem:DS-D-fvs-reverse} complete the proof of Theorem~\ref{thm:DS-D-fvs}. 

\section{Shortest Path Discovery}
\label{sec:sp}
Finally, we show that \textsc{SP-D} does not admit a polynomial kernel unless $\NP \subseteq \coNP/\poly$. The employed or-cross-composition is similar to the construction in the hardness proof of \textsc{SP-D} presented in~\cite[Theorem 4.2]{grobler2023solution}. 
We denote an instance of \textsc{SP-D} by $(G, S, b, a, b)$ to emphasize that the solution must be a shortest path between the vertices $a$ and $b$ in $V(G)$ (for consistency with the previous sections we do not speak of $s$-$t$-connectivity but use $t$ for the number of instances in the cross composition). 

\begin{theorem}\label{thm:SP-kb-general}
    There exists an or-cross-composition from \textsc{Hamiltonian Path} into \textsc{SP-D}, parameterized by $k + \budget$. Consequently, \textsc{SP-D} does not admit a polynomial kernel with respect to $k + \budget$, unless $\NP \subseteq \cp$.
\end{theorem}
\begin{proof}
    Let $\Rmc$ be the polynomial equivalence relation whose equivalence classes are defined by graphs with the same number of vertices, that is, two graphs $G$ and $H$ are equivalent with respect to $\Rmc$ if and only if $|V(G)| = |V(H)|$.
    Let $G_1, \dots, G_t$ be a sequence of instances of \textsc{Hamiltonian Path}, where every $G_j$, $j \leq t$, is an $n$-vertex graph, say $V(G_j) = \{1, \dots, n\}$.
    For every $G_j$ we create a new graph $H_j$ that consists of $n^2$ vertices, say $(x,y)$ for $x, y \leq n$. For every $x < n$ and $y, y' \leq n$, we connect the vertex $(x,y)$ with the vertex $(x+1, y')$ if and only if $yy' \in E(G_j)$.
    
    We construct the following graph $G$. First, $G$ consists of a disjoint union of all $H_j$, $j \leq t$. Furthermore, we add two fresh vertices~$a$ and $b$, as well as $n$ fresh vertices (we simply call them $\{1, \dots, n\}$, too) to the vertex set of $G$.
    For every $y \leq n$ we connect the vertex $a$ with every vertex $(1, y)$ in every $H_j$. Also, for every $y \leq n$ we connect every vertex $(n, y)$ in every $H_j$ with $b$. Finally, for every $x \leq n$ we connect the vertex $x$ in $G$ with every vertex $(x, y)$ in every $H_j$ for all $y \leq n$ with a path of length $n$. This finishes the construction of $G$.
    Let $S = \{a,b,1, \dots, n\}$, hence $k = n + 2$ and $\budget = n^2$.
    Observe that the size of every $G_j$ is (given a suitable encoding) bounded by $n^2$. Hence, the parameter $k + b = n^2 + n + 2$ is bounded by a polynomial in $\max_{j=1}^t |G_j| + \log t$.
    We claim that $(G, S, b, a, b)$ is a yes-instance of \textsc{SP-D} if and only if at least one $G_j$ admits a Hamiltonian path.

    We begin with the backward direction, that is let $G_j$ be a Hamiltonian graph with Hamiltonian path $i_1 \dots i_n$. 
    Then we can move the token on vertex $x$ in $G$ 
    to $(x,i_x)$ in $H_j$ using $n$ slides for each token. This forms a shortest $a$-$b$-path in $G$ which is discovered with the budget $b = n^2$. 

    For the other direction assume that $(G, S, b, a, b)$ is a yes-instance of \textsc{Shortest Path Discovery} and observe that every shortest $a$-$b$-path in $G$ (which is of length $n+1$ and hence uses $n$ internal vertices) uses internal vertices from one $H_j$ only. 
    By the choice of the budget and the connections between vertices $x$ and $(x, y)$, every solution can only move the token from vertex $x$ to a vertex of the form $(x, y)$ for some $y \leq n$ in $H_j$. Let $a (1, y_1) (2, y_2) \dots (n, y_n) t$ be the discovered $a$-$b$-path in $G$. By construction, we have $y_i \neq y_{i'}$ for $i \neq i'$. Hence $y_1 \dots y_n$ is a Hamiltonian path in $G_j$.
\end{proof}

\section{Matching Discovery}
\label{sec:mat}

Grobler et al.~\cite{grobler2023solution} show that \textsc{Mat-D} is $\W[1]$-hard with respect to the parameter $\budget$ on $3$-degenerate graphs, yet it is in $\FPT$ with respect to parameter $k$ on general graphs. 
We show that, similarly to \textsc{VC-D}, \textsc{Mat-D} admits a polynomial kernel with respect to the parameter $k$.

In a manner akin to \Cref{thm:nowheredense-IS-k}, our kernelization algorithm for \textsc{Mat-D} with respect to the parameter $k$ will remove from the graph vertices that are irrelevant for every token.
Here however, to find irrelevant vertices or edges, we will make use of a classical result of Erd\H{o}s and Rado~\cite{erdos1960intersection} known in the literature as the \emph{sunflower lemma}. 
\begin{theorem}[\cite{erdos1960intersection}]
\label{lem:sunflower-lemma}
Let $\mathcal{A}$ be a family of sets (without duplicates) over a universe $\mathcal{U}$, such that each set in $\mathcal{A}$ has cardinality at most $d$. If $|\mathcal{A}| > d!(p-1)^d$, then $\mathcal{A}$ contains a sunflower with $p$ petals and such a sunflower can be computed in time polynomial in $|A|$, $|U|$, and $p$. 
\end{theorem}

\begin{theorem}\label{thm:Mat-k-general}
\textsc{Mat-D} admits a kernel of size $\mathcal{O}(k^5)$.
\end{theorem}

\begin{proof}
Let $(G, S, \budget)$ be an instance of \textsc{Mat-D}.
Without loss of generality, we assume the graph $G$ to be connected.
For each vertex $s \in S$, and integer $i \in [3k]$, we compute $E(s,i)$.
We maintain the invariant that we remove from $E(s,i)$ for each $s \in S$ and $i \in [3k]$, irrelevant vertices with respect to~$s$.

We remove an irrelevant edge with respect to a vertex $s \in S$ from $E(s,i)$ for an integer $i \in [3k]$ as follows. 
From the sunflower lemma (\Cref{lem:sunflower-lemma}), if $|E(s,i)| > 8k^2$, then it has a sunflower with $2k + 1$ petals that can be computed in polynomial time in $k$. 
We arbitrarily choose one edge $e$ corresponding to one petal of the sunflower and remove it from $E(s,i)$.
To see why $e$ is irrelevant with respect to $s$, assume that the token on $s$ slides to $e \in C_\ell$, where $C_\ell$ is a matching in $G$. 
The $2k - 2$ vertices of $C_\ell \setminus \{e\}$ can be incident to at most $2k - 2$ of the edges corresponding to the petals of the sunflower, leaving at least one petal with an edge $e_1$ that can replace $e$ in the matching $C_\ell$ in $G$.
Since also all edges in $E(s,i)$ are at the same distance $i$ from $s$, replacing $e$ by $e_1$ will not increase the number of slides needed to achieve $C_\ell \setminus \{e\} \cup \{e_1\}$.

We form the kernel $(G', S, \budget)$ of the original instance $(G, S, \budget)$ as follows.
First, note that for a token $s \in S$ and an edge $e \in E(H) \cap C_\ell$ such that $d(s,e) > 3k$, the edges in $C_\ell \setminus \{e\}$ can appear in at most $k-1$ of the $3k$ sets of edges $E(s,i)$ for $i \in [3k]$ and every such edge that appears in a set $E(s,\mathfrak{i})$ for a specific $\mathfrak{i} \in [3k]$ can be incident to at most all the edges in $E(s,\mathfrak{i}-1)$ and $E(s,\mathfrak{i}+1)$. 
This implies that the token on $s$ cannot move towards any edge of at most $3k - 3$ of the $3k$ sets $E(s,i)$ for $i \in [3k]$ (as these contain tokens and thus might result in incident tokens) and thus there exists an edge $e_1$ which the token on $s$ can move to while maintaining a matching in $C_\ell \setminus \{e\} \cup \{e_1\}$.
Thus, in any solution to $(G, S, \budget)$, if a token on an edge $s \in S$ moves to an edge $e \in C_\ell$ such that $d(s, e) > 3k$, it can instead move towards an edge $e_1 \in E(H)$ such that $d(s, e_1) \le 3k$, while keeping the rest of the solution unchanged. 
Consequently, we set $E(G') = \bigcup_{s \in S,i \in [3k]} E(s,i) \text{ } \cup \text{ } S$ and for each edge $e \in E(s,i)$, for $s \in S$ and $i \in [3k]$, we add to $E(G')$ at most $i$ edges that are on the shortest path from $s$ to $e$ (if such edges are not already in $E(G')$).
$G'$ is the subgraph of $G$ induced by the edges in $E(G')$.
By the end of this process, $|E(G')| \le k + 9k^3 \cdot 8k^2$, as for each $s \in S$ and $i \in [3k]$, $E(s,i) \le 8k^2$ and for each edge of the latter $3k^2$ sets of edges, we added to $E(G')$ at most $3k - 1$ edges that are on a shortest path from that edge to the edge $s$.
$(G', S, b)$ is a kernel as only edges that are irrelevant with respect to every token in $S$ might not be in $E(G')$ and all edges needed to move tokens from edges in $S$ towards a matching using only $b$ slides are present in $E(G')$.
\end{proof}

\section{Vertex Cut Discovery}
\label{sec:cut}
Grobler et al.~\cite{grobler2023solution} showed that \textsc{VCut-D} is $\W[1]$-hard with respect to parameter~$\budget$ on $2$-degenerate bipartite graphs but is in $\FPT$ with respect to the parameter $k$ on general graphs. 
We show that the problem admits no polynomial kernels unless $\NP \subseteq \cp$.
We denote an instance of \textsc{VCut-D} by $(G, S, b, a_1, b_1)$ to emphasize that the solution must be a vertex cut between $a_1$ and~$b_1$ in $V(G)$. 

Given a graph $H$ and an edge coloring $\phi: E(H) \rightarrow [c]$, we say $\phi$ is proper if, for all distinct edges $e, e_1 \in E(H)$, $\phi(e) \neq \phi(e_1)$ whenever $e$ and $e_1$ share a vertex.
We form our or-cross-composition from the \textsc{Rainbow Matching} problem, which is \NP-complete even on properly colored 2-regular graphs and where every $i \in [c]$ is used exactly twice in the coloring~\cite{van2013complexity}:
\\[1ex]
\textsc{Rainbow Matching:}
\newline
\textbf{Input}: Graph $H$, a proper edge coloring $\phi$ and an integer $\kappa$.
\newline
\textbf{Question}: Does $H$ have a rainbow matching, i.e., a matching whose edges have distinct colors, with at least $\kappa$ edges?
\begin{theorem}\label{thm:cross_composition-Vcut-k}
There exists an or-cross-composition from \textsc{Rainbow Matching} into \textsc{VCut-D} where the parameter is the number of tokens, $k$. Consequently, \textsc{VCut-D} does not admit a polynomial kernel with respect to $k$, unless $\NP \subseteq \cp$.
\end{theorem}

\begin{proof}
By choosing an appropriate polynomial equivalence relation $\mathcal{R}$, we may assume that we are given a family of $t$ \textsc{Rainbow Matching} instances $(H_r, \phi_r, \kappa_r)$, where $H_r$ is a 2-regular graph, $|V(H_r)| = n$, $|E(H_r)| = m$, $\kappa_r = \kappa \in \mathbb{N}$, and $\phi_r : E(H_r) \rightarrow [c]$ is a mapping that properly colors $H_r$ and in which every $i \in [c]$ is used exactly twice.
We may duplicate some input instances so that $t = 2^s$ for some integer $s$. 
Note that this step at most doubles the number of input instances.
The construction of the instance $(G, S, \budget, a_1, b_1)$ of \textsc{VCut-D} is twofold. 

For each instance $(H_r, \phi_r, \kappa_r)$, we create $G_r$, formed of two vertices, $s_r$ and $t_r$, as well as $\kappa - 1$ sets $\{E^1_r, \ldots, E_r^{\kappa - 1}\}$ of $2m + 2$ vertices each.
A set $E^p_r$ for $p \in [\kappa - 1]$ contains $2m$ vertices, denoted \emph{edge-vertices}, that represent the edges in $H_r$ twice and two other vertices which are denoted by $s^p_r$ and $t^p_r$ (see \Cref{fig:vcut-k}).
We denote the edge-vertices in a set $E^p_r$ as $v_{e_h}^{p,r}(1)$ ($v_{e_h}^{p,r}(2)$) to refer to the first (second) vertex representing the same edge $e_h$ of $E(H_r)$ in $E^p_r$. 
We denote by $E^p_r(1)$ the set of all vertices $v_{e_h}^{p,r}(1)$, and by $E^p_r(2)$ the set of all vertices $v_{e_h}^{p,r}(2)$. 
In $G_r$, we connect:
\begin{itemize}[itemsep=0pt]
    \item through paths of length $m^3 + \log t$, $s_r$ to each of $s_r^p$ for $p \in [\kappa-1]$ and $t_r$ to each of $t_r^p$ for $p \in [\kappa-1]$,
    \item through paths of length $m^3 + \log t$, $s_r^p$ to all vertices $v_{e_h}^{p,r}(1)$ and $t_r^p$ to all vertices $v_{e_h}^{p,r}(2)$ for each $e_h \in E(H_r)$ and each $p \in [\kappa - 1]$, 
    \item through paths of length $m^3 + \log t$, all vertices $v_{e_h}^{p,r}(1)$ and $v_{e_g}^{q,r}(2)$ such that $\phi_r(e_h) = \phi_r(e_g)$ for each $p \le q \in [\kappa-1]$,
    \item through paths of length $m^3 + \log t$, $v_{e_h}^{p,r}(1)$ and $v_{e_g}^{q,r}(2)$, for each $p \le q \in [\kappa-1]$, whenever $e_h$ and $e_g$ are incident in $H_r$,
    \item through paths of length $m^3 + \log t$, $v_{e_h}^{p,r}(2)$ and $v_{e_g}^{q,r}(1)$, for each $p \in [\kappa - 2]$, $q = p + 1$, whenever $e_h \neq e_g$.
\end{itemize}

\afterpage{
\begin{figure}[H]
    \centering
    \begin{tikzpicture}
    [
      level distance=1cm,
      level 1/.style={sibling distance=4cm},
      level 2/.style={sibling distance=2cm},
      level 3/.style={sibling distance=1cm},
      edge from parent/.style={draw=black, dashed},
      every node/.style={circle, solid, draw=black, fill=white, inner sep=1.5pt},
      grow'=down
    ]
    
     % Draw 9 rectangles, each 3 aligned vertically, on the left of the tree
    \fill[yellow!10] (-10,-1.75) -- (-10,-0.75) -- (-4.25,-0.75) -- (-4.25,-1.75) -- cycle;
    \fill[red!10] (-10,-2.75) -- (-10,-1.75) -- (-4.25,-1.75) -- (-4.25,-2.75) -- cycle;
    \fill[blue!10] (-10,-3.75) -- (-10,-2.75) -- (-4.25,-2.75) -- (-4.25,-3.75) -- cycle;
    \draw[teal, line width=0.5mm] (-4.75,0.25) .. controls (-3,-1.25) .. (-4.75,-3);
    \draw[teal, line width=0.5mm] (-4.75,0.25) .. controls (-3.3,-0.25) .. (-4.75,-2);
    \draw[teal, line width=0.5mm] (-4.75,0.25) -- (-4.75,-1);
    \draw[green, line width=0.5mm] (-6.75,0.25) .. controls (-5,-1.25) .. (-6.75,-3);
    \draw[green, line width=0.5mm] (-6.75,0.25) .. controls (-5.3,-0.25) .. (-6.75,-2);
    \draw[green, line width=0.5mm] (-6.75,0.25) -- (-6.75,-1);
    \draw[red, line width=0.5mm] (-8.75,0.25) .. controls (-7,-1.25) .. (-8.75,-3);
    \draw[red, line width=0.5mm] (-8.75,0.25) .. controls (-7.3,-0.25) .. (-8.75,-2);
    \draw[red, line width=0.5mm] (-8.75,0.25) -- (-8.75,-1);
    \draw[teal] (-4.75,0.25) -- (-3.5,-3);
    \node[fill=none, draw=none] at (-9.5,-1.25) {...};
    \node[fill=none, draw=none] at (-9.5,-2.25) {...};
    \node[fill=none, draw=none] at (-9.5,-3.25) {...};
    
    % Draw each rectangle one by one and add text inside
    \draw[draw=blue] (-9,-1) rectangle ++(0.5,-0.5);
    \node[fill=none, draw=none, label=above:{\scalebox{0.4}{$E_1^2(1)$}}] at (-8.75,-1.75) {};
    \draw[draw=blue] (-9,-2) rectangle ++(0.5,-0.5);
    \node[fill=none, draw=none, label=above:{\scalebox{0.4}{$E_2^2(1)$}}] at (-8.75,-2.75) {};
    \draw[draw=blue] (-9,-3) rectangle ++(0.5,-0.5);
    \node[fill=none, draw=none, label=above:{\scalebox{0.4}{$E_3^2(1)$}}] at (-8.75,-3.75) {};
    \draw[draw=blue] (-7,-1) rectangle ++(0.5,-0.5);
    \node[fill=none, draw=none, label=above:{\scalebox{0.4}{$E_1^1(2)$}}] at (-6.75,-1.75) {};
    \draw[draw=blue] (-7,-2) rectangle ++(0.5,-0.5);
    \node[fill=none, draw=none, label=above:{\scalebox{0.4}{$E_2^1(2)$}}] at (-6.75,-2.75) {};
    \draw[draw=blue] (-7,-3) rectangle ++(0.5,-0.5);
    \node[fill=none, draw=none, label=above:{\scalebox{0.4}{$E_3^1(2)$}}] at (-6.75,-3.75) {};
    \draw[draw=blue] (-5,-1) rectangle ++(0.5,-0.5);
    \node[fill=none, draw=none, label=above:{\scalebox{0.4}{$E_1^1(1)$}}] at (-4.75,-1.75) {};
    \draw[draw=blue] (-5,-2) rectangle ++(0.5,-0.5);
    \node[fill=none, draw=none, label=above:{\scalebox{0.4}{$E_2^1(1)$}}] at (-4.75,-2.75) {};
    \draw[draw=blue] (-5,-3) rectangle ++(0.5,-0.5);
    \node[fill=none, draw=none, label=above:{\scalebox{0.4}{$E_3^1(1)$}}] at (-4.75,-3.75) {};
    \draw[blue, fill=yellow!10] (-3.7,-3.01) rectangle (-3.3,-3.99);
    \draw[blue, fill=red!10] (-2.7,-3.01) rectangle (-2.3,-3.99);
    \draw[blue, fill=blue!10] (-1.7,-3.01) rectangle (-1.3,-3.99);
    
    \foreach \i in {4,5,6,7,8} {
        % Draw a rectangle between the leaf and the corresponding end vertex
        \draw[blue] (\i*1-4.7,-3.01) rectangle (\i*1-4.3,-3.99);
    }
    \draw[teal] (-4.75,0.25) -- (-3.5,-4);
    \draw[teal] (-4.75,0.25) -- (-2.5,-3);
    \draw[teal] (-4.75,0.25) -- (-2.5,-4);
    \draw[teal] (-4.75,0.25) -- (-1.5,-3);
    \draw[pink] (0, -2) -- (-1,-2);
    \draw[pink] (0,-2) -- (1,-2);
    \draw[pink] (0,-2) .. controls (1.5,-1.5) .. (3,-2);
    \draw[pink] (0,-2) .. controls (-1.5,-1.5) .. (-3,-2);
    \draw[pink] (0, -3) -- (-0.5,-3);
    \draw[pink] (0,-3) -- (0.5,-3);
    \draw[pink] (0,-3) .. controls (-0.75,-2.5) .. (-1.5,-3);
    \draw[pink] (0,-3) .. controls (0.75,-2.5) .. (1.5,-3);
    \draw[pink] (0,-3) .. controls (-1.5,-2) .. (-3.5,-3);
    \draw[pink] (0,-3) .. controls (1.5,-2) .. (3.5,-3);
    \draw[pink] (0, -1) -- (-2,-1);
    \draw[pink] (0,-1) -- (2,-1);
    \node[fill=black, draw=black, label=above:{\scalebox{0.4}{$u(3,.)$}}] at (-8.75,0.25) {};
    \node[fill=black, draw=black, label=above:{\scalebox{0.4}{$u(2,.)$}}] at (-6.75,0.25) {};
    \node[fill=black, draw=black, label=above:{\scalebox{0.4}{$u(1,.)$}}] at (-4.75,0.25) {};
    
    % Root of the binary tree
    \node [label=above:{\tiny $a_1$}]{}
      child { node {}
        child { node {}
          child { node[circle, solid, draw=blue, fill=white, inner sep=1.5pt, label=below:{\scalebox{0.5}{$s_8$}}] {} }
          child { node[circle, solid, draw=blue, fill=white, inner sep=1.5pt, label=below:{\scalebox{0.5}{$s_7$}}] {} }
        }
        child { node {}
          child { node[circle, solid, draw=blue, fill=white, inner sep=1.5pt, label=below:{\scalebox{0.5}{$s_6$}}] {} }
          child { node[circle, solid, draw=blue, fill=white, inner sep=1.5pt, label=below:{\scalebox{0.5}{$s_5$}}] {} }
        }
      }
      child { node {}
        child { node {}
          child { node[circle, solid, draw=blue, fill=white, inner sep=1.5pt, label=below:{\scalebox{0.5}{$s_4$}}] {} }
          child { node[circle, solid, draw=blue, fill=white, inner sep=1.5pt, label=below:{\scalebox{0.5}{$s_3$}}] {} }
        }
        child { node {}
          child { node[circle, solid, draw=blue, fill=white, inner sep=1.5pt, label=below:{\scalebox{0.5}{$s_2$}}] {} }
          child { node[circle, solid, draw=blue, fill=white, inner sep=1.5pt, label=below:{\scalebox{0.5}{$s_1$}}] {} }
        }
      };
    
    \node at (0,-2) [circle, solid, draw=black, fill=black, inner sep=1.5pt, label=above:{\tiny $v^2$}] {};
    \node at (0,-1) [circle, solid, draw=black, fill=black, inner sep=1.5pt, label=above:{\tiny $v^1$}] {};
    \node at (0,-3) [circle, solid, draw=black, fill=black, inner sep=1.5pt, label=above:{\tiny $v^3$}] {};
    
    \foreach \i in {1,2,3,4,5,6,7,8} {
        \draw[black, dashed] (\i*1-4.5,-4) -- (0,-5);
        \node at (\i*1-4.5,-4) [circle, solid, draw=blue, fill=white, inner sep=1.5pt, label=above:{\scalebox{0.5}{$t_{\i}$}}] (end\i) {};
    }
    \node at (0,-5) [circle, solid, draw=black, fill=white, inner sep=1.5pt, label=below:{\tiny $b_1$}] {};
    \end{tikzpicture}
    \caption{\footnotesize An illustration of the graph $G$ formed as per the composition of \Cref{thm:cross_composition-Vcut-k} given input instances $(H_r, \phi_r, \kappa_r)$ for $r \in [8]$, where $\kappa_r = \kappa \ge 3$.
    For clarity, the graphs $G_r$ for $r \in [8]$ were replaced by rectangles with blue borders incident to two blue vertices $s_r$ and $t_r$ of $G_r$. 
    Pink edges are used to illustrate how the vertices $v^d$ for $d \in [3]$ connect to the vertices of $\mathcal{T}$. 
    Dotted lines represent paths of length $m^3 + \log t$ between the vertices and thick edges are used to represent that a vertex is adjacent to all vertices in a set of vertices.
    The yellow, pink, and blue rectangular areas on the left provide a zoomed-in view of some of the content of $G_1$, $G_2$, and $G_3$, respectively. Particularly, they show the sets of vertices $E_1^1(1)$, $E_1^1(2)$, $E_1^2(1)$, $E_2^1(1)$, $E_2^1(2)$, $E_2^2(1)$, $E_3^1(1)$, $E_3^1(2)$, and $E_3^2(1)$.
    For clarity, not all edges between vertices of the form $u(i,j)$ for $i \in [2(\kappa -1)]$ and $j \in [m-1]$, and both vertices $s_r$ and $t_r$ for $r \in [8]$ are shown.}
    \label{fig:vcut-general}
\end{figure}
}
We form $G$ of all $G_r$ for $r \in [t]$ as follows (see \Cref{fig:vcut-general}).
We create two global vertices $a_1$ and~$b_1$ such that $b_1$ is connected through paths of length $m^3+ \log t$ to $t_r$ for $r \in [t]$.
Additionally, we create a binary tree~$\mathcal{T}$ rooted at $a_1$, with $\log t + 1$ levels, and whose leaves constitute $s_r$ for $r \in [t]$.
For each depth $d$ of $\mathcal{T}$ for $d \in \{1, \ldots, \log t\}$, we create a vertex $v^d$ that contains a token and is connected through a single edge to each vertex of $\mathcal{T}$ that is at depth $d$. 
The edges of $\mathcal{T}$ are all replaced by paths of length $m^3 + \log t$.
Finally, we create $2(\kappa - 1)$ sets $\{M_1, \ldots, M_{2(\kappa - 1)}\}$, of $m - 1$ edges each. 
We connect each edge $e^{(i,j)} \in M_{i}$ for $i \in [2(\kappa-1)]$ and $j \in [m-1]$, from one of its endpoints, denoted $u^{(i,j)}$, to each vertex $v_{e_h}^{\ceil{i/2},r}(1)$ for each $r \in [t]$ if $i$ is odd, and to each vertex $v_{e_h}^{\ceil{i/2},r}(2)$ for each $r \in [t]$ if $i$ is even. 
Additionally, we connect through paths of length $m^3+ \log t$, each $s_r$ and $t_r$ for $r \in [t]$ to all of $u^{(i,j)}$ for $i \in [2(\kappa - 1)]$ and $j \in [m - 1]$.
All vertices in the sets $\{M_1, \ldots, M_{2(\kappa - 1)}\}$ contain tokens.
Setting $\budget = \log t + 2(2\kappa - 2) \cdot (m - 1)$ finalizes the construction of $(G, S, \budget, a_1, b_1)$. 
Since we perform only a polynomial number of operations per instance as well as some polynomial in $t$ other operations while creating the tree $\mathcal{T}$ and connecting some vertices, the reduction is polynomial in $\Sigma^t_{i=1} |x_i|$. 
Additionally, $k$ is $O(m^2 + \log t)$ since $\kappa \le m$.
\\[1ex]
\begin{figure}
    \centering
\begin{tikzpicture}
\node[circle, inner sep=2pt, fill=white, draw=black, label=above left:{\tiny $s^1_1$}] (s11) at (0, 0) {};
\node[circle, inner sep=2pt, fill=white, draw=black, label=above right:{\tiny $t^1_1$}] (t11) at (10, 0) {};
\node[circle, inner sep=2pt, fill=blue!30, label=above right:{\tiny $v_a^{(1,1)}(1)$}] (va11) at (3, 3) {};
\node[circle, inner sep=2pt, fill=blue!30, label=above right:{\tiny $v_b^{(1,1)}(1)$}] (vb11) at (3, 2.25) {};
\node[circle, inner sep=2pt, fill=orange!30, label=above right:{\tiny $v_c^{(1,1)}(1)$}] (vc11) at (3, 1.5) {};
\node[circle, inner sep=2pt, fill=orange!30, label=above right:{\tiny $v_d^{(1,1)}(1)$}] (vd11) at (3, 0.75) {};
\node[circle, inner sep=2pt, fill=green!30, label=above right:{\tiny $v_e^{(1,1)}(1)$}] (ve11) at (3, 0) {};
\node[circle, inner sep=2pt, fill=green!30, label=above right:{\tiny $v_f^{(1,1)}(1)$}] (vf11) at (3, -0.75) {};
\node[circle, inner sep=2pt, fill=red!30, label=above right:{\tiny $v_g^{(1,1)}(1)$}] (vg11) at (3, -1.5) {};
\node[circle, inner sep=2pt, fill=red!30, label=above right:{\tiny $v_h^{(1,1)}(1)$}] (vh11) at (3, -2.25) {};
\node[circle, inner sep=2pt, fill=brown!30, label=above right:{\tiny $v_i^{(1,1)}(1)$}] (vi11) at (3, -3) {};
\node[circle, inner sep=2pt, fill=brown!30, label=above right:{\tiny $v_j^{(1,1)}(1)$}] (vj11) at (3, -3.75) {};

\draw[draw=red, thick, decorate, decoration={brace,amplitude=10pt,mirror}] (2.5,3.5) -- (2.5,-4.25) node at (3.5, 4) {\textcolor{red}{\tiny $E_1^1(1)$}};
\draw[draw=red, thick, decorate, decoration={brace,amplitude=10pt}] (4.5,3.5) -- (4.5,-4.25);
\draw[draw=brown, thick, decorate, decoration={brace,amplitude=10pt,mirror}] (5.5,3.5) -- (5.5,-4.25) node at (6.5, 4) {\textcolor{brown}{\tiny $E_1^1(2)$}};
\draw[draw=brown, thick, decorate, decoration={brace,amplitude=10pt}] (7.5,3.5) -- (7.5,-4.25);

\node[circle, inner sep=2pt, fill=blue!30, label=above left:{\tiny $v_a^{(1,1)}(2)$}] (va12) at (7, 3) {};
\node[circle, inner sep=2pt, fill=blue!30, label=above left:{\tiny $v_b^{(1,1)}(2)$}] (vb12) at (7, 2.25) {};
\node[circle, inner sep=2pt, fill=orange!30, label=above left:{\tiny $v_c^{(1,1)}(2)$}] (vc12) at (7, 1.5) {};
\node[circle, inner sep=2pt, fill=orange!30, label=above left:{\tiny $v_d^{(1,1)}(2)$}] (vd12) at (7, 0.75) {};
\node[circle, inner sep=2pt, fill=green!30, label=above left:{\tiny $v_e^{(1,1)}(2)$}] (ve12) at (7, 0) {};
\node[circle, inner sep=2pt, fill=green!30, label=above left:{\tiny $v_f^{(1,1)}(2)$}] (vf12) at (7, -0.75) {};
\node[circle, inner sep=2pt, fill=red!30, label=above left:{\tiny $v_g^{(1,1)}(2)$}] (vg12) at (7, -1.5) {};
\node[circle, inner sep=2pt, fill=red!30, label=above left:{\tiny $v_h^{(1,1)}(2)$}] (vh12) at (7, -2.25) {};
\node[circle, inner sep=2pt, fill=brown!30, label=above left:{\tiny $v_i^{(1,1)}(2)$}] (vi12) at (7, -3) {};
\node[circle, inner sep=2pt, fill=brown!30, label=above left:{\tiny $v_j^{(1,1)}(2)$}] (vj12) at (7, -3.75) {};

\foreach \i in {a,b,c,d,e,f,g,h,i,j} {
  \draw (s11) -- (v\i11);
}

\foreach \i in {a,b,c,d,e,f,g,h,i,j} {
  \draw (t11) -- (v\i12);
}

\foreach \i in {a,b,c,d,e,f,g,h,i,j} {
  \draw[draw=blue] (v\i11) -- (v\i12);
}

\node[circle, inner sep=2pt, fill=white, draw=black, label=above left:{\tiny $s_1^2$}] (s12) at (0, -10) {};
\node[circle, inner sep=2pt, fill=white, draw=black, label=above right:{\tiny $t_1^2$}] (t12) at (10, -10) {};
\node[circle, inner sep=2pt, fill=blue!30, label=above right:{\tiny $v_a^{(2,1)}(1)$}] (va21) at (3, -7) {};
\node[circle, inner sep=2pt, fill=blue!30, label=above right:{\tiny $v_b^{(2,1)}(1)$}] (vb21) at (3, -7.75) {};
\node[circle, inner sep=2pt, fill=orange!30, label=above right:{\tiny $v_c^{(2,1)}(1)$}] (vc21) at (3, -8.5) {};
\node[circle, inner sep=2pt, fill=orange!30, label=above right:{\tiny $v_d^{(2,1)}(1)$}] (vd21) at (3, -9.25) {};
\node[circle, inner sep=2pt, fill=green!30, label=above right:{\tiny $v_e^{(2,1)}(1)$}] (ve21) at (3, -10) {};
\node[circle, inner sep=2pt, fill=green!30, label=above right:{\tiny $v_f^{(2,1)}(1)$}] (vf21) at (3, -10.75) {};
\node[circle, inner sep=2pt, fill=red!30, label=above right:{\tiny $v_g^{(2,1)}(1)$}] (vg21) at (3, -11.5) {};
\node[circle, inner sep=2pt, fill=red!30, label=above right:{\tiny $v_h^{(2,1)}(1)$}] (vh21) at (3, -12.25) {};
\node[circle, inner sep=2pt, fill=brown!30, label=above right:{\tiny $v_i^{(2,1)}(1)$}] (vi21) at (3, -13) {};
\node[circle, inner sep=2pt, fill=brown!30, label=above right:{\tiny $v_j^{(2,1)}(1)$}] (vj21) at (3, -13.75) {};
\node[circle, inner sep=2pt, fill=blue!30, label=above left:{\tiny $v_a^{(2,1)}(2)$}] (va22) at (7, -7) {};
\node[circle, inner sep=2pt, fill=blue!30, label=above left:{\tiny $v_b^{(2,1)}(2)$}] (vb22) at (7, -7.75) {};
\node[circle, inner sep=2pt, fill=orange!30, label=above left:{\tiny $v_c^{(2,1)}(2)$}] (vc22) at (7, -8.5) {};
\node[circle, inner sep=2pt, fill=orange!30, label=above left:{\tiny $v_d^{(2,1)}(2)$}] (vd22) at (7, -9.25) {};
\node[circle, inner sep=2pt, fill=green!30, label=above left:{\tiny $v_e^{(2,1)}(2)$}] (ve22) at (7, -10) {};
\node[circle, inner sep=2pt, fill=green!30, label=above left:{\tiny $v_f^{(2,1)}(2)$}] (vf22) at (7, -10.75) {};
\node[circle, inner sep=2pt, fill=red!30, label=above left:{\tiny $v_g^{(2,1)}(2)$}] (vg22) at (7, -11.5) {};
\node[circle, inner sep=2pt, fill=red!30, label=above left:{\tiny $v_h^{(2,1)}(2)$}] (vh22) at (7, -12.25) {};
\node[circle, inner sep=2pt, fill=brown!30, label=above left:{\tiny $v_i^{(2,1)}(2)$}] (vi22) at (7, -13) {};
\node[circle, inner sep=2pt, fill=brown!30, label=above left:{\tiny $v_j^{(2,1)}(2)$}] (vj22) at (7, -13.75) {};

\foreach \i in {a,b,c,d,e,f,g,h,i,j} {
  \draw (s12) -- (v\i21);
}

\foreach \i in {a,b,c,d,e,f,g,h,i,j} {
  \draw (t12) -- (v\i22);
}

\foreach \i in {a,b,c,d,e,f,g,h,i,j} {
  \draw[draw=blue] (v\i21) -- (v\i22);
}

\draw[draw=blue] (va11) -- (vb12);
\draw[draw=blue] (vb11) -- (va12);
\draw[draw=blue] (va11) -- (va22);
\draw[draw=blue] (va11) -- (vb22);
\draw[draw=blue] (va21) -- (vb22);
\draw[draw=blue] (vb21) -- (va22);
\draw[draw=blue] (va12) -- (va21);
\draw[draw=blue] (va12) -- (vb21);
\draw[draw=orange] (vh11) -- (vj22);
\draw[draw=black, dotted] (vj12) -- (va21);
\draw[draw=black, dotted] (vj12) -- (vb21);
\draw[draw=black, dotted] (vj12) -- (vc21);
\draw[draw=black, dotted] (vj12) -- (vd21);
\draw[draw=black, dotted] (vj12) -- (ve21);
\draw[draw=black, dotted] (vj12) -- (vf21);
\draw[draw=black, dotted] (vj12) -- (vg21);
\draw[draw=black, dotted] (vj12) -- (vh21);
\draw[draw=black, dotted] (vj12) -- (vi21);
\draw (-1,0) -- (s11);
\draw (11,0) -- (t11);
\draw (-1,-10) -- (s12);
\draw (11,-10) -- (t12);
\end{tikzpicture}
    \caption{\footnotesize An illustration of $E^1_1$, $E_1^2$, $s^1_1$, $t^1_1$, $s^2_1$, and $t^2_1$ of $G_1$ of the or-cross-composition of Theorem~\ref{thm:cross_composition-Vcut-k}. In $H_1$, the vertices are $a$, $b$, $c$, $d$, $e$, $f$, $g$, $h$, $i$, and $j$. For simplification purposes, the figure illustrates the types of edges but does not contain all edges between the illustrated vertices. Length $m^3+ \log t$ paths are represented by the edges (regular or dotted). Vertices in red brackets are in $E_1^1(1)$ and those in beige brackets are in $E_1^1(2)$. Blue edges are between vertices representing edges of the same color in $H^1$ and dotted ones between all $v_{e_h}^{p,r}(2)$ and $v_{e_g}^{q,r}(1)$ for $q = p + 1$, whenever $e_h \neq e_g$. Finally, orange edges show that the edges, represented by the adjacent edge-vertices, are adjacent in $H_1$. In $G_1$, length $m^3+ \log t$ paths exist between $s_1$ and both of $s^1_1$ and $s^2_1$ and between $t_1$ and both of $t^1_1$ and $t^2_1$. No vertex in this figure contains a token (colored vertices display the colors of the edges in the instance $(H_1, \phi_1, r_1)$).}
    \label{fig:vcut-k}
\end{figure}

\begin{claim}
If for some $\mathfrak{r} \in [t]$, $(H_\mathfrak{r}, \phi_\mathfrak{r}, \kappa_\mathfrak{r})$ is a yes-instance of \textsc{Rainbow Matching}, then the constructed instance $(G, S, \budget, a_1, b_1)$ is a yes-instance of \textsc{VCut-D}. 
\end{claim}

\begin{claimproof}
Let $\mathcal{M}_\mathfrak{r}$ be a solution to the instance $(H_\mathfrak{r}, \phi_\mathfrak{r}, \kappa_\mathfrak{r})$. 
$\mathcal{M}_\mathfrak{r} \subseteq E(H_\mathfrak{r})$ forms a matching in $H_\mathfrak{r}$ such that $\phi_\mathfrak{r}(e_h) \neq \phi_\mathfrak{r}(e_g)$, for all $e_h$, $e_g \in \mathcal{M}_\mathfrak{r}$.
We apply the following slides in $(G, S, \budget, a_1, b_1)$ to disconnect $a_1$ from $b_1$.
First, we choose one edge $e_h$ of $\mathcal{M}_\mathfrak{r}$ and using $m - 1$ slides, we slide the tokens on $u^{(1,j)}$ for $j \in [m-1]$ onto all vertices in $E^1_\mathfrak{r}(1)$ except $v^{1,\mathfrak{r}}_{e_h}(1)$.   
Then, using $(2\kappa - 3) \cdot (m - 1)$ slides, for each $i \in [\kappa-1]$, we choose one other edge $e_s \in \mathcal{M}_\mathfrak{r}$ and slide the tokens on $u^{(2i,j)}$ and~$u^{(2i+1,j)}$ (when applicable) for $j \in [m-1]$ onto all vertices in $E_\mathfrak{r}^{i}(2)$ and $E_\mathfrak{r}^{i+1}(1)$ except $v^{i,\mathfrak{r}}_{e_s}(2)$ and $v^{i+1,\mathfrak{r}}_{e_s}(1)$, respectively.   
We slide onto $u^{(i,j)}$ for all $i \in [2(\kappa - 1)]$ and $j \in [m-1]$ the tokens adjacent to the latter vertices, on the edges in $\{M_1, \ldots, M_{2(\kappa-1)}\}$, using $(2\kappa - 2) \cdot (m - 1)$ slides.
Finally, in~$\mathcal{T}$, we use the tokens on the vertices $v^d$ for $d \in \{1, \ldots, \log t\}$, to disconnect all paths from the root $a_1$ to all of $s_r$ for $r \in [t] - \{\mathfrak{r}\}$, using one slide per token.
This ensures that, through at most $\log t$ slides, all paths from $a_1$ to $b_1$ go through only both $s_\mathfrak{r}$ and $t_\mathfrak{r}$. 
Following the described steps, we have executed a total of $\budget$ slides. 
To see that $a_1$ and $b_1$ are now disconnected, note that after the slides of the tokens on $u^d$ for $d \in \{1, \ldots, \log t\}$ are performed, all paths from $a_1$ to $b_1$ in $G$ go through $s_\mathfrak{r}$ and $t_\mathfrak{r}$.
Thus it suffices to argue that the remaining $2(2\kappa - 2) \cdot (m - 1)$ slides disconnect $s_\mathfrak{r}$ and $t_\mathfrak{r}$.
First, if this is not the case, then no path between $s_\mathfrak{r}$ and $t_\mathfrak{r}$ goes through any $u^{(i,j)}$ for all $i \in [2(\kappa - 1)]$ and $j \in [m-1]$ since the tokens that left those vertices have been replaced. 
Also, the last four vertices on any path between $s_\mathfrak{r}$ and $t_\mathfrak{r}$ must be $v_{e_h}^{p,\mathfrak{r}}(1)$ for some $p \in [\kappa-1]$ and some $e_h \in E(H_\mathfrak{r})$, $v_{e_g}^{q,\mathfrak{r}}(2)$ for some $q \in \{p, \ldots, \kappa - 1\}$ and some $e_g \in E(H_\mathfrak{r})$, $t_\mathfrak{r}^{q}$ and $t_\mathfrak{r}$.
However, by construction, there exists no paths between all vertices $v_{e_h}^{p,\mathfrak{r}}(1)$ and $v_{e_g}^{q,\mathfrak{r}}(2)$ for each $p \le q \in [\kappa-1]$, such that $\phi_\mathfrak{r}(e_h) \neq \phi_\mathfrak{r}(e_g)$ and~$e_h$ and $e_g$ are non-adjacent. 
Thus, given our choice of the free vertices remaining in $E_\mathfrak{r}^p(.)$ for all $p \in [\kappa-1]$, no path exists between $s_\mathfrak{r}$ from $t_\mathfrak{r}$ and therefore between $a_1$ and $b_1$.    
\end{claimproof}

\begin{claim}
If $(G, S, \budget, a_1, b_1)$ is a yes-instance of \textsc{VCut-D}, then there exists an integer $\mathfrak{r} \in [t]$ for which $(H_\mathfrak{r}, \phi_\mathfrak{r}, \kappa_\mathfrak{r})$ is a yes-instance of \textsc{Rainbow Matching}. 
\end{claim}

\begin{claimproof}
Assume $C_\ell$ for $\ell \le b$, is a solution to $(G, S, \budget, a_1, b_1)$ that is reached with only $2(2\kappa - 2) \cdot (m - 1)+ \log t$ slides and disconnects $a_1$ from $b_1$, then any token that slides in $G$ slides at most once, given that everything except:
\begin{itemize}[itemsep=0pt]
    \item for $d \in \{1, \ldots, \log t\}$, the vertex $v^d$ and each vertex of $\mathcal{T}$ that is at level $d$, 
    \item $u^{(i,j)}$ for $i \in [2(\kappa - 1)]$ and $j \in [m - 1]$, to each vertex $v_{e_h}^{\ceil{i/2},r}(1)$ for each $r \in [t]$ if $i$ is odd, and to each vertex $v_{e_h}^{\ceil{i/2},r}(2)$ for each $r \in [t]$ if $i$ is even,
    \item the endpoints of each edge $e^{(i,j)} \in M_{i}$ for $i \in [2(\kappa - 1)]$ and $j \in [m-1]$, 
\end{itemize}
is connected by paths of length $(m^3+ \log t) > \budget$.
Thus, we know that the tokens on the vertices $v^d$ for $d \in \{1, \ldots, \log t\}$ will have to leave some paths that go from $a_1$ to $b_1$ at least through one pair of vertices $s_\mathfrak{r}$ and $t_\mathfrak{r}$ for some $\mathfrak{r} \in [t]$ and can use at most $\log t$ slides. 
We know that in $G \setminus C_\ell$, no path exists between $s_\mathfrak{r}$ and $t_\mathfrak{r}$.
Since no token can reach $s_\mathfrak{r}$ and $t_\mathfrak{r}$ in the allocated budget, the remaining slides can only disconnect $s_\mathfrak{r}$ from $t_\mathfrak{r}$. 
Note also that $u^{(i,j)} \in C_\ell$, for $i \in [2(\kappa - 1)]$ and $j \in [m - 1]$ as otherwise, a path from $a_1$ to $b_1$ that goes through $s_\mathfrak{r}$, $u^{(i,j)}$ and $t_\mathfrak{r}$ will remain tokens-free.
This implies that at most $m - 1$ tokens can be slid into any one level $\{E^1_\mathfrak{r}(\cdot), \ldots, E^{\kappa - 1}_\mathfrak{r}(\cdot)\}$. 
We show via an inductive argument that the set of edges in $H_\mathfrak{r}$ represented by the vertices in $\{E^1_\mathfrak{r}(\cdot), \ldots, E^{\kappa - 1}_\mathfrak{r}(\cdot)\}$ but not in $C_\ell$ must form a matching $\mathcal{M}_\mathfrak{r}$ in $H_\mathfrak{r}$ of size $\kappa_\mathfrak{r}=\kappa$, such that for $e_h \text{, } e_g \in \mathcal{M}_\mathfrak{r}$, $\phi_\mathfrak{r}(e_h) \neq \phi_\mathfrak{r}(e_g)$ and the claim follows.
Let $P(q)$ be the proposition that the set $\mathcal{E}_q$ of edges represented by vertices in $\{E^1_\mathfrak{r}(\cdot), \ldots, E^{q}_\mathfrak{r}(\cdot)\}$ but not in $C_\ell$ form a matching such that for $e_h \text{, } e_g \in \mathcal{E}_q$, $\phi_\mathfrak{r}(e_h) \neq \phi_\mathfrak{r}(e_g)$ and that vertices that remain free in $E^{q+1}_\mathfrak{r}(1)$ for $q < \kappa - 1$ represent the same edges as the vertices that remain free in $E^{q}_\mathfrak{r}(2)$.
We show that $P(q)$ holds by induction on the levels $q = \{1, \ldots, \kappa - 1\}$. 

We prove the base case by contradiction and assume that a vertex $v_{e_g}^{1,\mathfrak{r}}(2)$ that remains free in~$E^1_\mathfrak{r}(2)$ either represents an edge $e_g$ that is incident to an edge $e_h$ represented by a vertex $v_{e_h}^{1,\mathfrak{r}}(1)$ that remains free in $E^1_\mathfrak{r}(1)$ or it holds that $\phi_\mathfrak{r}(e_g) = \phi_\mathfrak{r}(e_h)$. 
This implies that there exists a path between $s_\mathfrak{r}$ and $t_\mathfrak{r}$ that goes from $s_\mathfrak{r}$ to $s^1_\mathfrak{r}$, to
$v_{e_h}^{1,\mathfrak{r}}(1)$, $v_{e_g}^{1,\mathfrak{r}}(2)$, $t^1_\mathfrak{r}$ and to $t_\mathfrak{r}$ and thus $C_\ell$ is not a solution to $(G, S, \budget, a_1, b_1)$.
As for the second part of the statement, assume that a vertex $v_{e_h}^{1,\mathfrak{r}}(2)$ that remains free in $E^1_\mathfrak{r}(2)$ does not represent the same edge as any of the vertices that remain free in $E^2_\mathfrak{r}(1)$, then there exists a path between $s_\mathfrak{r}$ and $t_\mathfrak{r}$ that goes through, $s^2_\mathfrak{r}$, then any of the latter vertices, followed by $v_{e_h}^{1,\mathfrak{r}}(2)$ and $t^1_\mathfrak{r}$ and thus $C_\ell$ is not a solution to $(G, S, \budget, a_1, b_1)$.
Note that the same arguments used in the base case apply for the inductive step.

In other words, given the second part of the statement, we may assume (for contradiction purposes) that a vertex $v_{e_g}^{i,\mathfrak{r}}(2)$ for $i \le q$ (that remains free in $E^i_\mathfrak{r}(2)$) either represents an edge $e_g$ that is incident to an edge $e_h$ represented by a vertex $v_{e_h}^{i',\mathfrak{r}}(1)$ for $i' \le i$ (that remains free in $E^{i'}_\mathfrak{r}(1)$) or it holds that $\phi_\mathfrak{r}(e_g) = \phi_\mathfrak{r}(e_h)$. 
By construction, this implies that there exists a path from $s_\mathfrak{r}$ and $t_\mathfrak{r}$ that goes from $s_\mathfrak{r}$ to $s^{i'}_\mathfrak{r}$, $v_{e_h}^{i',\mathfrak{r}}(1)$, $v_{e_g}^{i,\mathfrak{r}}(2)$, $t^i_\mathfrak{r}$, and to $t_\mathfrak{r}$ and thus $C_\ell$ is not a solution to $(G, S, \budget, a_1, b_1)$.
As for the second part of the statement, assume that a vertex~$v_{e_h}^{q,\mathfrak{r}}(2)$ (that remains free in $E^q_\mathfrak{r}(2)$) does not represent the same edge as any of the vertices that remain free in $E^{q+1}_\mathfrak{r}(1)$, then there exists a path between $s_\mathfrak{r}$ and $t_\mathfrak{r}$ that goes through, $s^{q+1}_\mathfrak{r}$, then any of the latter vertices, followed by~$v_{e_h}^{q,\mathfrak{r}}(2)$ and $t^q_\mathfrak{r}$ and thus $C_\ell$ is not a solution to $(G, S, \budget, a_1, b_1)$.

Thus, $P(\kappa - 1)$ holds and the set $\mathcal{E}_{\kappa - 1}$ of edges represented by vertices in $\{E^1_\mathfrak{r}(\cdot), \ldots, E^{\kappa-1}_\mathfrak{r}(\cdot)\}$ but not $C_\ell$ form a matching of size $\kappa$ such that for $e_h \text{, } e_g \in \mathcal{E}_{\kappa-1}$, $\phi_\mathfrak{r}(e_h) \neq \phi_\mathfrak{r}(e_g)$.    
\end{claimproof}
This concludes the proof of the theorem.
\end{proof}

\bibliographystyle{plain}
\bibliography{ref}
\end{document}